\documentclass[acmsmall,screen,nonacm]{acmart}

\AtBeginDocument{%
  }

\theoremstyle{definition}

\theoremstyle{definition}

\newtheorem*{theorem*}{Theorem}

\usepackage{mathpartir}
\usepackage{amsthm}
\usepackage{stmaryrd}
\usepackage{hyperref}
\usepackage{turnstile}
\usepackage{xspace}
\usepackage{listings}
\usepackage{enumitem}
\usepackage{formal-grammar}
\usepackage[dvipsnames]{xcolor}
\usepackage{soul}
\usepackage{anyfontsize}
\usepackage{framed}

\theoremstyle{definition}
\newtheorem{lemma}{Lemma}
\newtheorem{definition}{Definition}

\newtheorem{theorem}{Theorem}

\newcommand{\gco}{\textrm{Gradual C0}\xspace}
\newcommand{\co}{\textrm{C\textsubscript{0}}\xspace}

\newcommand{\ttt}[1]{\texttt{#1}}
\newcommand{\multiple}[1]{\overline{#1}}
\newcommand{\disableTttResize}[0]{\renewcommand{\small}[1]{##1}}

\newcommand{\dom}{\operatorname{dom}}
\newcommand{\notimplies}{\Longarrownot\Longrightarrow}

\newcommand{\pair}[2]{\langle #1, \, #2 \rangle}
\newcommand{\triple}[3]{\langle #1,\, #2, \, #3 \rangle}
\newcommand{\quadruple}[4]{\langle #1,\, #2,\, #3,\, #4 \rangle}
\newcommand{\quintuple}[5]{\langle #1,\, #2,\, #3,\, #4,\, #5 \rangle}
\newcommand{\sextuple}[6]{\langle #1,\, #2,\, #3,\, #4,\, #5,\, #6 \rangle}

\newcommand{\gprogram}{\mathrm{program}}
\newcommand{\gpredicate}{\mathcal{P}}
\newcommand{\gmethod}{\mathcal{M}}
\newcommand{\gstruct}{\mathcal{S}}
\newcommand{\gfunc}{\mathcal{F}}
\newcommand{\gform}{\tilde{\phi}}
\newcommand{\gstatement}{\mathit{s}}
\newcommand{\gcontract}{\mathrm{\Phi}}
\newcommand{\gfcontract}{\mathrm{\Psi}}
\newcommand{\gtype}{\mathit{T}}

\newcommand{\gexpression}{\mathit{e}}
\newcommand{\gvar}{\mathit{x}}

\newcommand{\gpreciseform}{\phi}

\newcommand{\kunfolding}{\ttt{unfolding}}
\newcommand{\kensures}{\ttt{ensures}}
\newcommand{\krequires}{\ttt{requires}}
\newcommand{\kalloc}{\ttt{alloc}}
\newcommand{\knull}{\ttt{null}}
\newcommand{\kacc}{\ttt{acc}}
\newcommand{\kskip}{\ttt{skip}}
\newcommand{\kif}{\ttt{if}}
\newcommand{\kthen}{\ttt{then}}
\newcommand{\kelse}{\ttt{else}}
\newcommand{\kassert}{\ttt{assert}}

\newcommand{\kwhile}{\ttt{while}}
\newcommand{\kinvariant}{\ttt{invariant}}
\newcommand{\kdo}{\ttt{do}}
\newcommand{\kfold}{\ttt{fold}}
\newcommand{\kunfold}{\ttt{unfold}}
\newcommand{\kint}{\ttt{int}}
\newcommand{\kchar}{\ttt{char}}
\newcommand{\kbool}{\ttt{bool}}
\newcommand{\ktrue}{\ttt{true}}
\newcommand{\kfalse}{\ttt{false}}
\newcommand{\kresult}{\ttt{result}}
\newcommand{\keq}{\mathbin{\ttt{==}}}

\newcommand{\kadd}{\mathbin{\ttt{+}}}
\newcommand{\ksub}{\mathbin{\ttt{-}}}
\newcommand{\kdiv}{\mathbin{\ttt{/}}}
\newcommand{\kmul}{\mathbin{\ttt{*}}}
\newcommand{\kand}{\mathbin{\ttt{\&\&}}}
\newcommand{\kor}{\mathbin{\ttt{||}}}
\newcommand{\kneg}{\mathop{\ttt{!}}}
\newcommand{\kassign}{\mathbin{\ttt{=}}}

\newcommand{\sblock}[1]{\{ ~#1~ \}}
\newcommand{\salloc}[1]{\kalloc(#1)}
\newcommand{\sseq}[2]{#1 \ttt{;} ~#2}
\newcommand{\sif}[3]{\kif ~#1~ \kthen ~#2~ \kelse ~#3}
\newcommand{\sassert}[1]{\kassert ~#1}
\newcommand{\swhile}[3]{\kwhile ~#1~ \kinvariant ~#2~ \kdo ~#3}
\newcommand{\sfold}[1]{\kfold ~#1}
\newcommand{\sunfold}[1]{\kunfold ~#1}
\newcommand{\simprecise}[1]{\ttt{?} * #1}
\newcommand{\smethdef}[5]{#1 ~ #2(#3) ~#4~ #5}
\newcommand{\sunfolding}[2]{\kunfolding ~ #1 ~ \ttt{in} ~ #2}

\newcommand{\Literal}{\textsc{Literal}}
\newcommand{\Value}{\textsc{Value}}
\newcommand{\Field}{\textsc{Field}}
\newcommand{\Formula}{\textsc{Formula}}
\newcommand{\GFormula}{\tilde{\textsc{F}}\textsc{ormula}}
\newcommand{\Location}{\textsc{Location}}
\newcommand{\Var}{\textsc{Var}}
\newcommand{\Expr}{\textsc{Expr}}
\newcommand{\Stmt}{\textsc{Stmt}}
\newcommand{\Method}{\textsc{Method}}
\newcommand{\Predicate}{\textsc{Predicate}}
\newcommand{\Function}{\textsc{Function}}
\newcommand{\Struct}{\textsc{Struct}}
\newcommand{\Type}{\textsc{Type}}

\newcommand{\Perm}{\textsc{Perm}}

\newcommand{\SExpr}{\textsc{SExpr}}
\newcommand{\SValue}{\textsc{SValue}}
\newcommand{\SField}{\textsc{SField}}
\newcommand{\SPredicate}{\textsc{SPredicate}}
\newcommand{\SCheck}{\textsc{SCheck}}
\newcommand{\SPerm}{\textsc{SPerm}}
\newcommand{\SState}{\textsc{SState}}

\newcommand{\veval}[4]{#1 \vdash #2 \Downarrow #3 \dashv #4}
\newcommand{\vcons}[3]{#1 \vdash #2 \rhd #3}

\newcommand{\eval}[4]{\pair{#1}{#2} \vdash #3 \Downarrow #4}

\newcommand{\assertion}[4]{\triple{#1}{#2}{#3} \vDash #4}

\newcommand{\foot}[4]{\lfloor #4 \rfloor_{\triple{#1}{#2}{#3}}}

\newcommand{\efoot}[3]{\llfloor #3 \rrfloor_{\pair{#1}{#2}}}

\newcommand{\vfoot}[3]{#1 \llparenthesis #3 \rrparenthesis_{#2}}

\newcommand{\frm}[4]{\triple{#1}{#2}{#3} \vdash_{\mathrm{frm}} #4}
\newcommand{\ifrm}[4]{\triple{#1}{#2}{#3} \vdash_{\mathrm{frmI}} #4}
\newcommand{\efrm}[4]{\triple{#1}{#2}{#3} \vdash_{\mathrm{frmE}} #4}

\newcommand{\dexec}[5]{\pair{#1}{#2},\, #3 \to \pair{#4}{#5}}

\newcommand{\dtrans}[4]{#1 \vdash #3,\, #2 \to #4}

\newcommand{\simheap}[4]{\pair{#3}{#4} \sdtstile{#1}{} #2}

\newcommand{\simenv}[3]{#3 \sdtstile{#1}{} #2}

\newcommand{\simstate}[5]{\triple{#3}{#4}{#5} \sdtstile{#1}{} #2}

\newcommand{\rtassert}[4]{\pair{#2}{#3} \vdash_{#1} #4}

\newcommand{\seval}[5]{#1 \vdash #2 \Downarrow #3 \dashv #4,\, #5}

\newcommand{\sproduce}[3]{#1 \vdash #2 \lhd #3}

\newcommand{\sconsume}[6]{#1,\, #2 \vdash #3 \rhd #4,\, \allowbreak #5,\, \allowbreak #6}

\newcommand{\scons}[4]{#1 \vdash #2 \rhd #3, ~#4}

\newcommand{\sexec}[4]{#1 \vdash #2 \to #3 \dashv #4}

\newcommand{\sguard}[4]{#1 \rightharpoonup #2 \dashv #3,\, #4}

\newcommand{\strans}[3]{#1 \vdash #2 \to #3}

\newcommand{\vstate}{\Sigma}

\newcommand{\existential}[2]{\exists\, #1 : #2}
\newcommand{\nexistential}[2]{\nexists\, #1 : #2}
\newcommand{\universal}[2]{\forall\, #1 : #2}

\newcommand{\heap}{H}
\newcommand{\env}{\rho}
\newcommand{\perms}{\alpha}
\newcommand{\xperms}{\hat{\perms}}
\newcommand{\stack}{\mathcal{S}}
\newcommand{\tlist}{\multiple{t}}
\newcommand{\prog}{\Pi}
\newcommand{\dstate}{\Gamma}

\newcommand{\initsym}{\mathsf{init}}
\newcommand{\finalsym}{\mathsf{final}}
\newcommand{\nilsym}{\mathsf{nil}}

\newcommand{\ffresh}{\operatorname{fresh}}
\newcommand{\fpre}{\operatorname{pre}}
\newcommand{\fpost}{\operatorname{post}}
\newcommand{\fbody}{\operatorname{body}}
\newcommand{\fparams}{\operatorname{params}}
\newcommand{\fpred}{\operatorname{predicate}}
\newcommand{\fpredparams}{\operatorname{predicate\_params}}
\newcommand{\ffuncpre}{\operatorname{func\_pre}}
\newcommand{\ffuncpost}{\operatorname{func\_post}}
\newcommand{\ffuncparams}{\operatorname{func\_params}}
\newcommand{\ffuncbody}{\operatorname{func\_body}}
\newcommand{\fstruct}{\operatorname{struct}}
\newcommand{\fdefault}{\operatorname{default}}
\newcommand{\frem}{\operatorname{rem}}
\newcommand{\fremf}{\operatorname{rem_f}}
\newcommand{\fremfp}{\operatorname{rem_{fp}}}
\newcommand{\falias}{\operatorname{alias}}
\newcommand{\fmodified}{\operatorname{modified}}
\newcommand{\fsep}{\operatorname{sep}}
\newcommand{\fsat}{\operatorname{sat}}
\newcommand{\powerset}[1]{\mathcal{P}(#1)}
\newcommand{\set}[1]{\{ #1 \}}
\newcommand{\iffdef}{\stackrel{\text{def}}{\iff}}
\newcommand{\pfunc}{\rightharpoonup}

\newcommand{\sstate}{\sigma}

\newcommand{\sheap}{\mathsf{H}}
\newcommand{\oheap}{\mathcal{H}}
\newcommand{\scheck}{\mathcal{R}}
\newcommand{\senv}{\gamma}
\newcommand{\pc}{g}
\newcommand{\imp}{\iota} 
\newcommand{\sperms}{\Theta}
\newcommand{\sperm}{\theta}
\newcommand{\svis}{\mathcal{V}}

\newcommand{\semanticsargs}[1][]{#1\end{mathpar}}
\newcommand{\semantics}[3][]{\begin{mathpar}\inferrule[#1]{#2}{#3}\semanticsargs}

\newcounter{defpartnum}[definition]
\newcounter{defsubpartnum}[defpartnum]
\newenvironment{defparts}
{
    \setcounter{defpartnum}{0}
    \setcounter{defsubpartnum}{0}
    \newcommand{\defpart}{\par\indent\refstepcounter{defpartnum}\textit{Part \thedefpartnum.}~}
    
}
{
    \par
}
\renewcommand*\thedefpartnum{\arabic{definition}.\arabic{defpartnum}}

\newcounter{casenum}
\newcounter{subcasenum}[casenum]
\newenvironment{enumcases}
{
    \setcounter{casenum}{0}
    \setcounter{subcasenum}{0}
    \newcommand{\case}{\par\indent\refstepcounter{casenum}\textbf{Case \thecasenum.}~}
    \let\SavedTheHcasenum=\theHcasenum
    \def\theHcasenum{\thelemma.\SavedTheHcasenum}
    \newcommand{\subcase}{\par\indent\refstepcounter{subcasenum}\textit{Case \thecasenum(\thesubcasenum).}~}
}
{
    \par
}
\renewcommand*\thecasenum{\arabic{casenum}}
\renewcommand*\thesubcasenum{\alph{subcasenum}}

\hypersetup{
    colorlinks=true,
    linkcolor = blue
}

\newcommand{\refrule}[1]{\hyperlink{#1}{\TirName {#1}}}

\definecolor{light-gray}{gray}{0.87}
\newcommand{\lightgray}[1]{\colorbox{light-gray}{#1}}

\definecolor{light-purple}{RGB}{229,204,255}

\definecolor{light-yellow}{RGB}{255,228,181}

\definecolor{light-blue}{RGB}{189,215,238}
\newcommand{\lightblue}[1]{\colorbox{light-blue}{#1}}

\definecolor{light-green}{RGB}{152,251,152}

\definecolor{light-red}{RGB}{255,204,204}
\newcommand{\lightred}[1]{\colorbox{light-red}{#1}}

\definecolor{rred}{RGB}{255,153,153}

\definecolor{light-pink}{RGB}{255,204,255}

\definecolor{light-orange}{RGB}{255,204,153}

\definecolor{neon-blue}{RGB}{153,255,255}

\definecolor{neon-yellow}{RGB}{255,255,153}

\setlist[itemize]{leftmargin=14pt}

\lstdefinestyle{c0style}{
    basicstyle=\scriptsize\ttfamily,
	mathescape=true,
	numbers=none,
    numbersep=5pt,
    classoffset=0,
    morekeywords={requires, in, ensures, assert, predicate, invariant, unfold, fold, loop_invariant, acc, unfolding, result},
    keywordstyle=\bfseries\color{Bittersweet},
    classoffset=1,
	morekeywords={for, if, then, else, return, let, void, int, class, while, NULL, alloc, struct, true},	
    keywordstyle=\bfseries\color{MidnightBlue},
    classoffset=0,
    escapeinside={(*@}{@*)}
}
\lstdefinestyle{gobra}{
	language=Go,
    basicstyle=\scriptsize\ttfamily,
    classoffset=0,
    keywordstyle=\bfseries\color{MidnightBlue},
	mathescape=true,
	numbers=none,
    numbersep=5pt,
    classoffset=1,
    morekeywords={requires, ensures, pred, invariant, unfold, fold, loop_invariant, acc, unfolding, nat, in},
    keywordstyle=\bfseries\color{Bittersweet},
    escapeinside={(*@}{@*)},
    commentstyle=\color{gray}
}

\lstdefinelanguage{Viper}{
    basicstyle=\scriptsize\ttfamily,
	mathescape=true,
	numbers=none,
    numbersep=5pt,
    classoffset=0,
    morekeywords={method, field, predicate, function, requires, ensures, unfolding, in, acc, result},
    keywordstyle=\bfseries\color{MidnightBlue},
    classoffset=1,
    morekeywords={Int, Bool, Ref},
    keywordstyle=\bfseries\color{OliveGreen}
}

\newcommand{\sfoot}[2]{\llbracket #1 \rrbracket_{#2}}
\newcommand{\fax}[3]{\operatorname{axiomatize}_{#1}[#2](#3)}
\newcommand{\fts}{\ensuremath{f(\multiple{t}, s)}}

\definecolor{newstuff}{rgb}{0.9, 0.9, 0.9}
\addtocontents{toc}{\protect\setcounter{tocdepth}{0}}

\newcommand{\pe}{\ensuremath{p(\multiple{e})}}
\newcommand{\fe}{\ensuremath{f(\multiple{e})}}

\begin{document}

\title{Gradually Verifying Unfolding Expressions \& Pure Functions}

\author{Hazel Torek}
\affiliation{%
  \institution{Clemson University}
  \city{Clemson}
  \state{South Carolina}
  \country{USA}
}
\email{ctorek@clemson.edu}

\author{Long Tien Nguyen}
\affiliation{%
  \institution{Carnegie Mellon University}
  \city{Pittsburgh}
  \state{Pennsylvania}
  \country{USA}
}
\email{longtien@cs.cmu.edu}

\author{Priyam Gupta}
\affiliation{%
  \institution{ETH Z\"{u}rich}
  \city{Z\"{u}rich}
  \country{Switzerland}
}
\email{priyam23g@gmail.com}

\author{Jenna DiVincenzo}
\affiliation{%
  \institution{Purdue University}
  \city{West Lafayette}
  \state{Indiana}
  \country{USA}
}
\email{jennad@purdue.edu}

\author{Jonathan Aldrich}
\affiliation{%
  \institution{Carnegie Mellon University}
  \city{Pittsburgh}
  \state{Pennsylvania}
  \country{USA}
}
\email{jonathan.aldrich@cs.cmu.edu}

\renewcommand{\shortauthors}{Torek et al.}

\begin{abstract}
    Unfolding expressions, which temporarily unfold a predicate to leverage its owned fields when evaluating a heap-dependent expression, and pure functions, which are heap-dependent functions that can be used in specifications, are used in deductive program verifiers based on implicit dynamic frames, such as Gradual C0, Gobra, Nagini, and SnaKt, to increase the modularity of specifications involving ownership. In this paper, we present the formal semantics for unfolding expressions and pure functions for a static verifier using symbolic execution, extend it for a gradual verifier, and provide a proof of soundness. To support Gradual C0, our proof is in the setting of gradual verification, a deductive program verification system that combines static and dynamic verification to allow partial specifications.  However, because the gradual verifier is a conservative extension of a static verifier, our results also apply to static verifiers that use symbolic execution, such as the Silicon symbolic execution backend for the Viper verification infrastructure used by Gobra, Nagini, and SnaKt.
\end{abstract}



\received{9 July 2026}

\maketitle

\section{Introduction}

Implicit dynamic frames (IDF)~\cite{smans2009implicit} is a verification technique related to separation logic that can provide modularity benefits by separating ownership specifications from assertions about heap contents.
In the context of IDF, \emph{unfolding expressions} are a type of expression that temporarily unfolds a predicate instance to gain ownership of additional fields during evaluation of a heap-dependent expression. These are of the form $\sunfolding{\pe}{e_0}$ where $p$ is a predicate, $\multiple{e}$ is the list of arguments passed to the predicate, and $e_0$ is the expression evaluated using the permissions encapsulated in the predicate. Unfolding expressions were originally introduced as part of the Chalice verifier \cite{leino2009verification} to enable more modular expressions of ownership. Similarly, \textit{pure functions} allow developers to define heap-dependent, side effect-free, recursive functions that can be used in specifications. These also increase the modularity of specifications by allowing developers to separate and reuse recurring patterns. Both of these constructs appear in Viper \cite{muller2016viper} and many Viper-based verifiers such as Gobra for Go \cite{wolf2021gobra}, Nagini for Python \cite{eilers2018nagini}, SnaKt for Kotlin \cite{protopapa}, \gco for C0 \cite{divincenzo2025gradual, gupta2025increasing, mutlu2025expanding}. We observe that unfolding expressions and pure functions are used repeatedly throughout the verified packages of the Go standard library \cite{jenny, laux}, the implementation of VerifiedSCION, a next-generation Internet architecture \cite{pereira2024protocols}, and the verified Go implementations of the Needham-Schroeder-Lowe, Diffie-Hellman, and WireGuard protocols \cite{ArquintSchwerhoffMehtaMueller23, ArquintWLSSWBM23}, demonstrating their use in verified codebases from the real-world.


Together, unfolding expressions and pure functions enable design patterns in verified code that greatly improve modularity; for instance, Gobra codebases often define predicates that encapsulate entire data structures, using unfolding expressions to access fields of these data structures and pure functions to represent properties of these data structures. Figure~\ref{fig:pattern} demonstrates an example of this pattern from the implementation of multi-precision signed integers in VerifiedSCION,\footnote{Our example has been slightly adapted for readability. For multi-precision signed integers, the original implementation and usage can be found in \href{https://github.com/viperproject/VerifiedSCION/blob/4524a9d68eea963cfd9a9c2e094f336d3430e78c/verification/dependencies/math/big/int.gobra}{int.gobra} in VerifiedSCION.} and the implementation of lists in the verified Go standard library provides a more complex example of this pattern.\footnote{For lists, the predicate is defined in \href{https://github.com/viperproject/verified_go_stdlib/blob/cbd7513c892592b33d9bcc8af2f678671057f5ab/src/container/list/list_spec.gobra}{list\_spec.gobra} and used in \href{https://github.com/viperproject/verified_go_stdlib/blob/cbd7513c892592b33d9bcc8af2f678671057f5ab/src/container/list/list.go}{list.go} in the verified Go standard library.} Without these constructs, developers would be required to specify ownership of all fields used in the specification and method separately and define separate predicates to express the required recursive properties, which would significantly increase the specification burden.

\begin{figure}
    \centering

    \begin{minipage}{0.42\linewidth}
        \begin{lstlisting}[style=gobra]
type Int struct {
    neg bool // sign
    abs nat  // absolute value
}

pred (i *Int) Mem() {
    acc(i) && i.abs.Mem()
}
        \end{lstlisting}
    \end{minipage}
    \begin{minipage}{0.49\linewidth}
        \begin{lstlisting}[style=gobra]
// In Gobra, requires and ensures clauses precede
// the function/method they are associated with
requires acc(x.Mem())
ensures acc(x.Mem()) &&
        unfolding acc(x.Mem()) in
            len(x.abs) <= 1
            ==> toInt(res) == x.abs.toInt()
func (x *Int) Uint64() (res uint64) {
    ... // returns uint64 representation of x
}
        \end{lstlisting}
    \end{minipage}
    
    \caption{Predicate definition and usage from the implementation of multi-precision signed integers in VerifiedSCION demonstrating a common design pattern enabled by unfolding expressions and pure functions.}
    \label{fig:pattern}
\end{figure}

Accordingly, these are important constructs in increasing the usability and expressiveness of implicit dynamic frames-based verifiers; however, their semantics have not yet been formalized despite existing formalizations \cite{dardinier2025formal, parthasarathy2004towards} and they have not yet been included in a proof of soundness of such a verifier. To address this gap and ensure the correctness of verifiers using this feature, we provide a formalization and proof of soundness for a generalization of static verification called \emph{gradual verification}. Gradual verification is a verification technique introduced by Bader et al. \cite{bader2018gradual} to address the high upfront cost of static verification which, inspired by gradual typing \cite{siek2007gradual}, enables developers to write explicitly partial specifications using the \emph{imprecise formula}, written as $\ttt{?}$, which acts as a placeholder wildcard during static verification. For instance, $\simprecise{x > 9}$ can refer to any non-contradictory strengthening of the formula $x > 9$, such $x \geq 15$. These \emph{imprecise specifications} are backed by a combination of static and dynamic checks \cite{wise2020gradual}. Notably, gradual verification is designed to maintain an important property called the \emph{gradual guarantee}, an idea adapted from gradual typing: reducing the precision of a correct specification will never cause the verifier to report a spurious error. Gradual verification is also a \emph{conservative extension} of static verification, meaning that a static verifier and the gradual verifier extended from it will coincide on all fully-specified programs \cite{wise2020gradual, zimmerman2024sound}.

The initial formulation by Bader et al. lacked the ability to verify programs with recursive data structures allocated on the heap, which limited its practicality for verifying real-world programs. Wise et al. \cite{wise2020gradual} and DiVincenzo et al. \cite{divincenzo2025gradual} later extended gradual verification with these capabilities using implicit dynamic frames \cite{smans2009implicit}, a variation of separation logic \cite{reynolds2002separation}, and recursive abstract predicates \cite{parkinson2005separation}. These extensions were accompanied by an implementation for the C0 subset of the C programming language \cite{arnold2010c0, divincenzo2025gradual} by extending the Silicon symbolic execution backend for Viper \cite{schwerhoff2016advancing, muller2016viper} into \emph{Gradual Viper}.

Since Gradual Viper uses symbolic execution rather than weakest liberal preconditions, Zimmerman et al. \cite{zimmerman2024sound} provided a full formal system corresponding to Gradual Viper, including static and dynamic semantics, definitions of the correspondence between static and dynamic verifier states, and a proof of soundness. Because gradual verification is a conservative extension of its static counterpart, this also proves the soundness of a subset of the underlying Viper features. However, this formalization lacks many other underlying Viper constructs that could improve verifier usability; for instance, Gupta \cite{gupta2025increasing} introduced unfolding expressions and Mutlu \cite{mutlu2025expanding} introduced pure functions to Gradual Viper's implementation, but the formalization and proof of soundness does not yet include either of these features.

Addressing this gap in the static and gradual verification ecosystems, this paper makes the following contributions:

\begin{itemize}
    \item Design and formalization of the symbolic evaluation semantics for unfolding expressions and pure functions in IDF-based verifiers such as Viper.
    \item Design and formalization of the static semantics, dynamic semantics, and their correspondence for unfolding expressions and pure functions in Gradual Viper.
    \item An updated proof of soundness for Gradual Viper, based on the formal system provided by Zimmerman et al \cite{zimmerman2024sound}, which also applies to a subset of Viper's features.
\end{itemize}

With these contributions, we increase our confidence in the correctness of Viper's implementation and Viper-based verification tools that use these constructs, both directly and in translation. Additionally, by expanding upon the gradual verification ecosystem rather than focusing solely on static verification, we maintain the same correctness guarantees while inheriting a rich body of work on both implementation and formalization for gradual verification.

\section{Viper Semantics}

\subsection{Ownership \& Predicates in Viper}

In static verification of imperative programs, ownership reasoning is used to maintain invariants regarding heap contents and aliasing. These reasoning capabilities are established with logics like separation logic \cite{reynolds2002separation}, which requires reasoning about heap ownership and contents simultaneously, or implicit dynamic frames \cite{smans2009implicit}, which enables specifying ownership and contents separately. The Viper verification infrastructure, built on implicit dynamic frames, enables ownership reasoning using several different verification primitives, the most fundamental of which is the \emph{accessibility predicate}. An accessibility predicate, written $\kacc(x.f)$, specifies read and write access to field $f$ of structure $x$.

In verifiers based on implicit dynamic frames, a specification is \emph{framed} if it only accesses heap locations it has permission to access. To account for this requirement regardless of the owned heap locations in a specific verification state, specifications must be \emph{self-framing}, meaning they must specify ownership of all the fields used in the formula to be well-defined. For instance, the specification $x.f \geq 9$ is not self-framed, but the specification $\kacc(x.f) * x.f \geq 9$ is.

In Viper, \emph{predicates} are named, parameterized assertions that are used to encapsulate ownership of and express properties of potentially-unbounded data structures like lists and trees. The logic for specifying these structures is provided by recursive abstract predicates \cite{parkinson2005separation}. However, there are two differing interpretations of these predicates that lead to different semantics for assertions and different ways of specifying ownership \cite{summers2013formal}. Under the \emph{equirecursive} approach, an instance of a predicate is considered equivalent to the information contained in its body. Although this is viewed as more intuitive, this poses a problem in implementation because recursive predicates become unbounded. In contrast, the \emph{isorecursive} approach (used in Viper \cite{schwerhoff2016advancing}) views predicate instances and the information contained therein as distinct, requiring developers to manually exchange between these by \emph{unfolding} and \emph{folding} predicate instances (some verifiers such as VeriFast \cite{jacobs2011verifast} use the terms \textit{opening} and \textit{closing} a predicate, respectively). Accordingly, unfolding expressions temporarily unfold a predicate, using additional permissions from its body to ensure the evaluation of the expression is framed with sufficient permissions. 

\begin{figure}
    \centering
    \begin{minipage}{0.52\linewidth}
    \begin{lstlisting}
struct Node { int data; struct Node *next; };

/*@
predicate sortedList(struct Node *x) =
 (x == NULL) ? true : (
   acc(x->data) && acc(x->next) &&
   sortedList(x->next) && (x->next == NULL ||
     unfolding sortedList(x->next) in
       x->data <= x->next->data)
  );
@*/
    \end{lstlisting}
    \end{minipage}
    \caption{Checking that a list is sorted \cite{gupta2025increasing}.}
    \label{fig:sorted}
\end{figure}

Without using unfolding expressions in combination with pure functions, specifications are more complex and less intuitive; for instance, in Figure~\ref{fig:sorted}, checking whether a linked list is sorted would require adding an additional parameter to the \verb|sortedList| predicate keeping track of the previous value. In some cases, writing the specification without unfolding expressions might even require additional ghost code/auxiliary data to facilitate verification \cite{gupta2025increasing}.

\subsection{Representation}\label{sec:representation}

The static semantics are based on the Silicon backend for Viper \cite{schwerhoff2016advancing}. During symbolic execution, the \emph{symbolic state} $\sstate$ tracks information about the execution of the program, and is formulated as a tuple with the following components:

\begin{itemize}
    \item The \emph{precise heap} $\sheap$ represents distinct heap chunks where ownership is statically verified.
    \item The \emph{symbolic store} $\senv$ maps variables to symbolic values.
    \item The \emph{path condition} $\pc$ identifies the path through branches in the program.
    \item The \emph{visited set} $\svis$ tracks predicates that have already been unfolded and pure functions that have already been evaluated during recursive evaluation of these constructs; the use of this is described in Section~\ref{sec:static-unfolding}.
\end{itemize}

The program is statically verified starting from an initial symbolic state $\sstate$, which is updated throughout the process by performing four fundamental operations:

\begin{itemize}
    \item Symbolically evaluating a symbolic expression $e$ into a symbolic value $t$, resulting in new state $\sstate'$. This is denoted by $\veval{\sstate}{e}{t}{\sstate'}$.
    \item Producing the information from formula $\gform$ into the symbolic heap, resulting in new state $\sstate'$; denoted $\sproduce{\sstate}{\gform}{\sstate'}$. Intuitively, this \emph{adds information} to the state.
    \item Consuming a formula $\gform$ from the symbolic state, resulting in new state $\sstate'$; denoted $\vcons{\sstate}{\gform}{\sstate'}$. Intuitively, this \emph{checks for and removes} information in the state.
    \item Symbolically executing a statement. Notation and definitions of this judgement are provided in the appendix since it is not directly used or extended here.
\end{itemize}

\subsection{Static Semantics for Unfolding Expressions in Viper}\label{sec:static-unfolding}

Because unfolding expressions are a type of expression, we extend the symbolic evaluation judgement, additionally relying on the consume and produce judgements. The example in figure \ref{fig:static-unfolding-example} demonstrates our formalization of the purely static approach in the Silicon backend for Viper.

\begin{figure}[H]
\begin{minipage}{\linewidth}
\begin{lstlisting}[language=Viper,xleftmargin=.2\textwidth,xrightmargin=.2\textwidth]
field val: Int

predicate pos(x: Ref) { acc(x.val) * x.val > 0 }

method caller(x: Ref, ...)
  requires pos(x)
  ensures ...
{
    ... unfolding pos(x) in (x.val < 256) ...
}
\end{lstlisting}
\end{minipage}
    \caption{Example of unfolding expression.}
    \label{fig:static-unfolding-example}
\end{figure}

\newcommand{\kmem}{\ttt{pos}}
\newcommand{\kval}{\ttt{val}}

We start with an initial symbolic state with $\ktrue$ as the path condition and $\emptyset$ as the visited set. Assume we have an object $x$ and an initial symbolic heap containing an instance of the $\ttt{pos}$ predicate applied to $x$. Note that the symbolic heap doesn't directly contain a field chunk pertaining to $x.\kval$ because it is encapsulated by the predicate $\kmem$. We use grey highlights to show the current symbolic state at a given step and blue highlights to show the operations being performed.

{\disableTttResize
\begin{lstlisting}[aboveskip=.25em, belowskip=0.5em,numbers=none,style=c0style]
(*@
\lightgray{$\sstate_1 = \quadruple{\sheap_1}{\senv}{\ktrue}{\emptyset} \quad \sheap_1 = \set{\pair{\kmem}{t_x}} \quad \senv = [\ttt{x} \mapsto t_x]$}
@*)
\end{lstlisting}
}

To begin evaluation of the unfolding expression $\sunfolding{\kmem(x)}{x.\kval < 256}$, we evaluate the argument $x$ passed to the predicate. Because $x$ is mapped to $t_x$ in the symbolic environment, evaluation simply looks up the underlying values and the symbolic state remains unchanged.

{\disableTttResize
\begin{lstlisting}[aboveskip=.25em, belowskip=0.5em,numbers=none,style=c0style]
(*@
\lightblue{$\veval{\sstate_1}{x}{t_x}{\sstate_1}$}
@*)
\end{lstlisting}
}

Then, we must temporarily exchange the predicate instance $\kmem(x)$ for its body to temporarily frame the inner expression $x.\kval < 256$ with additional permissions. To do so, we first consume the predicate instance to ensure that it holds in the current state; this removes the predicate chunk $\pair{\kmem}{t_x}$ from the symbolic heap.

{\disableTttResize
\begin{lstlisting}[aboveskip=.25em, belowskip=0.5em,numbers=none,style=c0style]
(*@
\lightblue{$\vcons{\sstate_1}{\kmem(x)}{\sstate_2}$}
@*)
(*@
\lightgray{$\sstate_2 = \quadruple{\sheap_2}{\senv}{\ktrue}{\emptyset} \quad \sheap_2 = \emptyset \quad \senv = [\ttt{x} \mapsto t_x]$}
@*)
\end{lstlisting}
}

Now, we produce the body of the predicate to obtain the permissions encapsulated within, yielding a symbolic heap containing the field chunk $\triple{\kval}{t_x}{a}$ and the path condition $a > 0$ where $a$ is a fresh symbolic value. This also requires updating the visited set $\svis$ to prevent unbounded recursion; more discussion of the handling of recursive unfolding expressions is provided at the end of this section.

{\disableTttResize
\begin{lstlisting}[aboveskip=.25em, belowskip=0.5em,numbers=none,style=c0style]
(*@
\lightblue{$\sproduce{\sstate_2}{\kacc(x.\kval)}{\sstate_3}$}
@*)
(*@
\lightgray{$\sstate_3 = \quadruple{\sheap_3}{\senv}{a > 0}{\set{\kmem}} \quad \sheap_3 = \set{\triple{\kval}{t_x}{a}} \quad \senv = [\ttt{x} \mapsto t_x]$}
@*)
\end{lstlisting}
}

Now, having exchanged the predicate instance for its body, we have a suitable environment for evaluating the inner expression. Since this evaluation recurs over the operation and requires only evaluating a variable and a literal, the state is unchanged except for having updated the visited set. 

{\disableTttResize
\begin{lstlisting}[aboveskip=.25em, belowskip=0.5em,numbers=none,style=c0style]
(*@
\lightblue{$\veval{\sstate_3[\svis = \emptyset]}{x.\kval < 256}{a < 256}{\sstate_4}$}
@*)
(*@
\lightgray{$\sstate_4 = \quadruple{\sheap_4}{\senv}{a > 0}{\emptyset} \quad \sheap_4 = \set{\triple{\kval}{t_x}{a}} \quad \senv = [\ttt{x} \mapsto t_x]$}
@*)
\end{lstlisting}
}

Now, because the predicate instance must be re-folded after the expression is evaluated, we simply reset the symbolic heap to its initial state to receive our final state.

{\disableTttResize
\begin{lstlisting}[aboveskip=.25em, belowskip=0.5em,numbers=none,style=c0style]
(*@
\lightgray{$\sstate_5 = \quadruple{\sheap_5}{\senv}{a > 0}{\emptyset} \quad \sheap_5 = \set{\pair{\kmem}{t_x}} \quad \senv = [\ttt{x} \mapsto t_x]$}
@*)
\end{lstlisting}
}

In the end, this procedure returns the expression $a < 256$ from the penultimate step. This entire process is described formally by the \textsc{SEvalUnfoldingExplicit} rule in Figure~\ref{fig:viper-unfolding}. Note that, due to the isorecursive semantics of symbolic evaluation, $\fpred(p)$ refers to the body of the predicate $p$ obtained by unfolding once, not by the full unrolling.

The \textsc{SEvalUnfoldingImplicit} rule in Figure~\ref{fig:viper-unfolding} is specific to our handling of recursive unfolding expressions. In cases where an unfolding expression is applied with a predicate that unfolds another instance of the same predicate, following the previously described process in its entirely can result in unbounded recursion. This is exemplified by Figure~\ref{fig:sorted}, where evaluating an unfolding expression with the $\ttt{sorted}$ predicate will then unfold another instance of the same predicate, and the recursion will not terminate because evaluation will branch over both the base and recursive cases to verify all possible paths. To address this, our approach is adapted from that of Schwerhoff for Viper \cite{schwerhoff2016advancing}; using the visited set $\svis$, we track which predicates have already been unfolded. Then, if a predicate being unfolded has already been visited, the recursion is terminated by replacing the result of the entire expression with an uninterpreted symbolic value. Although we don't unfold the predicate due to the possibility of unbounded recursion, we do still check that a predicate instance is present in the heap to ensure that this operation will succeed at run-time. This approach is theoretically incomplete, but it is still sound and this incompleteness is rarely observed in practice since it requires manually unfolding predicates to the depth where the recursion was cut off. Although we have formalized this behavior using a set, it could easily be exchanged for a multiset to enable a specified cutoff depth rather than only allowing a depth of one.

\begin{figure}
  {\footnotesize\disableTttResize
    \begin{mathpar}
        \inferrule[SEvalUnfoldingExplicit]{
            p \notin \svis(\sstate_1) \\
            \multiple{\veval{\sstate_1}{e}{t}{\sstate_2}} \\
            \vcons{\sstate_2}{p(\multiple{e})}{\sstate_3} \\
            \multiple{x} = \fpredparams(p) \\
            \sproduce{\sstate_3[\senv = \senv(\sstate_3)[\multiple{x \mapsto t}], \svis = \svis(\sstate_3) \cup \{ p \}]}{\fpred(p)}{\sstate_4} \\
            \veval{\sstate_4[\senv = \senv(\sstate_3)]}{e_0}{t_0}{\sstate_5} \\
            \sstate_6 = \sstate_5[\svis = \svis(\sstate_1), \sheap = \sheap(\sstate_2)]
        }{
            \veval{\sstate_1}{\sunfolding{p(\multiple{e})}{e_0}}{t_0}{\sstate_6}
        }

        \inferrule[SEvalUnfoldingImplicit]{
            p \in \svis(\sstate_1) \\
            \multiple{\veval{\sstate_1}{e}{t}{\sstate_2}} \\
            \pair{p}{\multiple{t}} \in \sheap(\sstate_2) \\
            t_0 = \ffresh
        }{
            \veval{\sstate}{\sunfolding{p(\multiple{e})}{e_0}}{t_0}{\sstate_2}
        }
    \end{mathpar}
   }

  \caption{Symbolic evaluation rules for unfolding expressions in Viper.}
  \label{fig:viper-unfolding}
\end{figure} 

\subsection{Static Semantics for Pure Functions in Viper}\label{sec:static-pure}

Our approach for the evaluation of pure functions is based on that of Smans et al. in the Smallfoot verifier \cite{smans2010heap} for when the bodies of pure functions are visible. However, in Viper pure functions can have both preconditions and post-conditions, whereas in Smallfoot they can only have preconditions. This requires us to expand our approach to axiomatize not only the result of the function evaluation, but also the post-condition. This also means that we must provide a new definition for well-formedness rather than using the one for methods, since methods may have a formula as their post-condition rather than an expression. Accordingly, we impose the additional constraint that evaluation should succeed in the symbolic state obtained from producing the precondition; this corresponds to the body being framed by the precondition, meaning the precondition contains all permissions necessary for evaluating the body. The use of precondition will also impact our handling of recursive pure functions, described later in this section.

We demonstrate our formalization of the purely static approach of Smans et al. using the example in figure \ref{fig:static-pure-func-example}.

\begin{figure}[H]
\begin{minipage}{\linewidth}
\begin{lstlisting}[language=Viper,xleftmargin=.2\textwidth,xrightmargin=.2\textwidth]
field val: Int

function add(x: Ref, y: Ref): Int
  requires acc(x.val) * acc(y.val)
  ensures (y.val < 0) || (result >= x.val)
  { x.val + y.val }

method caller(...)
  requires ...
  ensures ...
{
    ... add(x, y) ...
}
\end{lstlisting}
\end{minipage}
    \caption{Example of pure function.}
    \label{fig:static-pure-func-example}
\end{figure}

To demonstrate the use of including the post-condition in our approach, we use the post-condition $(y.\kval \geq 0) \implies (\kresult \geq x.\kval)$, so that it does not directly articulate the same information provided by the implementation. In the program, we write this as $(y.\kval < 0) \lor (\kresult \geq x.\kval)$, using the fact that $p \implies q \equiv \lnot p \lor q$ since implications are not formalized in prior work.

Then, assume we have objects $x$ and $y$ and an initial symbolic heap containing the field chunks $\sheap(\sstate_1) = \set{\triple{\kval}{t_x}{a}, \triple{\kval}{t_y}{b}}$, meaning field $\kval$ of object $t_x$ (corresponding to the symbolic variable $x$) has value $a$ and field $\kval$ of object $t_y$ (corresponding to $y$) has value $b$. Note that, since these objects are contained in the symbolic heap, they must not be aliases. 

Now, we demonstrate the process for symbolically evaluating these in our formalization. We start with a symbolic state containing only the necessary heap chunks and variable mappings for our assumptions that is otherwise empty. 

{\disableTttResize
\begin{lstlisting}[aboveskip=.25em, belowskip=0.5em,numbers=none,style=c0style]
(*@
\lightgray{$\sstate_1 = \quadruple{\sheap_1}{\senv}{\ktrue}{\emptyset} \quad \sheap_1 = \set{\triple{\kval}{t_x}{a}, \triple{\kval}{t_y}{b}} \quad \senv = [\ttt{x} \mapsto t_x, \ttt{y} \mapsto t_y]$}
@*)
\end{lstlisting}
}

Now, we begin with evaluation of the pure function application $\ttt{add}(x, y)$. First, we evaluate each of the arguments to the function application into corresponding symbolic values. Since these arguments are already variables in the symbolic environment, evaluation simply results in these underlying values and the symbolic state is not changed.

{\disableTttResize
\begin{lstlisting}[aboveskip=.25em, belowskip=0.5em,numbers=none,style=c0style]
(*@
\lightblue{$\veval{\sstate_1}{x}{t_x}{\sstate_1}$} \qquad \lightblue{$\veval{\sstate_1}{y}{t_y}{\sstate_1}$}
@*)
\end{lstlisting}
}

Then, since pure functions can depend on values in the heap, we compute a \emph{heap snapshot} to implicitly pass as an argument to $\ttt{add}$ that contains all the fields used in the function. Since the bodies and post-conditions of pure functions must be framed by the precondition, we can generate this snapshot using exclusively the precondition. In Smallfoot, these heap snapshots are obtained directly by consuming the precondition. However, since the formalization from Zimmerman et al. does not explicitly provide machinery to take heap snapshots, we introduce the notion of a \emph{symbolic footprint} of a formula $\phi$ to do so, written $\sfoot{\phi}{\sstate}$, which recurses over the structure of a formula to determine the heap chunks used when producing or consuming it.

{\disableTttResize
\begin{lstlisting}[aboveskip=.25em, belowskip=0.5em,numbers=none,style=c0style]
(*@
\lightblue{$\sfoot{\kacc(x.\kval) * \kacc(y.\kval)}{\sstate_1} 
= \sfoot{\kacc(x.\kval)}{\sstate_1} \cup \sfoot{\kacc(y.\kval)}{\sstate_1} 
= \set{\triple{\kval}{t_x}{a}, \triple{\kval}{t_y}{b}}$}
@*)
\end{lstlisting}
}

Then, to ensure the precondition holds, we consume the function's precondition. In the resulting symbolic state $\sstate_2$, access to both heap chunks has been removed. Because of the separating conjunction $*$, the consume operation also ensures that the removed heap chunks aren't aliases.

{\disableTttResize
\begin{lstlisting}[aboveskip=.25em, belowskip=0.5em,numbers=none,style=c0style]
(*@
\lightblue{$\vcons{\sstate_1}{\kacc(x.\kval) * \kacc(y.\kval)}{\sstate_2}$}
@*)
(*@
\lightgray{$\sstate_2 = \quadruple{\sheap_2}{\senv}{\ktrue}{\emptyset} \quad \sheap_2 = \emptyset \quad \senv = [\ttt{x} \mapsto t_x, \ttt{y} \mapsto t_y]$}
@*)
\end{lstlisting}
}

After consuming and ensuring the precondition holds, we produce the precondition so that we have a suitable environment for evaluating the function's body, resulting in a symbolic heap that again contains two heap chunks corresponding to the fields of our two objects.

{\disableTttResize
\begin{lstlisting}[aboveskip=.25em, belowskip=0.5em,numbers=none,style=c0style]
(*@
\lightblue{$\sproduce{\sstate_2}{\kacc(x.\kval) * \kacc(y.\kval)}{\sstate_3}$}
@*)
(*@
\lightgray{$\sstate_3 = \quadruple{\sheap_3}{\senv}{\ktrue}{\emptyset} \quad \sheap_3 = \set{\triple{\kval}{t_x}{a'}, \triple{\kval}{t_y}{b'}} \quad \senv = [\ttt{x} \mapsto t_x, \ttt{y} \mapsto t_y]$}
@*)
\end{lstlisting}
}

However, note that between $\sheap_1$ and $\sheap_3$, $a \neq a'$ and $b \neq b'$. Consuming and then producing the same specification ensures that we will end up with field chunks corresponding to the same field accesses, but it is critical to observe that these field chunks \emph{will contain distinct symbolic values}, since producing an accessibility predicate will initialize fresh symbolic values. So that the function body is evaluated in a way that respects constraints on the original symbolic values, we must reset these field chunks to use the correct values.

{\disableTttResize
\begin{lstlisting}[aboveskip=.25em, belowskip=0.5em,numbers=none,style=c0style]
(*@
\lightblue{$\sheap_4 = \sheap(\sstate_1)$}
@*)
(*@
\lightgray{$\sstate_4 = \quadruple{\sheap_4}{\senv}{\ktrue}{\emptyset} \quad \sheap_4 = \set{\triple{\kval}{t_x}{a}, \triple{\kval}{t_y}{b}} \quad \senv = [\ttt{x} \mapsto t_x, \ttt{y} \mapsto t_y]$}
@*)
\end{lstlisting}
}

Now, we are prepared to evaluate the body of the function. Since the evaluation of the body involves only looking up field chunks in the symbolic heap and looking up variables in the symbolic environment, the symbolic heap, symbolic environment, and path condition remain unchanged. However, we add the function name $\ttt{add}$ to the visited set in the resulting state to track which function bodies have been evaluated and prevent unbounded recursion.

{\disableTttResize
\begin{lstlisting}[aboveskip=.25em, belowskip=0.5em,numbers=none,style=c0style]
(*@
\lightblue{$\veval{\sstate_4[\svis = \set{\ttt{add}}]}{x.\kval + y.\kval}{a + b}{\sstate_5}$}
@*)
(*@
\lightgray{$\sstate_5 = \quadruple{\sheap_5}{\senv}{\ktrue}{\set{\ttt{add}}} \quad \sheap_5 = \set{\triple{\kval}{t_x}{a}, \triple{\kval}{t_y}{b}} \quad \senv = [\ttt{x} \mapsto t_x, \ttt{y} \mapsto t_y]$}
@*)
\end{lstlisting}
}

\newcommand{\fadd}{\ttt{add}}
Instead of directly returning the result of evaluation, we now encode the result of evaluation into the path condition because of the implicit dependency on the heap snapshot. Accordingly, we generate a fresh symbolic variable $\fadd(a, b, s)$, which represents the dependencies on the heap and the arguments, and require in the path condition that $\fadd(a, b, s) = a + b$. We additionally reset the value of $\svis$ because we only care to cut off recursive function calls, not sequential ones.

{\disableTttResize
\begin{lstlisting}[aboveskip=.25em, belowskip=0.5em,numbers=none,style=c0style]
(*@
\lightblue{$\fadd(a, b, s) = \ffresh$}
@*)
(*@
\lightgray{$\sstate_6 = \quadruple{\sheap_6}{\senv}{\fadd(a, b, s) == a + b}{\emptyset} \quad \sheap_6 = \set{\triple{\kval}{t_x}{a}, \triple{\kval}{t_y}{b}} \quad \senv = [\ttt{x} \mapsto t_x, \ttt{y} \mapsto t_y]$}
@*)
\end{lstlisting}
}

Furthermore, the post-condition may contain further restrictions or information about the resulting value. However, in its current form, the post-condition is an arbitrary expression containing field accesses and the $\kresult$ variable, so we must \emph{axiomatize} it to include it in the post-condition. Accordingly, we introduce a new function $\fax{\sstate}{\fadd(a, b, s)}{e}$, which recurses over the expression $e$, substituting $\kresult$ with the symbolic variable $\fadd(a, b, s)$ and substituting field accesses for their underlying symbolic values from the snapshot. 

{\disableTttResize
\begin{lstlisting}[aboveskip=.25em, belowskip=0.5em,numbers=none,style=c0style]
(*@
\lightblue{$\fax{\sstate_6}{\fadd(a, b, s)}{(y.\kval < 0) \kor (\kresult >= x.\kval)}
= (a < 0) \kor (\fadd(a, b, s) >= b)
$}@*)
(*@
\lightgray{$\sstate_7 = \quadruple{\sheap_7}{\senv}{(\fadd(a, b, s) == a + b) \kand ((a < 0) \kor (\fadd(a, b, s) >= b))}{\emptyset} \quad \sheap_7 = \set{\triple{\kval}{t_x}{a}, \triple{\kval}{t_y}{b}} \quad \senv = [\ttt{x} \mapsto t_x, \ttt{y} \mapsto t_y]$}
@*)
\end{lstlisting}
}

Finally, we return the generated symbol $\fadd(a, b, s)$, which is constrained by the path condition. 

Formally, this process can be represented by the rules in Figure~\ref{fig:viper-pure}. Similarly to unfolding expressions, pure functions are also susceptible to unbounded verification during static verification since a recursive function evaluation will branch over both the base case and recursive case due to the use of symbolic values rather than concrete values. To address recursive functions, we use the same approach as Section~\ref{sec:static-unfolding}, where \textsc{SEvalPureExplicit} refers to the case where the function has not yet recursively called itself and the body is evaluated and \textsc{SEvalPureImplicit} refers to the case where the function is being recursively called and the body is not evaluated to prevent unbounded recursion. This approach is theoretically incomplete; however, we axiomatize the post-condition of the recursive application to minimize this drawback.

Compared to the approach detailed here, Viper's symbolic execution engine uses a more complex approach based on passing axioms to the underlying SMT solver to enable more complete quantifier triggering strategies, which increases performance when functions with large bodies are called repeatedly or when using recursive functions \cite{schwerhoff2016advancing}. We instead opt for the strategy described by Smans et al. because the Zimmerman et al. formalization does not directly interface with the underlying SMT solver beyond implicitly assuming its soundness; since the path condition is provided to the SMT solver during verification, we instead opt to emulate this behavior by adding these axioms to the path condition.

\begin{figure}
    \centering
    {\footnotesize\disableTttResize
    \begin{mathpar}
        \inferrule[SEvalPureExplicit]{
            f \notin \svis(\sstate_1) \\
            \multiple{\veval{\sstate_1}{e}{t}{\sstate_2}} \\
            \multiple{x} = \fparams(f) \\
            s = \sfoot{\fpre(f)}{{\sstate_2}} \\
            \vcons{\sstate_2[\senv = \senv(\sstate_2)[\multiple{x \mapsto t}]]}{\fpre(f)}{\sstate_3} \\
            \sproduce{\sstate_3}{\fpre(f)}{\sstate_4} \\
            \sheap' = \{ \triple{f}{t_e}{t'} : \triple{f}{t_e}{t} \in \sheap(\sstate_4), \triple{f}{t_e}{t'} \in \sheap(\sstate_2) \} \\
            \veval{\sstate_4[\sheap = \sheap', \svis = \svis(\sstate_3) \cup \{ f \}]}{\fbody(f)}{t'}{\sstate_5}{\scheck_3} \\
            f(\multiple{t}, s) = \ffresh \\
            \sstate_6 = \sstate_2[\pc = \pc(\sstate_2) \kand (f(\multiple{t}, s) \keq t') \kand \fax{\sstate}{\fts}{\fpost(f)}]
        }{
            \veval{\sstate_1}{\fe}{f(\multiple{t}, s)}{\sstate_6}
        }
    
        \inferrule[SEvalPureImplicit]{
            f \in \svis(\sstate_1) \\
            \multiple{\veval{\sstate_1}{e}{t}{\sstate_2}} \\
            \multiple{x} = \fparams(f) \\
            s = \sfoot{\fpre(f)}{{\sstate_2}} \\
            \vcons{\sstate_2[\senv = \senv(\sstate_2)[\multiple{x \mapsto t}]]}{\fpre(f)}{\sstate_3} \\
            f(\multiple{t}, s) = \ffresh \\
            \sstate_4 = \sstate_2[\pc = \pc(\sstate_2) \kand \fax{\sstate}{\fts}{\fpost(f)}]
        }{
            \veval{\sstate_1}{\fe}{f(\multiple{t}, s)}{\sstate_4}
        }
    \end{mathpar}
    }
    \caption{Symbolic evaluation rules for pure function applications in Viper.}
    \label{fig:viper-pure}
\end{figure}

\section{Gradual Viper Semantics}

After gradual verification was implemented for the C0 subset of the C programming language and the underlying Gradual Viper verifier by DiVincenzo et al. \cite{divincenzo2025gradual}, it was proved to be sound by Zimmerman et al. \cite{zimmerman2024sound}. This work differs from the previous formalizations of gradual verification \cite{bader2018gradual, wise2020gradual} in that it specifically focuses on the Silicon symbolic execution backend for Viper rather than the weakest liberal preconditions approach used in previous iterations and that the semantics are formalized in non-deterministic inference rules rather than continuation passing-style rules. Overall, the formalization is structured as several components: the static semantics, the dynamic semantics, the correspondence, and the proof of soundness.

In this section, we provide and examine the formal semantics for symbolic evaluation, run-time evaluation, and framing alongside the definition of exact footprints for unfolding expressions. Additionally, we describe the changes we've made to the definitions of correspondence between static and dynamic states to account for the addition of unfolding expressions.

\subsection{Representation}\label{sec:gv-representation}

Gradual Viper enables imprecise specifications through the use of the $\ttt{?}$ operator; accordingly, we expand our definition of formulae to include imprecise formulas, written $\gform = \simprecise{\phi}$. In our context, this is noteworthy because an imprecise specification represents any non-contradictory strengthening of the precise portion, meaning that formulas like $\simprecise{x.f \geq 2}$ can also be strengthened with accessibility predicates ($\kacc(x.f) * x.f \geq 2$) or with arbitrary unfolding of user-defined predicates ($\sunfolding{\ttt{pred}(x)}{x.f >= 2}$ where $\ttt{pred}(x)$ specifies access to $x.f$) \cite{wise2020gradual}.

\subsubsection{Symbolic State}
Because a symbolic state can now be imprecise, the state now includes an \emph{imprecision flag}, written $\imp$, representing whether the current state is imprecise. Additionally, to store optimistic assumptions, the symbolic state also includes a \emph{optimistic heap}, written $\oheap$.

\begin{itemize}
    \item The \emph{imprecise flag} $\imp \in \{ \top, \bot \}$ indicates whether a specification containing $\ttt{?}$ was used during symbolic execution and thus whether the current state can make optimistic assumptions.
    \item The \emph{optimistic heap} $\oheap$ represents field chunks and predicate chunks that may not be distinct where ownership is optimistically assumed. If the current state is precise, this must be empty.
    \item The precise heap $\sheap$, symbolic store $\senv$, path condition $\pc$, and visited set $\svis$ behave as described previously.
\end{itemize}

The fundamental operations describing static verification are the same; however, because Gradual Viper also includes imprecise specifications and optimistic assumptions, these operations can now also produce sets of \emph{run-time checks}, written $\scheck$. Optimistic assumptions about ownership are backed by these runtime checks, meaning that ownership also has to be tracked dynamically rather than only statically. This specifically poses a challenge for methods with imprecise preconditions, since this requires taking ownership of all fields owned by the caller since any might be used by the callee \cite{wise2020gradual}.

\begin{itemize}
    \item Symbolically evaluating a symbolic expression $e$ into a symbolic value $t$, resulting in new state $\sstate'$ and run-time checks $\scheck$. This is denoted by $\seval{\sstate}{e}{t}{\sstate'}{\scheck}$. The run-time checks $\scheck$ will ensure that all fields and predicate instances necessary to evaluate the expression are available.
    \item Producing the information from formula $\gform$ into the symbolic heap, resulting in new state $\sstate'$; denoted $\sproduce{\sstate}{\gform}{\sstate'}$.
    \item Consuming a specification $\gform$ from the symbolic state, resulting in new state $\sstate'$ and run-time checks $\scheck$; denoted $\scons{\sstate}{\gform}{\sstate'}{\scheck}$. The run-time checks $\scheck$ will ensure that the specification holds if optimistic assumptions are made.
    \item Symbolically executing a statement (notation unnecessary since we neither extend nor rely upon this judgement).
\end{itemize}

The semantics for these operations will coincide with those for Silicon for when the program being verified is fully precise, which will result in no run-time checks; the formalization of Gradual Viper conservatively extends these to account for cases with imprecise specifications, resulting optimistic assumptions, and runtime checks.

\subsubsection{Dynamic State}
The execution semantics for \gco were derived from those of of the original \co language \cite{arnold2010c0}. In addition to these, the \gco dynamic semantics include the definition of the \emph{footprint} of a formula, which is the set of permissions necessary to assert that formula. 

Analogous to how the symbolic state tracks information about a symbolic execution path, we track information about the execution of the program at run time with a heap $\heap$, an environment $\rho$ mapping variables to values, and a permission set $\alpha$ tracking the owned fields. There are three main operations at run time that use this information.

\begin{itemize}
    \item Evaluating an expression $e$ into a concrete value $t$, denoted by $\eval{\heap}{\env}{e}{t}$.
    \item Asserting a formula $\gform$, denoted by $\assertion{\heap}{\perms}{\env}{\gform}$.
    \item Executing a statement (again, notation elided for simplicity).
\end{itemize}

\subsection{Static Semantics for Unfolding Expressions in Gradual Viper}

With the introduction of imprecise specifications, there are several cases in which the behavior of unfolding expressions change, depending on whether the initial state is precise, whether the body of the predicate is precise, and whether the predicate recursively unfolds another instance of itself. If the predicate has not recursively unfolded another instance of itself, the following cases apply.

\begin{itemize}
    \item If the initial state is precise, then the optimistic heap must remain empty regardless of whether the predicate being unfolded is precise or not. This is because, in a precise state, the optimistic heap must remain empty to ensure soundness, since the optimistic heap contains assumptions made from an imprecise state. This is represented by the \textsc{SEvalUnfoldingPrecise} rule in Figure~\ref{fig:gv-unfolding}.

    \item If the initial state is imprecise and the predicate body is precise, the precise heap is reset as usual but the information from the final optimistic heap and an instance of the predicate are retained in the optimistic heap. The assumptions from the final optimistic heap can be retained since they resulted from imprecision in the initial state and not the predicate body -- since the predicate body must be re-folded, its resulting optimistic assumptions must be discarded. The predicate instance is also included because it must be optimistically assumed for framing during run-time evaluation. This is represented by the \textsc{SEvalUnfoldingImpreciseA} rule.

    \item If both the initial state and the predicate body are precise, then the assumptions in the final optimistic heap are discarded because those assumptions could have arisen as a result of the imprecision in the predicate rather than the imprecision in the initial state. Conservatively, our design assumes that these assumptions were made resulting from the imprecision in the predicate body and thus must be removed when the predicate is re-folded. This is represented by the \textsc{SEvalUnfoldingImpreciseB} rule.
\end{itemize}

Although we could naively reset both the precise and optimistic heaps in any case, our design safely decreases the run-time overhead by retaining assumptions in the optimistic heap if possible.

Our handling of predicates that recursively unfold another instance of themselves is similar to the approach described in Section~\ref{sec:static-unfolding}, where the \textsc{SEvalUnfoldingImplicit} rule from Figure~\ref{fig:viper-unfolding} applies to cut off unbounded recursion. However, in the gradual case, additional complications arise from the possibility of an imprecise initial state, meaning that we have to account for the following cases:

\begin{itemize}
    \item The predicate instance is already in either the precise or optimistic heap. In this case, the unfolding operation will succeed, since access to the predicate instance has been statically proven. This behavior is represented by the \textsc{SEvalUnfoldingImplicitPrecise} rule.

    \item The initial state is imprecise and the predicate instance isn't in either heap. In this case, the unfolding operation will succeed, since access to the predicate instance can be optimistically assumed by adding a matching run-time check. This behavior is represented by the \textsc{SEvalUnfoldingImplicitImprecise} rule.

    \item The initial state is precise and the predicate instance isn't in either heap. In this case, the unfolding operation will not succeed since the predicate instance isn't available statically or through optimistic assumptions. This behavior is represented by the \textsc{SEvalUnfoldingImplicitFailure} rule, since returning an unsatisfiable run-time check set $\set{\bot}$ corresponds to a failure of static verification.
\end{itemize}

In order to avoid unbounded recursion in the rules, we don't actually unfold the predicate instance in any of these cases; however, soundness still requires checking whether the instance is available. Otherwise, a recursive unfolding expression might statically verify when the predicate being unfolded doesn't hold statically or at run time.

In Figure~\ref{fig:gradual-unfolding-example}, we specifically examine an example with a precise initial state and an imprecise predicate being unfolded to demonstrate our design of the unfolding expression semantics.

\begin{figure}[H]
\begin{minipage}{\linewidth}
\begin{lstlisting}[language=Viper,xleftmargin=.2\textwidth,xrightmargin=.2\textwidth]
field val: Int

predicate pos(x: Ref) { (*@\lightred{\ttt{?}}@*) * x.val > 0 }

method caller(x: Ref, ...)
  requires pos(x)
  ensures ...
{
    ... unfolding pos(x) in (x.val < 256) ...
}
\end{lstlisting}
\end{minipage}
    \caption{Example of unfolding expression with imprecise predicate.}
    \label{fig:gradual-unfolding-example}
\end{figure}

To demonstrate the importance of the imprecise predicate and the optimistic assumptions it enables, we assume that we have objects $t_y$ and $t_z$ with an initial symbolic heap $\sheap(\sstate_1) = \set{\pair{\kmem}{t_x}, \triple{\kval}{t_y}{b}}$, meaning predicate $\kmem$ holds with argument $t_x$ and field $\kval$ of object $t_y$ has value $b$. Although $t_y$ isn't used in the body of the predicate, we include it to highlight how the gradual case interacts with the heap. Again, we start with a precise (represented by $\bot$) symbolic state that contains only the described heap chunk and variable mappings for our assumptions.

{\disableTttResize
\begin{lstlisting}[aboveskip=.25em, belowskip=0.5em,numbers=none,style=c0style]
(*@
\lightgray{$\sstate_1 = \sextuple{\bot}{\sheap_1}{\oheap_1}{\senv}{\ktrue}{\emptyset} \quad \sheap_1 = \set{\pair{\kmem}{t_x}, \triple{\kval}{t_y}{b}} \quad \oheap_1 = \emptyset \quad \senv = [\ttt{x} \mapsto t_x, \ttt{y} \mapsto t_y]$}
@*)
\end{lstlisting}
}

We evaluate the predicate argument, resulting in the same symbolic state. However, as described in Section~\ref{sec:gv-representation}, the form of the evaluation judgement has changed to include an additional set of run-time checks accumulated through evaluation. In this case, the set is empty because the evaluation only requires looking up a variable in the symbolic environment. Note that even though we do not have a heap chunk corresponding to $x.\kval$, we can evaluate $x$ itself since we are only looking up the underlying value $t_x$ in the symbolic environment, not accessing a field of $t_x$.

{\disableTttResize
\begin{lstlisting}[aboveskip=.25em, belowskip=0.5em,numbers=none,style=c0style]
(*@
\lightblue{$\seval{\sstate_1}{x}{t_x}{\sstate_1}{\emptyset}$}
@*)
\end{lstlisting}
}

Now, we consume an instance of the predicate. Since we are currently in a precise state and consuming a predicate instance does not interact with the imprecise body, we remain in a precise state afterwards. Note that this also includes a set of run-time checks, which is again empty because the predicate instance is statically proven.

{\disableTttResize
\begin{lstlisting}[aboveskip=.25em, belowskip=0.5em,numbers=none,style=c0style]
(*@
\lightblue{$\scons{\sstate_1}{\kmem(t_x)}{\sstate_2}{\emptyset}$}
@*)
(*@
\lightgray{$\sstate_2 = \sextuple{\bot}{\sheap_2}{\oheap_2}{\senv}{\ktrue}{\emptyset} \quad \sheap_2 = \set{\triple{\kval}{t_y}{b}} \quad \oheap_2 = \emptyset \quad \senv = [\ttt{x} \mapsto t_x, \ttt{y} \mapsto t_y]$}
@*)
\end{lstlisting}
}

To achieve the desired unfolding, we now produce the predicate's body so that the information contained therein is accessible while evaluating the inner expression. We also set the visited set to $\set{\kmem}$ to avoid unbounded recursion if the predicate unfolds itself in its body. Because the predicate body is imprecise, this results in an imprecise state; while evaluating the inner expression, this also optimistically assumes access to $x.\kval$, adding it to the optimistic heap. 


{\disableTttResize
\begin{lstlisting}[aboveskip=.25em, belowskip=0.5em,numbers=none,style=c0style]
(*@
\lightblue{$\sproduce{\sstate_2[\svis = \set{\kmem}]}{x.\kval > 0}{\sstate_3}$}
@*)
(*@
\lightgray{$\sstate_3 = \sextuple{\top}{\sheap_3}{\oheap_3}{\senv}{a > 0}{\set{\kmem}} \quad \sheap_3 = \set{\triple{\kval}{t_y}{b}} \quad \oheap_3 = \set{\triple{\kval}{t_x}{a}} \quad \senv = [\ttt{x} \mapsto t_x, \ttt{y} \mapsto t_y]$}
@*)
\end{lstlisting}
}

Now, we are ready to evaluate the inner expression. Again, the state remains unchanged (except for having reset $\svis$) and there are no run-time checks because we simply look up $x.\kval$ in the optimistic heap.

{\disableTttResize
\begin{lstlisting}[aboveskip=.25em, belowskip=0.5em,numbers=none,style=c0style]
(*@
\lightblue{$\seval{\sstate_3[\svis = \emptyset]}{x.\kval < 256}{a < 256}{\sstate_4}{\emptyset}$}
@*)
(*@
\lightgray{$\sstate_4 = \sextuple{\top}{\sheap_4}{\oheap_4}{\senv}{a > 0}{\emptyset} \quad \sheap_4 = \set{\triple{\kval}{t_y}{b}} \quad \oheap_4 = \set{\triple{\kval}{t_x}{a}} \quad \senv = [\ttt{x} \mapsto t_x, \ttt{y} \mapsto t_y]$}
@*)
\end{lstlisting}
}

Because the imprecision arose as a result of unfolding an imprecise predicate, which must now be re-folded, we must also discard the optimistic assumptions that were made while it was unfolded. Because the initial state was precise, the optimistic heap must be emptied, and we also must reset the precise heap so that it once again contains the predicate instance and the fields it initially held.

{\disableTttResize
\begin{lstlisting}[aboveskip=.25em, belowskip=0.5em,numbers=none,style=c0style]
(*@
\lightblue{$\sheap_5 = \sheap_1 \quad \oheap_5 = \emptyset$}
@*)
(*@
\lightgray{$\sstate_5 = \sextuple{\top}{\sheap_5}{\oheap_5}{\senv}{a > 0}{\emptyset} \quad \sheap_5 = \set{\pair{\kmem}{t_x}, \triple{\kval}{t_y}{b}} \quad \oheap_5 = \emptyset \quad \senv = [\ttt{x} \mapsto t_x, \ttt{y} \mapsto t_y]$}
@*)
\end{lstlisting}
}

Now, we can return the expression $a < 256$ from the penultimate step. 


\begin{figure}[t]
  {\footnotesize\disableTttResize
    \begin{mathpar}
        \inferrule[SEvalUnfoldingPrecise]{
                \lnot \imp(\sstate_1) \\
                p \notin \svis(\sstate_1) \\
                \multiple{\seval{\sstate_1}{e}{t}{\sstate_2}{\scheck_1}} \\
                \scons{\sstate_2}{p(\multiple{e})}{\sstate_3}{\scheck_2} \\
                \multiple{x} = \fpredparams(p) \\
                \sproduce{\sstate_3[\senv = \senv(\sstate_3)[\multiple{x \mapsto t}], \svis = \svis(\sstate_3) \cup \{ p \}]}{\fpred(p)}{\sstate_4} \\
                \seval{\sstate_4[\senv = \senv(\sstate_3)]}{e_0}{t_0}{\sstate_5}{\scheck_3} \\
                \sstate_6 = \sstate_5[\svis = \svis(\sstate_1), \sheap = \sheap(\sstate_2), \oheap = \emptyset]
            }{
                \seval{\sstate_1}{\sunfolding{p(\multiple{e})}{e_0}}{t_0}{\sstate_6}{\scheck_1 \cup \scheck_2 \cup \scheck_3}
            }

            \inferrule[SEvalUnfoldingImpreciseA]{
                \imp(\sstate_1) \\
                \fpred(p)\; \text{is precise} \\
                p \notin \svis(\sstate_1) \\
                \multiple{\seval{\sstate_1}{e}{t}{\sstate_2}{\scheck_1}} \\
                \scons{\sstate_2}{p(\multiple{e})}{\sstate_3}{\scheck_2} \\
                \multiple{x} = \fpredparams(p) \\
                \sproduce{\sstate_3[\senv = \senv(\sstate_3)[\multiple{x \mapsto t}], \svis = \svis(\sstate_3) \cup \{ p \}]}{\fpred(p)}{\sstate_4} \\
                \seval{\sstate_4[\senv = \senv(\sstate_3)]}{e_0}{t_0}{\sstate_5}{\scheck_3} \\
                \oheap' = \oheap(\sstate_2) \cup \oheap(\sstate_5) \cup \{ \pair{p}{\multiple{t}} \} \\
                \sstate_6 = \sstate_5[\svis = \svis(\sstate_1), \sheap = \sheap(\sstate_2), \oheap = \oheap']
            }{
                \seval{\sstate_1}{\sunfolding{p(\multiple{e})}{e_0}}{t_0}{\sstate_6}{\scheck_1 \cup \scheck_2 \cup \scheck_3}
            }

            \inferrule[SEvalUnfoldingImpreciseB]{
                \imp(\sstate_1) \\
                \fpred(p)\; \text{is imprecise} \\
                p \notin \svis(\sstate_1) \\
                \multiple{\seval{\sstate_1}{e}{t}{\sstate_2}{\scheck_1}} \\
                \scons{\sstate_2}{p(\multiple{e})}{\sstate_3}{\scheck_2} \\
                \multiple{x} = \fpredparams(p) \\
                \sproduce{\sstate_3[\senv = \senv(\sstate_3)[\multiple{x \mapsto t}], \svis = \svis(\sstate_3) \cup \{ p \}]}{\fpred(p)}{\sstate_4} \\
                \seval{\sstate_4[\senv = \senv(\sstate_3)]}{e_0}{t_0}{\sstate_5}{\scheck_3} \\
                \sstate_6 = \sstate_5[\svis = \svis(\sstate_1), \sheap = \sheap(\sstate_2), \oheap = \oheap(\sstate_2) \cup \{ \pair{p}{\multiple{t}} \}]
            }{
                \seval{\sstate_1}{\sunfolding{p(\multiple{e})}{e_0}}{t_0}{\sstate_6}{\scheck_1 \cup \scheck_2 \cup \scheck_3}
            }

            \inferrule[SEvalUnfoldingImplicitPrecise]{
                p \in \svis(\sstate_1) \\
                \seval{\sstate_1}{e}{t}{\sstate_2}{\scheck} \\
                \pair{p}{\multiple{t}} \in \sheap(\sstate_2) \cup \oheap(\sstate_2) \\
                t_0 = \ffresh
            }{
                \seval{\sstate_1}{\sunfolding{p(\multiple{e})}{e_0}}{t_0}{\sstate_2}{\scheck}
            }

            \inferrule[SEvalUnfoldingImplicitImprecise]{
                p \in \svis(\sstate_1) \\
                \seval{\sstate_1}{e}{t}{\sstate_2}{\scheck} \\
                \imp(\sstate_2) \\
                \pair{p}{\multiple{t}} \notin \sheap(\sstate_2) \cup \oheap(\sstate_2) \\
                \oheap' = \oheap(\sstate_2) \cup \set{\pair{p}{\multiple{t}}} \\
                \scheck' = \scheck \cup \set{\pair{p}{\multiple{t}}} \\
                \sstate_3 = \sstate_2[\oheap = \oheap'] \\
                t_0 = \ffresh
            }{
                \seval{\sstate_1}{\sunfolding{\pe}{e_0}}{t_0}{\sstate_3}{\scheck'}
            }

            \inferrule[SEvalUnfoldingImplicitFailure]{
                p \in \svis(\sstate_1) \\
                \seval{\sstate_1}{e}{t}{\sstate_2}{\scheck} \\
                \lnot \imp(\sstate_2) \\
                \pair{p}{\multiple{t}} \notin \sheap(\sstate_2) \cup \oheap(\sstate_2) \\
                t_0 = \ffresh
            }{
                \seval{\sstate_1}{\sunfolding{\pe}{e_0}}{t_0}{\sstate_2}{\set{\bot}}
            }
    \end{mathpar}
  }
  \caption{Symbolic evaluation rules for unfolding expressions in Gradual Viper.}
  \label{fig:gv-unfolding}
\end{figure}

\subsection{Static Semantics for Pure Functions in Gradual Viper}

To demonstrate our contribution in the static setting, in Figure~\ref{fig:gradual-pure-func-example} we use an example similar to the fully-static example used in Section~\ref{sec:static-pure}, but with an imprecise precondition.

\begin{figure}[H]
\begin{minipage}{\linewidth}
\begin{lstlisting}[language=Viper,xleftmargin=.2\textwidth,xrightmargin=.2\textwidth]
field val: Int

function add(x: Ref, y: Ref): Int
  requires (*@\lightred{\ttt{?}}@*) * acc(y.val)
  ensures (y.val < 0) || (result >= x.val)
  { x.val + y.val }

method caller(...)
  requires ...
  ensures ...
{
    ... add(x, y) ...
}
\end{lstlisting}
\end{minipage}
    \caption{Example of pure function with imprecise precondition.}
    \label{fig:gradual-pure-func-example}
\end{figure}

To demonstrate the behavior of the imprecise precondition, we assume an initial symbolic heap $\sheap(\sstate_1) = \set{\triple{\kval}{t_y}{b}, \triple{\kval}{t_z}{c}}$ as in the previous example. Although $t_z$ again isn't used in our code sample, we include it to highlight the intricacies of the gradual case and its interaction with the heap. Again, we start with a precise symbolic state with the aforementioned heap and variable mappings.

{\disableTttResize
\begin{lstlisting}[aboveskip=.25em, belowskip=0.5em,numbers=none,style=c0style]
(*@
\lightgray{$\sstate_1 = \sextuple{\bot}{\sheap_1}{\oheap_1}{\senv}{\ktrue}{\emptyset} \quad \sheap_1 = \set{\triple{\kval}{t_y}{b}, \triple{\kval}{t_z}{c}} \quad \oheap_1 = \emptyset \quad \senv = [\ttt{x} \mapsto t_x, \ttt{y} \mapsto t_y]$}
@*)
\end{lstlisting}
}

We evaluate the function arguments, resulting in the same symbolic state. However, as described in Section~\ref{sec:gv-representation}, the form of the evaluation judgement has changed to include an additional set of run-time checks accumulated through evaluation. In this case, the sets are both empty because evaluation only requires looking up these variables in the symbolic environment. Note that, even though we don't have a heap chunk corresponding to $x.\kval$, we can evaluate $x$ itself since we are only looking up the underlying value $t_x$ in the symbolic environment, not accessing a field of $t_x$.

{\disableTttResize
\begin{lstlisting}[aboveskip=.25em, belowskip=0.5em,numbers=none,style=c0style]
(*@
\lightblue{$\seval{\sstate_1}{x}{t_x}{\sstate_1}{\emptyset}$} \qquad \lightblue{$\seval{\sstate_1}{y}{t_y}{\sstate_1}{\emptyset}$}
@*)
\end{lstlisting}
}

We must take a snapshot of the heap locations used during evaluation to pass as an implicit parameter; however, since the state during evaluation is imprecise, the verifier can make optimistic assumptions about ownership of any fields that are then checked at run-time. To account for this possibility, the symbolic footprint of an imprecise specification captures the entire heap, including the unused heap chunk $\triple{\kval}{t_z}{c}$.

{\disableTttResize
\begin{lstlisting}[aboveskip=.25em, belowskip=0.5em,numbers=none,style=c0style]
(*@
\lightblue{$\sfoot{\simprecise{\kacc(y.\kval)}}{\sstate_1} = \sheap_1$}
@*)
\end{lstlisting}
}

To ensure that the precondition holds in the initial state (which is still precise), we again consume the precondition. Because an imprecise specification might represent access to any field, consuming an imprecise formula empties both the precise and optimistic heaps and ends in an imprecise state (reprsented by $\top$). Note that, although the consume operation generally may produce a run-time check, the resulting run-time check set here is empty since $y.\kval$ is the only field mentioned so far and access to it is statically proven.

{\disableTttResize
\begin{lstlisting}[aboveskip=.25em, belowskip=0.5em,numbers=none,style=c0style]
(*@
\lightblue{$\scons{\sstate_1}{\simprecise{\kacc(y.\kval)}}{\sstate_2}{\emptyset}$}
@*)
(*@
\lightgray{$\sstate_2 = \sextuple{\top}{\sheap_2}{\oheap_2}{\senv}{\ktrue}{\emptyset} \quad \sheap_2 = \emptyset \quad \oheap_2 = \emptyset \quad \senv = [\ttt{x} \mapsto t_x, \ttt{y} \mapsto t_y]$} 
@*)
\end{lstlisting}
}

So that we have a suitable environment to evaluate the function's body, we must now produce the function's precondition. This adds a heap chunk corresponding to the accessibility predicate $\kacc(y.\kval)$ back into the precise heap, but the resulting state is still imprecise.

{\disableTttResize
\begin{lstlisting}[aboveskip=.25em, belowskip=0.5em,numbers=none,style=c0style]
(*@
\lightblue{$\sproduce{\sstate_2}{\simprecise{\kacc(y.\kval)}}{\sstate_3}$}
@*)
(*@
\lightgray{$\sstate_3 = \sextuple{\top}{\sheap_3}{\oheap_3}{\senv}{\ktrue}{\emptyset} \quad \sheap_3 = \set{ \triple{\kval}{t_y}{b'} } \quad \oheap_3 = \emptyset \quad \senv = [\ttt{x} \mapsto t_x, \ttt{y} \mapsto t_y]$} 
@*)
\end{lstlisting}
}

However, $b \neq b'$ because consuming and producing resulted in new symbolic values instead of the ones used in the original symbolic state. To ensure that the final constraints in the path condition are expressed in terms of known symbolic values, we revert these to refer to those from a previous state. Note that, unlike the previous example, we can't directly reset to $\sheap_1$ since this would provide static access to the field chunk $\triple{\kval}{t_z}{c}$, even though the initial heaps were emptied after consuming and access to $t_z.\kval$ was never obtained again through the produce.

{\disableTttResize
\begin{lstlisting}[aboveskip=.25em, belowskip=0.5em,numbers=none,style=c0style]
(*@
\lightblue{$\sheap_4 = \set{\triple{f}{t_e}{t'} : \triple{f}{t_e}{t} \in \sheap(\sstate_3), \triple{f}{t_e}{t'} \in \sheap(\sstate_1)}$}
@*)
(*@
\lightgray{$\sstate_4 = \sextuple{\top}{\sheap_4}{\oheap_4}{\senv}{\ktrue}{\emptyset} \quad \sheap_4 = \set{\triple{\kval}{t_y}{b}} \quad \oheap_4 = \emptyset \quad \senv = [\ttt{x} \mapsto t_x, \ttt{y} \mapsto t_y]$}
@*)
\end{lstlisting}
}

Now, we proceed to evaluation of the function body, which is where the imprecision of the symbolic state comes into play. Since the field $x.\kval$ isn't available in the precise heap and we are currently in an imprecise state, the verifier creates a new heap chunk $\triple{\kval}{t_x}{a}$, adds this to the optimistic heap $\oheap_5$, and assumes that ownership of the required field will hold. This also adds the run-time check $\pair{t_x}{\kval}$ to the set of run-time checks at the end, which will check ownership of the field $\kval$ of object $t_x$ at this point during execution.

{\disableTttResize
\begin{lstlisting}[aboveskip=.25em, belowskip=0.5em,numbers=none,style=c0style]
(*@
\lightblue{$\seval{\sstate_4[\svis = \set{\fadd}]}{x.\kval + y.\kval}{a + b}{\sstate_5}{\set{\pair{t_x}{\kval}}}$}
@*)
(*@
\lightgray{$\sstate_5 = \sextuple{\top}{\sheap_5}{\oheap_5}{\senv}{\ktrue}{\set{\ttt{add}}} \quad \sheap_5 = \set{\triple{\kval}{t_y}{b}} \quad \oheap_5 = \set{\triple{\kval}{t_x}{a}} \quad \senv = [\ttt{x} \mapsto t_x, \ttt{y} \mapsto t_y] \quad \scheck = \set{\pair{t_x}{\kval}}$}
@*)
\end{lstlisting}
}

Then, axiomatizing the return value and post-condition is the same as in the fully-precise case.

{\disableTttResize
\begin{lstlisting}[aboveskip=.25em, belowskip=0.5em,numbers=none,style=c0style]
(*@
\lightblue{$\fadd(a, b, s) = \ffresh$}
@*)
(*@
\lightblue{$\fax{\sstate_5}{\fadd(a, b, s)}{(y.\kval < 0) \kor (\kresult >= x.\kval)}
= (b < 0) \kor (\fadd(a, b, s) >= a)
$}@*)
(*@
\lightgray{$\sstate_6 = \sextuple{\top}{\sheap_6}{\oheap_6}{\senv}{(\fadd(a, b, s) == a + b) \kand ((b < 0) \kor (\fadd(a, b, s) >= a))}{\emptyset} \quad \sheap_6 = \set{\triple{\kval}{t_y}{b}} \quad \oheap_6 = \set{\triple{\kval}{t_x}{a}} $}
@*)
(*@
\lightgray{$\senv = [\ttt{x} \mapsto t_x, \ttt{y} \mapsto t_y] \quad \scheck = \set{\pair{t_x}{\kval}}$}
@*)
\end{lstlisting}
}

When evaluating a pure function with an imprecise precondition, the final symbolic state after evaluation remains imprecise; although this breaks a previous assumption by Zimmerman et al. \cite{zimmerman2024sound} that evaluation will not change the precision of a state, once imprecise assumptions are introduced to a state, it should remain imprecise. This assumption arose only because the previous formalization didn't address features like pure functions that rely on the evaluation judgement while directly interacting with imprecision.

Overall, this process is captured by the symbolic evaluation rules in Figure~\ref{fig:gv-pure}. Notably, our handling of recursive pure functions interacts differently with imprecise preconditions or imprecise initial states; since recursive calls to the same function aren't evaluated directly, the post-condition of a recursive pure function may rely on optimistic assumptions that haven't yet been made. For instance, let us consider the example in Figure~\ref{fig:recursive-pure-func-example}. When $\ttt{sum}$ is called for the second time, no optimistic assumptions will be made since the body isn't called and the post-condition isn't directly produced. 


\newcommand{\ksum}{\ttt{sum}}
\begin{figure}[H]
\begin{minipage}{\linewidth}
\begin{lstlisting}[language=Viper,xleftmargin=.2\textwidth,xrightmargin=.2\textwidth]
field val: Int
field next: Ref

function sum(node: Ref): Int
  requires (*@\lightred{\ttt{?}}@*)
  ensures (*@\lightred{$\ttt{x.val}~\ttt{<=}~\kresult$}@*)
  { (node != null) ? node.val + sum(node.next) : 0 }
\end{lstlisting}
\end{minipage}
    \caption{Example of recursive pure function with imprecise precondition.}
    \label{fig:recursive-pure-func-example}
\end{figure}

In this case, we must take the symbolic footprint of the post-condition, add any missing fields to the optimistic heap, and add the necessary run-time checks to the final result of evaluation. When we pass the symbolic state to the symbolic footprint, we must ensure we pass it in an imprecise state so that any of the necessary fields can still be obtained without being statically proven. After this step, we can axiomatize the post-condition as usual.

{\disableTttResize
\begin{lstlisting}[aboveskip=.25em, belowskip=0.5em,numbers=none,style=c0style]
(*@
\lightblue{$s' = \sfoot{\ffuncpost(\ksum)}{\sstate[\imp = \top]} = \sfoot{x.\kval}{\sstate[\imp = \top]} = \set{\triple{\kval}{t_x}{a}}$}
@*)
(*@
\lightblue{$\scheck' = \set{\pair{t_e}{f} : \triple{f}{t_e}{t} \in s'} = \set{\pair{t_x}{\kval}}$} \quad
\lightblue{$\oheap' = \oheap \cup (s' \setminus \sheap) = \set{\triple{\kval}{t_x}{a}}$}
@*)
\end{lstlisting}
}

This behavior is exhibited in the \textsc{SEvalPureImplicit} rule in Figure~\ref{fig:gv-pure}.

\begin{figure}
    \centering
    {\footnotesize\disableTttResize
    \begin{mathpar}
        \inferrule[SEvalFunctionExplicit]{
            f \notin \svis(\sstate_1) \\
            \multiple{\seval{\sstate_1}{e}{t}{\sstate_2}{\scheck_1}} \\
            \multiple{x} = \ffuncparams(f) \\
            s = \sfoot{\ffuncpre(f)}{{\sstate_2}} \\
            \scons{\sstate_2[\senv = \senv(\sstate_2)[\multiple{x \mapsto t}]]}{\ffuncpre(f)}{\sstate_3}{\scheck_2} \\
            \sproduce{\sstate_3}{\ffuncpre(f)}{\sstate_4} \\
            \sheap' = \{ \triple{f}{t_e}{t'} : \triple{f}{t_e}{t} \in \sheap(\sstate_4), \triple{f}{t_e}{t'} \in \sheap(\sstate_2) \} \\
            \oheap' = \set{ \triple{f}{t_e}{t'} : \triple{f}{t_e}{t} \in \oheap(\sstate_4), \triple{f}{t_e}{t'} \in \oheap(\sstate_2) } \\
            \seval{\sstate_4[\sheap = \sheap', \oheap = \oheap', \svis = \svis(\sstate_3) \cup \{ f \}]}{\ffuncbody(f)}{t'}{\sstate_5}{\scheck_3} \\
            f(\multiple{t}, s) = \ffresh \\
            \sstate_6 = \sstate_2[\imp = \imp(\sstate_5), \oheap = \oheap(\sstate_5), \pc = \pc(\sstate_2) \kand (f(\multiple{t}, s) \keq t') \kand \fax{\sstate}{\fts}{\ffuncpost(f)}]
        }{
            \seval{\sstate_1}{\fe}{f(\multiple{t}, s)}{\sstate_6}{\scheck_1 \cup \scheck_2 \cup \scheck_3}
        }
    
        \inferrule[SEvalFunctionImplicit]{
            f \in \svis(\sstate_1) \\
            \multiple{\seval{\sstate_1}{e}{t}{\sstate_2}{\scheck_1}} \\
            \multiple{x} = \ffuncparams(f) \\
            s = \sfoot{\ffuncpre(f)}{{\sstate_2}} \\
            s' = \sfoot{\ffuncpost(f)}{\sstate_2} \\
            \oheap' = \oheap(\sstate_2) \cup (s' \setminus \sheap(\sstate_2)) \\
            \scons{\sstate_2[\senv = \senv(\sstate_2)[\multiple{x \mapsto t}]]}{\ffuncpre(f)}{\sstate_3}{\scheck_2} \\
            f(\multiple{t}, s) = \ffresh \\
            \sstate_4 = \sstate_2[\imp = \imp(\sstate_3), \oheap = \oheap', \pc = \pc(\sstate_2) \kand \fax{\sstate}{\fts}{\ffuncpost(f)}] \\
            \scheck_3 = \set{\pair{t_e}{f} : \triple{f}{t_e}{t} \in s'}
        }{
            \seval{\sstate_1}{\fe}{f(\multiple{t}, s)}{\sstate_4}{\scheck_1 \cup \scheck_2 \cup \scheck_3}
        }
    \end{mathpar}
    }
    
    \caption{Symbolic evaluation rules for pure function applications in Gradual Viper.}
    \label{fig:gv-pure}
\end{figure}

\subsection{Dynamic Semantics}

Because proving that static verification is sound also requires providing semantics for execution and because our static verification scheme results in run-time checks, specifying the execution semantics is necessary to formalizing implementation. Consequently, our extensions to the dynamic semantics (Figures ~\ref{fig:dynamic-rules-unfolding} and \ref{fig:dynamic-rules-func}) govern the evaluation, the framing checks and the footprints of unfolding expressions and pure function applications.

The \textsc{EvalUnfolding} rule, which describes the runtime evaluation of an unfolding expression, evaluates the inner expression without using the predicate instance. Because \gco uses equirecursive semantics at run time, it is not necessary to unfold the predicate to gain ownership of additional fields while evaluating the expression as the predicate instance will already be fully unrolled and the fields will already be accessible. 

In addition to the evaluation rule, the dynamic semantics also include the framing rule and the definition of the exact footprint for unfolding expressions. The framing rule, \textsc{FrameUnfolding}, gives the conditions where the evaluation of an unfolding expression is well-defined; similarly to framing rules for other types of assertions, this checks that all the requisite memory locations are accessible. However, this also performs an additional check that the predicate assertion holds, which poses significant challenges in proving soundness since other framing checks rely solely on ownership. The exact footprint defines the set of memory locations accessed to assert and frame the unfolding expression; because the assertion relies on evaluating the expression and framing relies on the predicate assertion, the exact footprint contains both the exact footprint of the predicate instance and the exact footprint of the inner expression.

Similarly for pure functions, the \textsc{EvalFunction} rule directly evaluates the body. At run time, it isn't necessary to annotate the function body with the same implicit dependency on the heap. The framing rule for pure functions relies on the precondition holding, similarly to how unfolding expressions rely on the predicate assertion holding, and the exact footprint for pure function applications depends on the arguments and the precondition to represent the minimal set of locations the function may access. However, the maximal footprint differs depending on whether the precondition is imprecise; if the precondition is imprecise, then the maximal footprint is the entire set of permissions since it may implicitly assume access to new fields.

\begin{figure}
    {\footnotesize\disableTttResize
     \begin{mathpar}
        \inferrule[EvalUnfolding]{
            \eval{\heap}{\env}{e_0}{v}
        }{
            \eval{\heap}{\env}{\sunfolding{p(\multiple{e})}{e_0}}{v}
        }

        \inferrule[FrameUnfolding]{
            \ifrm{\heap}{\perms}{\env}{p(\multiple{e})} \\
            \assertion{\heap}{\perms}{\env}{p(\multiple{e})} \\
            \frm{\heap}{\perms}{\env}{e_0}
        }{
            \frm{\heap}{\perms}{\env}{\sunfolding{p(\multiple{e})}{e_0}}
        }
            
        \efoot{\heap}{\env}{\sunfolding{p(\multiple{e})}{e_0}} = \efoot{\heap}{\env}{p(\multiple{e})} \cup \efoot{\heap}{\env}{e_0}

     \end{mathpar}
    }
    \caption{Dynamic evaluation and framing rules for unfolding expressions}
    \label{fig:dynamic-rules-unfolding}
\end{figure}

\begin{figure}
    {\footnotesize\disableTttResize
     \begin{mathpar}
        \inferrule[EvalFunction]{
        \multiple{x} = \fparams(f) \\
        \multiple{\eval{\heap}{\env}{e}{v}} \\
        \env' = [\multiple{x \mapsto v}] \\
        \eval{\heap}{\env'}{\fbody(f)}{v'}
    }{
        \eval{\heap}{\env}{\fe}{v'}
    }

    \inferrule[FrameFunction]{
        \ifrm{\heap}{\perms}{\env}{\fpre(f)} \\
        \assertion{\heap}{\perms}{\env}{\fpre(f)} \\
        \multiple{\frm{\heap}{\perms}{\env}{e}}
    }{
        \frm{\heap}{\perms}{\env}{\fe}
    }

    \efoot{\heap}{\env}{\fe} = \bigcup \multiple{\efoot{\heap}{\env}{e}} \cup \efoot{\heap}{[\multiple{x \mapsto v}]}{\ffuncpre(f)} \cup \efoot{\heap}{[\multiple{x \mapsto v}]}{\ffuncbody(f)} \\ \text{where}~ \multiple{x} = \fparams(f) ~\text{and}~ \multiple{\eval{\heap}{\env}{e}{v}}
     \end{mathpar}
    }
    \caption{Dynamic evaluation and framing rules for pure functions}
    \label{fig:dynamic-rules-func}
\end{figure}

\subsection{Correspondence}

To prove soundness, it is necessary to specify how states of the verifier correspond to states at run time and how run-time checks behave dynamically. Mapping between verification states and dynamic states requires mapping between symbolic and concrete values. The primary mechanism that \citet{zimmerman2024sound} use to do so is \emph{corresponding valuations}. This valuation is a partial function mapping symbolic values to concrete values for a specific program trace, starting from the valuation function defined by Khoo et al. \cite{khoo2010mixing}.  The initial valuation function is written as $V$, and the corresponding valuation after evaluating an expression under dynamic heap $\heap$ and environment $\env$ is written as $V[\seval{\sstate}{e}{t}{\sstate'}{\scheck} \mid \heap, \env]$. This expands upon the initial valuation by mapping additional symbolic values to concrete values based on the result of evaluation. By chaining these together, we can obtain a valuation corresponding to a full program trace with all the resultant mappings between symbolic and concrete values.

For a dynamic state to model a static state $\sstate$ under a certain valuation $V$, the dynamic environment $\env$ must model the symbolic store $\senv(\sstate)$ under $V$ and the dynamic heap $\heap$ and permission set $\perms$ must model the precise heap $\sheap$ and optimistic heap $\oheap$ under $V$. This is notated as $\simstate{V}{\sstate}{\heap}{\perms}{\env}$.

To define the corresponding valuation for \textsc{SEvalUnfoldingPrecise}, \textsc{SEvalUnfoldingImpreciseA}, and \textsc{SEvalUnfoldingImpreciseB}, we simply follow the steps in the symbolic evaluation rule because symbolic and concrete evaluation proceed in lockstep. For \textsc{SEvalUnfoldingImplicitImprecise}, the same approach wouldn't work; since the expression isn't actually being evaluated statically, we rely on the result of run time evaluation instead (see the appendix for details). The cases for functions are similar.

Additionally, our design of unfolding expressions for Gradual Viper changed the definition of the optimistic heap so that it is able to hold both predicate instances and heap locations rather than just heap locations. Accordingly, when formalizing these changes, it was necessary to change the definition of correspondence between the optimistic heap and the dynamic heap to assert that each of the predicate instances holds in the corresponding dynamic state. This fits alongside the previous definition's requirements that field values align with those in the dynamic heap and that each field is included in the permission set of the dynamic state. These updated definitions are in the appendix. 

\section{Soundness}

To show that the verifier only accepts programs which meet their specifications, we expand upon the proof of soundness provided by Zimmerman et al. \cite{zimmerman2024sound}, who derive and prove a formulation of soundness from the progress-preservation formulation commonly used in type systems \cite{wright1994syntactic}.

In these theorems, we use the following notation:
\begin{itemize}
    \item A \emph{dynamic state} $\dstate$ is either the abstract symbol $\initsym$, the abstract symbol $\finalsym$, or a pair $\pair{\heap}{\stack}$ of a heap $\heap$ and a stack $\stack$.
    
    \item A \emph{verification state} $\vstate$ is either the abstract symbol $\initsym$, the abstract symbol $\finalsym$, or a triple $\triple{\sstate}{s}{\gform}$ containing a symbolic state, a statement to execute, and a post-condition to be proven after execution.
    
    \item The \emph{guard} judgement $\sguard{\vstate}{\sstate'}{\scheck}{\sperms}$ determines that run-time check set $\scheck$ must be satisfied at state $\sstate'$ with exclusion frame $\sperms$ to continue symbolic execution.
    
    \item The \textit{run-time assertion} judgement $\rtassert{V'}{\heap}{\perms}{\scheck}$ represents that each run-time check in $\scheck$ succeeds under heap $\heap$, permission set $\perms$, and valuation $V'$.

    \item The set of permissions in heap $\heap$ necessary to satisfy symbolic permission set $\sperms$ is given by $\vfoot{V}{\heap}{\sperms}$.
\end{itemize}

The definition of progress is split into two theorems: first, dynamic execution will continue when the run-time checks corresponding to the current path through the program are satisfied, and second, there exists a set of run-time checks corresponding to any valid path condition. Together, these imply that dynamic execution will only halt if a run-time check fails.

\begin{theorem}[Progress, part 1]
  Let $\Gamma$ be some dynamic state validated by $\vstate$ and valuation $V$. If $\sguard{\vstate}{\sstate'}{\scheck}{\sperms}$ with corresponding valuation $V'$ extending $V$, $V'(\pc(\sstate')) = \ktrue$, and $\pair{\heap}{\perms(\Gamma)} \vdash_{V'} \scheck$ then
  $$\dtrans{\prog}{\vfoot{V'}{\heap(\Gamma)}{\sperms}}{\Gamma}{\Gamma'}$$
  for some $\Gamma'$.
\end{theorem}

\begin{theorem}[Progress, part 2]
  Let $\Gamma$ be some well-formed dynamic state validated by $\vstate$ and valuation $V$. Then if $\Gamma \ne \finalsym$ and $\rtassert{V'}{\heap(\Gamma)}{\perms}{\scheck}$,
  $$\vstate \rightharpoonup \sstate', \scheck, \sperms$$
  for some $\sstate'$, $\scheck$, $\sperms$ such that $V'(\pc(\sstate')) = \ktrue$ where $V'$ is the corresponding valuation extending $V'$.
\end{theorem}

The definition of preservation is similar to the original formulation: when there is a dynamic execution step to be taken from a valid dynamic state, the resulting dynamic state will also be valid. 

\begin{theorem}[Preservation]
  Let $\pair{\heap}{\stack}$ be some dynamic state validated by $\vstate$ and valuation $V$ for some program $\prog$. If $\sguard{\vstate}{\sstate'}{\scheck}{\sperms}$ with $V' = V[\sguard{\vstate}{\sstate'}{\scheck}{\sperms} \mid \heap]$, $V'(\pc(\sstate')) = \ktrue$, $\rtassert{V'}{\heap}{\perms(\stack)}{\scheck}$, and $\dexec{\heap}{\stack}{\vfoot{V'}{\heap}{\sperms}}{\heap'}{\stack'}$
  then $\Gamma' = \pair{\heap'}{\stack'}$ is a valid state.
\end{theorem}

The full proof of soundness is structured as a set of lemmas for each of the fundamental operations in the static verifier, culminating in proofs of progress and preservation. Previously, the proof was structured using lemmas for each of the fundamental operations, culminating in the proof of progress and preservation. In this structure, \citet{zimmerman2024extended} first proved lemmas relating to expressions and operations on them such as evaluation and framing. However, incorporating pure functions and unfolding expressions into this framework results in a distinct issue. Previously the syntax forms for expressions didn't rely on formulas at all, instead exclusively using simpler forms like field accesses, binary operations, unary operations, literals, and variables. However, unfolding expressions are a type of expression that rely on predicates, which are formulas, and pure functions have preconditions, which are formulas. Accordingly, expressions that perform induction on expressions now rely on those that perform induction on formulas, which in turn rely on those that perform induction on expressions since formulas include heap-dependent expressions. Similarly, defining symbolic evaluation for these constructs uses the produce and consume judgements, which are defined in terms of symbolic evaluation. 

For example, in the proof the following are written as separate theorems. In prior work, the second of these theorems depended on the first but not vice versa. However, because of how we mix formulas and expressions, they rely on each other.

\begin{theorem*}
    If $\frm{\heap}{\perms}{\env}{e}$ and $\perms \subseteq \perms'$, then $\frm{\heap}{\perms'}{\env}{e}$.
\end{theorem*}

\begin{theorem*}
    If $\efrm{\heap}{\perms}{\env}{\gform}$ and $\perms \subset \perms'$, then $\efrm{\heap}{\perms'}{\env}{\gform}$.
\end{theorem*}

However, to demonstrate the logical soundness of our proof, we can interpret these separate but mutually-recursive theorems as the following single theorem.

\begin{theorem*}
    Let $x \in \Expr \cup \GFormula$ and $\perms \subseteq \perms'$. Then, if $x \in \Expr$ and if $\frm{\heap}{\perms}{\env}{x}$, then $\frm{\heap}{\perms'}{\env}{x}$. If $x \in \GFormula$ and if $\efrm{\heap}{\perms}{\env}{x}$, then $\efrm{\heap}{\perms'}{\env}{x}$.
\end{theorem*}

Conceptually, this theorem can be proved by induction on the size of the term.  Our formalization uses this intuition but concretely realizes it via mutual induction over the previous decoupled theorems.

\section{Related Work}

Gradual verification was originally inspired by gradual typing \cite{siek2007gradual}, where type annotations are optional and missing annotations are backed by checks at run time. Accordingly, previous formulations of gradual verification \cite{bader2018gradual, wise2020gradual} draw on work from gradual typing such as the Abstracting Gradual Typing methodology from Garcia et al. \cite{garcia2016abstracting}.

The original formulation of gradual verification did not have any support for unfolding expressions; Wise et al.'s extension \cite{wise2020gradual} supports \emph{unfolding formulae}, which are more restricted than unfolding expressions since they can only appear in specifications and predicate bodies. Similarly, although Mutlu \cite{mutlu2025expanding} designed an approach for handling pure functions in gradual verification based on the approach in Silicon \cite{schwerhoff2016advancing}, this work lacks a proof of soundness. However, our design is instead inspired by that of Smans et al. \cite{smans2010heap} for its higher level of abstraction. Gradual verification systems without these features using weakest liberal precondition reasoning have been previously proven sound \cite{wise2020gradual} and mechanized in Rocq \cite{bader2018gradual}. However, these proofs of soundness are not applicable to Gradual Viper verifier because it is built on top of the Viper symbolic execution backend.

We found that unfolding expressions and pure functions are used extensively in real-world code bases verified using symbolic execution verifiers, often in conjunction -- for example, in the router for the VerifiedSCION architecture and the verified Go standard library, which are both verified using Gobra \cite{wolf2021gobra}. Our work is therefore a critical building block for extending gradual verification to support verifying realistic programs written in mainstream programming languages.

Many static verifiers based on symbolic execution, such as Chalice \cite{leino2009verification}, lack formal semantics and proofs of soundness. A notable exception to this is VeriFast, which is a static verifier for C, Rust, and Java based on symbolic execution \cite{jacobs2011verifast}. However, VeriFast is based on separation logic rather than implicit dynamic frames, meaning heap ownership and contents are specified simultaneously and unfolding expressions are unnecessary. Another notable exception is Viper; Dardinier et al. \cite{dardinier2025formal} proved the soundness of the Viper verification infrastructure by providing a formal framework for proving the soundness of translational verifiers based on separation logic that can use more than one verification backend and applying this framework to Viper. However, this only addresses a subset of Viper's features not including unfolding expressions or pure functions. Similarly, Parthasarathy et al. \cite{parthasarathy2004towards} provide semantics for a subset of Viper's features to automatically validate translations into the Viper IVL; however, this subset again does not include unfolding expressions or pure functions and their method focuses on the Boogie verification condition generation backend rather than the symbolic execution backend for Viper.

Previous work by Summers and Drossopoulou has connected equirecursive and isorecursive interpretations of predicates in the context of implicit dynamic frames \cite{summers2013formal}, but treated the dichotomy as a result of the difference between theory and implementation. In contrast, Wise et al. \cite{wise2020gradual} connected isorecursive and equirecursive interpretations in the context of gradual verification because of how it combines static and dynamic verification and the connection between symbolic and concete values.

\section{Future Work}

One limitation of our work is its handling of recursive pure functions and recursive unfolding expressions, resulting in incompleteness. Formalizing the behavior of the Silicon symbolic execution verifier, which better handles recursion for pure functions using SMT axioms to improve quantifier triggering strategies, would result in a more complete treatment of pure functions that is closer to the implementation.

Several other features that are available in the Viper verification backend have not yet been implemented or formalized for gradual verification. For example, fractional permissions have been implemented \cite{liu2024design} but not formalized, and quantified permissions have been neither implemented nor formalized. Adding these features to the gradual verification ecosystem expands the range of programs that can be verified; for example, adding quantified permissions enables verifications of programs that use data structures such as arrays or graphs.

Although the underlying verifications primitives from Viper are sufficiently expressive as part of an intermediate verification language, directly exposing these primitives in the \gco front-end is not a suitable level of abstraction for realistic usable verification. To make gradual verification more practical, we have the long-term goal of incorporating it into a verifier for a language with a substructural type system, such as Prusti for Rust \cite{astrauskas2019leveraging} or the SnaKt uniqueness type system and verifier for Kotlin \cite{protopapa}. Substructural type systems such as Rust's type system provide information that allow the verifier to automatically generate auxiliary specifications, such as folding and unfolding predicates, reducing the specification burden on the developer \cite{astrauskas2019leveraging}. Therefore, introducing more features from Viper to the gradual verification ecosystem ensures that we have a sound and practical foundation with the tools needed for such an extension.

Although the initial semantics and properties of gradual verification were mechanized in Rocq, these are no longer applicable to \gco due to the use of symbolic execution rather than weakest liberal preconditions. The semantics and proof of soundness for \gco have not been formalized using a proof assistant. However, Dardinier et al. \cite{dardinier2025formal} provide a formal framework for proving the soundness of translational separation logic verifiers mechanized in the Isabelle/HOL proof assistant. Since this framework accounts for using different verification back-ends and different front-end languages, this could be expanded to include gradual verification using the gradual version of the Silicon symbolic execution backend and the gradual version of the Silver intermediate representation language for Viper. Incorporating the formalization of gradual verification into the mechanized proof could make it easier to expand upon the formal system for gradual verification, discover possible errors in the proofs, and verify that the implementation corresponds to the formal system.

\section{Conclusion}

In this paper, we provide the formal semantics for the evaluation of unfolding expressions and pure function applications in gradual verification by extending the formalization provided by Zimmerman et al. \cite{zimmerman2024sound}, as well as an extension to the proof of soundness to include the semantics for unfolding expressions and pure functions. 

By formalizing and proving the soundness of gradually-verified unfolding expressions and pure function applications, we demonstrate the soundness of both \gco and the underlying symbolic execution backend for Viper. Overall, introducing unfolding expressions and pure functions to the gradual verification ecosystem ensures that gradual verification becomes more practical for verifying real-world programs while still remaining sound.

\begin{acks}
    We thank Conrad Zimmerman and Chanhee Cho for their valuable feedback. Additionally, we thank Doruk Alp Mutlu and Craig Liu for their valuable contributions to the gradual verification ecosystem.
    
    This work was supported by the
    \grantsponsor{nsa}{National Security Agency}{https://www.nsa.gov/} under the \href{https://sos-vo.org/projects/continuous-reasoning-gradual-verification}{Continuous Reasoning with Gradual Verification} grant, by the \grantsponsor{nsf}{National Science Foundation}{https://www.nsf.gov/} under Grant No. \grantnum[https://www.nsf.gov/awardsearch/show-award?AWD_ID=2447499]{NSF}{CCF-2447499}, and by \grantsponsor{meta}{Meta}{https://www.meta.com/} through the \href{https://www.cmu.edu/scs/s3d/reuse/}{Research Experiences for Undergraduates in Software Engineering} program at Carnegie Mellon University. Any opinions, findings, conclusions, or recommendations are those of the authors and do not necessarily reflect the views of the National Security Agency, the National Science Foundation, or Meta.
\end{acks}

\bibliographystyle{ACM-Reference-Format}
\bibliography{main}

\clearpage
\appendix

\renewcommand*\contentsname{Appendices}

\tableofcontents
\addtocontents{toc}{\protect\setcounter{tocdepth}{2}}
\setcounter{theorem}{0}

\section{Grammar}\label{sec:grammar}

In the following appendices, our changes from Zimmerman et al. \cite{zimmerman2024extended} are \sethlcolor{newstuff}\hl{highlighted light gray}.
\begin{grammar}
  \firstcase
    {\gprogram}
    {\multiple{\gstruct} ~ \multiple{\gpredicate} ~ \multiple{\gfunc} ~ \multiple{\gmethod} ~ \gstatement}
    {Program definition}

  \firstcase
    {\gstruct}
    {S~ \sblock{\multiple{\gtype ~ f}}}
    {Struct definition}

  \firstcase
    {\gpredicate}
    {p({\multiple{\gtype ~\gvar}}) = \gform}
    {Predicate definition}

  \firstcase
    {\gfunc}
    {T~g(\multiple{T~x})~\gfcontract~e}
    {\sethlcolor{newstuff}\hl{Pure function definition}}

  \firstcase
    {\gfcontract}
    {\krequires ~ \gform ~ \kensures ~ e}
    {\sethlcolor{newstuff}\hl{Pure function contract}}

  \firstcase
    {\gmethod}
    {\smethdef{T}{m}{\multiple{T ~ x}}{\gcontract}{s}}
    {Method definition}
  
  \firstcase
    {\gcontract}
    {\krequires ~ \gform ~ \kensures ~ \gform}
    {Method contract}

  \firstcase
    {\gtype}
    {S \gralt \kint \gralt \kbool \gralt \kchar}
    {Type}

  \firstcase
    {\gstatement}
    {\sseq{\gstatement}{\gstatement}}
    {Statement sequence}
  \otherform
    {\kskip}
    {No-op}
  \otherform
    {x \kassign e}
    {Variable assignment}
  \otherform
    {x.f \kassign e}
    {Field assignment}
  \otherform
    {x \kassign \salloc{S}}
    {Allocation}
  \otherform
    {x \kassign m({\multiple{e}})}
    {Method invocation}
  \otherform
    {\sassert{\gform}}
    {Assertion}
  \otherform
    {\sif{e}{s}{s}}
    {Conditional}
  \otherform
    {\swhile{e}{\gform}{s}}
    {Loop}
  \otherform
    {\sfold{p(\multiple{e})}}
    {Fold predicate}
  \otherform
    {\sunfold{p(\multiple{e})}}
    {Unfold predicate}

  \firstcase
    {\gexpression}
    {l \gralt x \gralt e.f \gralt e \oplus e}
    {Expression}
  \otherform
    {e \kor e \gralt e \kand e \gralt \kneg e}
    {}
  \otherform
    {\sunfolding{p(\multiple{e})}{e}}
    {\sethlcolor{newstuff}\hl{Unfolding expression}}
  \otherform
    {g(\multiple{e})}
    {\sethlcolor{newstuff}\hl{Pure function call}}

  \firstcase
    {\gvar}
    {\kresult \gralt id}
    {Variable}

  \firstcase
    {\textit{l}}
    {n \gralt c}
    {Value}
  \otherform{\knull \gralt \ktrue \gralt \kfalse}{}

  \firstcase
    {\gform}
    {\simprecise{\phi} \gralt \phi}
    {Gradual formula}

  \firstcase
    {\phi}
    {\phi * \phi \gralt p(\multiple{e}) \gralt e}
    {Precise formula}
    \otherform
    {\sif{e}{\phi}{\phi}}
    {}
  \otherform
    {\kacc(e.f)}
    {}
\end{grammar}

Where $n \in \mathbb{Z}$, $c \in \textsc{Char}$, $id \in \textsc{Identifier}$, $f \in \Field$, $m \in \Method$, $p \in \Predicate$, $g \in \textsc{Function}$, $\oplus \in \{ +, -, <, >, \le, \ge, = \}$

\begin{definition}\label{def:well-formed-prog}
  A program is \textbf{well-formed} if all the following requirements are satisified:
  \begin{itemize}
    \item It is properly typed.
    \item All loop invariants, method pre-conditions, function pre-conditions, and method post-conditions are specifications (definition \ref{def:specification}).
    \item The free variables of any method are a subset of its parameters.
    \item The special variable $\kresult$ is not a free variable of any method.
    \item No parameters appear on the left side of a variable assignment.
    \item Formulas in pre-conditions only reference parameters.
    \item Formulas in post-conditions only reference parameters and the special variable $\kresult$.
    \item If a pre-condition is imprecise, the post-condition is also imprecise.
    \item Function bodies and post-conditions are framed (section~\ref{sec:framing}) by their pre-conditions and function post-conditions can be proven to hold for the result.
  \end{itemize}
\end{definition}

\section{Run-time Semantics}\label{sec:dynamic-exec}

\subsection{Definitions}

The rules in the following section reference an ambient program with elements denoted as follows:
\begin{itemize}
  \item Structs: $S \in \Struct$
  \item Predicates: $p \in \Predicate$
  \definecolor{shadecolor}{gray}{0.9}
  \begin{snugshade}
  \item Functions: $g \in \textsc{Function}$
  \end{snugshade}
  \item Methods: $m \in \Method$
  \item Types: $T \in \Type$
  \item Variables: $x \in \Var$
  \item Field identifiers: $f \in \Field$
  \item Locations (opaque values): $\ell \in \Location$
  \item Literals (integers, characters, booleans, null): $l \in \Literal$
  \item Values: $v \in \Value = \Location \cup \Literal$
  \item Gradual formulas: $\gform \in \GFormula$
  \item Precise formulas: $\phi \in \Formula$
  \item Statements: $s \in \Stmt$
  \item Heap: $\heap : \Location \times \Field \to \Value$
  \item Permissions: $\perms \in \powerset{\Location \times \Field}$
  \item Environment: $\env : \Var \pfunc \Value$
\end{itemize}

The following functions are defined to access elements in the program:
\begin{itemize}
  \item $\fdefault : \Type \to \Value$ -- Gets the default value of the given type ($0$, $\knull$, etc.)
  \item $\fstruct : \Struct \to \multiple{\Field}$ -- Gets the list of fields from the declaration of the specified struct.
  \item $\fpre : \Method \to \GFormula$ -- Gets the precondition from the declaration of the specified method.
  \item $\fpost : \Method \to \GFormula$ -- Gets the post-condition from the declaration of the specified method.
  \item $\fbody : \Method \to \Stmt$ -- Gets the body from the declaration of the specified method.
  \item $\fparams : \Method \to \multiple{\Var}$ -- Gets the list of parameters from the declaration of the specified predicate.
  \item $\fpred : \Predicate \to \GFormula$ -- Gets the body from the declaration of the specified predicate.
  \item $\fpredparams : \Predicate \to \multiple{\Var}$ -- Gets the list of parameters from the declaration of the specified predicate.
  \definecolor{shadecolor}{gray}{0.9}
  \begin{snugshade}
  \item $\ffuncpre : \Function \to \GFormula$ -- Gets the precondition from the declaration of the specified pure function.
  \item $\ffuncpost : \Function \to \Expr$ -- Gets the post-condition, which is an expression rather than a formula, from the declaration of the specified pure function.
  \item $\ffuncbody : \Function \to \Expr$ -- Gets the body from the declaration of the specified pure function.
  \item $\ffuncparams : \Function to \multiple{\Var}$ -- Gets the list of parameters from the declaration of the specified pure function.
  \end{snugshade}
\end{itemize}

\subsection{Evaluation}

The relation
$$\eval{\heap}{\env}{e}{v}$$
denotes the evaluation of an expression $e \in \Expr$ to a value $v \in \Value$
where $\heap : \Value \times \Field \to \Value$ represents the heap, and $\env : \Var \pfunc \Value$ represents the local variable environment.

\semantics[EvalLiteral]
  { }
  {\eval{\heap}{\env}{l}{l}}
\semantics[EvalVar]
  { }
  {\eval{\heap}{\env}{x}{\env(x)}}
\semantics[EvalAndA]
  {\eval{\heap}{\env}{e_1}{\kfalse}}
  {\eval{\heap}{\env}{e_1 \kand e_2}{\kfalse}}
\semantics[EvalAndB]
  {\eval{\heap}{\env}{e_1}{\ktrue} \\ \eval{\heap}{\env}{e_2}{v_2}}
  {\eval{\heap}{\env}{e_1 \kand e_2}{v_2}}
\semantics[EvalOrA]
  {\eval{\heap}{\env}{e_1}{\ktrue}}
  {\eval{\heap}{\env}{e_1 \kor e_2}{\ktrue}}
\semantics[EvalOrB]
  {\eval{\heap}{\env}{e_1}{\kfalse} \\ \eval{\heap}{\env}{e_2}{v_2}}
  {\eval{\heap}{\env}{e_1 \kor e_2}{v_2}}
\semantics[EvalOp]
  {\eval{\heap}{\env}{e_1}{v_1} \\ \eval{\heap}{\env}{e_2}{v_2}}
  {\eval{\heap}{\env}{e_1 \oplus e_2}{v_1 \oplus v_2}}
\semantics[EvalNeg]
  {\eval{\heap}{\env}{e}{v}}
  {\eval{\heap}{\env}{\kneg e}{\neg v}}
\semantics[EvalField]
  {\eval{\heap}{\env}{e}{\ell}}
  {\eval{\heap}{\env}{e.f}{\heap(\ell, f)}}
\definecolor{shadecolor}{gray}{0.9}
\begin{snugshade}
\semantics[EvalUnfolding]
    {\eval{\heap}{\env}{e_0}{v}}
    {\eval{\heap}{\env}{\sunfolding{p(\multiple{e})}{e_0}}{v}}
\semantics[EvalFunction]
    {
        \multiple{x} = \fparams(f) \\
        \multiple{\eval{\heap}{\env}{e}{v}} \\
        \env' = [\multiple{x \mapsto v}] \\
        \eval{\heap}{\env'}{\fbody(f)}{v'}
    }
    {
        \eval{\heap}{\env}{\fe}{v'}
    }
\end{snugshade}

\subsection{Formulas}

$\Formula$ is the set of all $\gpreciseform$ elements in the grammar, while $\GFormula$ is the set of all $\gform$ elements in the grammar.

\begin{definition}\label{def:precise}
  An \textbf{imprecise} formula $\gform$ is any formula in $\GFormula$ of the form $\simprecise{\phi}$ where $\phi \in \Formula$.
  
  Otherwise, a formula $\phi$ is \textbf{precise} and $\phi \in \Formula$.
\end{definition}

\begin{definition}\label{def:precise-formula}
  A formula $\gform$ is \textbf{completely precise} if there is no $\heap, \perms, \env$ such that \textnormal{\refrule{AssertImprecise}} applies at some step in the derivation of $\assertion{\heap}{\perms}{\env}{\gform}$.

  In other words, a completely precise formula is precise and all predicate bodies referenced in its equi-recursive unrolling are also precise.
\end{definition}

\begin{definition}\label{def:specification}
  A formula $\gform$ is a \textbf{specification} if either $\gform$ is imprecise or $\gform$ is precise and self-framed (definition \ref{def:self-framed}).
\end{definition}

\subsection{Footprints}

\begin{definition}\label{def:efoot}
  The \textbf{exact footprint} of a formula $\gform \in \GFormula$, denoted $\efoot{\heap}{\env}{\gform}$, or of an expression $e$, denoted $\efoot{\heap}{\env}{e}$, is the set of permissions that must be accessed when asserting $\gform$ or evaluating $e$.

  By lemmas \ref{lem:efoot-subset-spec} and \ref{lem:efoot-assert}, if $\gform$ is a specification, this set is the lower bound of permissions that satisfy $\gform$.
\end{definition}

The calculation of exact footprints is defined as follows:
\begin{align*}
  \efoot{\heap}{\env}{l} &:= \emptyset \\
  \efoot{\heap}{\env}{x} &:= \emptyset \\
  \efoot{\heap}{\env}{e.f} &:= \efoot{\heap}{\env}{e}; \pair{\ell}{f} &\text{if }\eval{\heap}{\env}{e}{\ell} \\
  \efoot{\heap}{\env}{e_1 \oplus e_2} &:= \efoot{\heap}{\env}{e_1} \cup \efoot{\heap}{\env}{e_2} \\
  \efoot{\heap}{\env}{e_1 \kor e_2} &:= \efoot{\heap}{\env}{e_1} &\text{if }\eval{\heap}{\env}{e_1}{\ktrue} \\
  \efoot{\heap}{\env}{e_1 \kor e_2} &:= \efoot{\heap}{\env}{e_1} \cup \efoot{\heap}{\env}{e_2} &\text{if }\eval{\heap}{\env}{e_1}{\kfalse} \\
  \efoot{\heap}{\env}{e_1 \kand e_2} &:= \efoot{\heap}{\env}{e_1} &\text{if }\eval{\heap}{\env}{e_1}{\kfalse} \\
  \efoot{\heap}{\env}{e_1 \kand e_2} &:= \efoot{\heap}{\env}{e_1} \cup \efoot{\heap}{\env}{e_2} &\text{if }\eval{\heap}{\env}{e_1}{\ktrue} \\
  \efoot{\heap}{\env}{\kneg e} &:= \efoot{\heap}{\env}{e} \\
  \efoot{\heap}{\env}{\simprecise{\phi}} &:= \efoot{\heap}{\env}{\phi} \\
  \efoot{\heap}{\env}{\phi_1 * \phi_2} &:= \efoot{\heap}{\env}{\phi_1} \cup \efoot{\heap}{\env}{\phi_2} \\
  \efoot{\heap}{\env}{p(\multiple{e})} &:= \efoot{\heap}{[\multiple{x \mapsto v}]}{\fpred(p)} ~\cup &\text{if } \multiple{x} = \fpredparams(p) \\
  &\quad\quad \bigcup \multiple{\efoot{\heap}{\env}{e}} &\text{and }\multiple{\eval{\heap}{\env}{e}{v}} \\
  \efoot{\heap}{\env}{\sif{e}{\phi_1}{\phi_2}} &:= \efoot{\heap}{\env}{e} \cup \efoot{\heap}{\env}{\phi_1} &\text{if }\eval{\heap}{\env}{e_1}{\ktrue} \\
  \efoot{\heap}{\env}{\sif{e}{\phi_1}{\phi_2}} &:= \efoot{\heap}{\env}{e} \cup \efoot{\heap}{\env}{\phi_2} &\text{if }\eval{\heap}{\env}{e_1}{\kfalse} \\
  \efoot{\heap}{\env}{\kacc(e.f)} &:= \efoot{\heap}{\env}{e}; \pair{\ell}{f} &\text{if }\eval{\heap}{\env}{e}{\ell}
\end{align*}
\definecolor{shadecolor}{gray}{0.9}
\begin{snugshade}
\begin{align*}
  \efoot{\heap}{\env}{\sunfolding{\pe}{e_0}} &:= \efoot{\heap}{\env}{\pe} \cup \efoot{\heap}{\env}{e_0} \\
  \efoot{\heap}{\env}{\fe} &:= \bigcup \multiple{\efoot{\heap}{\env}{e}} \cup &\text{if } \multiple{x} = \ffuncparams(f) \\
  &\quad\quad \efoot{\heap}{[\multiple{x \mapsto v}]}{\ffuncpre(f)} \cup &\text{and } \multiple{\eval{\heap}{\env}{e}{v}} \\
  &\quad\quad \efoot{\heap}{[\multiple{x \mapsto v}]}{\ffuncbody(f)}
\end{align*}
\end{snugshade}

\begin{definition}\label{def:footprint}
  The \textbf{maximal footprint} of a formula, denoted $\foot{\heap}{\perms}{\env}{\gform}$, is the set of all permissions that $\gform$ may represent in the context of a heap $\heap$, permission set $\perms$, and variable environment $\env$.

  The footprint of a completely precise formula (definition \ref{def:precise-formula}) is its exact footprint, while the footprint of a formula which is not completely precise is the current set of permissions.

  \begin{equation*}
    \foot{\heap}{\perms}{\env}{\gform} := \begin{cases}
      \efoot{\heap}{\env}{\gform} & \text{if $\gform$ is completely precise} \\
      \perms & \text{otherwise}
    \end{cases}
  \end{equation*}
\end{definition}

\subsection{Framing}\label{sec:framing}

The relation $\frm{\heap}{\perms}{\env}{e}$ denotes that $e \in \Expr$ is framed by the permissions contained in $\perms \in \powerset{\Perm}$.

\semantics[FrameLiteral]
  { }
  {\frm{\heap}{\perms}{\env}{l}}
\semantics[FrameVar]
  { }
  {\frm{\heap}{\perms}{\env}{x}}
\semantics[FrameField]
  {
    \frm{\heap}{\perms}{\env}{e} \\
    \assertion{\heap}{\perms}{\env}{\kacc(e.f)}
  }
  {
    \frm{\heap}{\perms}{\env}{e.f}
  }
\semantics[FrameOp]
  {
    \frm{\heap}{\perms}{\env}{e_1} \\
    \frm{\heap}{\perms}{\env}{e_2}
  }
  {
    \frm{\heap}{\perms}{\env}{e_1 \oplus e_2}
  }
\semantics[FrameOrA]
  {
    \eval{\heap}{\env}{e_1}{\ktrue} \\
    \frm{\heap}{\perms}{\env}{e_1}
  }
  {
    \frm{\heap}{\perms}{\env}{e_1 \kor e_2}
  }
\semantics[FrameOrB]
  {
    \eval{\heap}{\env}{e_1}{\kfalse} \\
    \frm{\heap}{\perms}{\env}{e_1} \\
    \frm{\heap}{\perms}{\env}{e_2}
  }
  {
    \frm{\heap}{\perms}{\env}{e_1 \kor e_2}
  }
\semantics[FrameAndA]
  {
    \eval{\heap}{\env}{e_1}{\kfalse} \\
    \frm{\heap}{\perms}{\env}{e_1}
  }
  {
    \frm{\heap}{\perms}{\env}{e_1 \kand e_2}
  }
\semantics[FrameAndB]
  {
    \eval{\heap}{\env}{e_1}{\ktrue} \\
    \frm{\heap}{\perms}{\env}{e_1} \\
    \frm{\heap}{\perms}{\env}{e_2}
  }
  {
    \frm{\heap}{\perms}{\env}{e_1 \kand e_2}
  }
\semantics[FrameNeg]
  {\frm{\heap}{\perms}{\env}{e}}
  {\frm{\heap}{\perms}{\env}{\kneg e}}
\definecolor{shadecolor}{gray}{0.9}
\begin{snugshade}
\semantics[FrameUnfolding]
  {
    \ifrm{\heap}{\perms}{\env}{p(\multiple{e})} \\
    \assertion{\heap}{\perms}{\env}{p(\multiple{e})} \\
    \frm{\heap}{\perms}{\env}{e_0}
  }
  {
    \frm{\heap}{\perms}{\env}{\sunfolding{p(\multiple{e})}{e_0}}
  }
\semantics[FrameFunction]
  {
    \ifrm{\heap}{\perms}{\env}{\fpre(f)} \\
    \assertion{\heap}{\perms}{\env}{\fpre(f)} \\
    \multiple{\frm{\heap}{\perms}{\env}{e}}
  }
  {
    \frm{\heap}{\perms}{\env}{\fe}
  }
\end{snugshade}

The relation $\ifrm{\heap}{\perms}{\env}{\phi}$ denotes that $\phi \in \Formula$ is framed by the permissions in $\perms \in \powerset{\Perm}$ using an iso-recursive interpretation of predicates (i.e., without unrolling predicate instances).

\semantics[IFrameExpression]
  {
    \frm{\heap}{\perms}{\env}{e}
  }
  {
    \ifrm{\heap}{\perms}{\env}{e}
  }
\semantics[IFrameConjunction]
  {
    \ifrm{\heap}{\perms}{\env}{\phi_1} \\
    \ifrm{\heap}{\perms}{\env}{\phi_2}
  }
  {
    \ifrm{\heap}{\perms}{\env}{\phi_1 * \phi_2}
  }
\semantics[IFramePredicate]
  {
    \multiple{\frm{\heap}{\perms}{\env}{e}}
  }
  {
    \ifrm{\heap}{\perms}{\env}{p(\multiple{e})}
  }
\semantics[IFrameConditionalA]
  {
    \eval{\heap}{\env}{e}{\ktrue} \\
    \frm{\heap}{\perms}{\env}{e} \\
    \ifrm{\heap}{\perms}{\env}{\phi_1}
  }
  {
    \ifrm{\heap}{\perms}{\env}{\sif{e}{\phi_1}{\phi_2}}
  }
\semantics[IFrameConditionalB]
  {
    \eval{\heap}{\env}{e}{\kfalse} \\
    \frm{\heap}{\perms}{\env}{e} \\
    \ifrm{\heap}{\perms}{\env}{\phi_2}
  }
  {
    \ifrm{\heap}{\perms}{\env}{\sif{e}{\phi_1}{\phi_2}}
  }
\semantics[IFrameAcc]
  {
    \frm{\heap}{\perms}{\env}{e}
  }
  {
    \ifrm{\heap}{\perms}{\env}{\kacc(e.f)}
  }

Define the relation $\efrm{\heap}{\perms}{\env}{\phi}$  denotes that $\phi \in \Formula$ is framed by the permissions in $\perms \in \powerset{\Perm}$ using an equi-recursive interpretation of predicates (i.e., unrolling predicate instances).

\semantics[EFrameExpression]
  {
    \frm{\heap}{\perms}{\env}{e}
  }
  {
    \efrm{\heap}{\perms}{\env}{e}
  }
\semantics[EFrameConjunction]
  {
    \efrm{\heap}{\perms}{\env}{\phi_1} \\
    \efrm{\heap}{\perms}{\env}{\phi_2}
  }
  {
    \efrm{\heap}{\perms}{\env}{\phi_1 * \phi_2}
  }
\semantics[EFramePredicate]
  {
    \multiple{\frm{\heap}{\perms}{\env}{e}} \\
    \multiple{\eval{\heap}{\env}{e}{v}} \\
    \multiple{x} = \fpredparams(p) \\
    \efrm{\heap}{\perms}{[\multiple{x \mapsto v}]}{\fpred(p)}
  }
  {
    \efrm{\heap}{\perms}{\env}{p(\multiple{e})}
  }
\semantics[EFrameConditionalA]
  {
    \eval{\heap}{\env}{e}{\ktrue} \\
    \frm{\heap}{\perms}{\env}{e} \\
    \efrm{\heap}{\perms}{\env}{\phi_1}
  }
  {
    \efrm{\heap}{\perms}{\env}{\sif{e}{\phi_1}{\phi_2}}
  }
\semantics[EFrameConditionalB]
  {
    \eval{\heap}{\env}{e}{\kfalse} \\
    \frm{\heap}{\perms}{\env}{e} \\
    \efrm{\heap}{\perms}{\env}{\phi_2}
  }
  {
    \efrm{\heap}{\perms}{\env}{\sif{e}{\phi_1}{\phi_2}}
  }
\semantics[EFrameAcc]
  {
    \frm{\heap}{\perms}{\env}{e}
  }
  {
    \efrm{\heap}{\perms}{\env}{\kacc(e.f)}
  }

\begin{definition}\label{def:self-framed}
  A \textbf{self-framed} formula is a precise formula $\phi \in \Formula$ such that for all $\heap, \perms, \env$,
  $$\assertion{\heap}{\perms}{\env}{\phi} \implies \ifrm{\heap}{\perms}{\env}{\phi}.$$
\end{definition}

\subsection{Assertions}

The relation $\assertion{\heap}{\perms}{\env}{\gform}$ denotes the validity of $\gform \in \GFormula$ for the state represented by $\triple{\heap}{\perms}{\env}$.

\semantics[AssertImprecise]
  {
    \assertion{\heap}{\perms}{\env}{\phi} \\
    \efrm{\heap}{\perms}{\env}{\phi}
  }
  {\assertion{\heap}{\perms}{\env}{\simprecise{\phi}}}
\semantics[AssertValue]
  {
    \eval{\heap}{\env}{e}{\ktrue}
  }
  {
    \assertion{\heap}{\perms}{\env}{\gexpression}
  }
\semantics[AssertIfA]
  {
    \eval{\heap}{\env}{e}{\ktrue} \\
    \assertion{\heap}{\perms}{\env}{\phi_1}
  }
  {
    \assertion{\heap}{\perms}{\env}{\sif{e}{\phi_1}{\phi_2}}
  }
\semantics[AssertIfB]
  {
    \eval{\heap}{\env}{e}{\kfalse} \\
    \assertion{\heap}{\perms}{\env}{\phi_2}
  }
  {
    \assertion{\heap}{\perms}{\env}{\sif{e}{\phi_1}{\phi_2}}
  }
\semantics[AssertAcc]
  {
    \eval{\heap}{\env}{e}{\ell} \\
    \pair{\ell}{f} \in \perms
  }
  {
    \assertion{\heap}{\perms}{\env}{\kacc(e.f)}
  }
\semantics[AssertConjunction]
  {
    \assertion{\heap}{\perms_1}{\env}{\phi_1} \\
    \assertion{\heap}{\perms_2}{\env}{\phi_2} \\
    \perms_1 \cup \perms_2 \subseteq \perms \\
    \perms_1 \cap \perms_2 = \emptyset
  }
  {
    \assertion{\heap}{\perms}{\env}{\phi_1 * \phi_2}
  }
\semantics[AssertPredicate]
  {
    \multiple{x} = \fpredparams(p) \\
    \multiple{\eval{\heap}{\env}{e}{v}} \\
    \assertion{\heap}{\perms}{[\multiple{x \mapsto v}]}{\fpred(p)}
  }
  {
    \assertion{\heap}{\perms}{\env}{p(\multiple{e})}
  }

\subsection{Execution}\label{sec:exec-rules}

\begin{definition}
  A \textbf{stack} $\stack$ is a list of the form
  $$\triple{\perms_n}{\env_n}{s_n} \cdot \ldots \cdot \triple{\perms_1}{\env_1}{s_1} \cdot \nilsym$$
  where $n\ge 1$, $\perms_n, \cdots, \perms_1$ are permission sets, $\env_n, \cdots, \env_1$ are variable environments, and $s_n, \cdots, s_1$ are statements.

  $\perms(\stack)$, $\env(\stack)$, and $s(\stack)$ may be used to denote the values $\perms_n$, $\env_n$, and $s_n$, respectively.
\end{definition}

\begin{definition}
  An \textbf{exclusion frame} $\xperms$ a set of permissions that may not be transferred to a callee stack frame. This is necessary to ensure that the permissions represented by the imprecise specifications of a callee cannot overlap with some predicate instance that is owned by the caller.
\end{definition}

Small-step execution is denoted by the judgement
$$\dexec{\heap}{\stack}{\xperms}{\heap'}{\stack'}$$
for stacks $\stack, \stack'$, heap $\heap$, and exclusion frame $\xperms$.

\semantics[ExecSeq]
  { }
  {
    \dexec
      {\heap}{\triple{\perms}{\env}{\sseq{\kskip}{s}} \cdot \stack}
      {\xperms}
      {\heap}{\triple{\perms}{\env}{s} \cdot \stack}
  }
\semantics[ExecAssign]
  {
    \eval{\heap}{\env}{e}{v} \\
    \frm{\heap}{\perms}{\env}{e}
  }
  {
    \dexec
      {\heap}{\triple{\perms}{\env}{\sseq{x = e}{s}}  \cdot \stack}
      {\xperms}
      {\heap}{\triple{\perms}{\env[x \mapsto v]}{s} \cdot \stack}
  }
\semantics[ExecAssignField]
  {
    \eval{\heap}{\env}{x}{\ell} \\
    \eval{\heap}{\env}{e}{v} \\
    \assertion{\heap}{\perms}{\env}{\kacc(x.f)} \\
    \frm{\heap}{\perms}{\env}{e} \\
    \heap' = \heap[\pair{\ell}{f} \mapsto v]
  }
  {
    \dexec
      {\heap}{\triple{\perms}{\env}{\sseq{x.f = e}{s}} \cdot \stack}
      {\xperms}
      {\heap'}{\triple{\perms}{\env}{s} \cdot \stack}
  }
\semantics[ExecAlloc]
  {
    \ell = \ffresh \\
    \multiple{T ~f} = \fstruct(S) \\
    \heap' = \heap[\multiple{\pair{\ell}{f} \mapsto \fdefault(T)}] \\
    \perms' = \perms \cup \set{\multiple{\pair{\ell}{f}}}
  }
  {
    \dexec
      {\heap}{\triple{\perms}{\env}{\sseq{x = \salloc{S}}{s}} \cdot \stack}
      {\xperms}
      {\heap'}{\triple{\perms'}{\env[x \mapsto \ell]}{s} \cdot \stack}
  }
\semantics[ExecCallEnter]
  {
    \multiple{x} = \fparams(m) \\
    \multiple{\eval{\heap}{\env}{e}{v}} \\
    \multiple{\frm{\heap}{\perms}{\env}{e}} \\
    \env' = [\multiple{x \mapsto v}] \\
    \assertion{\heap}{\perms \setminus \xperms}{\env'}{\fpre(m)} \\
    \perms' = \foot{\heap}{\perms \setminus \xperms}{\env'}{\fpre(m)}
  }
  {
    \pair
      {\heap}
      {\triple{\perms}{\env}{\sseq{y \kassign m(\multiple{e})}{s}} \cdot \stack},
    \,\xperms
    \\\\ \to \\\\
    \pair
      {\heap}
      {
        \triple{\perms'}{\env'}{\sseq{\fbody(m)}{\kskip}} \cdot
        \triple{\perms \setminus \perms'}{\env}{\sseq{y \kassign m(\multiple{e})}{s}} \cdot \stack
      }
  }
\semantics[ExecCallExit]
  {
    \assertion{\heap}{\perms'}{\env'}{\fpost(m)} \\
    \env'' = \env[y \mapsto \env'(\kresult)] \\
    \perms'' = \perms \cup \foot{\heap}{\perms'}{\env'}{\fpost(m)}
  }
  {
    \dexec
      {\heap}
      {
        \triple{\perms'}{\env'}{\kskip} \cdot
        \triple{\perms}{\env}{\sseq{y \kassign m(\multiple{e})}{s}} \cdot
        \stack
      }
      {\xperms}
      {\heap}
      {\triple{\perms''}{\env''}{s} \cdot \stack}
  }
\semantics[ExecAssert]
  {
    \assertion{\heap}{\perms}{\env}{\simprecise{\phi}}
  }
  {
    \dexec
      {\heap}{\triple{\perms}{\env}{\sseq{\sassert{\gform}}{s}} \cdot \stack}
      {\xperms}
      {\heap}{\triple{\perms}{\env}{s} \cdot \stack}
  }
\semantics[ExecIfA]
  {
    \eval{\heap}{\env}{e}{\ktrue} \\
    \frm{\heap}{\perms}{\env}{e}
  }
  {
    \dexec
      {\heap}{\triple{\heap}{\env}{\sseq{\sif{e}{s_1}{s_2}}{s}} \cdot \stack}
      {\xperms}
      {\heap}{\triple{\heap}{\env}{\sseq{s_1}{s}} \cdot \stack}
  }
\semantics[ExecIfB]
  {
    \eval{\heap}{\env}{e}{\kfalse} \\
    \frm{\heap}{\perms}{\env}{e}
  }
  {
    \dexec
      {\heap}{\triple{\heap}{\env}{\sseq{\sif{e}{s_1}{s_2}}{s}} \cdot \stack}
      {\xperms}
      {\heap}{\triple{\heap}{\env}{\sseq{s_2}{s}} \cdot \stack}
  }
\semantics[ExecWhileEnter]
  {
    \eval{\heap}{\env}{e}{\ktrue} \\
    \frm{\heap}{\perms}{\env}{e} \\
    \assertion{\heap}{\perms \setminus \xperms}{\env}{\gform} \\
    \perms' = \foot{\heap}{\perms \setminus \xperms}{\env}{\gform}
  }
  {
    \pair
      {\heap}
      {\triple{\perms}{\env}{\sseq{\swhile{e}{\gform}{s'}}{s}} \cdot \stack},
    \,\xperms
    \\\\ \to \\\\
    \pair
      {\heap}
      {\triple{\perms'}{\env}{\sseq{s'}{\kskip}} \cdot \triple{\perms \setminus \perms'}{\env}{\sseq{\swhile{e}{\gform}{s'}}{s}} \cdot \stack}
  }
\semantics[ExecWhileSkip]
  {
    \eval{\heap}{\env}{e}{\kfalse} \\
    \frm{\heap}{\perms}{\env}{e} \\
    \assertion{\heap}{\perms \setminus \xperms}{\env}{\gform}
  }
  {
    \dexec
      {\heap}{\triple{\perms}{\env}{\sseq{\swhile{e}{\gform}{s'}}{s}} \cdot \stack}
      {\xperms}
      {\heap}{\triple{\perms}{\env}{s} \cdot \stack}
  }
\semantics[ExecWhileFinish]
  {
    \assertion{\heap}{\perms'}{\env'}{\gform} \\
    \perms'' = \perms \cup \foot{\heap}{\perms'}{\env'}{\gform}
  }
  {
    \pair
      {\heap}
      {
        \triple{\perms'}{\env'}{\kskip} \cdot
        \triple{\perms}{\env}{\sseq{\swhile{e}{\gform}{s'}}{s}} \cdot
        \stack
      },
    \,\xperms
    \\\\ \to \\\\
    \pair
      {\heap}
      {\triple{\perms''}{\env'}{\sseq{\swhile{e}{\gform}{s'}}{s}} \cdot \stack}
  }
\semantics[ExecFold]
  { }
  {
    \dexec
      {\heap}
      {\triple{\perms}{\env}{\sseq{\sfold{p(\multiple{e})}}{s}} \cdot S}
      {\xperms}
      {\heap}
      {\triple{\perms}{\env}{s} \cdot S}
  }
\semantics[ExecUnfold]
  { }
  {
    \dexec
      {\heap}
      {\triple{\perms}{\env}{\sseq{\sunfold{p(\multiple{e})}}{s}} \cdot S}
      {\xperms}
      {\heap}
      {\triple{\perms}{\env}{s} \cdot S}
  }

\subsection{Reachable transitions}\label{sec:dynamic-reachability}

\begin{definition}
  An \textbf{execution state} $\Gamma$ is either one of the abstract symbols $\finalsym$ or $\initsym$, or a pair $\pair{\heap}{\stack}$ of a heap $\heap$ and a stack $\stack$.
\end{definition}

\begin{definition}
  An execution state $\Gamma$ is \textbf{well-formed} if $\Gamma$ is either one of the abstract symbols $\initsym$ or $\finalsym$, or of the form $\pair{\heap}{\triple{\perms_n}{\env_n}{s_n} \cdot \ldots \cdot \triple{\perms_1}{\env_1}{s_1} \cdot \nilsym}$ and
  \begin{itemize}
    \item $\perms_i \cap \perms_j = \emptyset$ for all $1 \le i < j \le n$.
    \item $s_n = \sseq{s}{\kskip}$ for some statement $s$ or $s_n = \kskip$.
    \item For all $1 \le i < n$, $s_i = \sseq{s}{\sseq{s'}{\kskip}}$ or $s_i  = \sseq{s}{\kskip}$ for some statements $s$ and $s'$ where $s$ is of the form $m(\multiple{e})$ for some $m$, $\multiple{e}$ or $\swhile{e}{\gform}{s_{body}}$ for some $e$, $\gform$, $s_{body}$.
  \end{itemize}
\end{definition}

Examining the execution rules shows that well-formedness of states is preserved by the execution rules defined above.

A dynamic execution transition $\Gamma \to \Gamma'$ is reachable under a program $\prog$, using the exclusion frame $\xperms$, when the following judgement holds:
$$\dtrans{\prog}{\xperms}{\Gamma}{\Gamma'}$$
\semantics[ExecInit]
{ }
{
  \dtrans{\quadruple{s}{M}{P}{S}}{\_}{\initsym}{\pair{\emptyset}{\triple{\emptyset}{\emptyset}{s} \cdot \nilsym}}
}
\semantics[ExecStep]
{
  \dtrans{\prog}{\_}{\_}{\pair{\heap}{\stack}} \\
  \dexec{\heap}{\stack}{\xperms}{\heap'}{\stack'}
}
{
  \dtrans{\prog}{\xperms}{\pair{\heap}{\stack}}{\pair{\heap'}{\stack'}}
}
\semantics[ExecFinal]
{ }
{
  \dtrans{\prog}{\_}{\pair{\_}{\triple{\_}{\_}{\kskip} \cdot \nilsym}}{\finalsym}
}
\begin{definition}\label{def:dstate-reachable}
  An execution state $\dstate$ is \textbf{reachable} from program $\prog$ if $\dstate = \initsym$ or $\dtrans{\prog}{\_}{\_}{\dstate}$.
\end{definition}

\section{Symbolic Execution}\label{sec:symbolic-exec}

\subsection{Definitions}

\begin{definition}
  A \textbf{symbolic value} $\nu \in \SValue$ is an abstract value that represents an unknown character, boolean, integer, or location value. We leave the concrete type of $\SValue$ undefined, but assume that an infinite number of distinct new values can be produced by the $\ffresh$ function.
\end{definition}

\begin{definition}
  A \textbf{symbolic expression} $t \in \SExpr$ is a symbolic value or symbolic expressions combined using operators. Note that the binary operators $\oplus$ are the same as in \S \ref{sec:grammar}.
    $$t ::= \nu \gralt l \gralt \kneg t \gralt t_1 \kand t_2 \gralt t_1 \kor t_2 \gralt t_1 \oplus t_2$$
\end{definition}

\begin{definition}
  A \textbf{path condition} $\pc \in \SExpr$ is a symbolic expression consisting of conjuncts added at every branch point during a particular symbolic execution path.
\end{definition}

\begin{definition}
  An \textbf{imprecise flag} $\imp \in \set{ \top, \bot }$ is a flag that indicates whether a state is imprecise.
\end{definition}

\begin{definition}
  A \textbf{symbolic evironment} $\senv : \Var \pfunc \SExpr$ is a partial function mapping variable names to symbolic expressions.
\end{definition}

\begin{definition}
  A \textbf{field chunk} $\triple{f}{t}{t'} \in \SField$ denotes the mapping of the field $f$ of instance $t$ to the value $t'$.
\end{definition}

\begin{definition}
  A \textbf{predicate chunk} $\pair{p}{\multiple{t}} \in \SPredicate$ represents an isorecursive predicate $p$ with symbolically-evaluated arguments $\multiple{t}$.
\end{definition}

\begin{definition}
  A \textbf{heap chunk} $h \in \SField \cup \SPredicate$ is either a field chunk or a predicate chunk.
\end{definition}

\begin{definition}
  A \textbf{precise symbolic heap} (usually abbreviated as precise heap) $\sheap \in \powerset{\SField \cup \SPredicate}$ is a set of heap chunks where all heap chunks must occupy distinct heap locations at run time.
\end{definition}

\definecolor{shadecolor}{gray}{0.9}
\begin{snugshade}
\begin{definition}
  An \textbf{optimistic symbolic heap} (usually abbreviated as optimistic heap) $\oheap \in \powerset{\SField \cup \SPredicate}$ is a set of field chunks and predicate chunks where distinct chunks may coincide on the heap at run time (i.e. object references that are distinct symbolic expressions may be represent the same object value at run time).
\end{definition}

\begin{definition}
    A \textbf{visited set} $\svis \in \powerset{\Predicate~\cup~\Function}$ contains predicates that have already been unfolded and pure functions that have already been evaluated during recursive symbolic evaluation.
\end{definition}
\end{snugshade}

\begin{definition}
  A \textbf{symbolic permission} $\sperm \in \SPerm$ represents a particular heap location or predicate instance.
  $$\sperm ::= \pair{f}{t} \gralt \pair{p}{\multiple{t}}$$
  A set of symbolic permissions is denoted by $\sperms$.
\end{definition}

\begin{definition}
  A \textbf{run-time check} $r \in \SCheck$ is a symbolic value that must be asserted at run time, a symbolic permission whose access must be asserted at run time, a set of symbolic permissions whose disjointness must be asserted at run time, or an unsatisfiable check.
  $$r ::= t \gralt \sperm \gralt \fsep(\sperms_1, \sperms_2) \gralt \bot$$
  A set of run-time checks is denoted by $\scheck$.
\end{definition}

\definecolor{shadecolor}{gray}{0.9}
\begin{snugshade}
\begin{definition}
  A \textbf{symbolic state} $\sstate \in \SState$ consists of an imprecise flag, a path condition, a precise heap, an optimistic heap, a symbolic environment, and a visited set.
  $$\sstate ::= \sextuple{\imp}{\pc}{\sheap}{\oheap}{\senv}{\svis}$$
  $\imp(\sstate)$, $\pc(\sstate)$, $\sheap(\sstate)$, $\oheap(\sstate)$, $\senv(\sstate)$, and $\svis(\sstate)$ each denote a reference to the respective component of $\sstate$.
\end{definition}
\end{snugshade}

\subsection{Valuations}\label{sec:valuations}

In order to prove soundness with respect to the dynamic semantics, we must first define a correspondence between the two representations.

\begin{definition}
  A \textbf{valuation} $V : \SValue \pfunc \Value$ is a partial function mapping symbolic values to concrete values.
\end{definition}

A base valuation $V : \SValue \pfunc \Value$ is implicitly extended to $V : \SExpr \pfunc \Value$ for all possible symbolic expressions composed of literals and symbolic values in the domain of $V$:
\begin{align*}
  V(l) &:= l \\
  V(t_1 \kadd t_2) &:= V(t_1) + V(t_2) \\
  V(t_1 \ksub t_2) &:= V(t_1) - V(t_2) \\
  V(t_1 \kmul t_2) &:= V(t_1) \cdot V(t_2) \\
  V(t_1 \kdiv t_2) &:= \frac{V(t_1)}{V(t_2)} \\
  V(t_1 \keq t_2) &:= \begin{cases}
    \ktrue &\text{if}~ V(t_1) = V(t_2) \\
    \kfalse &\text{otherwise}
  \end{cases} \\
  V(\kneg t) &:= \begin{cases}
    \ktrue &\text{if}~ V(t_1) = \kfalse \\
    \kfalse &\text{otherwise}
  \end{cases} \\
  V(t_1 \kor t_2) &:= \begin{cases}
    \ktrue &\text{if}~ V(t_1) = \ktrue ~\text{or}~ V(t_2) = \ktrue \\
    \kfalse &\text{otherwise}
  \end{cases} \\
  V(t_1 \kand t_2) &:= \begin{cases}
    \ktrue &\text{if}~ V(t_1) = \ktrue ~\text{and}~ V(t_2) = \ktrue \\
    \kfalse &\text{otherwise}
  \end{cases}
\end{align*}

\begin{definition}\label{def:implication}
  A symbolic expression $t_1$ \textbf{implies} another symbolic expression $t_2$ (denoted $t_1 \implies t_2$) if, for all valuations for which $V(t_1)$ and $V(t_2)$ are defined, $(V(t_1) = \ktrue) \implies (V(t_2) = \ktrue)$.
\end{definition}

\subsection{Footprints}

$\vfoot{V}{\heap}{\sperm}$ and $\vfoot{V}{\heap}{\sperms}$ denote the footprint (i.e. set of permissions) necessary to satisfy the given symbolic permission or symbolic permission set, respectively, given some heap $\heap$.
\begin{align*}
  \vfoot{V}{\heap}{\pair{f}{t}} &:= \pair{V(t)}{f} \\
  \vfoot{V}{\heap}{\pair{p}{\multiple{t}}} &:= \efoot{\heap}{[\multiple{x \mapsto V(t)}]}{\fpred(p)} \\
  \vfoot{V}{\heap}{\sperms} &:= \bigcup_{\sperm \in \sperms} \vfoot{V}{\heap}{\sperm}
\end{align*}

\definecolor{shadecolor}{gray}{0.9}
\begin{snugshade}
Additionally, $\sfoot{\gform}{\sstate}$ represents the set of field chunks in the precise and imprecise heaps represented by an imprecise specification. This is used when taking heap snapshots during the evaluation of pure function applications.
\begin{align*}
    \sfoot{l}{\sstate} &= \varnothing \\
    \sfoot{x}{\sstate} &= \varnothing \\
    \sfoot{e.f}{\sstate} &= \set{\triple{f}{t_e}{t}} &\text{where}~ \seval{\sstate}{e}{t_e}{\_}{\_} \\
    &\quad\quad &\text{and}~ \seval{\sstate}{t_e.f}{t}{\_}{\_} \\
    \sfoot{\kacc(e.f)}{\sstate} &= \set{\triple{f}{t_e}{t}} &\text{where}~ \seval{\sstate}{e}{t_e}{\_}{\_} \\
    &\quad\quad &\text{and}~ \seval{\sstate}{t_e.f}{t}{\_}{\_} \\
    \sfoot{e_1 \oplus e_2}{\sstate} &= \sfoot{e_1}{\sstate} \cup \sfoot{e_2}{\sstate} \\
    \sfoot{e_1 \kor e_2}{\sstate} &= \sfoot{e_1}{\sstate} \cup \sfoot{e_2}{\sstate} \\
    \sfoot{e_1 \kand e_2}{\sstate} &= \sfoot{e_1}{\sstate} \cup \sfoot{e_2}{\sstate} \\
    \sfoot{\kneg e}{\sstate} &= \sfoot{e}{\sstate} \\
    \sfoot{\simprecise{\phi}}{\sstate} &= \sheap(\sstate) \cup \oheap(\sstate) \\
    \sfoot{\phi_1 * \phi_2}{\sstate} &= \sfoot{\phi_1}{\sstate} \cup \sfoot{\phi_2}{\sstate} \\
    \sfoot{\pe}{\sstate} &= \sfoot{\fpred(p)}{\sstate[ \senv = [\multiple{x \mapsto t}] ] } \cup &\text{where}~ \multiple{x} = \fpredparams(p) \\
    &\quad\quad \bigcup \multiple{\sfoot{e}{\sstate}} \cup \set{\pair{p}{\multiple{t}}} &\text{and}~ \multiple{\seval{\sstate}{e}{t}{\_}{\_}} \\
    \sfoot{\fe}{\sstate} &= \sfoot{\ffuncpre(f)}{\sstate} \cup \\
    &\quad\quad \sfoot{\ffuncbody(f)}{\sstate} \cup \bigcup \multiple{\sfoot{e}{\sstate}} \\
    \sfoot{\sunfolding{\pe}{e_0}}{\sstate} &= \sfoot{\pe}{\sstate} \cup \sfoot{e_0}{\sstate}
\end{align*}
\end{snugshade}

\subsection{Correspondence}

The relations $\simheap{V}{\sheap}{\heap}{\perms}$, $\simheap{V}{\oheap}{\heap}{\perms}$, $\simenv{V}{\senv}{\env}$, and $\simstate{V}{\sstate}{\heap}{\perms}{\env}$ denote correspondence between symbolic states and run-time states:
\begin{align}
  \begin{split}\label{eq:sheap-correspondence}
    \simheap{V}{\sheap}{\heap}{\perms} \iffdef
      & (\universal{\triple{f}{t}{t'} \in \sheap}{\heap(V(t), f) = V(t')}) ~\wedge \\
      & (\universal{\triple{f}{t}{t'} \in \sheap}{\pair{V(t)}{f} \in \perms}) ~\wedge \\
      & (\universal{\pair{p}{\tlist} \in \sheap}{\assertion{\heap}{\perms}{[\multiple{x \mapsto V(t)}]}{\fpred(p)}}) ~\wedge \\
      & (\universal{h_1, h_2 \in \sheap^2}{h_1 \ne h_2 \implies \vfoot{V}{\heap}{h_1} \cap \vfoot{V}{\heap}{h_2} = \emptyset})
  \end{split} \\
  \begin{split}
    \simheap{V}{\oheap}{\heap}{\perms} \iffdef
      & (\universal{\triple{f}{t}{t'} \in \oheap}{\heap(V(t), f) = V(t')}) ~\wedge \\
      & (\universal{\triple{f}{t}{t'} \in \oheap}{\pair{V(t)}{f} \in \perms}) ~\wedge \\
      & (\universal{\pair{p}{\multiple{t}} \in \oheap}{\assertion{\heap}{\perms}{[\multiple{x \mapsto V(t)}]}{\fpred(p)}})
  \end{split} \label{eq:oheap-correspondence} \\
  \simenv{V}{\senv}{\env} \iffdef &\universal{x \in \dom(\senv)}{\env(x) = V(t)} \label{eq:senv-correspondence} \\
  \begin{split}\label{eq:sstate-correspondence}
    \simstate{V}{\sstate}{\heap}{\perms}{\env} \iffdef
      & (\simheap{V}{\sheap(\sstate)}{\heap}{\perms}) ~\wedge \\
      & (\simheap{V}{\oheap(\sstate)}{\heap}{\perms}) ~\wedge \\
      & (\simenv{V}{\senv(\sstate)}{\env}) ~\wedge \\
      & (V(\pc(\sstate)) = \ktrue)
  \end{split}
\end{align}

\subsection{Run-time checks}

The judgement $\rtassert{V}{\heap}{\perms}{r}$ denotes that a symbolic runtime check $r \in \SCheck$ is satisfied at run time by a heap $\heap$ and permission set $\perms$ through a valuation $V$. Note that there is no rule for $\bot$; by design it can never be satisfied.
\semantics[CheckValue]
  {V(t) = \ktrue}
  {\rtassert{V}{\heap}{\perms}{t}}
\semantics[CheckAcc]
  {\pair{V(t)}{f} \in \perms}
  {\rtassert{V}{\heap}{\perms}{\pair{f}{t}}}
\semantics[CheckPred]
  {
    \multiple{x} = \fpredparams(p) \\
    \assertion{\heap}{\perms}{[\multiple{x \mapsto V(t)}]}{\fpred(p)}
  }
  {
    \rtassert{V}{\heap}{\perms}{\pair{p}{\multiple{t}}}
  }
\semantics[CheckSep]
  {
    \vfoot{V}{\heap}{\sperms_1} \cap \vfoot{V}{\heap}{\sperms_2} = \emptyset
  }
  {
    \rtassert{V}{\heap}{\perms}{\fsep(\sperms_1, \sperms_2)}
  }

This judgement is naturally extended for a set of runtime checks $\scheck$:
$$\rtassert{V}{\heap}{\perms}{\scheck} \iffdef \universal{r \in \scheck}{\rtassert{V}{\heap}{\perms}{r}}$$

\subsection{Evaluation}\label{sec:seval-rules}

The judgement
$$\seval{\sstate}{e}{t}{\sstate'}{\scheck}$$
denotes the evaluation of the expression $e \in \Expr$ in the symbolic state $\sstate \in \SState$. It yields the symbolic expression $t \in \SExpr$, a new symbolic state $\sstate'$ which must be satisfied to produce the resulting value, and a set of run-time checks $\scheck \in \powerset{\SCheck}$.

Note that for any given $\sstate$, there may be multiple values of $t, \sstate', \scheck$ for which the relation is satisfied. Therefore, the path condition of $\sstate'$ must be satisfied before assuming that $t$ corresponds to an actual value.

Also note that unsatisfiable paths are not pruned during evaluation. These paths may be pruned by checking the satisfiability of $\pc(\sstate')$.

For a list of expressions $\multiple{e}$, $\multiple{\seval{\sstate}{e}{t}{\sstate'}{\scheck}}$ represents a sequence of judgements
$$\seval{\sstate_0}{e_1}{t_1}{\sstate_1}, \cdots, \seval{\sstate_{n-1}}{e_n}{t_n}{\sstate_n}{\scheck_n}$$
where $\sstate_0 = \sstate$, $e_1, \cdots, e_n = \multiple{e}$, and $\scheck = \scheck_1 \cup \cdots \cup \scheck_n$.

\semantics[SEvalLiteral]
  { }
  {\seval{\sstate}{l}{l}{\sstate}{\emptyset}}
\semantics[SEvalVar]
  { }
  {\seval{\sstate}{x}{\senv(\sstate)(x)}{\sstate}{\emptyset}}
\semantics[SEvalOrA]
  {
    \seval{\sstate}{e_1}{t_1}{\sstate'}{\scheck} \\
    \sstate'' = \sstate'[\pc = \pc(\sstate') \kand t_1]
  }
  {
    \seval{\sstate}{e_1 \kor e_2}{t_1}{\sstate''}{\scheck}
  }
\semantics[SEvalOrB]
  {
    \seval{\sstate}{e_1}{t_1}{\sstate'}{\scheck_1} \\
    \seval{\sstate'[\pc = \pc(\sstate') \kand \kneg t_1]}{e_2}{t_2}{\sstate''}{\scheck_2}
  }
  {
    \seval{\sstate}{e_1 \kor e_2}{t_2}{\sstate''}{\scheck_1 \cup \scheck_2}
  }
\semantics[SEvalAndA]
  {
    \seval{\sstate}{e_1}{t_1}{\sstate'}{\scheck} \\
    \sstate'' = \sstate'[\pc = \pc(\sstate') \kand \kneg t_1]
  }
  {
    \seval{\sstate}{e_1 \kand e_2}{t_1}{\sstate''}{\scheck}
  }
\semantics[SEvalAndB]
  {
    \seval{\sstate}{e_1}{t_1}{\sstate'}{\scheck_1} \\
    \seval{\sstate'[\pc = \pc(\sstate') \kand t_1]}{e_2}{t_2}{\sstate''}{\scheck_2}
  }
  {
    \seval{\sstate}{e_1 \kand e_2}{t_2}{\sstate''}{\scheck_1 \cup \scheck_2}
  }
\semantics[SEvalOp]
  {
    \seval{\sstate}{e_1}{t_1}{\sstate'}{\scheck_1} \\
    \seval{\sstate'}{e_2}{t_2}{\sstate''}{\scheck_2}
  }
  {
    \seval{\sstate}{e_1 \oplus e_2}{t_1 \oplus t_2}{\sstate''}{\scheck_1 \cup \scheck_2}
  }
\semantics[SEvalNeg]
  {
    \seval{\sstate}{e}{t}{\sstate'}{\scheck}
  }
  {
    \seval{\sstate}{\kneg e}{\kneg t}{\sstate'}{\scheck}
  }
\semantics[SEvalField]
  {
    \seval{\sstate}{e}{t_e}{\sstate'}{\scheck} \\
    \pc(\sstate') \implies t_e \keq t_e' \\
    \triple{f}{t_e'}{t} \in \sheap(\sstate')
  }
  {
    \seval{\sstate}{e.f}{t}{\sstate'}{\scheck}
  }
\semantics[SEvalFieldOptimistic]
  {
    \seval{\sstate}{e}{t_e}{\sstate'}{\scheck} \\
    \nexistential{t_e', t}{\triple{f}{t_e'}{t} \in \sheap(\sstate) \wedge \pc(\sstate') \implies t_e' \keq t_e} \\
    \triple{f}{t_e'}{t} \in \oheap(\sstate) \\
    \pc(\sstate') \implies t_e' \keq t_e
  }
  {
    \seval{\sstate}{e.f}{t}{\sstate'}{\scheck}
  }
\semantics[SEvalFieldImprecise]
  {
    \imp(\sstate) \\
    \seval{\sstate}{e}{t_e}{\sstate'}{\scheck} \\
    \nexistential{t_e', t}{\triple{f}{t_e'}{t} \in \sheap(\sstate) \cup \oheap(\sstate) \wedge \pc(\sstate') \implies t_e' \keq t_e} \\
    t = \ffresh \\
    \sstate'' = \sstate'[\oheap = \oheap(\sstate'); \triple{f}{t_e}{t}]
  }
  {
    \seval{\sstate}{e.f}{t}{\sstate''}{\scheck; \pair{t_e}{f}}
  }
\semantics[SEvalFieldFailure]
  {
    \neg \imp(\sstate) \\
    \seval{\sstate}{e}{t_e}{\sstate'}{\scheck} \\
    \nexistential{t_e', t}{\triple{f}{t_e'}{t} \in \sheap(\sstate) \cup \oheap(\sstate) \wedge \pc(\sstate') \implies t_e' \keq t_e} \\
    t = \ffresh
  }
  {
    \seval{\sstate}{e.f}{t}{\sstate'}{\set{\bot}}
  }

\definecolor{shadecolor}{gray}{0.9}
\begin{snugshade}
\semantics[SEvalUnfoldingPrecise]{
    \lnot \imp(\sstate_1) \\
    p \notin \svis(\sstate_1) \\
    \multiple{\seval{\sstate_1}{e}{t}{\sstate_2}{\scheck_1}} \\
    \scons{\sstate_2}{p(\multiple{e})}{\sstate_3}{\scheck_2} \\
    \multiple{x} = \fpredparams(p) \\
    \sproduce{\sstate_3[\senv = \senv(\sstate_3)[\multiple{x \mapsto t}], \svis = \svis(\sstate_3) \cup \{ p \}]}{\fpred(p)}{\sstate_4} \\
    \seval{\sstate_4[\senv = \senv(\sstate_3)]}{e_0}{t_0}{\sstate_5}{\scheck_3} \\
    \sstate_6 = \sstate_5[\svis = \svis(\sstate_1), \sheap = \sheap(\sstate_2), \oheap = \emptyset]
}{
    \seval{\sstate_1}{\sunfolding{p(\multiple{e})}{e_0}}{t_0}{\sstate_6}{\scheck_1 \cup \scheck_2 \cup \scheck_3}
}

            \semantics[SEvalUnfoldingImpreciseA]{
                \imp(\sstate_1) \\
                \fpred(p)\; \text{is precise} \\
                p \notin \svis(\sstate_1) \\
                \multiple{\seval{\sstate_1}{e}{t}{\sstate_2}{\scheck_1}} \\
                \scons{\sstate_2}{p(\multiple{e})}{\sstate_3}{\scheck_2} \\
                \multiple{x} = \fpredparams(p) \\
                \sproduce{\sstate_3[\senv = \senv(\sstate_3)[\multiple{x \mapsto t}], \svis = \svis(\sstate_3) \cup \{ p \}]}{\fpred(p)}{\sstate_4} \\
                \seval{\sstate_4[\senv = \senv(\sstate_3)]}{e_0}{t_0}{\sstate_5}{\scheck_3} \\
                \oheap' = \oheap(\sstate_2) \cup \oheap(\sstate_5) \cup \{ \pair{p}{\multiple{t}} \} \\
                \sstate_6 = \sstate_5[\svis = \svis(\sstate_1), \sheap = \sheap(\sstate_2), \oheap = \oheap']
            }{
                \seval{\sstate_1}{\sunfolding{p(\multiple{e})}{e_0}}{t_0}{\sstate_6}{\scheck_1 \cup \scheck_2 \cup \scheck_3}
            }

            \semantics[SEvalUnfoldingImpreciseB]{
                \imp(\sstate_1) \\
                \fpred(p)\; \text{is imprecise} \\
                p \notin \svis(\sstate_1) \\
                \multiple{\seval{\sstate_1}{e}{t}{\sstate_2}{\scheck_1}} \\
                \scons{\sstate_2}{p(\multiple{e})}{\sstate_3}{\scheck_2} \\
                \multiple{x} = \fpredparams(p) \\
                \sproduce{\sstate_3[\senv = \senv(\sstate_3)[\multiple{x \mapsto t}], \svis = \svis(\sstate_3) \cup \{ p \}]}{\fpred(p)}{\sstate_4} \\
                \seval{\sstate_4[\senv = \senv(\sstate_3)]}{e_0}{t_0}{\sstate_5}{\scheck_3} \\
                \sstate_6 = \sstate_5[\svis = \svis(\sstate_1), \sheap = \sheap(\sstate_2), \oheap = \oheap(\sstate_2) \cup \{ \pair{p}{\multiple{t}} \}]
            }{
                \seval{\sstate_1}{\sunfolding{p(\multiple{e})}{e_0}}{t_0}{\sstate_6}{\scheck_1 \cup \scheck_2 \cup \scheck_3}
            }

            \semantics[SEvalUnfoldingImplicitPrecise]{
                p \in \svis(\sstate_1) \\
                \seval{\sstate_1}{e}{t}{\sstate_2}{\scheck} \\
                \pair{p}{\multiple{t}} \in \sheap(\sstate_2) \cup \oheap(\sstate_2) \\
                t_0 = \ffresh
            }{
                \seval{\sstate_1}{\sunfolding{p(\multiple{e})}{e_0}}{t_0}{\sstate_2}{\scheck}
            }

            \semantics[SEvalUnfoldingImplicitImprecise]{
                p \in \svis(\sstate_1) \\
                \seval{\sstate_1}{e}{t}{\sstate_2}{\scheck} \\
                \imp(\sstate_2) \\
                \pair{p}{\multiple{t}} \notin \sheap(\sstate_2) \cup \oheap(\sstate_2) \\
                \oheap' = \oheap(\sstate_2) \cup \set{\pair{p}{\multiple{t}}} \\
                \scheck' = \scheck \cup \set{\pair{p}{\multiple{t}}} \\
                \sstate_3 = \sstate_2[\oheap = \oheap'] \\
                t_0 = \ffresh
            }{
                \seval{\sstate_1}{\sunfolding{\pe}{e_0}}{t_0}{\sstate_3}{\scheck'}
            }

            \semantics[SEvalUnfoldingImplicitFailure]{
                p \in \svis(\sstate_1) \\
                \seval{\sstate_1}{e}{t}{\sstate_2}{\scheck} \\
                \lnot \imp(\sstate_2) \\
                \pair{p}{\multiple{t}} \notin \sheap(\sstate_2) \cup \oheap(\sstate_2) \\
                t_0 = \ffresh
            }{
                \seval{\sstate_1}{\sunfolding{\pe}{e_0}}{t_0}{\sstate_2}{\set{\bot}}
            }
\semantics[SEvalFunctionExplicit]{
            f \notin \svis(\sstate_1) \\
            \multiple{\seval{\sstate_1}{e}{t}{\sstate_2}{\scheck_1}} \\
            \multiple{x} = \ffuncparams(f) \\
            s = \sfoot{\ffuncpre(f)}{{\sstate_2}} \\
            \scons{\sstate_2[\senv = \senv(\sstate_2)[\multiple{x \mapsto t}]]}{\ffuncpre(f)}{\sstate_3}{\scheck_2} \\
            \sproduce{\sstate_3}{\ffuncpre(f)}{\sstate_4} \\
            \sheap' = \{ \triple{f}{t_e}{t'} : \triple{f}{t_e}{t} \in \sheap(\sstate_4), \triple{f}{t_e}{t'} \in \sheap(\sstate_2) \} \\
            \oheap' = \set{ \triple{f}{t_e}{t'} : \triple{f}{t_e}{t} \in \oheap(\sstate_4), \triple{f}{t_e}{t'} \in \oheap(\sstate_2) } \\
            \seval{\sstate_4[\sheap = \sheap', \oheap = \oheap', \svis = \svis(\sstate_3) \cup \{ f \}]}{\ffuncbody(f)}{t'}{\sstate_5}{\scheck_3} \\
            f(\multiple{t}, s) = \ffresh \\
            \sstate_6 = \sstate_2[\imp = \imp(\sstate_5), \oheap = \oheap(\sstate_5), \pc = \pc(\sstate_2) \kand (f(\multiple{t}, s) \keq t') \kand \fax{\sstate}{\fts}{\ffuncpost(f)}]
        }{
            \seval{\sstate_1}{\fe}{f(\multiple{t}, s)}{\sstate_6}{\scheck_1 \cup \scheck_2 \cup \scheck_3}
        }
        \semantics[SEvalFunctionImplicit]{
            f \in \svis(\sstate_1) \\
            \multiple{\seval{\sstate_1}{e}{t}{\sstate_2}{\scheck_1}} \\
            \multiple{x} = \ffuncparams(f) \\
            s = \sfoot{\ffuncpre(f)}{{\sstate_2}} \\
            s' = \sfoot{\ffuncpost(f)}{\sstate_2} \\
            \oheap' = \oheap(\sstate_2) \cup (s' \setminus \sheap(\sstate_2)) \\
            \scons{\sstate_2[\senv = \senv(\sstate_2)[\multiple{x \mapsto t}]]}{\ffuncpre(f)}{\sstate_3}{\scheck_2} \\
            f(\multiple{t}, s) = \ffresh \\
            \sstate_4 = \sstate_2[\imp = \imp(\sstate_3), \oheap = \oheap', \pc = \pc(\sstate_2) \kand \fax{\sstate}{\fts}{\ffuncpost(f)}] \\
            \scheck_3 = \set{\pair{t_e}{f} : \triple{f}{t_e}{t} \in s'}
        }{
            \seval{\sstate_1}{\fe}{f(\multiple{t}, s)}{\sstate_4}{\scheck_1 \cup \scheck_2 \cup \scheck_3}
        }

The function $\operatorname{axiomatize}$, used in evaluating pure function applications to convert function postconditions into valid path conditions, is defined as follows.

\begin{align*}
    \fax{\sstate}{\fts}{\kresult} &= \fts \\
    \fax{\sstate}{\fts}{l} &= l \\
    \fax{\sstate}{\fts}{x} &= \senv(\sstate)(x) \\
    \fax{\sstate}{\fts}{e.f} &= t ~\text{where}~ \triple{f}{t_e}{t} \in s ~\text{and}~ \seval{\sstate}{e}{t_e}{\_}{\_} \\
    \fax{\sstate}{\fts}{e_1 \oplus e_2} &= \fax{\sstate}{\fts}{e_1} \oplus \fax{\sstate}{\fts}{e_2} \\
    \fax{\sstate}{\fts}{e_1 \kor e_2} &= \fax{\sstate}{\fts}{e_1} \kor \fax{\sstate}{\fts}{e_2} \\
    \fax{\sstate}{\fts}{e_1 \kand e_2} &= \fax{\sstate}{\fts}{e_1} \kand \fax{\sstate}{\fts}{e_2} \\
    \fax{\sstate}{\fts}{\kneg e} &= \kneg \fax{\sstate}{\fts}{e} \\
    \fax{\sstate}{\fts}{\sunfolding{\pe}{e_0}} &= \fax{\sstate}{\fts}{e_0} \\
\end{align*}
\end{snugshade}

\subsection{Produce}
\label{sec:produce-rules}

The \textbf{produce} operation adds the information contained in a formula $\gform$ to the symbolic state $\sstate$, resulting in a new symbolic state $\sstate'$. This is denoted by the judgement
$$\sproduce{\sstate}{\phi}{\sstate'}.$$

\semantics[SProduceImprecise]
  {\sproduce{\sstate[\imp = \top]}{\phi}{\sstate'}}
  {\sproduce{\sstate}{\simprecise{\phi}}{\sstate'}}
\semantics[SProduceExpr]
  {
    \seval{\sstate}{e}{t}{\sstate'}{\_} \\
    \sstate'' = \sstate'[\pc = \pc(\sstate') \kand t]
  }
  {
    \sproduce{\sstate}{e}{\sstate''}
  }
\semantics[SProducePredicate]
  {
    \multiple{\seval{\sstate}{e}{t}{\sstate'}{\_}} \\
    \sstate'' = \sstate'[\sheap = \sheap(\sstate'); \pair{p}{\tlist}]
  }
  {
    \sproduce{\sstate}{p(\multiple{e})}{\sstate''}
  }
\semantics[SProduceField]
  {
    \seval{\sstate}{e}{t_e}{\sstate'}{\_} \\
    t = \ffresh \\
    \sstate'' = \sstate'[\sheap = \sheap(\sstate'); \triple{f}{t_e}{t}]
  }
  {
    \sproduce{\sstate}{\kacc(e.f)}{\sstate''}
  }
\semantics[SProduceConjunction]
  {
    \sproduce{\sstate}{\phi_1}{\sstate'} \\
    \sproduce{\sstate'}{\phi_2}{\sstate''}
  }
  {
    \sproduce{\sstate}{\phi_1 * \phi_2}{\sstate''}
  }
\semantics[SProduceIfA]
  {
    \seval{\sstate}{e}{t}{\sstate'}{\_} \\
    \sproduce{\sstate'[\pc = \pc(\sstate') \kand t]}{\phi_1}{\sstate''}
  }
  {
    \sproduce{\sstate}{\sif{e}{\phi_1}{\phi_2}}{\sstate''}
  }
\semantics[SProduceIfB]
  {
    \seval{\sstate}{e}{t}{\sstate'}{\_} \\
    \sproduce{\sstate'[\pc = \pc(\sstate') \kand \kneg t]}{\phi_2}{\sstate''}
  }
  {
    \sproduce{\sstate}{\sif{e}{\phi_1}{\phi_2}}{\sstate'}
  }

\subsection{Consume}\label{sec:consume-rules}

The \textbf{consume} operation checks whether a formula $\gform$ is established by the symbolic state, collects runtime checks that are minimally sufficient to establish $\gform$, and removes permissions asserted in $\gform$ from the symbolic state. This is denoted by the judgement
$$\sconsume{\sstate}{\sstate_E}{\gform}{\sstate'}{\scheck}{\sperms}$$
where $\sstate$ is the symbolic state containing the currently remaining permissions during consume, and $\sstate_E$ is the symbolic state containing the original permissions which may used for evaluating expressions.

\semantics[SConsumeImprecision]
  {
    \sconsume{\sstate}{\sstate_E[\imp = \top]}{\phi}{\sstate'}{\scheck}{\sperms}
  }
  {
    \sconsume{\sstate}{\sstate_E}{\simprecise{\phi}}{\quintuple{\top}{\pc(\sstate')}{\senv(\sstate')}{\emptyset}{\emptyset}}{\scheck}{\sperms}
  }
\semantics[SConsumeValue]
  {
    \seval{\sstate_E}{e}{t}{\_}{\scheck} \\
    \pc(\sstate) \implies t
  }
  {
    \sconsume{\sstate}{\sstate_E}{e}{\sstate}{\scheck}{\emptyset}
  }
\semantics[SConsumeValueImprecise]
  {
    \imp(\sstate) \\
    \seval{\sstate_E}{e}{t}{\_}{\scheck} \\
    \pc(\sstate) \notimplies t
  }
  {
    \sconsume{\sstate}{\sstate_E}{e}{\sstate[\pc = \pc(\sstate) \kand t]}{\scheck; t}{\emptyset}
  }
\semantics[SConsumeValueFailure]
  {
    \neg \imp(\sstate) \\
    \seval{\sstate_E}{e}{t}{\_}{\scheck} \\
    \pc(\sstate) \notimplies t
  }
  {
    \sconsume{\sstate}{\sstate_E}{e}{\sstate}{\set{\bot}}{\emptyset}
  }
\semantics[SConsumePredicate]
  {
    \multiple{\seval{\sstate_E}{e}{t}{\_}{\scheck}} \\
    \multiple{\pc(\sstate) \implies t \keq t'} \\
    \sheap(\sstate) = \sheap'; \pair{p}{\multiple{t'}}
  }
  {
   \sconsume{\sstate}{\sstate_E}{p(\multiple{e})}{\sstate[\sheap = \sheap', \oheap = \emptyset]}{\bigcup \multiple{\scheck}}{\set{\pair{p}{\multiple{t}}}}
  }
\semantics[SConsumePredicateImprecise]
  {
    \imp(\sstate) \\
    \multiple{\seval{\sstate_E}{e}{t}{\_}{\scheck}} \\
    \nexistential{\pair{p}{\multiple{t'}} \in \sheap(\sstate)}{\bigwedge \multiple{\pc(\sstate) \implies t \keq t'}}
  }
  {
    \sconsume{\sstate}{\sstate_E}{p(\multiple{e})}{\sstate[\heap = \emptyset, \oheap = \emptyset]}{\bigcup \multiple{\scheck}; \pair{p}{\multiple{t}}}{\set{\pair{p}{\multiple{t}}}}
  }
\semantics[SConsumePredicateFailure]
  {
    \neg \imp(\sstate) \\
    \multiple{\seval{\sstate_E}{e}{t}{\_}{\scheck}} \\
    \nexistential{\pair{p}{\multiple{t'}} \in \sheap(\sstate)}{\bigwedge \multiple{\pc(\sstate) \implies t \keq t'}}
  }
  {
    \sconsume{\sstate}{\sstate_E}{p(\multiple{e})}{\sstate}{\set{\bot}}{\set{\pair{p}{\multiple{t}}}}
  }
\semantics[SConsumeAcc]
  {
    \seval{\sstate_E}{e}{t_e}{\_}{\scheck} \\
    \pc(\sstate) \implies t_e' \keq t_e \\
    \triple{f}{t_e'}{t} \in \sheap(\sstate) \\
    \sheap' = \fremfp(\sheap(\sstate), \sstate, t_e, f) \\
    \oheap' = \fremf(\oheap(\sstate), \sstate, t_e, f)
  }
  {
    \sconsume{\sstate}{\sstate_E}{\kacc(e.f)}{\sstate[\sheap = \sheap', \oheap = \oheap']}{\scheck}{\set{\pair{t_e}{f}}}
  }
\semantics[SConsumeAccOptimistic]
  {
    \seval{\sstate_E}{e}{t_e}{\_}{\scheck} \\
    \pc(\sstate) \implies t_e' \keq t_e \\
    \triple{f}{t_e'}{t} \in \sheap(\sstate) \\
    \nexistential{t_e', t}{\triple{f}{t_e}{t} \in \sheap(\sstate) \wedge (\pc(\sstate) \implies t_e' \keq t_e)} \\
    \sheap' = \fremf(\sheap(\sstate), \sstate, t_e, f) \\
    \oheap' = \fremf(\oheap(\sstate), \sstate, t_e, f)
  }
  {
    \sconsume{\sstate}{\sstate_E}{\kacc(e.f)}{\sstate[\sheap = \sheap', \oheap = \oheap']}{\scheck}{\set{\pair{t_e}{f}}}
  }
\semantics[SConsumeAccImprecise]
  {
    \imp(\sstate) \\
    \seval{\sstate_E}{e}{t_e}{\_}{\scheck} \\
    \nexistential{t_e', t}{\triple{f}{t_e}{t} \in \sheap(\sstate) \cup \oheap(\sstate) \wedge (\pc(\sstate) \implies t_e' \keq t_e)} \\
    \sheap' = \fremf(\sheap(\sstate), \sstate, t_e, f) \\
    \oheap' = \fremf(\oheap(\sstate), \sstate, t_e, f)
  }
  {
    \sconsume{\sstate}{\sstate_E}{\kacc(e.f)}{\sstate[\sheap = \sheap', \oheap = \oheap']}{\scheck; \pair{t_e}{f}}{\set{\pair{t_e}{f}}}
  }
\semantics[SConsumeAccFailure]
  {
    \neg \imp(\sstate) \\
    \seval{\sstate_E}{e}{t_e}{\_}{\scheck} \\
    \nexistential{t_e', t}{\triple{f}{t_e}{t} \in \sheap(\sstate) \cup \oheap(\sstate) \wedge (\pc(\sstate) \implies t_e' \keq t_e)}
  }
  {
    \sconsume{\sstate}{\sstate_E}{\kacc(e.f)}{\sstate}{\set{\bot}}{\set{\pair{t_e}{f}}}
  }
\semantics[SConsumeConjunction]
  {
    \sconsume{\sstate}{\sstate_E}{\phi_1}{\sstate'}{\scheck_1}{\sperms_1} \\
    \sconsume{\sstate'}{\sstate_E[\pc = \pc(\sstate')]}{\phi_2}{\sstate''}{\scheck_2}{\sperms_2} \\\\
    (\scheck_1 \cup \scheck_2) \cap \SPerm = \emptyset
  }
  {
    \sconsume{\sstate}{\sstate_E}{\phi_1 * \phi_2}{\sstate''}{\scheck_1 \cup \scheck_2}{\sperms_1 \cup \sperms_2}
  }
\semantics[SConsumeConjunctionImprecise]
  {
    \sconsume{\sstate}{\sstate_E}{\phi_1}{\sstate'}{\scheck_1}{\sperms_1} \\
    \sconsume{\sstate'}{\sstate_E[\pc = \pc(\sstate')]}{\phi_2}{\sstate''}{\scheck_2}{\sperms_2} \\\\
    (\scheck_1 \cup \scheck_2) \cap \SPerm \ne \emptyset
  }
  {
    \sconsume{\sstate}{\sstate_E}{\phi_1 * \phi_2}{\sstate''}{\scheck_1 \cup \scheck_2; \fsep(\sperms_1, \sperms_2)}{\sperms_1 \cup \sperms_2}
  }
\semantics[SConsumeConditionalA]
  {
    \seval{\sstate_E}{e}{t}{\_}{\scheck} \\
    \pc' = \pc(\sstate') \kand t \\
    \sconsume{\sstate[\pc = \pc']}{\sstate_E[\pc = \pc']}{\phi_1}{\sstate'}{\scheck'}{\sperms}
  }
  {
    \sconsume{\sstate}{\sstate_E}{\sif{e}{\phi_1}{\phi_2}}{\sstate'}{\scheck \cup \scheck'}{\sperms}
  }
\semantics[SConsumeConditionalB]
  {
    \seval{\sstate_E}{e}{t}{\_}{\scheck} \\
    \pc' = \pc(\sstate') \kand \kneg t \\
    \sconsume{\sstate[\pc = \pc']}{\sstate_E[\pc = \pc']}{\phi_2}{\sstate'}{\scheck'}{\sperms}
  }
  {
    \sconsume{\sstate}{\sstate_E}{\sif{e}{\phi_1}{\phi_2}}{\sstate'}{\scheck \cup \scheck'}{\sperms}
  }

The functions $\fremf$, $\fremfp$, and $\falias$ are defined as follows:
\begin{align*}
  \fremf(\sheap, \sstate, t, f) &= \set{\triple{f'}{t'}{t''} \in \sheap : \neg \falias(\sstate, t, f, t', f')} \\
  \fremfp(\sheap, \sstate, t, f) &= \fremf(\sheap, \sstate, t, f) \cup \set{\pair{p}{\multiple{t}} \in \sheap} \\
  \falias(\sstate, t, f, t', f') &= \begin{cases}
    f = f' \wedge (\pc(\sstate) \implies t \keq t') & \neg \imp(\sstate) \\
    (f = f') \wedge \fsat(\pc(\sstate) \kand t \keq t') & \imp(\sstate)
  \end{cases}
\end{align*}

For ease of notation, the judgement $\scons{\sstate}{\gform}{\sstate'}{\scheck}$ applies the rules above, initializing additional parameters and ignoring the parameters that are internal to consume.
\semantics[SConsume]
  {
    \sconsume{\sstate}{\sstate}{\gform}{\sstate'}{\scheck}{\_}
  }
  {
    \scons{\sstate}{\gform}{\sstate'}{\scheck}
  }

\subsection{Execute}\label{sec:sexec-rules}

Symbolic execution is denoted by the small-step judgement
$$\sexec{\sstate}{s}{s'}{\sstate'}.$$
where $\sstate$ is the symbolic state prior to execution, $s$ is the statement to execute, $s'$ is the statement remaining after this execution step and $\sstate'$ is the symbolic state after the execution step.

Note that as in \S\ref{sec:seval-rules} there may be multiple resulting symbolic states $\sstate'$ for which the judgement applies, and therefore the path condition of $\sstate'$ must be satisfied before assuming that $\sstate'$ corresponds to a particular execution step.

\semantics[SExecSeq]
  {
  }
  {
    \sexec{\sstate}{\sseq{\kskip}{s}}{s}{\sstate}
  }
\semantics[SExecAssign]
  {
    \seval{\sstate}{e}{t}{\sstate'}{\_} \\
    \senv' = \senv(\sstate)[x \mapsto t]
  }
  {
    \sexec{\sstate}{\sseq{x = e}{s}}{s}{\sstate'[\senv = \senv']}
  }
\semantics[SExecAssignField]
  {
    \seval{\sstate}{e}{t}{\sstate'}{\_} \\
    \scons{\sstate'}{\kacc(x.f)}{\sstate''}{\_} \\
    \sheap' = \sheap(\sstate); \triple{\senv(\sstate'')(x)}{f}{t}
  }
  {
    \sexec{\sstate}{\sseq{x.f = e}{s}}{s}{\sstate''[\sheap = \sheap']}
  }
\semantics[SExecAlloc]
  {
    t = \ffresh \\
    \multiple{T ~ f} = \fstruct(S) \\
    \sheap' = \sheap(\sstate); \multiple{\triple{f}{t}{\fdefault(T)}}
  }
  {
    \sexec{\sstate}{\sseq{x = \salloc{S}}{s}}{s}{\sstate[\sheap = \sheap']}
  }
\semantics[SExecCall]
  {
    \multiple{\seval{\sstate}{e}{t_e}{\sstate'}{\_}} \\
    \multiple{x} = \fparams(m) \\
    \scons{\sstate'[\senv = [\multiple{x \mapsto t_e}]]}{\fpre(m)}{\sstate'}{\_} \\
    t = \ffresh \\
    \sproduce{\sstate'[\senv = [\multiple{x \mapsto t_e}, \kresult \mapsto t]]}{\fpost(m)}{\sstate''}
  }
  {
    \sexec{\sstate}{\sseq{y \kassign m(\multiple{e})}{s}}{s}{\sstate''[\senv = \senv(\sstate)[y \mapsto t]]}
  }
\semantics[SExecAssert]
  {
    \scons{\sstate}{\simprecise{\phi}}{\sstate'}{\_} \\
    \sproduce{\sstate'}{\simprecise{\phi}}{\sstate''}
  }
  {
    \sexec{\sstate}{\sseq{\sassert{\phi}}{s}}{s}{\sstate[\pc = \pc(\sstate'')]}
  }
\semantics[SExecFold]
  {
    \multiple{\seval{\sstate}{e}{t}{\sstate'}{\_}} \\
    \multiple{x} = \fpredparams(p) \\
    \scons{\sstate'[\senv = [\multiple{x \mapsto t}]]}{\fpred(p)}{\sstate''}{\_} \\
    \sstate''' = \sstate''[\senv = \senv(\sstate), \sheap = \sheap(\sstate''); \pair{p}{\multiple{t}}]
  }
  {
    \sexec{\sstate}{\sseq{\sfold{p(\multiple{e})}}{s}}{s}{\sstate'''}
  }
\semantics[SExecUnfold]
  {
    \multiple{\seval{\sstate}{e}{t}{\sstate'}{\_}} \\
    \scons{\sstate'}{p(\multiple{e})}{\sstate''}{\_} \\
    \multiple{x} = \fpredparams(p) \\
    \sproduce{\sstate''[\senv = [\multiple{x \mapsto t}]]}{\fpred(p)}{\sstate'''}
  }
  {
    \sexec{\sstate}{\sseq{\sunfold{p(e_1, \cdots, e_n)}}{s}}{s}{\sstate'''[\senv = \senv(\sstate)]}
  }
\semantics[SExecIfA]
  {
    \seval{\sstate}{e}{t}{\sstate'}{\_}
  }
  {
    \sexec{\sstate}{\sseq{\sif{e}{s_1}{s_2}}{s}}{\sseq{s_1}{s}}{\sstate'[\pc = \pc(\sstate') \kand t]}
  }
\semantics[SExecIfB]
  {
    \seval{\sstate}{e}{t}{\sstate'}{\_}
  }
  {
    \sexec{\sstate}{\sseq{\sif{e}{s_1}{s_2}}{s}}{\sseq{s_2}{s}}{\sstate'[\pc = \pc(\sstate') \kand \kneg t]}
  }

\semantics[SExecWhileSkip]
  {
    \scons{\sstate}{\gform}{\sstate'}{\_} \\
    \multiple{x} = \fmodified(s) \\
    \sproduce{\sstate'[\senv = \senv(\sstate')[\multiple{x \mapsto \ffresh}]]}{\gform}{\sstate''} \\
    \seval{\sstate''}{e}{t}{\_}{\_} \\
  }
  {
    \sexec{\sstate}{\sseq{\swhile{e}{\gform}{s}}{s'}}{s'}{\sstate''[\pc = \pc(\sstate'') \kand \kneg t]}
  }

\begin{definition}\label{def:frem}
  The helper function $\frem(\sstate, \gform)$ returns the set of all permissions remaining in $\sstate$ if $\gform$ is not completely precise:
  $$\frem(\sstate, \gform) := \begin{cases}
    \emptyset &\text{if $\gform$ is completely precise} \\
    \set{ \pair{t}{f} : \triple{f}{t}{t'} \in \sheap(\sstate) \cup \oheap(\sstate) } ~\cup \\
    \quad \set{ \pair{p}{\multiple{t}} : \pair{p}{\multiple{t}} \in \sheap(\sstate) \cup \oheap(\sstate) } &\text{otherwise}
  \end{cases}$$
  This is used to symbolically calculate the required exclusion frame.
\end{definition}

\subsection{Verification states}

\begin{definition}
  A \textbf{verification state} $\vstate$ is either an abstract symbol or a triple consisting of a symbolic state, a statement, and a post-condition.
  $$\vstate ::= \initsym \gralt \finalsym \gralt \triple{\sstate}{s}{\gform}$$
\end{definition}

The reachability of verification transitions, as determined by modular verification for a given program $\prog$, is determined by the judgement
$$\strans{\prog}{\vstate}{\vstate'}$$
where $\vstate$ is the beginning verification state and $\vstate'$ is the next verification state.

\semantics[SVerifyInit]
{
}
{
  \strans{\quadruple{s}{M}{P}{S}}{\initsym}{\triple{\quintuple{\bot}{\emptyset}{\emptyset}{\emptyset}{\ktrue}}{s}{\ktrue}}
}
\semantics[SVerifyMethod]
{
  m \in M \\
  \multiple{x} = \fparams(m) \\
  \sproduce{\quintuple{\bot}{\emptyset}{\emptyset}{[\multiple{x \mapsto \ffresh}]}{\ktrue}}{\fpre(m)}{\sstate}
}
{
  \strans{\quadruple{s}{M}{P}{S}}{\initsym}{\triple{\sstate}{\sseq{\fbody(m)}{\kskip}}{\fpost(m)}}
}

\semantics[SVerifyLoopBody]
{
  \strans{\prog}{\_}{\triple{\sstate_0}{\sseq{\swhile{e}{\gform}{s}}{s'}}{\gform_0}} \\
  \multiple{x} = \fmodified(s) \\
  \sproduce{\quintuple{\bot}{\senv(\sstate_0)[\multiple{x \mapsto \ffresh}]}{\emptyset}{\emptyset}{\pc(\sstate_0)}}{\gform}{\sstate_0'} \\
  \seval{\sstate_0'}{e}{t}{\_}{\scheck}
}
{
  \strans{\prog}{\triple{\sstate}{\sseq{\swhile{e}{\gform}{s}}{s'}}{\gform_0}}{\triple{\sstate_0'[\pc = \pc(\sstate_0') \kand t]}{\sseq{s}{\kskip}}{\gform}}
}
\semantics[SVerifyLoop]
{
  \strans{\prog}{\_}{\triple{\sstate}{\sseq{\swhile{e}{\gform'}{s}}{s'}}{\gform}} \\
  \scons{\sstate}{\gform'}{\sstate'}{\_} \\
  \multiple{x} = \fmodified(s) \\
  \sproduce{\sstate'[\senv = \senv(\sstate)[\multiple{x \mapsto \ffresh}]]}{\gform'}{\sstate''}
}
{
  \prog \vdash \triple{\sstate}{\sseq{\swhile{e}{\gform'}{s}}{s'}}{\gform} \to \\\\
  \hspace{3em} \triple{\sstate''}{\sseq{\swhile{e}{\gform'}{s}}{s'}}{\gform}
}
\semantics[SVerifyStep]
{
  \strans{\prog}{\_}{\triple{\sstate}{s}{\gform}} \\
  \sexec{\sstate}{s}{s'}{\sstate'}
}
{
  \strans{\prog}{\triple{\sstate}{s}{\gform}}{\triple{\sstate'}{s'}{\gform}}
}
\semantics[SVerifyFinal]
{
  \strans{\prog}{\_}{\triple{\sstate}{\kskip}{\gform}} \\
  \scons{\sstate}{\gform}{\sstate'}{\_}
}
{
  \prog \vdash \triple{\sstate}{\kskip}{\gform} \to \finalsym
}

\begin{definition}\label{def:vstate-reachable}
  A verification state $\vstate$ is \textbf{reachable} from program $\prog$ if $\vstate = \initsym$ or $\prog \vdash \_ \to \vstate, \_, \_, \_$.

  A verification state $\vstate$ is \textbf{reachable} from program $\prog$ \textbf{with valuation} $V$ if $\vstate$ is reachable from $\prog$ and $V$ is defined for all symbolic values contained in $\vstate$.
\end{definition}

\subsection{Guards}

A guard judgement determines the set of runtime checks that must be satisfied at the given verification state before taking the next execution step, and determines the exclusion frame required to preserve the validity of heap chunks contained by the current symbolic state. This is denoted by the judgement
$$\vstate \rightharpoonup \sstate', \scheck, \sperms$$
where $\vstate$ is the current verification state, $\sstate'$ is the intermediate state, $\scheck$ is the set of required checks, and $\sperms$ is the exclusion frame (represented as a set of symbolic permissions).

\semantics[SGuardInit]
  {
  }
  {
    \sguard{\initsym}{\quintuple{\bot}{\emptyset}{\emptyset}{\emptyset}{\ktrue}}{\emptyset}{\emptyset}
  }
\semantics[SGuardSeq]
  {
  }
  {
    \sguard{\triple{\sstate}{\sseq{\kskip}{s}}{\gform}}{\sstate}{\emptyset}{\emptyset}
  }
\semantics[SGuardAssign]
  {
    \seval{\sstate}{e}{\_}{\sstate'}{\scheck}
  }
  {
    \sguard{\triple{\sstate}{\sseq{x = e}{s}}{\gform}}{\sstate}{\scheck}{\emptyset}
  }
\semantics[SGuardAssignField]
  {
    \seval{\sstate}{e}{\_}{\sstate'}{\scheck'} \\
    \scons{\sstate'}{\kacc(x.f)}{\sstate''}{\scheck''}
  }
  {
    \sguard{\triple{\sstate}{\sseq{x.f = e}{s}}{\gform}}{\sstate''}{\scheck' \cup \scheck''}{\emptyset}
  }
\semantics[SGuardAlloc]
  {
  }
  {
    \sguard{\triple{\sstate}{\sseq{x = \salloc{S}}{s}}{\gform}}{\sstate}{\emptyset}{\emptyset}
  }
\semantics[SGuardCall]
  {
    \multiple{
      \seval{\sstate}{e}{t}{\sstate'}{\scheck}
    } \\
    \multiple{x} = \fparams(m) \\
    \scons{\sstate'[\senv = [\multiple{x \mapsto t}]]}{\fpre(m)}{\sstate''}{\scheck'}
  }
  {
    \sguard{\triple{\sstate}{\sseq{y \kassign m(\multiple{e})}{s}}{\gform}}{\sstate''[\senv = \senv(\sstate)]}{\multiple{\scheck} \cup \scheck'}{\frem(\sstate'', \fpre(m))}
  }
\semantics[SGuardAssert]
  {
    \scons{\sstate}{\simprecise{\phi}}{\sstate'}{\scheck}
  }
  {
    \sguard{\triple{\sstate}{\sseq{\sassert{\phi}}{s}}{\gform}}{\sstate'}{\scheck}{\emptyset}
  }
\semantics[SGuardFold]
  {
    \multiple{\seval{\sstate}{e}{t}{\sstate'}{\scheck}} \\
    \multiple{x} = \fpredparams(p) \\
    \scons{\sstate'[\senv = [\multiple{x \mapsto t}]]}{\fpred(p)}{\sstate''}{\scheck'}
  }
  {
    \sguard{\triple{\sstate}{\sseq{\sfold{p(\multiple{e})}}{s}}{\gform}}{\sstate''[\senv = \senv(\sstate)]}{\scheck' \cup \bigcup \multiple{\scheck}}{\emptyset}
  }
\semantics[SGuardUnfold]
  {
    \multiple{\seval{\sstate}{e}{t}{\sstate'}{\scheck}} \\
    \scons{\sstate'}{p(\multiple{e})}{\sstate''}{\scheck'}
  }
  {
    \sguard{\triple{\sstate}{\sseq{\sunfold{p(\multiple{e})}}{s}}{\gform}}{\sstate''}{\scheck' \cup \bigcup \multiple{\scheck}}{\emptyset}
  }
\semantics[SGuardIf]
  {
    \seval{\sstate}{e}{\_}{\sstate'}{\scheck}
  }
  {
    \sguard{\triple{\sstate}{\sseq{\sif{e}{s_1}{s_2}}{s}}{\gform}}{\sstate'}{\scheck}{\emptyset}
  }
  
\semantics[SGuardWhile]
  {
    \scons{\sstate}{\gform}{\sstate'}{\scheck'} \\
    \multiple{x} = \fmodified(s) \\
    \sproduce{\sstate'[\senv = \senv(\sstate')[\multiple{x \mapsto \ffresh}]]}{\gform}{\sstate''} \\
    \seval{\sstate''}{e}{\_}{\_}{\scheck''}
  }
  {
    \sguard{\triple{\sstate}{\sseq{\swhile{e}{\gform}{s}}{s'}}{\gform'} }{\sstate'[\pc = \pc(\sstate'')]}{\scheck' \cup \scheck''}{\frem(\sstate', \gform)}
  }

\semantics[SGuardFinish]
  {
    \scons{\sstate}{\gform}{\sstate'}{\scheck}
  }
  {
    \sguard{\triple{\sstate}{\kskip}{\gform}}{\sstate'}{\scheck}{\emptyset}
  }

\subsection{Valid states}

\begin{definition}\label{def:vstate-corresponds}
  A verification state $\vstate$ \textbf{corresponds} with valuation $V$ to an execution state $\Gamma$ if $\vstate = \Gamma$, or $\Gamma = \pair{\heap}{\triple{\perms}{\env}{s} \cdot \stack}$ for some $\heap, \perms, \env, s$ and $\vstate = \triple{\sstate}{s}{\_}$ for some $\sstate$, such that $\simstate{V}{\sstate}{\heap}{\perms}{\env}$.
\end{definition}

\begin{definition}\label{def:partial-valid}
  A \textit{partial state} $\Gamma = \pair{\heap}{\stack}$ is \textbf{validated} by a verification state $\vstate$ with valuation $V$ if one of the following cases apply:

  \begin{defparts}
    \defpart\label{def:partial-valid-nil}
      $\stack = \nilsym$

    \defpart\label{def:partial-valid-call}
      $\stack = \triple{\env}{\perms}{\sseq{y \kassign m(e_1, \cdots, e_k)}{s}} \cdot \stack^*$ for some $\env$, $\perms$, $y$, $m$, $k$, $e_1, \cdots, e_k$, $s$, and $\stack^*$, and there exists some $\vstate'$, $V'$, $x_1, \cdots, x_k$, $t_1, \cdots, t_k$, $\sstate_0, \cdots, \sstate_k$ and $\sstate'$ such that:
      \begin{gather}
        \text{The partial state $\pair{\heap}{\stack^*}$ is validated by $\vstate'$ and $V'$}, \\
        \vstate' \text{ is reachable from $\prog$ with valuation $V'$}, \quad s(\vstate') = s(\stack), \\
        x_1, \cdots, x_k = \fparams(m), \\
        \sstate_0 = \sstate(\vstate'), \quad \seval{\sstate_0}{e_1}{t_1}{\sstate_1}{\_}, \quad\cdots,\quad \seval{\sstate_{k-1}}{e_k}{t_k}{\sstate_k}{\_}, \label{eq:partial-valid-call-eval}\\
        \universal{1 \le i \le k}{V(\senv(\vstate)(x_i)) = V'(t_i)}, \label{eq:partial-valid-call-params} \\
        \scons{\sstate_k}{\fpre(m)}{\sstate'}{\_}, \quad \simstate{V'}{\sstate'[\senv = \senv(\sstate_0)]}{\heap}{\perms}{\env}, \quad\text{and} \label{eq:partial-valid-call-sim} \\
        \gform(\vstate) = \fpost(m)
      \end{gather}

    \defpart\label{def:partial-valid-while}
      $\stack = \triple{\env}{\perms}{\sseq{\swhile{e}{\gform}{s}}{s'}} \cdot \stack^*$ for some $\env$, $\perms$, $e$, $\gform$, $s$, $s'$, $\stack^*$, and there exists some $\vstate'$, $V'$, and $\sstate'$ such that:
      \begin{gather}
        \text{The partial state $\pair{\heap}{\stack^*}$ is validated by $\vstate'$ and $V'$} \label{eq:partial-valid-while-valid}\\
        \vstate' \text{ is reachable from $\prog$ with valuation $V'$}, \quad s(\vstate') = s(\stack) \label{eq:partial-valid-while-reachable}\\
        \scons{\sstate}{\gform}{\sstate'}{\_}, \quad\text{and}\quad
        \simstate{V'}{\sstate'}{\heap}{\perms}{\env} \label{eq:partial-valid-while-sim} \\
        \gform(\vstate) = \gform \label{eq:partial-valid-while-post}
      \end{gather}
  \end{defparts}
\end{definition}

\begin{definition}\label{def:state-valid}
  For a program $\prog$, a dynamic state $\Gamma$ is \textbf{validated} by $\vstate$ and valuation $V$ if all the following are true:
  
  \begin{defparts}
    \defpart\label{def:state-valid-reachable} $\vstate$ is reachable from $\prog$ with $V$
    
    \defpart\label{def:state-valid-correspond} $\Gamma$ corresponds to $\vstate$ with $V$

    \defpart\label{def:state-valid-partial} If $\Gamma = \pair{\heap}{\triple{\perms}{\env}{s} \cdot \stack^*}$,
    then the partial state $\pair{\heap}{\stack^*}$ is validated by $\vstate$ and $V$.

  \end{defparts}
\end{definition}

\begin{definition}
  $\Gamma$ is a \textbf{valid state} if $\Gamma$ is validated by some $\vstate$.
\end{definition}
  




\section{Soundness}
\label{sec:soundness}

\definecolor{shadecolor}{gray}{0.9}
\begin{snugshade}
In this version of the proof, we forgo the deterministic evaluation judgement used by \citet{zimmerman2024extended} to more closely align with the implementation of Gradual Viper. Lemmas 1 to 27 in our proof have the same numbering as in the proof of soundness in \cite{zimmerman2024extended}. Lemma 28 is a modified version of lemma 29 in \cite{zimmerman2024extended}. Lemmas 29 and above in our proof correspond to lemmas 31 and above in \cite{zimmerman2024extended}.

Lemmas 1, 14, 24, 25, 26, 27, 28, 49, 50, 54 have been modified;
lemmas and theorems stated without proof have the same proof as \citet{zimmerman2024extended}.
\end{snugshade}

\subsection{Cross-cutting lemmas}
\label{sec:cross-cutting}

\begin{lemma}[Relating expression framing and exact footprints]\label{lem:efoot-framing}
    If $\frm{\heap}{\perms}{\env}{e}$ then $\efoot{\heap}{\env}{e} \subseteq \perms$, and conversely, if $\frm{\heap}{\perms}{\env}{e}$ and $\efoot{\heap}{\env}{e} \subseteq \perms'$ then $\frm{\heap}{\perms'}{\env}{e}$.
\end{lemma}

\begin{proof}
    By induction on the derivation of $\frm{\heap}{\perms}{\env}{e}$.

    \begin{enumcases}
        \case $l$: Then $\efoot{\heap}{\env}{l} = \emptyset$, thus $\efoot{\heap}{\env}{l} \subseteq \perms$, and $\frm{\heap}{\perms}{\env}{l}$ by \refrule{FrameLiteral}.
        \case $v$: Then $\efoot{\heap}{\env}{v} = \emptyset$, thus $\efoot{\heap}{\env}{v} \subseteq \perms$, and $\frm{\heap}{\perms}{\env}{v}$ by \refrule{FrameVar}.
        \case $e.f$:
    
          Suppose that $\efoot{\heap}{\env}{e.f} \subseteq \perms$. Then by definition $\efoot{\heap}{\env}{e} \subseteq \perms$ and $\pair{\ell}{f} \in \perms$ for some $\ell$ such that $\eval{\heap}{\env}{e}{\ell}$. By induction $\frm{\heap}{\perms}{\env}{e}$, by \refrule{AssertAcc} $\assertion{\heap}{\perms}{\env}{\kacc(e.f)}$, and thus by \refrule{FrameField}, $\frm{\heap}{\perms}{\env}{e.f}$.
    
          Suppose that $\frm{\heap}{\perms}{\env}{e.f}$. Then by \refrule{FrameField}, $\frm{\heap}{\perms}{\env}{e}$ and $\assertion{\heap}{\perms}{\env}{\kacc(e.f)}$, thus by induction $\efoot{\heap}{env}{e} \subseteq \perms$ and by \refrule{AssertAcc}, $\pair{\ell}{f} \in \perms$ for some $\ell$ such that $\eval{\heap}{\env}{e}{\ell}$. Thus $\efoot{\heap}{\env}{e.f} = \efoot{\heap}{\env}{e}; \pair{\ell}{f} \subseteq \perms$.
    
        \case $e_1 \oplus e_2$:
          \begin{align*}
            &\efoot{\heap}{\env}{e_1 \oplus e_2} \subseteq \perms \\
            &\iff \efoot{\heap}{\env}{e_1} \subseteq \perms ~\text{and}~ \efoot{\heap}{\env}{e_2} \subseteq \perms &\text{by definition} \\
            &\iff \frm{\heap}{\perms}{\env}{e_1} ~\text{and}~ \frm{\heap}{\perms}{\env}{e_2} &\text{by induction} \\
            &\iff \frm{\heap}{\perms}{\env}{e_1 \oplus e_2} &\text{by \refrule{FrameOp}}
          \end{align*}
    
        \case\label{case:efoot-framing-or} $e_1 \kor e_2$:
          \begin{align*}
            &\efoot{\heap}{\env}{e_1 \kor e_2} \subseteq \perms \implies  \\
            &~\textbf{either}~ \eval{\heap}{\env}{e_1}{\ktrue} ~\text{and}~ \efoot{\heap}{\env}{e_1} \subseteq \perms &\text{by definition}\\
            &\quad\implies \eval{\heap}{\env}{e_1}{\ktrue} ~\text{and}~ \frm{\heap}{\perms}{\env}{e_1} &\text{by induction} \\
            &\quad\implies \frm{\heap}{\perms}{\env}{e_1 \kor e_2} &\text{by \refrule{FrameOrA}} \\
            &~\textbf{or}~ \eval{\heap}{\env}{e_1}{\kfalse} ~\text{and}~ \efoot{\heap}{\env}{e_1} \cup \efoot{\heap}{\env}{e_2} \subseteq \perms &\text{by definition}\\
            &\quad\implies \eval{\heap}{\env}{e_1}{\kfalse}, ~\frm{\heap}{\perms}{\env}{e_1}, \\
              &\hspace{3.5em} \text{and}~ \frm{\heap}{\perms}{\env}{e_2} &\text{by induction} \\
            &\quad\implies \frm{\heap}{\perms}{\env}{e_1 \kor e_2} &\text{by \refrule{FrameOrB}}
          \end{align*}
          \begin{align*}
            &\frm{\heap}{\perms}{\env}{e_1 \kor e_2} \implies \\
            &~\textbf{either}~ \eval{\heap}{\env}{e_1}{\ktrue} ~\text{and}~ \frm{\heap}{\perms}{\env}{e_1} &\text{by \refrule{FrameOrA}} \\
            &\quad\implies \eval{\heap}{\env}{e_1}{\ktrue} ~\text{and}~ \efoot{\heap}{\env}{e_1} \subseteq \perms &\text{by induction} \\
            &\quad\implies \efoot{\heap}{\env}{e_1 \kor e_2} \subseteq \perms &\text{by definition} \\
            &~\textbf{or}~ \eval{\heap}{\env}{e_1}{\kfalse}, ~\frm{\heap}{\perms}{\env}{e_1}, \\
              &\hspace{2.25em}\text{and}~ \frm{\heap}{\perms}{\env}{e_2} &\text{by \refrule{FrameOrB}}\\
            &\quad\implies \eval{\heap}{\env}{e_1}{\kfalse} ~\text{and}~ \efoot{\heap}{\env}{e_1} \cup \efoot{\heap}{\env}{e_2} \subseteq \perms &\text{by induction}\\
            &\quad\implies \efoot{\heap}{\env}{e_1 \kor e_2} \subseteq \perms &\text{by definition}
          \end{align*}
    
        \case $e_1 \kand e_2$: Similar to case \ref{case:efoot-framing-or}.
    
        \case $\kneg e$:
          \begin{align*}
            &\efoot{\heap}{\env}{\kneg e} \subseteq \perms \\
            &\iff \efoot{\heap}{\env}{e} \subseteq \perms &\text{by definition} \\
            &\iff \frm{\heap}{\perms}{\env}{e} &\text{by induction} \\
            &\iff \frm{\heap}{\perms}{\env}{\kneg e} &\text{by \refrule{FrameNeg}}
          \end{align*}

        \definecolor{shadecolor}{gray}{0.9}
        \begin{snugshade} \case \textsc{FrameUnfolding} -- Assume that $\frm{\heap}{\perms}{\env}{\sunfolding{\pe}{e_0}}$. Then, by definition of \refrule{IFramePredicate}, $\multiple{\frm{\heap}{\perms}{\env}{e}}$, so inductively $\bigcup \multiple{\efoot{\heap}{\env}{e}} \subseteq \perms$. Inductively, $\efoot{\heap}{\env}{e_0} \subseteq \perms$. By definition, 

        $$\efoot{\heap}{\env}{\sunfolding{\pe}{e_0}} = \efoot{\heap}{\env}{e_0} \cup \bigcup \multiple{\efoot{\heap}{\env}{e}} \subseteq \perms$$.

        Assume $\frm{\heap}{\perms}{\env}{\fe}$ and $\efoot{\heap}{\env}{\fe} \subseteq \perms'$. By definition, $\ifrm{\heap}{\perms}{\env}{\pe}$, $\assertion{\heap}{\perms}{\env}{\pe}$, and $\frm{\heap}{\perms}{\env}{e_0}$. By the induction hypothesis and \refrule{IFramePredicate}, $\ifrm{\heap}{\perms'}{\env}{\pe}$ and $\frm{\heap}{\perms'}{\env}{e_0}$. Then, by definition of \refrule{AssertPredicate}, $\assertion{\heap}{\perms}{[\multiple{x \mapsto v}]}{\fpred(p)}$ where $\multiple{x} = \fpredparams(p)$ and $\multiple{\eval{\heap}{\env}{e}{v}}$. Then, by Lemma~\ref{lem:efoot-assert} and Lemma~\ref{lem:assert-monotonicity}, $\assertion{\heap}{\perms'}{\env}{\pe}$.
        
        \case \textsc{FrameFunction} -- Assume $\frm{\heap}{\perms}{\env}{\fe}$. Then, by definition, $\ifrm{\heap}{\perms}{[\multiple{x \mapsto v}]}{\ffuncpre(f)}$, $\assertion{\heap}{\perms}{[\multiple{x \mapsto v}]}{\ffuncpre(f)}$, and $\multiple{\frm{\heap}{\perms}{\env}{e}}$. By Definition~\ref{def:well-formed-prog}, this implies $\frm{\heap}{\perms}{[\multiple{x \mapsto v}]}{\ffuncbody(f)}$. Inductively, $\bigcup \multiple{\efoot{\heap}{\env}{e}} \subseteq \perms$ and $\efoot{\heap}{[\multiple{x \mapsto v}]}{\ffuncbody}$. Then, by Lemma~\ref{lem:ifrm-implies-efrm}, $\efrm{\heap}{\perms}{[\multiple{x \mapsto v}]}{\ffuncpre(f)}$, so by Lemma~\ref{lem:efoot-subset-framed}, $\efoot{\heap}{[\multiple{x \mapsto v}]}{\ffuncpre(f)} \subseteq \perms$. Therefore, 
        
        $$\efoot{\heap}{\perms}{\fe} = \efoot{\heap}{[\multiple{x \mapsto v}]}{\ffuncpre(f)} \cup \efoot{\heap}{[\multiple{x \mapsto v}]}{\ffuncbody(f)} \cup \bigcup \multiple{\efoot{\heap}{\env}{e}} \subseteq \perms.$$

        Now, assume $\frm{\heap}{\perms}{\env}{\fe}$ and $\efoot{\heap}{\perms}{\fe} \subseteq \perms'$. Inductively, $\multiple{\frm{\heap}{\perms'}{\env}{e}}$ and $\ifrm{\heap}{\perms'}{\env}{\ffuncpre(f)}$. Since $\assertion{\heap}{\perms}{\env}{\ffuncpre(f)}$ and $\efoot{\heap}{\env}{\ffuncpre(f)} \subseteq \perms'$, by Lemma~\ref{lem:assert-efoot-subset}, $\assertion{\heap}{\perms'}{\env}{\ffuncpre(f)}$. Therefore, $\frm{\heap}{\perms'}{\env}{\fe}$.
        \end{snugshade}
    \end{enumcases}
\end{proof}

\begin{lemma}[Relating formula assertion/framing and exact footprints]\label{lem:efoot-subset-framed}
  If $\assertion{\heap}{\perms'}{\env}{\gform}$, $\efrm{\heap}{\perms}{\env}{\gform}$, and $\perms' \subseteq \perms$, then $\efoot{\heap}{\env}{\gform} \subseteq \perms$.
\end{lemma}

\begin{lemma}[Relating iso- and equi-recursive framing]\label{lem:ifrm-implies-efrm}
  If $\ifrm{\heap}{\perms}{\env}{\gform}$, $\assertion{\heap}{\perms'}{\env}{\gform}$, and $\perms' \subseteq \perms$, then $\efrm{\heap}{\perms}{\env}{\gform}$.
\end{lemma}

\begin{lemma}[Relating specification assertion and exact footprints]\label{lem:efoot-subset-spec}
  If $\gform$ is a specification and $\assertion{\heap}{\perms}{\env}{\gform}$, then $\efoot{\heap}{\env}{\gform} \subseteq \perms$.
\end{lemma}

\begin{lemma}[Relating specification assertion and footprints]\label{lem:foot-subset-spec}
  If $\gform$ is a specification and $\assertion{\heap}{\perms}{\env}{\gform}$, then $\foot{\heap}{\perms}{\env}{\gform} \subseteq \perms$.
\end{lemma}

\begin{lemma}[Relating specification exact footprints and footprints]\label{lem:foot-subset-efoot-spec}
  If $\gform$ is a specification and $\assertion{\heap}{\perms}{\env}{\gform}$, then $\efoot{\heap}{\env}{\gform} \subseteq \foot{\heap}{\perms}{\env}{\gform}$.
\end{lemma}

\begin{lemma}[Monotonicity of expression framing WRT permissions]\label{lem:framing-monotonicity}
  If $\frm{\heap}{\perms}{\env}{e}$ and $\perms \subseteq \perms'$, then $\frm{\heap}{\perms'}{\env}{e}$.
\end{lemma}

\begin{proof}
    By Lemma~\ref{lem:efoot-framing}, $\efoot{\heap}{\env}{e} \subseteq \perms \subseteq \perms'$, so again by Lemma~\ref{lem:efoot-framing}, $\frm{\heap}{\perms'}{\env}{e}$.
\end{proof}

\begin{lemma}[Monotonicity of equi-recursive framing WRT permissions]\label{lem:efrm-monotonicity}
  If $\efrm{\heap}{\perms}{\env}{\gform}$ and $\perms \subseteq \perms'$, then $\efrm{\heap}{\perms'}{\env}{\gform}$.
\end{lemma}

\begin{lemma}[Monotonicity of assertions WRT permissions]\label{lem:assert-monotonicity}
  If $\assertion{\heap}{\perms}{\env}{\gform}$ and $\perms \subseteq \perms'$, then $\assertion{\heap}{\perms'}{\env}{\gform}$.
\end{lemma}

\begin{lemma}[Exact footprint preserves equi-recursive framing]\label{lem:efoot-efrm}
  If $\efrm{\heap}{\perms}{\env}{\gform}$, then \\
  $\efrm{\heap}{\efoot{\heap}{\env}{\gform} \cap \perms}{\env}{\gform}$.
\end{lemma}

\begin{lemma}[Exact footprint preserves assertions]\label{lem:efoot-assert}
  If $\assertion{\heap}{\perms}{\env}{\gform}$, then $\assertion{\heap}{\efoot{\heap}{\env}{\gform} \cap \perms}{\env}{\gform}$.
\end{lemma}

\begin{lemma}[Supersets of exact footprints preserve assertions]\label{lem:assert-efoot-subset}
  If $\assertion{\heap}{\perms}{\env}{\gform}$ and $\efoot{\heap}{\env}{\gform} \subseteq \perms'$, then $\assertion{\heap}{\perms'}{\env}{\gform}$.
\end{lemma}

\begin{lemma}[Footprints preserve specification assertions]\label{lem:foot-assert}
  If $\gform$ is a specification and $\assertion{\heap}{\perms}{\env}{\gform}$, then $\assertion{\heap}{\foot{\heap}{\perms}{\env}{\gform}}{\env}{\gform}$.
\end{lemma}

\begin{lemma}[Modeling implies ownership]\label{lem:sim-heap-contains}
  If $\simstate{V}{\sstate}{\heap}{\perms}{\env}$ then
  $$\universal{h \in \sheap(\sstate) \cup \oheap(\sstate)}{\vfoot{V}{\heap}{h} \subseteq \perms}.$$
\end{lemma}

\begin{proof}
    Let $h \in \sheap(\sstate) \cup \oheap(\sstate)$. Then, by definition of the symbolic and optimistic heaps, one of the following two cases applies:

    \begin{enumcases}
        \case $h$ is a field chunk -- Then, $h = \triple{f}{t}{t'}$. By the definition of state correspondence, $\simheap{V}{\sheap(\sstate)}{\heap}{\perms}$ and $\simheap{V}{\oheap(\sstate)}{\heap}{\perms}$. Therefore, $\pair{V(t)}{f} \in \perms$, so $\vfoot{V}{\heap}{\triple{f}{t}{t'}} = \{ \pair{V(t)}{f} \} \subseteq \perms$.

        \case $h$ is a predicate chunk -- Then, $h = \pair{p}{\multiple{t}}$. Since $\simheap{V}{\sheap(\sstate)}{\heap}{\perms}$ \sethlcolor{newstuff}\hl{and $\simheap{V}{\oheap(\sstate)}{\heap}{\perms}$ (since the optimistic heap can now contain predicate chunks as well)}, $\assertion{\heap}{\perms}{[\multiple{x \mapsto V(t)}]}{\fpred(p)}$ where $\multiple{x} = \fpredparams(p)$. Then, since $\fpred(p)$ is a specification, $\vfoot{V}{\heap}{\pair{p}{\multiple{t}}} = \efoot{\heap}{[\multiple{x \mapsto V(t)}]}{\fpred(p)} \subseteq \perms$ by Lemma~\ref{lem:efoot-subset-spec}.
    \end{enumcases}
\end{proof}

\begin{lemma}[Correspondence with exclusion implies disjointness]\label{lem:sim-heap-disjoint}
  If $\simstate{V}{\sstate}{\heap}{\perms \setminus \perms'}{\env}$, then
  $$\universal{h \in \sheap(\sstate)}{\vfoot{V}{\heap}{h} \cap \perms' = \emptyset}.$$
\end{lemma}

\begin{lemma}[Disjointness implies correspondence with exclusion]\label{lem:disjoint-sim-heap-subset}
  If $\simheap{V}{\sheap}{\heap}{\perms}$ and $\universal{h \in \sheap}{\vfoot{V}{\heap}{h} \cap \perms' = \emptyset}$, then $\simheap{V}{\sheap}{\heap}{\perms \setminus \perms'}$.
\end{lemma}

\begin{lemma}\label{lem:sim-sheap-monotonicity}
  If $\simheap{V}{\sheap}{\heap}{\perms}$ and $\perms \subseteq \perms'$, then $\simheap{V}{\sheap}{\heap}{\perms}$.
\end{lemma}

\begin{lemma}\label{lem:sim-oheap-monotonicity}
  If $\simheap{V}{\oheap}{\heap}{\perms}$ and $\perms \subseteq \perms'$, then $\simheap{V}{\oheap}{\heap}{\perms}$.
\end{lemma}

\begin{lemma}\label{lem:simstate-monotonicity}
  If $\simstate{V}{\sstate}{\heap}{\perms}{\env}$ and $\perms \subseteq \perms'$, then $\simstate{V}{\sstate}{\heap}{\perms'}{\env}$.
\end{lemma}

\begin{lemma}[Statement rearrangement]\label{lem:stmt-rearrangement}
  If $s = s_1; s_2$, then $s = s_1'; s_2'$ such that $s_1'$ is not a sequence statement and $s_2' = s_2$ or $s_2' = s_1''; s_2$ where $s_1 = s_1'; s_1''$.
\end{lemma}

\begin{lemma}[Run-time check monotonicity WRT permissions]\label{lem:scheck-perms-monotonicity-elem}
  If $\rtassert{V}{\heap}{\perms}{r}$ and $\perms \subseteq \perms'$, then $\rtassert{V}{\heap}{\perms'}{r}$.
\end{lemma}

\begin{lemma}[Run-time checks monotonicity WRT permissions]\label{lem:scheck-perms-monotonicity}
  If $\rtassert{V}{\heap}{\perms}{\scheck}$ and $\perms \subseteq \perms'$, then $\rtassert{V}{\heap}{\perms'}{\scheck}$.
\end{lemma}

\begin{lemma}[Run-time check set implies subsets]\label{lem:scheck-monotonicity}
  If $\rtassert{V}{\heap}{\perms}{\scheck}$ and $\scheck' \subseteq \scheck$, then $\rtassert{V}{\heap}{\perms}{\scheck'}$.
\end{lemma}

\subsection{Evaluation}
\label{sec:eval-soundness}

\begin{definition}\label{def:eval-valuation}
  For a judgement $\seval{\sstate}{e}{t}{\sstate'}{\scheck}$, given an initial valuation $V$, heap $\heap$, and environment $\env$, the \textbf{corresponding valuation} is denoted
  $$V[\seval{\sstate}{e}{t}{\sstate'}{\scheck} \mid \heap, \env].$$
  This valuation is defined as follows, depending on the rule that proves the derivation. Values are referenced using the respective name from the rule definition.

  Note that the corresponding valuation always extends the initial valuation, and is defined for all fresh symbolic values in the judgement.

  \begin{itemize}
    \item \refrule{SEvalLiteral}:
      $$V[\seval{\sstate}{l}{l}{\sstate}{\_} \mid \heap, \env] := V$$
    \item \refrule{SEvalVar}:
      $$V[\seval{\sstate}{x}{\_}{\sstate}{\_} \mid \heap, \env] := V$$
    \item \refrule{SEvalNeg}:
      $$V[\seval{\sstate}{\kneg e}{\kneg t}{\sstate'}{\scheck} \mid \heap, \env] := V[\seval{\sstate}{e}{t}{\sstate'}{\scheck} \mid \heap, \env]$$
    \item \refrule{SEvalOrA}:
      $$V[\seval{\sstate}{e_1 \kor e_2}{t_1}{\sstate''}{\scheck} \mid \heap, \env] := V[\seval{\sstate}{e_1}{t_1}{\sstate'}{\scheck} \mid \heap, \env]$$
    \item \refrule{SEvalOrB}:
      \begin{align*}
        &V[\seval{\sstate}{e_1 \kor e_2}{t_2}{\sstate''}{\scheck_1 \cup \scheck_2} \mid \heap, \env] := \\
        &\hspace{2em} V[\seval{\sstate}{e_1}{t_1}{\sstate'}{\scheck_1} \mid \heap, \env] \\
        &\hspace{3em} [\seval{\sstate'[\pc = \pc(\sstate') \kand \kneg t_1]}{e_2}{t_2}{\sstate''}{\scheck_2} \mid \heap, \env]
      \end{align*}
    \item \refrule{SEvalAndA}:
      $$V[\seval{\sstate}{e_1 \kand e_2}{t_1}{\sstate''}{\scheck} \mid \heap, \env] := V[\seval{\sstate}{e_1}{t_1}{\sstate'}{\scheck} \mid \heap, \env]$$
    \item \refrule{SEvalAndB}:
      \begin{align*}
        &V[\seval{\sstate}{e_1 \kand e_2}{t_2}{\sstate''}{\scheck_1 \cup \scheck_2} \mid \heap, \env] := \\
        &\hspace{2em} V[\seval{\sstate}{e_1}{t_1}{\sstate'}{\scheck_1} \mid \heap, \env] \\
        &\hspace{3em} [\seval{\sstate'[\pc = \pc(\sstate') \kand t_1]}{e_2}{t_2}{\sstate''}{\scheck_2} \mid \heap, \env]
      \end{align*}
    \item \refrule{SEvalOp}:
      \begin{align*}
        &V[\seval{\sstate}{e_1 \oplus e_2}{t_1 \oplus t_2}{\sstate''}{\scheck_1 \cup \scheck_2} \mid \heap, \env] := \\
        &\hspace{2em} V[\seval{\sstate}{e_1}{t_1}{\sstate'}{\scheck_1} \mid \heap, \env] \\
        &\hspace{3em} [\seval{\sstate'}{e_2}{t_2}{\sstate''}{\scheck_2} \mid \heap, \env]
      \end{align*}
    \item \refrule{SEvalField} or \refrule{SEvalFieldOptimistic}:
      $$V[\seval{\sstate}{e.f}{t}{\sstate''}{\_} \mid \heap, \env] := V[\seval{\sstate}{e}{t_e}{\sstate'}{\scheck} \mid \heap, \env]$$
    \item \refrule{SEvalFieldImprecise}:
      $$V[\seval{\sstate}{e.f}{t}{\sstate''}{\_} \mid \heap, \env] := V[\seval{\sstate}{e}{t_e}{\sstate'}{\scheck} \mid \heap, \env][t \mapsto \heap(V(t_e), f)]$$
    \item \refrule{SEvalFieldFailure}:
      $$V[\seval{\sstate}{e.f}{t}{\sstate'}{\_} \mid \heap, \env] := V[\seval{\sstate}{e}{t_e}{\sstate'}{\scheck} \mid \heap, \env][t \mapsto \heap(V(t_e), f)]$$
    \definecolor{shadecolor}{gray}{0.9}
    \begin{snugshade} \item \refrule{SEvalUnfoldingPrecise}, \refrule{SEvalUnfoldingImpreciseA}, or \refrule{SEvalUnfoldingImpreciseB}:
        \begin{align*}
           &V[\seval{\sstate_1}{\sunfolding{\pe}{e_0}}{t}{\sstate_5[\ldots]}{\scheck} \mid \heap, \env] := \\ 
           &\hspace{2em} V[\multiple{\seval{\sstate_1}{e}{t}{\sstate_2}{\scheck_1}} \mid \heap, \env] \\
           &\hspace{3em} [\scons{\sstate_2}{\pe}{\sstate_3}{\scheck_2} \mid \heap, \env] \\
           &\hspace{3em} [\sproduce{\sstate_3[\ldots]}{\fpred(p)}{\sstate_4} \mid \heap, \env] \\
           &\hspace{3em} [\seval{\sstate_4}{e_0}{t_0}{\sstate_5}{\scheck_3} \mid \heap, \env]
        \end{align*}
    \item \refrule{SEvalUnfoldingImplicitPrecise} or \refrule{SEvalUnfoldingImplicitImprecise}:
        $$V[\seval{\sstate_1}{\sunfolding{\pe}{e_0}}{t}{\sstate_2}{\scheck} \mid \heap, \env] := 
            V[\multiple{\seval{\sstate_1}{e}{t}{\sstate_2}{\scheck}} \mid \heap, \env][t \mapsto v]$$
        where $\eval{\heap}{\env}{e_0}{v}$.
    \item \refrule{SEvalFunctionExplicit}:
        \begin{align*}
          &V[\seval{\sstate_1}{\fe}{t}{\sstate_6}{\scheck_1 \cup \scheck_2 \cup \scheck_3} \mid \heap, \env] := \\
          &\hspace{2em} V[\multiple{\seval{\sstate_1}{e}{t}{\sstate_2}{\scheck_1}} \mid \heap] \\
          &\hspace{3em} [\scons{\sstate_2[\senv = \senv(\sstate_2)[\multiple{x \mapsto t}]]}{\fpre(f)}{\sstate_3}{\scheck_2} \mid \heap, \env] \\
          &\hspace{3em} [\sproduce{\sstate_3}{\fpre(f)}{\sstate_4} \mid \heap, \env] \\
          &\hspace{3em} [\seval{\sstate_4[\sheap = \sheap', \oheap = \oheap(\sstate_2), \svis = \svis(\sstate_3) \cup \{ f \}]}{\fbody(f)}{t'}{\sstate_5}{\scheck_3} ~|~ \heap, \env] \\
          &\hspace{3em} [\fts \mapsto V(t')]
        \end{align*}
    \item \refrule{SEvalFunctionImplicit}:
         \begin{align*}
           &V[\seval{\sstate_1}{\fe}{t}{\sstate_4}{\scheck_1 \cup \scheck_2} \mid \heap, \env] := \\
           &\hspace{2em} V[\multiple{\seval{\sstate_1}{e}{t}{\sstate_2}{\scheck_1}} \mid \heap, \env] \\
           &\hspace{3em} [\scons{\sstate_2[\senv = \senv(\sstate_2)[\multiple{x \mapsto t}]]}{\fpre(f)}{\sstate_3}{\scheck_2} \mid \heap, \env]
           &\hspace{3em} [\fts \mapsto v]
         \end{align*}
         where $\eval{\heap}{[x \mapsto V(t)]}{\ffuncbody(f)}{v}$.
    \end{snugshade}
  \end{itemize}
\end{definition}

\begin{lemma}\label{lem:eval-subpath}
  If $\seval{\sstate}{e}{\_}{\sstate'}{\_}$ then $\pc(\sstate') \implies \pc(\sstate)$.
\end{lemma}

\begin{proof}
    We proceed by induction on the derivation of $\seval{\sstate}{e}{\_}{\sstate'}{\_}$.

    \begin{enumcases}
        \case \refrule{SEvalLiteral} -- $\seval{\sstate}{l}{\_}{\sstate}{\_}$; \refrule{SEvalVar} -- $\seval{\sstate}{x}{\_}{\sstate}{\_}$: Trivial since $\pc(\sstate) \implies \pc(\sstate)$.
    
        \case\label{case:eval-subpath-ora} \refrule{SEvalOrA} -- $\seval{\sstate}{e_1 \vee e_2}{\_}{\sstate''}{\_}$:
          By \refrule{SEvalOrA} $\seval{\sstate}{e_1}{\_}{\sstate'}{\_}$, thus by induction $\pc(\sstate') \implies \pc(\sstate)$. Therefore $\pc(\sstate'') = \pc(\sstate') \wedge t_1 \implies \pc(\sstate') \implies \pc(\sstate)$.
    
        \case\label{case:eval-subpath-orb} \refrule{SEvalOrB} -- $\seval{\sstate}{e_1 \vee e_2}{\_}{\sstate''}{\_}$:
          By \refrule{SEvalOrB} $\seval{\sstate}{e_1}{\_}{\sstate'}{\_}$ and $\seval{\sstate'}{e_2}{\_}{\sstate''}{\_}$, thus by induction $\pc(\sstate'') \implies \pc(\sstate') \implies \pc(\sstate)$. Therefore $\pc(\sstate''') = \pc(\sstate'') \kand \kneg t_1 \implies \pc(\sstate'') \implies \pc(\sstate)$.
    
        \case \refrule{SEvalAndA} -- $\seval{\sstate}{e_1 \wedge e_2}{\_}{\sstate''}{\_}$: Similar to case \ref{case:eval-subpath-ora}.
    
        \case \refrule{SEvalAndB} -- $\seval{\sstate}{e_1 \wedge e_2}{\_}{\sstate''}{\_}$: Similar to case \ref{case:eval-subpath-orb}.
    
        \case \refrule{SEvalOp} -- $\seval{\sstate}{e_1 \oplus e_2}{\_}{\sstate''}{\_}$; \refrule{SEvalNeg} -- $\seval{\sstate}{\kneg e}{\_}{\sstate'}{\_}$; \refrule{SEvalField}, \refrule{SEvalFieldOptimistic} -- $\seval{\sstate}{e.f}{\_}{\sstate'}{\_}$; \refrule{SEvalFieldImprecise} -- $\seval{\sstate}{e.f}{\_}{\sstate''}{\_}$: By induction.

        \definecolor{shadecolor}{gray}{0.9}
        \begin{snugshade} \case \refrule{SEvalUnfoldingPrecise} -- By definition of \textsc{SEvalUnfoldingPrecise}, $\pc(\sstate_6) \implies \pc(\sstate_5)$ because $\pc(\sstate_6) = \pc(\sstate_5)$. Then, $\pc(\sstate_5) \implies \pc(\sstate_4)$ by the induction hypothesis, $\pc(\sstate_4) \implies \pc(\sstate_3)$ by Lemma 31, $\pc(\sstate_3) \implies \pc(\sstate_2)$ by Lemma 36, and $\pc(\sstate_2) \implies \pc(\sstate_1)$ by the induction hypothesis. Therefore, $\pc(\sstate_6) \implies \pc(\sstate_1)$ transitively.
        
        \case \refrule{SEvalUnfoldingImpreciseA} or \textsc{SEvalUnfoldingImpreciseB} -- By definition, $\pc(\sstate_6) \implies \pc(\sstate_5)$ because $\pc(\sstate_6) = \pc(\sstate_5)$. Then, $\pc(\sstate_5) \implies \pc(\sstate_4)$ by the induction hypothesis, $\pc(\sstate_4) \implies \pc(\sstate_3)$ by Lemma 31, $\pc(\sstate_3) \implies \pc(\sstate_2)$ by Lemma 36, and $\pc(\sstate_2) \implies \pc(\sstate_1)$ by the induction hypothesis. Therefore, $\pc(\sstate_6) \implies \pc(\sstate_1)$ transitively.
        
        \case \refrule{SEvalUnfoldingImplicitPrecise} -- By induction, $\pc(\sstate_1) \implies \pc(\sstate_2)$.

        \case \refrule{SEvalUnfoldingImplicitImprecise} -- By induction and \refrule{SEvalUnfoldingImplicitImprecise}, $\pc(\sstate_1) \implies \pc(\sstate_2) = \pc(\sstate_3)$.
        
        \case \refrule{SEvalFunctionExplicit} -- Inductively, $\pc(\sstate_2) \implies \pc(\sstate_1)$, and by definition, $\pc(\sstate_6) \implies \pc(\sstate_2)$. Transitively, $\pc(\sstate_6) \implies \pc(\sstate_1)$ as desired.

        \case \refrule{SEvalFunctionImplicit} -- Inductively, $\pc(\sstate_2) \implies \pc(\sstate_1)$, and by definition, $\pc(\sstate_4) \implies \pc(\sstate_2)$. Transitively, $\pc(\sstate_4) \implies \pc(\sstate_1)$ as desired.
        \end{snugshade}
    \end{enumcases}
\end{proof}

\begin{lemma}\label{lem:eval-unchanged}
  If $\seval{\sstate}{e}{t}{\sstate'}{\scheck}$ then $\sheap(\sstate') = \sheap(\sstate)$, and $\senv(\sstate') = \senv(\sstate)$.
\end{lemma}

\begin{proof}
    We proceed by induction on the derivation of $\seval{\sstate}{e}{t}{\sstate'}{\scheck}$.

    \begin{enumcases}
        \case \refrule{SEvalLiteral}; \refrule{SEvalVar}; \refrule{SEvalNeg}; \refrule{SEvalOrA}; \refrule{SEvalOrB}; \refrule{SEvalAndA}; \refrule{SEvalAndB}; \refrule{SEvalOp}; \refrule{SEvalField}; \refrule{SEvalFieldOptimistic}; \refrule{SEvalFieldImprecise}; \refrule{SEvalFieldFailure}: Trivial by induction on $\seval{\sstate}{e}{t}{\sstate'}{\scheck}$.
        \definecolor{shadecolor}{gray}{0.9}
        \begin{snugshade} \case \refrule{SEvalUnfoldingPrecise} -- By definition, $\sheap(\sstate_6) = \sheap(\sstate_2)$, and by the induction hypothesis, $\sheap(\sstate_2) = \sheap(\sstate_1)$. By definition and the induction hypothesis, $\senv(\sstate_6) = \senv(\sstate_5) = \senv(\sstate_3)$. By the definition of \textsc{SConsumePredicate}, $\senv(\sstate_3) = \senv(\sstate_2)$. Then, by the induction hypothesis, $\senv(\sstate_2) = \senv(\sstate_1)$.
        
        \case \refrule{SEvalUnfoldingImpreciseA} -- By definition, $\sheap(\sstate_6) = \sheap(\sstate_2)$, and by the induction hypothesis, $\sheap(\sstate_2) = \sheap(\sstate_1)$.
        
        By definition and the induction hypothesis, $\senv(\sstate_6) = \senv(\sstate_5) = \senv(\sstate_3)$. By the definition of \textsc{SConsumePredicate}, $\senv(\sstate_3) = \senv(\sstate_2)$. Then, by the induction hypothesis, $\senv(\sstate_2) = \senv(\sstate_1)$.
        
        \case \refrule{SEvalUnfoldingImpreciseB} -- By definition, $\sheap(\sstate_6) = \sheap(\sstate_2)$, and by the induction hypothesis, $\sheap(\sstate_2) = \sheap(\sstate_1)$.
        By definition and the induction hypothesis, $\senv(\sstate_6) = \senv(\sstate_5) = \senv(\sstate_3)$. By the definition of \textsc{SConsumePredicate}, $\senv(\sstate_3) = \senv(\sstate_2)$. Then, by the induction hypothesis, $\senv(\sstate_2) = \senv(\sstate_1)$.
        
        \case \refrule{SEvalUnfoldingImplicitPrecise} -- Trivially, $\gamma(\sigma) = \gamma(\sigma)$, and $\mathsf{H}(\sigma) = \mathsf{H}(\sigma)$.

        \case \refrule{SEvalUnfoldingImplicitImprecise} -- Similar to \refrule{SEvalUnfoldingImplicitPrecise}.

        \case \refrule{SEvalFunctionExplicit} -- Inductively, $\sheap(\sstate_2) = \sheap(\sstate_1)$ and $\senv(\sstate_2) = \senv(\sstate_1)$. By definition, $\sheap(\sstate_6) = \sheap(\sstate_2)$ and $\senv(\sstate_6) = \senv(\sstate_2)$.

        \case \refrule{SEvalFunctionImplicit} -- Inductively, $\sheap(\sstate_2) = \sheap(\sstate_1)$ and $\senv(\sstate_2) = \senv(\sstate_1)$. By definition, $\sheap(\sstate_4) = \sheap(\sstate_2)$ and $\senv(\sstate_4) = \senv(\sstate_2)$.
        \end{snugshade}
    \end{enumcases}
\end{proof}

\begin{lemma}[Soundness]\label{lem:seval-soundness}
  Let $V$ be some initial valuation and $\triple{\heap}{\perms}{\env}$ be some well-formed evaluation state such that $\simstate{V}{\sstate}{\heap}{\perms}{\env}$.

  If $\seval{\sstate}{e}{t}{\sstate'}{\scheck}$, $\rtassert{V'}{\heap}{\perms}{\scheck}$, and $V'(\pc(\sstate')) = \ktrue$ where $V' \supseteq V[\seval{\sstate}{e}{t}{\sstate'}{\scheck} \mid \heap, \env]$, then
  $$\simstate{V'}{\sstate'}{\heap}{\perms}{\env}, \quad
    \eval{\heap}{\env}{e}{V'(t)}, \quad \text{and}~
    \frm{\heap}{\perms}{\env}{e}.$$
\end{lemma}

\begin{proof}
    By induction on $\seval{\sstate}{e}{t}{\sstate'}{\scheck}$.  
    
    \begin{enumcases}
        \case \refrule{SEvalLiteral} -- $\seval{\sstate}{l}{l}{\sstate}{\emptyset}$:
          Then $V' = V$, thus $\simstate{V'}{\sstate}{\heap}{\perms}{\env}$ by assumption.
    
          By \refrule{EvalLiteral} $\eval{\heap}{\env}{l}{l}$, and by definition $V'(l) = l$.
    
          By \refrule{FrameLiteral} $\frm{\heap}{\perms}{\env}{l}$.
    
        \case \refrule{SEvalVar} -- $\seval{\sstate}{x}{\senv(\sstate)(x)}{\sstate}{\emptyset}$:
          Then $V' = V$, thus $\simstate{V'}{\sstate}{\heap}{\perms}{\env}$ by assumption.
    
          By \refrule{EvalVar} $\eval{\heap}{\env}{x}{\env(x)}$, and $\env(x) = V'(\senv(\sstate)(x))$ since $\simenv{V'}{\senv(\sstate)}{\env}$.
    
          By \refrule{FrameVar} $\frm{\heap}{\perms}{\env}{x}$.
    
        \case\label{case:seval-sim-ora} \refrule{SEvalOrA} -- $\seval{\sstate}{e_1 \kor e_2}{t_1}{\sstate''}{\scheck}$:
    
          By \refrule{SEvalOrA} $\seval{\sstate}{e_1}{t_1}{\sstate'}{\scheck}$.
    
          Suppose $\rtassert{V'}{\heap}{\perms}{\scheck}$ and $V'(\pc(\sstate'')) = \ktrue$. Then $V'(\pc(\sstate')) = \ktrue$ since $\pc(\sstate'') = \pc(\sstate') \kand t_1 \implies \pc(\sstate')$ by \refrule{SEvalOrA}. Then by induction $\simstate{V'}{\sstate'}{\heap}{\perms}{\env}$.
    
          Now $\simstate{V'}{\sstate''}{\heap}{\perms}{\env}$ since $\sstate'' = \sstate'[\pc = \pc(\sstate') \kand t_1]$ by \refrule{SEvalOrA} and $V'(\pc(\sstate'')) = \ktrue$ by assumption.
    
          By induction $\eval{\heap}{\env}{e_1}{V'(t_1)}$. But now $V'(t_1) = \ktrue$ since $\pc(\sstate'') = \pc(\sstate') \kand t_1 \implies t_1$. Thus $\eval{\heap}{\env}{e_1}{\ktrue}$.
    
          Then by \refrule{EvalOrA} $\eval{\heap}{\env}{e_1 \kor e_2}{\ktrue}$ and $V'(t_1) = \ktrue$.
    
          By induction $\frm{\heap}{\perms}{\env}{e_1}$. Thus $\frm{\heap}{\perms}{\env}{e_1 \kor e_2}$ by \refrule{FrameOrA} since $\eval{\heap}{\env}{e_1}{\ktrue}$.
    
        \case\label{case:seval-sim-orb} \refrule{SEvalOrB} -- $\seval{\sstate}{e_1 \vee e_2}{t_2}{\sstate''}{\scheck_1 \cup \scheck_2}$:
    
          By \refrule{SEvalOrB} $\seval{\sstate}{e_1}{t_1}{\sstate'}{\scheck_1}$ and $\seval{\hat{\sstate}'}{e_2}{t_2}{\sstate''}{\scheck_2}$, where $\hat{\sstate}' = \sstate'[\pc = \pc(\sstate') \kand \kneg t_1]$.
    
          Now suppose $\rtassert{V'}{\heap}{\perms}{\scheck_1 \cup \scheck_2}$ and $V'(\pc(\sstate'')) = \ktrue$. By lemma \ref{lem:scheck-monotonicity} $\rtassert{V'}{\heap}{\perms}{\scheck_1}$ (thus $\rtassert{V'}{\heap}{\perms}{\scheck_1}$) and $\rtassert{V'}{\heap}{\perms}{\scheck_2}$.
    
          Also, by lemma \ref{lem:eval-subpath}, $\pc(\sstate'') \implies \pc(\hat{\sstate}') = \pc(\sstate') \kand \kneg t_1 \implies \pc(\sstate')$. Therefore $V'(\pc(\sstate')) = V'(\pc(\sstate')) = \ktrue$, and $V'(\pc(\sstate'')) = \ktrue$ by assumption.
    
          Now by induction $\simstate{V'}{\sstate'}{\heap}{\perms}{\env}$. Also $V'(\pc(\hat{\sstate}')) = \ktrue$ since $\pc(\sstate'') \implies \pc(\sstate')$. Therefore $\simstate{V'}{\hat{\sstate}'}{\heap}{\perms}{\env}$. Then also by induction $\simstate{V'}{\sstate''}{\heap}{\perms}{\env}$.
    
          By induction $\eval{\heap}{\env}{e_1}{V'(t_1)}$ and $\eval{\heap}{\env}{e_2}{V'(t_2)}$. But now $\pc(\sstate'') \implies \pc(\hat{\sstate}') = \pc(\sstate') \kand \kneg t_1 \implies \kneg t_1$, thus $V'(t_1) = \kfalse$. Therefore $\eval{\heap}{\env}{e_1}{\kfalse}$.
    
          Thus $\eval{\heap}{\env}{e_1 \kor e_2}{V'(t_2)}$ by \refrule{EvalOrB}.
    
          By induction $\frm{\heap}{\perms}{\env}{e_1}$ and $\frm{\heap}{\perms}{\env}{e_2}$. Thus $\frm{\heap}{\perms}{\env}{e_1 \kor e_2}$ by \refrule{FrameOrB} since $\eval{\heap}{\env}{e_1}{\kfalse}$.
    
        \case \refrule{SEvalAndA} -- $\seval{\sstate}{e_1 \wedge e_2}{t_1}{\sstate''}{\scheck}$: Similar to case \ref{case:seval-sim-ora}.
    
        \case \refrule{SEvalAndB} -- $\seval{\sstate}{e_1 \wedge e_2}{t_2}{\sstate''}{\scheck_1 \cup \scheck_2}$: Similar to case \ref{case:seval-sim-orb}.
    
        \case \refrule{SEvalOp} -- $\seval{\sstate}{e_1 \oplus e_2}{t_1 \oplus t_2}{\sstate''}{\scheck_1 \cup \scheck_2}$:
    
          By \refrule{SEvalOp} $\seval{\sstate}{e_1}{t_1}{\sstate'}{\scheck_1}$ and $\seval{\sstate'}{e_2}{t_2}{\sstate''}{\scheck_2}$.
    
          Now suppose $\rtassert{V'}{\heap}{\perms}{\scheck_1 \cup \scheck_2}$ and $V'(\pc(\sstate'')) = \ktrue$. By lemma \ref{lem:scheck-monotonicity} $\rtassert{V'}{\heap}{\perms}{\scheck_1}$ and $\rtassert{V'}{\heap}{\perms}{\scheck_2}$.
    
          Also, by lemma \ref{lem:eval-subpath}, $\pc(\sstate'') \implies \pc(\sstate')$. Therefore $V'(\pc(\sstate')) = V'(\pc(\sstate')) = \ktrue$, and $V'(\pc(\sstate'')) = \ktrue$ by assumption.
    
          Now by induction $\simstate{V'}{\sstate'}{\heap}{\perms}{\env}$, and then also by induction $\simstate{V'}{\sstate''}{\heap}{\perms}{\env}$.
    
          By induction $\eval{\heap}{\env}{e_1}{V'(t_1)}$ and $\eval{\heap}{\env}{e_2}{V'(t_2)}$. Therefore $\eval{\heap}{\env}{e_1 \oplus e_2}{V'(t_1) \oplus V'(t_2)}$ and $V'(t_1) \oplus V'(t_2) = V'(t_1 \oplus t_2)$.
    
          By induction $\frm{\heap}{\perms}{\env}{e_1}$ and $\frm{\heap}{\perms}{\env}{e_2}$. Thus $\frm{\heap}{\perms}{\env}{e_1 \oplus e_2}$ by \refrule{FrameOp}.
    
        \case \refrule{SEvalNeg} -- $\seval{\sstate}{\kneg e}{\kneg t}{\sstate'}{\scheck}$:
        
          By \refrule{SEvalNeg} $\seval{\sstate}{e}{t}{\sstate'}{\scheck}$.
    
          Suppose that $\rtassert{V'}{\heap}{\perms}{\scheck}$ and $V'(\pc(\sstate')) = \ktrue$.
    
          Now $\simstate{V'}{\sstate'}{\heap}{\perms}{\env}$ by induction.
    
          Also by induction $\eval{\heap}{\env}{e}{V'(t)}$. Thus $\eval{\heap}{\env}{\kneg e}{\neg V'(t)}$ by \refrule{EvalNeg} and $V'(\kneg t) = \neg V'(t)$.
    
          By induction $\frm{\heap}{\perms}{\env}{e}$, and thus $\frm{\heap}{\perms}{\env}{\kneg e}$ by \refrule{FrameNeg}.
    
        \case\label{case:seval-sim-field} \refrule{SEvalField}: $\seval{\sstate}{e.f}{t}{\sstate'}{\scheck}$
    
          By \refrule{SEvalField} $\seval{\sstate}{e}{t_e}{\sstate'}{\scheck}$.
    
          Suppose that $\rtassert{V'}{\heap}{\perms}{\scheck}$ and $V'(\pc(\sstate')) = \ktrue$.
    
          Then $\simstate{V'}{\sstate'}{\heap}{\perms}{\env}$ by induction.
    
          Also by induction $\eval{\heap}{\env}{e}{V'(t_e)}$. By \refrule{SEvalField} $\pc(\sstate') \implies t_e \keq t_e'$, therefore $V'(t_e) = V'(t_e')$. Also by \refrule{SEvalField} $\triple{f}{t_e'}{t} \in \sheap(\sstate')$. Therefore $V'(t) = \heap(V'(t_e'), f) = \heap(V'(t_e), f)$ since $\simstate{V'}{\sstate'}{\heap}{\env}{\perms}$.
    
          Thus $\eval{\heap}{\env}{e.f}{\heap(V'(t_e), f)}$ by \refrule{EvalField} and $\heap(V'(t_e), f) = V_{\sstate'}(t)$.
    
          By induction $\frm{\heap}{\perms}{\env}{e}$. Also, since $\triple{f}{t_e'}{t} \in \sheap(\sstate')$ and $\simstate{V'}{\sstate'}{\heap}{\env}{\perms}$, $\pair{V'(t_e')}{f} = \pair{V'(t_e)}{f} \in \perms$. By \refrule{AssertAcc}, $\assertion{\heap}{\perms}{\env}{\kacc(e.f)}$ since $\eval{\heap}{\env}{e}{V'(t_e)}$. Thus by \refrule{FrameField} $\frm{\heap}{\perms}{\env}{e.f}$.
    
        \case \refrule{SEvalFieldOptimistic} -- $\seval{\sstate}{e.f}{t}{\sstate'}{\scheck}$: Similar to case \ref{case:seval-sim-field}.
    
        \case \refrule{SEvalFieldImprecise} -- $\seval{\sstate}{e.f}{t}{\sstate'}{\scheck; \pair{t_e}{f}}$:
    
          By \refrule{SEvalFieldImprecise} $\seval{\sstate}{e}{t_e}{\sstate'}{\scheck}$, $t = \ffresh$, and $\sstate'' = \sstate'[\oheap = \oheap(\sstate'); \triple{f}{t_e}{t}]$.
    
          Suppose $\rtassert{V'}{\heap}{\perms}{\scheck; \pair{t_e}{f}}$ and $V'(\pc(\sstate'')) = \ktrue$, thus $\rtassert{V'}{\heap}{\perms}{\scheck}$ by lemma \ref{lem:scheck-monotonicity}. Then $V'(\pc(\sstate'')) = V'(\pc(\sstate')) = \ktrue$.
    
          Then by induction $\simstate{V'}{\sstate'}{\heap}{\perms}{\env}$, thus $\simstate{V'}{\sstate'}{\heap}{\perms}{\env}$.
    
          By lemma \ref{lem:scheck-monotonicity} $\rtassert{V'}{\heap}{\perms}{\set{\pair{t_e}{f}}}$ and thus $\rtassert{V'}{\heap}{\perms}{\pair{t_e}{f}}$. Then by \refrule{CheckAcc} $\pair{V'(t_e)}{f} \in \perms$.
          
          By definition \ref{def:eval-valuation} $\heap(V'(t_e), f) = V'(t)$, and also $\pair{V'(t_e)}{f} \in \perms$ and $\simheap{V'}{\oheap(\sstate')}{\heap}{\perms}$, thus $\simheap{V'}{\oheap(\sstate'); \triple{f}{t_e}{t}}{\heap}{\perms}$.
          
          Then $\simstate{V'}{\sstate''}{\heap}{\perms}{\env}$ since $\sstate'$ and $\sstate''$ differ only in their $\oheap$ components and $\oheap(\sstate'') = \oheap(\sstate'); \triple{f}{t_e}{t}$.
    
          By induction $\eval{\heap}{\env}{e}{V'(t_e)}$. As shown before, $\heap(V'(t_e), f) = V'(t)$. Thus by \refrule{EvalField} $\eval{\heap}{\env}{e.f}{V'(t)}$.
    
          By induction $\frm{\heap}{\perms}{\env}{e}$. Also, as shown before, $\pair{V'(t_e)}{f} \in \perms$. By \refrule{AssertAcc}, $\assertion{\heap}{\perms}{\env}{\kacc(e.f)}$ since $\eval{\heap}{\env}{e}{V'(t_e)}$. Thus by \refrule{FrameField} $\frm{\heap}{\perms}{\env}{e.f}$.
    
        \case \refrule{SEvalFieldFailure} -- $\seval{\sstate}{e.f}{t}{\sstate'}{\set{\bot}}$:
        
          $\rtassert{V}{\heap}{\perms}{\set{\bot}}$ cannot hold. Since the assumptions cannot be satisfied, the statement holds vacuously.
    
        \definecolor{shadecolor}{gray}{0.9}
        \begin{snugshade} 
        \case \refrule{SEvalUnfoldingPrecise} -- $\seval{\sstate_1}{\sunfolding{\pe}{e_0}}{t_0}{\sstate_6}{\scheck_1 \cup \scheck_2 \cup \scheck_3}$:

        By the induction hypothesis, $\simstate{V'}{\sstate_2}{\heap}{\perms}{\env}$, $\multiple{\eval{\heap}{\env}{e}{V'(t)}}$, and $\multiple{\frm{\heap}{\perms}{\env}{e}}$ with $V'(\pc(\sstate_2)) = \ktrue$. Then, by Lemma~\ref{lem:cons-soundness}, $\assertion{\heap}{\perms}{\env}{\pe}$ and $\simstate{V'}{\sstate_3}{\heap}{\perms \setminus \efoot{\heap}{\env}{\pe}}{\env}$ with $V'(\pc(\sstate_3)) = \ktrue$. By Lemma~\ref{lem:produce-soundness}, $\simstate{V'}{\sstate_4}{\heap}{\perms}{\env}$ with $V'(\pc(\sstate_4)) = \ktrue$. Then, by the induction hypothesis, $\simstate{V'}{\sstate_5}{\heap}{\perms}{\env}$ with $V'(\pc(\sstate_5)) = \ktrue$. Then, $\simstate{V'}{\sstate_6}{\heap}{\perms}{\env}$ since $\sheap(\sstate_6) \subseteq \sheap(\sstate_5)$ and $\oheap(\sstate_6) = \emptyset$. Similarly, $\pc(\sstate_6) = \pc(\sstate_5)$ so $V'(\pc(\sstate_6)) = V'(\pc(\sstate_5)) = \ktrue$, as desired. Furthermore, by \refrule{EvalUnfolding}, $\eval{\heap}{\perms}{\sunfolding{\pe}{e_0}}{V'(t_0)}$.
        By the induction hypothesis, $\eval{\heap}{\env}{e_0}{V'(t_0)}$ and $\frm{\heap}{\perms}{\env}{e_0}$. By Lemma~\ref{lem:cons-soundness}, $\assertion{\heap}{\perms}{\env}{\pe}$, so by definition of \refrule{FrameUnfolding}, $\frm{\heap}{\perms}{\env}{\sunfolding{\pe}{e_0}}$ as desired.
        
        \case \refrule{SEvalUnfoldingImplicitPrecise} -- $\seval{\sstate_1}{\sunfolding{\pe}{e_0}}{t_0}{\sstate_2}{\scheck}$: 
        
        Since $\simstate{V}{\sstate_1}{\heap}{\perms}{\env}$, by the induction hypothesis $\simstate{V'}{\sstate_2}{\heap}{\perms}{\env}$. By Definition~\ref{def:eval-valuation}, $\eval{\heap}{\env}{e_0}{V'(t)}$, so therefore $\eval{\heap}{\env}{\sunfolding{\pe}{e_0}}{V'(t)}$. 
        
        Inductively, $\frm{\heap}{\perms}{\env}{e_0}$ and $\ifrm{\heap}{\perms}{\env}{\pe}$, and since $\pair{p}{\multiple{t}} \in \sheap(\sstate_2) \cup \oheap(\sstate_2)$, by the definition of heap correspondence and \refrule{AssertPredicate}, $\assertion{\heap}{\perms}{\env}{\pe}$. Therefore, $\frm{\heap}{\perms}{\env}{\sunfolding{\pe}{e_0}}$.

        \case \refrule{SEvalUnfoldingImplicitImprecise} -- $\seval{\sstate_1}{\sunfolding{\pe}{e_0}}{t_0}{\sstate_3}{\scheck'}$:

        Similar to the case for \refrule{SEvalUnfoldingImplicitPrecise}. To show heap correspondence with $\oheap'$, recall that $\rtassert{V'}{\heap}{\perms}{\pair{p}{\multiple{t}}}$, showing that the predicate assertion holds in the corresponding state.

        \case \refrule{SEvalUnfoldingImplicitFailure} -- $\seval{\sstate_1}{\sunfolding{\pe}{e_0}}{t_0}{\sstate_2}{\set{\bot}}$:

        The run-time check set $\set{\bot}$ is unsatisfiable, so the conclusion holds vacuously.

        \case \refrule{SEvalFunctionExplicit} -- $\seval{\sstate_1}{\fe}{t'}{\sstate_6}{\scheck_1 \cup \scheck_2 \cup \scheck_3}$:

        By the induction hypothesis, $\simstate{V'}{\sstate_2}{\heap}{\perms}{\env}$, $\multiple{\eval{\heap}{\env}{e}{V'(t)}}$, and $\multiple{\frm{\heap}{\perms}{\env}{e}}$ with $V'(\pc(\sstate_2)) = \ktrue$. Then, by Lemma~\ref{lem:cons-soundness}, $\assertion{\heap}{\perms}{\env}{\ffuncpre(f)}$ and $\simstate{V'}{\sstate_3}{\heap}{\perms \setminus \efoot{\heap}{\env}{\ffuncpre(f)}}{\env}$ with $V'(\pc(\sstate_3)) = \ktrue$. By Lemma~\ref{lem:produce-soundness}, $\simstate{V'}{\sstate_4}{\heap}{\perms}{\env}$ with $V'(\pc(\sstate_4)) = \ktrue$. Then, $\simstate{V'}{\sstate_5}{\heap}{\perms}{\env}$ by the induction hypothesis. Then, $\simstate{V'}{\sstate_6}{\heap}{\perms}{\env}$ because $\oheap(\sstate_6) = \oheap(\sstate_5)$ and $\sheap(\sstate_6) = \sheap(\sstate_2)$.

        Since $\multiple{\frm{\heap}{\perms}{\env}{e}}$ and $\ifrm{\heap}{\perms}{\env}{\ffuncpre(f)}$ inductively and $\assertion{\heap}{\perms}{\env}{\ffuncpre(f)}$ from the consume step, $\frm{\heap}{\perms}{\env}{\ffuncpre(f)}$. 
        
        
        \case \refrule{SEvalFunctionImplicit} -- $\seval{\sstate_1}{\fe}{t'}{\sstate_4}{\scheck_1 \cup \scheck_2 \cup \scheck_3}$:

        Because $\rtassert{V}{\heap}{\perms}{\scheck_2}$ where $\scheck_2$ arose as a result of consuming $\ffuncpre(f)$, in the corresponding state it must be the case that $\assertion{\heap}{\perms}{\env}{\ffuncpre(f)}$. Inductively, $\multiple{\frm{\heap}{\perms}{\env}{e}}$ and $\ifrm{\heap}{\perms}{\env}{\ffuncpre(f)}$, so $\frm{\heap}{\perms}{\env}{\fe}$ by \refrule{FrameFunction}. By Definition~\ref{def:eval-valuation}, $\eval{\heap}{\env'}{\ffuncbody(f)}{V'(t')}$ so by \refrule{EvalFunction}, $\eval{\heap}{\env}{\fe}{V'(t')}$. Then, to show $\simstate{V'}{\sstate_4}{\heap}{\perms}{\env}$, it suffices to show that $\simheap{V'}{\oheap(\sstate_4)}{\heap}{\perms}$ because $\sheap(\sstate_4) = \sheap(\sstate_1)$ and $\senv(\sstate_4) = \senv(\sstate_1)$ by Lemma~\ref{lem:eval-unchanged} where $\simstate{V}{\sstate_1}{\heap}{\perms}{\env}$. Furthermore, since $\oheap(\sstate_4) = \oheap(\sstate_2) \cup (s' \setminus \sheap(\sstate_2))$ where $\simheap{V}{\oheap(\sstate_2)}{\heap}{\perms}$, it suffices to show $\simheap{V'}{s' \setminus \sheap(\sstate_2)}{\heap}{\perms}$. For each heap chunk in $s' \setminus \sheap(\sstate_2)$, there is a corresponding run-time check $\rtassert{V'}{\heap}{\perms}{r}$ showing that the heap chunk holds in the corresponding state by \refrule{CheckAcc} and \refrule{CheckPred}. Accordingly, $\simheap{V'}{s' \setminus \sheap(\sstate_2)}{\heap}{\perms}$ as desired.
        \end{snugshade}
    \end{enumcases}
\end{proof}

\begin{lemma}[Progress]\label{lem:eval-progress}
  If $V$ is some initial valuation and $V(\pc(\sstate)) = \ktrue$, then for some $\sstate'$, $t$, and $\scheck$,
  $$\seval{\sstate}{e}{t}{\sstate'}{\scheck} ~\text{and}~ V'(\pc(\sstate')) = \ktrue$$
  where $V' = V[\seval{\sstate}{e}{t}{\sstate'}{\scheck} \mid \heap, \env]$ for some $\heap$.
\end{lemma}

\begin{proof}
    By induction on $e$.
    \begin{enumcases}
        \case $l \in \Literal$:
          By \refrule{SEvalLiteral} $\seval{\sstate}{l}{\_}{\sstate}{\_}$. Then $V'(\pc(\sstate)) = \ktrue$ by assumptions.
    
        \case $x \in \Var$:
          By \refrule{SEvalVar} $\seval{\sstate}{x}{\_}{\sstate}{\_}$. Then $V'(\pc(\sstate)) = \ktrue$ by assumptions.
    
        \case $e.f$ where $e \in \Expr, f \in \Field$:
          By induction, $\seval{e}{t_e}{\sstate}{\scheck}$ and $V'(\pc(\sstate')) = \ktrue$ where $V'$ is the corresponding valuation. Then one of the following must apply:
          
          \subcase If $\existential{t_e', t}{\big[\pc(\sstate') \implies t_e' \keq t_e \big] \wedge \triple{f}{t_e'}{t} \in \sheap(\sstate')}$, then \refrule{SEvalField} applies and thus $\seval{\sstate}{e.f}{\_}{\sstate'}{\_}$. Let $\sstate'' = \sstate'$.
          \subcase Otherwise, $\nexistential{t_e', t}{\big[\pc(\sstate') \implies t_e' \keq t_e \big] \wedge \triple{f}{t_e'}{t} \in \sheap(\sstate')}$. Then, if $\existential{t_e', t}{\big[\pc(\sstate') \implies t_e' \keq t_e \big] \wedge \triple{f}{t_e'}{t} \in \oheap(\sstate')}$, \refrule{SEvalFieldOptimistic} applies and thus $\seval{\sstate}{e.f}{\_}{\sstate'}{\_}$. Let $\sstate'' = \sstate'$.
          \subcase Otherwise, $\nexistential{t_e', t}{\big[\pc(\sstate') \implies t_e' \keq t_e \big] \wedge \triple{f}{t_e'}{t} \in \sheap(\sstate') \cup \oheap(\sstate')}$. Then if $\imp(\sstate')$, \refrule{SEvalFieldImprecise} applies. In this case, $\seval{\sstate}{e.f}{\_}{\sstate''}{\_}$ where $\pc(\sstate'') = \pc(\sstate')$.
          \subcase Otherwise, $\neg \imp(\sstate')$ and $\nexistential{t_e', t}{\big[\pc(\sstate') \implies t_e' \keq t_e \big] \wedge \triple{f}{t_e'}{t} \in \sheap(\sstate') \cup \oheap(\sstate')}$. Thus \refrule{SEvalFieldFailure} applies and thus $\seval{\sstate}{e.f}{\_}{\sstate'}{\_}$. Let $\sstate'' = \sstate'$.
    
          In all of these subcases, $\seval{\sstate}{e.f}{\_}{\sstate''}{\_}$ where $\pc(\sstate'') = \pc(\sstate')$, and thus $V'(\pc(\sstate'')) = V'(\pc(\sstate')) = \ktrue$. By definition \ref{def:eval-valuation}, in all of these subcases $V[\seval{\sstate}{e.f}{\_}{\sstate''}{\_} \mid \heap]$ extends $V'$.
    
        \case $e_1 \oplus e_2$; $e_1, e_2 \in \Expr$:
          By induction $\seval{\sstate}{e_1}{\_}{\sstate'}{\_}$ for some $e_1, \sstate'$ where $V'$ is the corresponding valuation and $V'(\pc(\sstate')) = \ktrue$. Then by induction $\seval{\sstate'}{e_2}{\_}{\sstate''}$ where $V''$ is the corresponding valuation extending $V'$ and $V''(\pc(\sstate'')) = \ktrue$. Finally, by \refrule{SEvalOp} $\seval{\sstate}{e_1 \oplus e_2}{\_}{\sstate''}{\_}$. By definition \ref{def:eval-valuation}, $V[\seval{\sstate}{e_1 \oplus e_2}{\_}{\sstate''}{\_} \mid \heap]$ extends $V''$.
    
        \case\label{case:seval-progress-or} $e_1 \kor e_2$; $e_1, e_2 \in \Expr$:
          By induction $\seval{\sstate}{e_1}{t_1}{\sstate'}{\_}$ for some $e_1, t_e, \sstate'$ where $V'$ is the corresponding valuation and $V'(\pc(\sstate')) = \ktrue$. Then one of the following cases must apply since the program is well-typed:
    
          \subcase $V'(t_1) = \ktrue$: Then by \refrule{SEvalOrA}, $\seval{\sstate}{e_1 \vee e_2}{t_1}{\sstate''}{\_}$ where $\pc(\sstate'') = \pc(\sstate') \kand t_1$. Let $V''$ be the corresponding valuation. Now $V''(\pc(\sstate'')) = V''(\pc(\sstate')) \wedge V''(t_1) = \ktrue$, which completes the proof.
          
          \subcase $V'(t_1) = \kfalse$: Let $\hat{\sstate}' = \sstate'[\pc = \pc(\sstate') \kand \kneg t_1]$. Then $V'(\pc(\hat{\sstate}')) = V'(\sstate') \wedge \neg V'(t_1) = \ktrue$. By induction $\seval{\hat{\sstate}'}{e_2}{t_2}{\sstate''}{\_}$ for some $e_2, t_2, \sstate''$ where $V''$ is the corresponding valuation and $V''(\pc(\sstate'')) = \ktrue$. Finally, by \refrule{SEvalOrB}, $\seval{\sstate}{e_1 \kor e_2}{\_}{\sstate''}{\_}$. By definition \ref{def:eval-valuation} $V[\seval{\sstate}{e_1 \kor e_2}{\_}{\sstate''}{\_} \mid \heap]$ extends $V''$.
    
        \case $e_1 \kand e_2$; $e_1, e_2 \in \Expr$: Similar to case \ref{case:seval-progress-or}.
    
        \case $\kneg e$; $e \in \Expr$:
          By induction $\seval{\sstate}{e}{\_}{\sstate'}{\_}$ and $V'(\sstate') = \ktrue$ where $V'$ is the corresponding derivation. Then by \refrule{SEvalNeg}, $\seval{\sstate}{\kneg e}{\_}{\sstate'}{\_}$ and the corresponding valuation extends $V'$ by definition \ref{def:eval-valuation}.

        \definecolor{shadecolor}{gray}{0.9}
        \begin{snugshade} \case $\sunfolding{\pe}{e_0}$ where $p \in \Predicate$ and $\multiple{e}, e_0 \in \Expr$: 
        
        By the induction hypothesis, $\multiple{\seval{\sstate_1}{e}{t}{\sstate_2}{\scheck_1}}$ with $V'(\pc(\sstate_2)) = \ktrue$.

            \subcase $p \in \svis(\sstate_1), \pair{p}{\multiple{t}} \in \sheap(\sstate_2) \cup \oheap(\sstate_2)$:
            In this case, \refrule{SEvalUnfoldingImplicitPrecise} applies. The final state is $\sstate_2$ so $V'(\pc(\sstate_2)) = \ktrue$ as desired.

            \subcase $p \in \svis(\sstate_1), \pair{p}{\multiple{t}} \notin \sheap(\sstate_2) \cup \oheap(\sstate_2), \imp(\sstate_2)$:
            In this case, \refrule{SEvalUnfoldingImplicitImprecise} applies. Since $\pc(\sstate_3) = \pc(\sstate_2)$, $V'(\pc(\sstate_3)) = V'(\pc(\sstate_2)) = \ktrue$.

            \subcase $p \in \svis(\sstate_1), \pair{p}{\multiple{t}} \notin \sheap(\sstate_2) \cup \oheap(\sstate_2), \lnot \imp(\sstate_2)$:
            In this case, \refrule{SEvalUnfoldingImplicitFailure} applies. The final state is $\sstate_2$ so $V'(\pc(\sstate_2)) = \ktrue$ as desired.

            \subcase $p \notin \svis(\sstate_1), \lnot \imp(\sstate_1)$:
            In this case, \refrule{SEvalUnfoldingPrecise} applies, so by Lemma~\ref{lem:cons-progress}, $\scons{\sstate_2}{\pe}{\sstate_3}{\scheck_2}$ with $V'(\pc(\sstate_3)) = \ktrue$. Since $\pair{p}{\multiple{t}}$ holds in $\sstate_2$ due to the consume operation, in the corresponding state the predicate assertion must hold, so by Lemma~\ref{lem:produce-progress}, $\sproduce{\sstate_3}{\fpred(p)}{\sstate_4}$ with $V'(\pc(\sstate_4))$. By the induction hypothesis, $\seval{\sstate_4}{e_0}{t_0}{\sstate_5}{\scheck_3}$ with $V'(\pc(\sstate_5) - \ktrue$. Then, since $\pc(\sstate_5) = \pc(\sstate_6)$, $V'(\pc(\sstate_6)) = V'(\pc(\sstate_5)) = \ktrue$.

            \subcase $p \notin \svis(\sstate_1), \imp(\sstate_1)$:
            In this case, \refrule{SEvalUnfoldingImpreciseA} or \refrule{SEvalUnfoldingImpreciseB}. This case is similar to the previous case, except for the information that remains in the optimistic heap at the end.


        \case $\fe$ where $f \in \Function$ and $\multiple{e} \in \Expr$: By the induction hypothesis, $\multiple{\seval{\sstate_1}{e}{t}{\sstate_2}{\scheck_1}}$ with $V_2(\pc(\sstate_2)) = \ktrue$. Then, by Lemma~\ref{lem:cons-progress}, $\scons{\sstate_2[\senv = \senv(\sstate_2)[\multiple{x \mapsto t}]]}{\ffuncpre(f)}{\sstate_3}{\_}$ where $V_3(\pc(\sstate_3)) = \ktrue$. 
    
            \subcase $f \notin \svis(\sstate_1)$: Because the precondition holds due to the consume operation, in the corresponding state the dynamic assertion must hold, so by Lemma~\ref{lem:produce-progress}, $\sproduce{\sstate_3}{\fpred(p)}{\sstate_4}$ with $V'(\pc(\sstate_4))$. By induction, $\seval{\sstate_4[\sheap = \sheap', \oheap = \oheap', \svis = \svis(\sstate_3) \cup \set{f}]}{\ffuncbody(f)}{t'}{\sstate_5}{\_}$ with $V_5(\pc(\sstate_5))$. By definition, $\pc(\sstate_5) \implies \pc(\sstate_6)$ so $V_5(\pc(\sstate_6)) = \ktrue$.

            \subcase $f \in \svis(\sstate_1)$: By definition, $\pc(\sstate_3) \implies \pc(\sstate_4)$ so $V_3(\pc(\sstate_4)) = \ktrue$.
        \end{snugshade}
    \end{enumcases}
\end{proof}

\definecolor{shadecolor}{gray}{0.9}
\begin{snugshade}
\begin{lemma}[Correspondence]\label{lem:eval-correspondence}
  Let $\seval{\sstate}{e}{t}{\sstate'}{\_}$ with corresponding valuation $V'$. If $\simstate{V}{\sstate}{\heap}{\perms}{\env}$ such that $\efoot{\heap}{\env}{e} \subseteq \perms$ and $\eval{\heap}{\env}{e}{v}$ then 
  \[ \simstate{V'}{\sstate'}{\heap}{\perms}{\env} \quad v = V'(t) \]
\end{lemma}
  \begin{proof}
    By induction on $\seval{\sstate}{e}{t}{\sstate'}{\_}$:
    \begin{enumcases}
    \case \refrule{SEvalLiteral} -- $\seval{\sstate}{l}{l}{\sstate}{\_}$: 
    
    Then $V' = V$, so trivially $\simstate{V'}{\sstate'}{\heap}{\perms}{\env}$. In addition $\eval{\heap}{\env}{l}{l}$ and by definition $V(l) = V'(l) = l$. 
    
    \case \refrule{SEvalVar} -- ${\seval{\sstate}{x}{\senv(\sstate)(x)}{\sstate}{\_}}$: 
    
    Then $V' = V$, so trivially $\simstate{V'}{\sstate'}{\heap}{\perms}{\env}$. Therefore $\simenv{V'}{\senv(\sstate)}{\env}$, and so $\env(x) = V'(\senv(\sstate)(x))$. Then since $\eval{\heap}{\perms}{x}{\env(x)}$, $\eval{\heap}{\perms}{x}{V'(\senv(\sstate)(x))}$.
    
    \case \refrule{SEvalOrA} -- $\seval{\sstate}{e_1 \kor e_2}{t_1}{\sstate''}{\_}$: 
    
    By \refrule{SEvalOrA}, $\seval{\sstate}{e_1}{t_1}{\sstate'}{\_}$, and the corresponding valuation for this case is $V' = V[\seval{\sstate}{e_1}{t_1}{\sstate'}{\_} \mid H]$. Since $\simstate{V}{\sstate}{\heap}{\perms}{\env}$, by lemma \ref{lem:eval-progress}, $V'(\pc(\sstate'')) = \ktrue$ and $V'(\pc(\sstate') = \ktrue$. Then, since $\pc(\sstate'') = \pc(\sstate') \kand t_1$, $V'(t_1) = \ktrue$. 

    Now since $\eval{\heap}{\env}{e_1 \kor e_2}{v}$, either \refrule{EvalOrA} or \refrule{EvalOrB} must apply. In either case, $\eval{\heap}{\env}{e_1}{v_1}$ for some $v_1$. Since $\efoot{\heap}{\env}{e_1 \kor e_2} \subseteq \perms$, by definition \ref{def:efoot}, $\efoot{\heap}{\env}{e_1} \subseteq \perms$. Then by induction $\simstate{V'}{\sstate'}{\heap}{\perms}{\env}$ and $\eval{\heap}{\env}{e_1}{V'(t_1)}$. 

    Since $\sstate''$ and $\sstate'$ differ only in their $\pc$ components, $\simstate{V'}{\sstate''}{\heap}{\perms}{\env}$ as $V'(\pc(\sstate')) = \ktrue$. In addition, since $V'(t_1) = \ktrue$,  $\eval{\heap}{\env}{e_1}{\ktrue}$, which means \refrule{EvalOrA} must apply. Then $\eval{\heap}{\env}{e_1 \kor e_2}{\ktrue}$, which is $V'(t_1)$.
    
    \case \refrule{SEvalOrB} -- $\seval{\sstate}{e_1 \kor e_2}{t_2}{\sstate''}{\_}$: 
    
    By \refrule{SEvalOrB}, $\seval{\sstate}{e_1}{t_1}{\sstate'}{\_}$. Let $V' = V[\seval{\sstate}{e_1}{t_1}{\sstate'}{\_} \mid \heap]$. Let $\hat{\sigma} = \sigma'[\pc = \pc(\sigma') \kand \kneg t_1]$. Also by \refrule{SEvalOrB}, $\seval{\hat{\sstate}}{e_2}{t_2}{\sstate''}{\_}$ and the corresponding valuation for this case is $V'' = V'[\seval{\hat{\sstate}}{e_2}{t_2}{\sstate''}{\_} \mid H]$. 
    
    Since $\simstate{V}{\sstate}{\heap}{\perms}{\env}$, by lemma \ref{lem:eval-progress}, $V'(\pc(\sstate')) = \ktrue$ and $V''(\pc(\sstate'')) = \ktrue$. Then by lemma \ref{lem:eval-subpath}, $\pc(\sstate'') \implies \pc(\hat{\sstate}) = \pc(\sstate') \kand \kneg t_1$. Since corresponding valuations must extend, $V''(t_1) = V'(t_1) = \kfalse$. 

    Now since $\eval{\heap}{\env}{e_1 \kor e_2}{v}$, either \refrule{EvalOrA} or \refrule{EvalOrB} must apply. In either case, $\eval{\heap}{\env}{e_1}{v_1}$ for some $v_1$. Since $\efoot{\heap}{\env}{e_1 \kor e_2} \subseteq \perms$, by definition \ref{def:efoot}, $\efoot{\heap}{\env}{e_1} \subseteq \perms$. Then by induction $\simstate{V'}{\sstate'}{\heap}{\perms}{\env}$ and $\eval{\heap}{\env}{e_1}{V'(t_1)}$. 

    By lemma \ref{lem:eval-progress} $V'(\pc(\sstate')) = \ktrue$. Then $\hat{\sstate}$ and $\sstate'$ differ only in their $\pc$ components, $\simstate{V'}{\hat{\sstate}}{\heap}{\perms}{\env}$. In addition, since $V'(t_1) = \kfalse$, $\eval{\heap}{\env}{e_1}{\kfalse}$, which means that \refrule{EvalOrB} must apply. Therefore, $\eval{\heap}{\env}{e_2}{v_2}$ for some $v_2$. In addition, by definition \ref{def:efoot} and since $\efoot{\heap}{\env}{e_1 \kor e_2} \subseteq \perms$, $\efoot{\heap}{\env}{e_2} \subseteq \perms$. Then by induction, $\simstate{V''}{\sstate''}{\heap}{\perms}{\env}$ and $\eval{\heap}{\env}{e_2}{V''(t_2)}$. By \refrule{EvalOrB}, $\eval{\heap}{\env}{e_1 \kor e_2}{v_2}$, and $v_2 = V''(t_2)$. 
     
    \case \refrule{SEvalAndA} -- $\seval{\sstate}{e_1 \kand e_2}{t_1}{\sstate''}{\_}$: 
    
    By \refrule{SEvalAndA}, $\seval{\sstate}{e_1}{t_1}{\sstate'}{\_}$, and the corresponding valuation for this case is $V' = V[\seval{\sstate}{e_1}{t_1}{\sstate'}{\_} \mid \heap]$. Since $\simstate{V}{\sstate}{\heap}{\perms}{\env}$, by lemma \ref{lem:eval-progress} $V'(\pc(\sstate'')) = \ktrue$ and $V'(\pc(\sstate')) = \ktrue$. Then, since $\pc(\sstate'') = \pc(\sstate') \kand \kneg t_1$, $V'(t_1) = \kfalse$.

    Now since $\eval{\heap}{\env}{e_1 \kand e_2}{v}$, either \refrule{EvalAndA} or \refrule{EvalAndB} must apply. In either case, $\eval{\heap}{\env}{e_1}{v_1}$ for some $v_1$. Since $\efoot{\heap}{\env}{e_1 \kand e_2} \subseteq \perms$, by definition \ref{def:efoot}, $\efoot{\heap}{\env}{e_1} \subseteq \perms$. Then by induction $\simstate{V'}{\sstate'}{\heap}{\perms}{\env}$ and $\eval{\heap}{\env}{e_1}{V'(t_1)}$.

    Since $\sstate''$ and $\sstate'$ differ only in their $\pc$ components, $\simstate{V'}{\sstate''}{\heap}{\perms}{\env}$ as $V'(\pc(\sstate')) = \ktrue$. In addition, since $V'(t_1) = \kfalse$, $\eval{\heap}{\env}{e_1}{\kfalse}$, which means \refrule{EvalAndA} must apply. Then $\eval{\heap}{\env}{e_1 \kand e_2}{\kfalse}$, which is $V'(t_1)$.
 
    \case \refrule{SEvalAndB} -- $\seval{\sstate}{e_1 \kand e_2}{t_2}{\sstate''}{\_}$: 
    
    By \refrule{SEvalAndB} $\seval{\sstate}{e_1}{t_1}{\sstate'}{\_}$. Let $V' = V[\seval{\sstate}{e_1}{t_1}{\sstate'}{\_} \mid \heap]$. Let $\hat{\sstate} = \sstate'[\pc = \pc(\sstate') \kand t_1]$. Also by \refrule{SEvalAndB}, $\seval{\hat{\sstate}}{e_2}{t_2}{\sstate''}{\_}$ and the corresponding valuation for this case is $V'' = V'[\seval{\hat{\sstate}}{e_2}{t_2}{\sstate''}{\_} \mid \heap]$.

    Since $\simstate{V}{\sstate}{\heap}{\perms}{\env}$, by lemma \ref{lem:eval-progress}, $V'(\pc(\sstate')) = \ktrue = V''(\pc(\sstate'))$ and $V''(\pc(\sstate'')) = \ktrue$. Then by lemma \ref{lem:eval-subpath}, $\pc(\sstate'') \implies \pc(\hat{\sstate'}) = \pc(\sstate') \kand t_1$. Since corresponding valuations must extend, $V''(t_1) = V'(t_1) = \ktrue$. 

    Now since $\eval{\heap}{\env}{e_1 \kand e_2}{v}$, either \refrule{EvalAndA} or \refrule{EvalAndB} must apply. In either case, $\eval{\heap}{\env}{e_1}{v_1}$ for some $v_1$. Since $\efoot{\heap}{\env}{e_1 \kand e_2} \subseteq \perms$, by definition \ref{def:efoot}, $\efoot{\heap}{\env}{e_1} \subseteq \perms$. Then by induction $\simstate{V'}{\sstate'}{\heap}{\perms}{\env}$ and $\eval{\heap}{\env}{e_1}{V'(t_1)}$.

    By lemma \ref{lem:eval-progress} $V'(\pc(\sstate')) = \ktrue$. Then since $\hat{\sstate}$ and $\sstate'$ differ only in their $\pc$ components, $\simstate{V'}{\hat{\sstate}}{\heap}{\perms}{\env}$. In addition, since $V'(t_1) = \ktrue$, $\eval{\heap}{\env}{e_1}{\ktrue}$, which means that \refrule{EvalAndB} must apply. Therefore, $\eval{\heap}{\env}{e_2}{v_2}$ for some $v_2$. Also since $\eval{\heap}{\env}{e_1}{\ktrue}$ and  $\efoot{\heap}{\env}{e_1 \kand e_2} \subseteq \perms$, $\efoot{\heap}{\env}{e_2} \subseteq \perms$. Then by induction, $\simstate{V''}{\sigma''}{\heap}{\perms}{\env}$ and $\eval{\heap}{\env}{e_2}{V''(t_2)}$. By \refrule{EvalAndB}, $\eval{\heap}{\perms}{e_1 \kand e_2}{v_2}$, and $v_2 = V''(t_2)$.

    \case \refrule{SEvalOp} -- $    \seval{\sstate}{e_1 \oplus e_2}{t_1 \oplus t_2}{\sstate''}{\_}$: 
    
    By \refrule{SEvalOp} $\seval{\sstate}{e_1}{t_1}{\sstate'}{\_}$. Let $V' = V[\seval{\sstate}{e_1}{t_1}{\sstate'}{\_} \mid \heap]$. Also by \refrule{SEvalOP}, $\seval{\sstate'}{e_2}{t_2}{\sstate''}{\_}$. Let $V'' = V'[\seval{\sstate'}{e_2}{t_2}{\sstate''}{\_} \mid \heap]$. Then $V''$ is the corresponding valuation for this case.  

    Since $\eval{\heap}{\env}{e_1 \oplus e_2}{v}$ for some $v$, by \refrule{EvalOP} $\eval{\heap}{\env}{e_1}{v_1}$ for some $v_1$ and $\eval{\heap}{\env}{e_2}{v_2}$ for some $v_2$. In addition, since $\efoot{\heap}{\perms}{e_1 \oplus e_2} \subseteq \perms$, by definition \ref{def:efoot}, $\efoot{\heap}{\perms}{e_1} \subseteq \perms$ and $\efoot{\heap}{\perms}{e_2} \subseteq \perms$. Then by induction $\simstate{V'}{\sstate'}{\heap}{\perms}{\env}$ and $\eval{\heap}{\env}{e_1}{V'(t_1)}$. As such, by induction $\simstate{V''}{\sstate''}{\heap}{\perms}{\env}$ and $\eval{\heap}{\env}{e_2}{V''(t_2)}$. Since $V''$ extends $V'$, by \refrule{EvalOP}, $\eval{\heap}{\env}{e_1 \oplus e_2}{V''(t_1) \oplus V''(t_2)}$ and $V''(t_1) \oplus V''(t_2) = V''(t_1 \oplus t_2)$.
    
    \case \refrule{SEvalNeg} -- $    \seval{\sstate}{\kneg e}{\kneg t}{\sstate'}{\_}$: 
    
    By \refrule{SEvalNeg} $\seval{\sstate}{e}{t}{\sstate'}{\_}$, and the corresponding valuation for this case is $V' = V[\seval{\sstate}{e}{t}{\sstate'}{\_} \mid \heap]$. Since $\eval{\heap}{\env}{\kneg e}{v}$ for some $v$, by \refrule{EvalNeg} $\eval{\heap}{\env}{e}{v'}$ for some $v'$. In addition, since $\efoot{\heap}{\perms}{\kneg e} \subseteq \perms$, by definition \ref{def:efoot}, $\efoot{\heap}{\perms}{e} \subseteq \perms$. Then by induction $\simstate{V'}{\sstate'}{\heap}{\perms}{\env}$ and $\eval{\heap}{\env}{e}{V'(t)}$. Then by \refrule{EvalNeg} $\eval{\heap}{\env}{\kneg e}{\neg V'(t)}$ and $\neg V'(t) = V'(\kneg t)$. 
    
    \case\label{case:sevalfield-correspondence} \refrule{SEvalField} -- $\seval{\sstate}{e.f}{t}{\sstate'}{\_}$: 
    
    By \refrule{SEvalField}, $\seval{\sstate}{e}{t_e}{\sstate'}{\_}$, and the corresponding valuation for this case is $V' = V[\seval{\sstate}{e}{t_e}{\sstate'}{\_} \mid \heap]$. 

    Since $\efoot{\heap}{\env}{e.f} \subseteq \perms$, by definition \ref{def:efoot} $\efoot{\heap}{\env}{e} \subseteq \perms$. In addition, since $\eval{\heap}{\env}{e.f}{v}$, by \refrule{EvalField} $\eval{\heap}{\env}{e}{\ell}$ for some $\ell$. Then by induction $\simstate{V'}{\sstate'}{\heap}{\perms}{\env}$, and $\eval{\heap}{\env}{e}{V'(t_e)}$. Then by \refrule{EvalField}, $\eval{\heap}{\env}{e.f}{\heap(V'(t_e), f)}$. 

    By \refrule{SEvalField}, $\pc(\sstate') \implies t_e \keq t'_e$, therefore $V'(t_e) = V'(t'_e)$. Also by \refrule{SEvalField} $\triple{f}{t_e'}{t} \in \sheap(\sstate')$. Therefore $V'(t) = \heap(V'(t_e'), f) = \heap(V'(t_e), f)$ since $\simstate{V'}{\sstate'}{\heap}{\perms}{\env}$. 
    
    \case \refrule{SEvalFieldOptimistic} -- $    \seval{\sstate}{e.f}{t}{\sstate'}{\_}$: Similar to case \ref{case:sevalfield-correspondence}.
    
    \case \refrule{SEvalFieldImprecise} -- $\seval{\sstate}{e.f}{t}{\sstate''}{\_}$: 
    
    By \refrule{SEvalFieldImprecise}, $\seval{\sstate}{e}{t_e}{\sstate'}{\_}$, $t = \ffresh$, and $\sstate'' = \sstate'[\oheap = \oheap(\sstate'); \triple{f}{t_e}{t}]$. Let $V' = V[\seval{\sstate}{e}{t_e}{\sstate'}{\_} \mid H][t \mapsto \heap(V(t_e), f)]$. Then $V'$ is the corresponding valuation for this case. 

    Since $\efoot{\heap}{\env}{e.f} \subseteq \perms$, by definition \ref{def:efoot} $\efoot{\heap}{\env}{e} \subseteq \perms$. In addition, since $\eval{\heap}{\env}{e.f}{v}$, by \refrule{EvalField} $\eval{\heap}{\env}{e}{\ell}$ for some $\ell$. Then, by induction $\simstate{V'}{\sstate'}{\heap}{\perms}{\env}$ and $\eval{\heap}{\env}{e}{V'(t_e)}$ (so $\ell = V'(t_e)$). In addition, $V'(t) = \heap(V'(t_e), f)$ by definition of $V'$. Then by \refrule{EvalField} $\eval{\heap}{\env}{e.f}{\heap(\ell, f)}$, and $V'(t) = \heap(\ell, f)$. 
    
    Since $\efoot{\heap}{\env}{e.f} \subseteq \perms$, by definition \ref{def:efoot} $\pair{\ell}{f} = \pair{V'(t_e)}{f} \subseteq \perms$. Then since $\simstate{V'}{\sstate'}{\heap}{\perms}{\env}$, $\simheap{V'}{\oheap(\sstate')}{\heap}{\perms}$. Therefore, $\simheap{V'}{\oheap(\sstate'); \triple{f}{t_e}{t}}{\heap}{\perms}$.

    Then $\simstate{V'}{\sstate''}{\heap}{\perms}{\env}$ since $\sstate'$ and $\sstate''$ differ only in their $\oheap$ components and $\oheap(\sstate'') = \oheap(\sstate'); \triple{f}{t_e}{t}$.

    \case \refrule{SEvalFieldFailure} -- $    \seval{\sstate}{e.f}{t}{\sstate'}{\_}$: 
    
    By \refrule{SEvalFieldFailure}, $\seval{\sstate}{e}{t_e}{\sstate'}{\_}$ and $t = \ffresh$. Let $V' = V[\seval{\sstate}{e}{t_e}{\sstate'}{\_}]$. Then the corresponding valuation for this case is $V'' = V'[t \mapsto \heap(V(t_e), f)]$. 
    
    Since $\efoot{\heap}{\env}{e.f} \subseteq \perms$, by definition \ref{def:efoot} $\efoot{\heap}{\env}{e} \subseteq \perms$. In addition, since $\eval{\heap}{\env}{e.f}{v}$, by \refrule{EvalField} $\eval{\heap}{\env}{e}{\ell}$ for some $\ell$. Then, by induction $\simstate{V'}{\sstate'}{\heap}{\perms}{\env}$ and $\eval{\heap}{\env}{e}{V'(t_e)}$ (meaning $\ell = V'(t_e)$). Since $V''$ extends $V'$, $\simstate{V''}{\sstate'}{\heap}{\perms}{\env}$ and $\eval{\heap}{\env}{e}{V''(t_e)}$. 

    Then by \refrule{EvalField}, $\eval{\heap}{\env}{e.f}{H(V''(t_e), f)}$. Since $V''$ extends $V'$ which extends $V$, $V''(t) = V(t) = H(V(t_e), f)$.

    \case \refrule{SEvalUnfoldingPrecise} -- $\seval{\sstate_1}{\sunfolding{\pe}{e_0}}{t_0}{\sstate_6}{\_}$:

    Inductively, $\simstate{V'}{\sstate_2}{\heap}{\perms}{\env}$. Since consuming $\pe$ only removes from the precise and optimistic heaps, $\sheap(\sstate_3) \subseteq \sheap(\sstate_2)$ and $\oheap(\sstate_3) \subseteq \oheap(\sstate_2)$. Therefore, using the definition of state correspondence, $\simstate{V'}{\sstate_3}{\heap}{\perms}{\env}$. Then, since $\assertion{\heap}{\perms}{[\multiple{x \mapsto V(t)}]}{\fpred(p)}$ in the run-time semantics due to the previous state correspondence which uses equirecursive semantics, all obtained predicate chunks and field chunks must hold after unfolding the predicate once, meaning $\simstate{V'}{\sstate_4}{\heap}{\perms}{\env}$. Then, by induction, $\simstate{V'}{\sstate_5}{\heap}{\perms}{\env}$. Since $\eval{\heap}{\env}{\sunfolding{\pe}{e_0}}{v}$, by \refrule{EvalUnfolding} $\eval{\heap}{\env}{e_0}{v}$. Inductively, $\simstate{V'}{\sstate_5}{\heap}{\perms}{\env}$ and $\eval{\heap}{\env}{\sunfolding{\pe}{e_0}}{V'(t_0)}$.
    
    \case \refrule{SEvalUnfoldingImpreciseA} and \refrule{SEvalUnfoldingImpreciseB} -- $\seval{\sstate_1}{\sunfolding{\pe}{e_0}}{t_0}{\sstate_6}{\_}$:
    
    Both of these cases are similar to the one for \refrule{SEvalUnfoldingPrecise}, since all that changes is the information retained in the optimistic heap at the end of evaluation.
    
    \case \refrule{SEvalUnfoldingImplicitPrecise} -- $\seval{\sstate_1}{\sunfolding{\pe}{e_0}}{t_0}{\sstate_2}{\_}$:

    By the induction hypothesis, $\simstate{V'}{\sstate_2}{\heap}{\perms}{\env}$ and $\eval{\heap}{\env}{e_0}{v_0}$ where $v_0 = V'(t_0)$. Therefore, by \refrule{EvalUnfolding}, $\eval{\heap}{\env}{\sunfolding{\pe}{e_0}}{v_0}$.

    \case \refrule{SEvalUnfoldingImplicitImprecise} -- $\seval{\sstate_1}{\sunfolding{\pe}{e_0}}{t_0}{\sstate_3}{\_}$:

    Similar to the case for \refrule{SEvalUnfoldingImplicitPrecise}.

    \case \refrule{SEvalUnfoldingImplicitFailure} -- $\seval{\sstate_1}{\sunfolding{\pe}{e_0}}{t_0}{\sstate_2}{\_}$:

    Similar to the case for \refrule{SEvalUnfoldingImplicitPrecise}.

    \case \refrule{SEvalFunctionExplicit} -- $\seval{\sstate_1}{\fe}{t'}{\sstate_6}{\_}$:
    
    Inductively, $\simstate{V'}{\sstate_2}{\heap}{\perms}{\env}$. Since consuming $\ffuncpre(f)$ only removes from the precise and optimistic heaps, $\sheap(\sstate_3) \subseteq \sheap(\sstate_2)$ and $\oheap(\sstate_3) \subseteq \oheap(\sstate_2)$. Therefore, using the definition of state correspondence, $\simstate{V'}{\sstate_3}{\heap}{\perms}{\env}$. Since producing directly after consuming results in precise and optimistic heaps that are subsets of their counterparts from $\sstate_1$, again $\simstate{V'}{\sstate_4}{\heap}{\perms}{\env}$. Inductively, $\simstate{V'}{\sstate_5}{\heap}{\perms}{\env}$. Since $\oheap(\sstate_6) = \oheap(\sstate_5)$ and $\sheap(\sstate_6) = \sheap(\sstate_2)$, $\simstate{V'}{\sstate_6}{\heap}{\perms}{\env}$ as desired.

    \case \refrule{SEvalFunctionImplicit} -- $\seval{\sstate_1}{\fe}{t'}{\sstate_5}{\_}$:

    Inductively, $\simstate{V'}{\sstate_2}{\heap}{\perms}{\env}$. Then, since $\sstate_4 = \sstate_2[\imp = \imp(\sstate_3), \oheap = \oheap', \pc = \pc(\sstate_2) \kand \ldots]$, it suffices to show $\simheap{V'}{\oheap'}{\heap}{\perms}$. If $\ffuncpre(f)$ is precise, then $s' \setminus \sheap(\sstate_2) = \emptyset$ because all heap chunks used in the body are optimistically proven, meaning that they already exist in the initial heap and $\simheap{V'}{\emptyset}{\heap}{\perms}$ trivially. Otherwise, since $\efoot{\heap}{\env}{\fe} \subseteq \perms$ by Definition~\ref{def:efoot}, $\simstate{V'}{\heap}{\perms}{\env}{\sstate'}$ where $\oheap(\sstate') \supseteq \oheap'$. Therefore, $\simheap{V'}{\oheap'}{\heap}{\perms}$.

     Since $\eval{\heap}{\env}{\fe}{v'}$ where $\eval{\heap}{\env'}{\ffuncbody(f)}{v'}$, then by induction $v' = V(t') = V(\fts)$.
    
    \end{enumcases}
  \end{proof}
\end{snugshade}

\subsection{Produce}
\label{sec:produce-soundness}

\begin{definition}
  For a judgement $\sproduce{\sstate}{\gform}{\sstate'}$, given an initial valuation $V$ and heap $\heap$, the \textbf{corresponding valuation} is denoted
  $$V[\sproduce{\sstate}{\gform}{\sstate'} \mid \heap, \env].$$
  This valuation is defined as follows, depending on the rule that proves the derivation. Values are referenced using the respective name from the rule definition.

  Note that the corresponding valuation always extends the initial valuation and is defined for all $\ffresh$ symbolic values in the judgement.
  \begin{itemize}
    \item \refrule{SProduceImprecise}:
      $$V[\sproduce{\sstate}{\simprecise{\phi}}{\sstate'} \mid \heap, \env] := V[\sproduce{\sstate[\imp = \top]}{\phi}{\sstate'} \mid \heap, \env]$$
    \item \refrule{SProduceExpr}:
      $$V[\sproduce{\sstate}{e}{\sstate''} \mid \heap, \env] := V[\seval{\sstate}{e}{t}{\sstate'}{\_} \mid \heap, \env]$$
    \item \refrule{SProducePredicate}:
      $$V[\sproduce{\sstate}{p(\multiple{e})}{\sstate''} \mid \heap, \env] := V[\multiple{\seval{\sstate}{e}{t}{\sstate'}{\_} \mid \heap, \env}]$$
    \item \refrule{SProduceField}:
      $$V[\sproduce{\sstate}{\kacc(e.f)}{\sstate''} \mid \heap, \env] := V'[t \mapsto \heap(V'(t_e), f)]$$
      where $V' = V[\seval{\sstate}{e}{t_e}{\sstate'}{\_} \mid \heap, \env]$.
    \item \refrule{SProduceConjunction}:
      $$V[\sproduce{\sstate}{\phi_1 * \phi_2}{\sstate''} \mid \heap, \env] := V[\sproduce{\sstate}{\phi_1}{\sstate'} \mid \heap, \env][\sproduce{\sstate'}{\phi_2}{\sstate''} \mid \heap, \env]$$
    \item \refrule{SProduceIfA}:
      \begin{align*}
        &V[\sproduce{\sstate}{\sif{e}{\phi_1}{\phi_2}}{\sstate''} \mid \heap, \env] := \\&\quad V[\seval{\sstate}{e}{t}{\sstate}{\_} \mid \heap, \env][\sproduce{\sstate'[\pc = \pc(\sstate') \kand t]}{\phi_1}{\sstate''} \mid \heap, \env]
      \end{align*}
    \item \refrule{SProduceIfB}:
      \begin{align*}
        &V[\sproduce{\sstate}{\sif{e}{\phi_1}{\phi_2}}{\sstate''} \mid \heap, \env] := \\&\quad V[\seval{\sstate}{e}{t}{\sstate'}{\_} \mid \heap, \env][\sproduce{\sstate'[\pc = \pc(\sstate') \kand \kneg t]}{\phi_2}{\sstate''} \mid \heap, \env]
      \end{align*}
  \end{itemize}
\end{definition}

\begin{lemma}\label{lem:produce-subpath}
  If $\sproduce{\sstate}{\phi}{\sstate'}$, then $\pc(\sstate') \implies \pc(\sstate)$.
\end{lemma}

\begin{proof}
  By induction on $\sproduce{\sstate}{\phi}{\sstate'}$:

  \begin{enumcases}
    \case \refrule{SProduceImprecise} -- $\sproduce{\sstate}{\simprecise{\phi}}{\sstate'}$:
    
      By \refrule{SProduceImprecise} $\sproduce{\sstate[\imp = \top]}{\phi}{\sstate'}$, thus by induction $\pc(\sstate') \implies \pc(\sstate)$.

    \case \refrule{SProduceExpr} -- $\sproduce{\sstate}{e}{\sstate''}$: 
    
    By \refrule{SProduceExpr} $\sstate'' = \sstate'[\pc = \pc(\sstate') \kand t]$, where $\seval{\sigma}{e}{t}{\sigma'}{\_}$. Trivially $\pc(\sstate') \kand t \implies \pc(\sstate')$ and by lemma \ref{lem:eval-subpath} $\pc(\sstate') \implies \pc(\sigma)$.

    \case\label{case:produce-subpath-pred} \refrule{SProducePredicate} -- $\sproduce{\sstate}{p(\multiple{e})}
    {\sstate''}$:
    
      By \refrule{SProducePredicate} $\sstate'' = \sstate'[\sheap = \cdots]$ where $\multiple{\seval{\sstate}{e}{t}{\sstate'}{\_}}\,$, so $\pc(\sstate'') = \pc(\sstate')$. Then by lemma \ref{lem:eval-subpath} $\pc(\sstate') \implies \pc(\sstate)$.

    \case \refrule{SProduceField} -- $\sproduce{\sstate}{\kacc(e.f)}{\sstate'}$: Similar to case \ref{case:produce-subpath-pred}.

    \case \refrule{SProduceConjunction} -- $\sproduce{\sstate}{\phi_1 * \phi_2}{\sstate''}$:
    
      By \refrule{SProduceConjunction} $\sproduce{\sstate}{\phi_1}{\sstate'}$ and $\sproduce{\sstate'}{\phi_2}{\sstate''}$. By induction $\pc(\sstate'') \implies \pc(\sstate')$ and $\pc(\sstate') \implies \pc(\sstate)$, therefore $\pc(\sstate'') \implies \pc(\sstate)$.

    \case\label{case:produce-subpath-ifa} \refrule{SProduceIfA} -- $\sproduce{\sstate}{\sif{e}{\phi_1}{\phi_2}}{\sstate''}$:
    
      By \refrule{SProduceIfA} $\sproduce{\sstate'[\pc = \pc(\sstate') \kand t]}{\phi_1}{\sstate''}$ where $\seval{\sigma}{e}{t}{\sigma'}{\_}$ thus by induction $\pc(\sstate'') \implies \pc(\sstate') \kand t \implies \pc(\sstate')$. Then by lemma \ref{lem:eval-subpath} $\pc(\sstate') \implies \pc(\sstate)$.

    \case \refrule{SProduceIfB} -- $\sproduce{\sstate}{\sif{e}{\phi_1}{\phi_2}}{\sstate''}$: Similar to case \ref{case:produce-subpath-ifa}.
  \end{enumcases}
\end{proof}

\begin{lemma}\label{lem:produce-unchanged}
  If $\sproduce{\sstate}{\gform}{\sstate'}$ then $\senv(\sstate') = \senv(\sstate)$.
\end{lemma}
\begin{proof}
  Trivial by induction on $\sproduce{\sstate}{\gform}{\sstate'}$ and by using lemma \ref{lem:eval-unchanged}.
\end{proof}

\begin{lemma}\label{lem:produce-soundness}
  Suppose $\simstate{V}{\sstate}{\heap}{\perms}{\env}$ where $\simheap{V}{\sheap(\sstate)}{\heap}{\perms \setminus \efoot{\heap}{\env}{\gform}}$. 
  
  If $\sproduce{\sstate}{\gform}{\sstate'}$, $\assertion{\heap}{\perms}{\env}{\gform}$, and $V'(\pc(\sstate')) = \ktrue$ where $V' = V[\sproduce{\sstate}{\gform}{\sstate'} \mid \heap, \env]$, then
  $$\simstate{V'}{\sstate'}{\heap}{\perms}{\env}$$
\end{lemma}
\begin{proof}
  By induction on $\sproduce{\sstate}{\gform}{\sstate'}$:
  \begin{enumcases}
    \case \refrule{SProduceImprecise} -- $\sproduce{\sstate}{\simprecise{\phi}}{\sstate'}$:

      By \refrule{SProduceImprecise} $\sproduce{\sstate[\imp = \top]}{\phi}{\sstate'}$. Let $V'$ be the corresponding valuation, thus $V'$ is the corresponding valuation for this case.

      Then, since $\assertion{\heap}{\perms}{\env}{\simprecise{\phi}}$, by \refrule{AssertImprecise} $\assertion{\heap}{\perms}{\env}{\phi}$. Also, $\efoot{\heap}{\env}{\simprecise{\phi}} = \efoot{\heap}{\env}{\phi}$. Therefore $\simheap{V}{\sheap(\sstate)}{\heap}{\perms \setminus \efoot{\heap}{\env}{\phi}}$, and furthermore $\simheap{V}{\sheap(\sstate[\imp = \top])}{\heap}{\perms \setminus \efoot{\heap}{\env}{\phi}}$.

      Therefore $\simstate{V'}{\sstate'}{\heap}{\perms}{\env}$ by induction.

    \case \refrule{SProduceExpr} -- $\sproduce{\sstate}{e}{\sstate''}$:
    
      By \refrule{SProduceExpr} $\seval{\sstate}{e}{t}{\sstate'}{\_}$. Let $V'$ be the corresponding valuation, therefore $V'$ is the corresponding valuation for this case.

      Since $\assertion{\heap}{\perms}{\env}{\env}$, by \refrule{AssertValue} $\eval{\heap}{\env}{e}{\ktrue}$. Since $\simheap{V}{\sheap(\sstate)}{\heap}{\perms \setminus \efoot{\heap}{\env}{e}}$ and $\simstate{V}{\sstate}{\heap}{\perms}{\env}$ clearly $\efoot{\heap}{\env}{e} \subseteq \perms$. Then by lemma \ref{lem:eval-correspondence}, $\simstate{V'}{\sstate'}{\heap}{\perms}{\env}$.

      Let $\sstate'' = \sstate'[\pc(\sstate') \kand t]$.  Then, since $\sstate'$ and $\sstate$ differ only in their $\pc$ components $\simstate{V'}{\sstate''}{\heap}{\perms}{\env}$ since $V'(\pc(\sstate')) = \ktrue$ as $\simstate{V'}{\sstate'}{\heap}{\perms}{\env}$.
      
    \case \refrule{SProducePredicate} -- $\sproduce{\sstate}{p(\multiple{e})}{\sstate'}$:

      By \refrule{SProducePredicate}, for each $e$, $\seval{\sstate}{e}{t}{\sstate'}{\_}$ for some $t_i$. Let $V'$ be the corresponding valuation for this case, thus $V'$ extends the corresponding valuation corresponding for each $\seval{\sstate}{e}{t}{\sstate'}{\_}$. 

      By assumptions, $\assertion{\heap}{\perms}{\env}{p(\multiple{e})}$. Then by \refrule{AssertPredicate}, for each $e$, $\eval{\heap}{\perms}{e}{v}$. Since $\simheap{V}{\sheap(\sstate)}{\heap}{\perms \setminus \efoot{\heap}{\env}{p(\multiple{e})}}$ and $\simstate{V}{\sstate}{\heap}{\perms}{\env}$ clearly $\efoot{\heap}{\perms}{p(\multiple{e}} \subseteq \perms$. Then $\multiple{\efoot{\heap}{\perms}{e}} \subseteq {\perms}$ by definition \ref{def:efoot}. Then by lemma \ref{lem:eval-correspondence}, $\simstate{V'}{\sstate'}{\heap}{\perms}{\env}$ and $\eval{\heap}{\env}{e}{V'(t)}$ for all $e$. In addition, since $\sheap(\sstate') = \sheap(\sstate)$ by lemma \ref{lem:eval-unchanged}, $\simheap{V'}{\sheap(\sstate')}{\heap}{\perms \setminus \efoot{\heap}{\env}{p(\multiple{e})}}$.
      
      By \refrule{SProducePredicate} $\sstate'' = \sstate'[\sheap = \sheap(\sstate'); \pair{p}{\multiple{t}}]$. Since $\sstate'$ and $\sstate''$ differ only in their $\sheap$ components, proving that \eqref{eq:sheap-correspondence} holds for $\sheap(\sstate'')$ is sufficient to prove that $\simstate{V'}{\sstate''}{\heap}{\perms}{\env}$. Then
      \[
        \efoot{\heap}{\env}{p(\multiple{e})} \supseteq \efoot{\heap}{\multiple{x \mapsto V'(t)}}{\fpred(p)} = \vfoot{V'}{\heap}{\pair{p}{\multiple{t}}},
      \]
      where $\multiple{x} = \fpredparams(p)$
      thus $\simheap{V'}{\sheap(\sstate')}{\heap}{\perms \setminus \vfoot{V}{\heap}{\pair{p}{\multiple{t}}}}$ by lemma \ref{lem:simstate-monotonicity}. Now by lemma \ref{lem:sim-heap-disjoint} and since $V'$ extends $V$,
      $$\universal{h \in \sheap(\sstate')}{\vfoot{V'}{\heap}{h} \cap \vfoot{V'}{\heap}{\pair{p}{\multiple{t}}} = \emptyset}.$$
      Since $\pair{p}{\multiple{t}}$ is the only addition to $\sheap(\sstate'')$ relative to $\sheap(\sstate')$ and $V'$ extends $V$,
      $$\universal{h_1, h_2 \in \sheap(\sstate'')}{h_1 \ne h_2 \implies \vfoot{V'}{\heap}{h_1} \cap \vfoot{V'}{\heap}{h_2} = \emptyset}.$$
      Also by \refrule{AssertPredicate} $\assertion{\heap}{\perms}{[\multiple{x \mapsto V'(t)}]}{\fpred(p)}$. Therefore, since $\pair{p}{\multiple{t}}$ is the only addition to $\sheap(\sstate'')$ relative to $\sheap(\sstate')$ and $V'$ extends $V$,
      $$\universal{\pair{p}{\multiple{t}} \in \sheap(\sstate'')}{\assertion{\heap}{\perms}{[\multiple{x \mapsto V'(t)}]}{\fpred(p)}}.$$
      Therefore all requirements for \eqref{eq:sheap-correspondence} are satisfied, and therefore $\simstate{V'}{\sstate''}{\heap}{\perms}{\env}$.

    \case \refrule{SProduceField} -- $\sproduce{\sstate}{\kacc(e.f)}{\sstate''}$:

      By \refrule{SProduceField} $\seval{\sstate}{e}{t_e}{\sstate'}{\_}$ for some $t_e$. Let $V_e$ be the corresponding valuation, and let $V'$ be the corresponding valuation for this case, thus $V'$ extends $V_e$.
      
      Also by \refrule{SProduceField} $\sstate'' = \sstate'[\sheap = \sheap(\sstate'); \triple{f}{t_e}{t}]$ where $t = \ffresh$. By definition $V'(t) = \heap(V_e(t_e), f)$.

      By assumptions $\assertion{\heap}{\perms}{\env}{\kacc(e.f)}$, thus by \refrule{AssertAcc} $\eval{\heap}{\env}{e}{v_e}$ for some $v_e$ such that $\pair{v_e}{f} \in \perms$.  Then since $\simheap{V}{\sheap{\sstate}}{\heap}{\perms 
      \setminus \efoot{\heap}{\perms}{\kacc(e.f)}}$ and $\simstate{V}{\sstate}{\heap}{\perms}{\env}$, $\efoot{\heap}{\perms}{\kacc(e.f)} \subseteq \perms$. Then $\efoot{\heap}{\perms}{e} \subseteq \perms$ by definition \ref{def:efoot}. Then by lemma \ref{lem:eval-correspondence}, $\simstate{V_e}{\sstate'}{\heap}{\perms}{\env}$ and $\eval{\heap}{\env}{e}{V_e(t_e)}$. In addition, $\pair{V'(t_e)}{f} \in \perms$ as $V'$ extends $V_e$. Also by lemma \ref{lem:eval-unchanged}, $\sheap(\sstate') = \sheap(\sstate)$, so $\simheap{V_e}{\sheap(\sstate')}{\heap}{\perms \setminus \efoot{\heap}{\perms}{\kacc(e.f)}}$.

      Since $\sstate'$ and $\sstate''$ differ only in their $\sheap$ components, proving that \eqref{eq:sheap-correspondence} holds for $\sheap(\sstate'')$ is sufficient to prove that $\simstate{V'}{\sstate''}{\heap}{\perms}{\env}$.

      As shown before, $V'(t) = \heap(V'(t_e), f)$. Therefore, since $\triple{f}{t_e}{t}$ is the only addition to $\sheap(\sstate'')$ relative to $\sheap(\sstate')$ and $V'$ extends $V_e$,
      $$\universal{\triple{f}{t}{t'} \in \sheap(\sstate'')}{\heap(V'(t), f) = V'(t')} ~\text{and}$$
      $$\universal{\triple{f}{t}{t'} \in \sheap(\sstate'')}{\pair{V'(t)}{f} \in \perms}$$
      Since $\eval{\heap}{\env}{e}{V'(t_e)}$ as shown before,
      $$\efoot{\heap}{\env}{\kacc(e.f)} \supseteq \set{\pair{V'(t_e)}{f}} = \vfoot{V}{\heap}{(t_e.f; t)}.$$
      Therefore $\simheap{V_e}{\sheap(\sstate')}{\heap}{\perms \setminus \vfoot{V'}{\heap}{(t.f; t')}}$ by lemma \ref{lem:simstate-monotonicity}. Now by lemma \ref{lem:sim-heap-disjoint}, and since $V'$ extends $V_e$,
      $$\universal{h \in \sheap(\sstate'')}{\vfoot{V'}{\heap}{h} \cap \vfoot{V'}{\heap}{\triple{f}{t_e}{t}} = \emptyset}.$$
      Finally, since $\triple{f}{t_e}{t}$ is the only addition to $\sheap(\sstate'')$ relative to $\sheap(\sstate')$ and $V'$ extends $V_e$,
      $$\universal{h_1, h_2 \in \sheap(\sstate'')^2}{h_1 \ne h_2 \implies \vfoot{V'}{\heap}{h_1} \cap \vfoot{V'}{\heap}{h_2} = \emptyset}.$$
      Therefore all requirements for \eqref{eq:sheap-correspondence} are satisfied and therefore $\simstate{V'}{\sstate''}{\heap}{\perms}{\env}$.

    \case \refrule{SProduceConjunction} -- $\sproduce{\sstate}{\phi_1 * \phi_2}{\sstate''}$:

      By \refrule{SProduceConjunction} $\sproduce{\sstate}{\phi_1}{\sstate'}$ and $\sproduce{\sstate'}{\phi_2}{\sstate''}$. Let $V_1$ and $V_2$ be the respective corresponding valuations, extending $V$ and $V_1$, respectively. Then $V_2$ is the corresponding valuation for this case.

      Since $\assertion{\heap}{\perms}{\env}{\phi_1 * \phi_2}$, by \refrule{AssertConjunction} $\assertion{\heap}{\perms}{\env}{\phi_1}$ and $\assertion{\heap}{\perms}{\env}{\phi_2}$ where $\efoot{\heap}{\env}{\phi_1} \cap \efoot{\heap}{\env}{\phi_2} = \emptyset$. Also, $\efoot{\heap}{\env}{\phi_1 * \phi_2} = \efoot{\heap}{\env}{\phi_1} \cup \efoot{\heap}{\env}{\phi_2}$.

      Let $\perms' = \perms \setminus \efoot{\heap}{\env}{\phi_2}$. Then $\assertion{\heap}{\perms'}{\env}{\phi_1}$ by lemma \ref{lem:assert-efoot-subset} since $\efoot{\heap}{\env}{\phi_1} \subseteq \perms'$. Also, $\simheap{V}{\sheap(\sstate)}{\heap}{\perms' \setminus \efoot{\heap}{\env}{\phi_1}}$ since $\perms' \setminus \efoot{\heap}{\env}{\phi_1} = \perms \setminus \efoot{\heap}{\env}{\phi_1 * \phi_2}$. Finally, by lemma \ref{lem:produce-subpath}, $\pc(\sstate'') \implies \pc(\sstate')$, and therefore $V_2(\pc(\sstate')) = V_1(\pc(\sstate')) = \ktrue$.

      Now $\simstate{V_1}{\sstate'}{\heap}{\perms'}{\env}$ by induction. Therefore, $\simheap{V_1}{\sheap(\sstate')}{\heap}{\perms'}$. Also, $\assertion{\heap}{\perms}{\env}{\phi_2}$ by lemma \ref{lem:assert-efoot-subset} since $\efoot{\heap}{\env}{\phi_2} \subseteq \perms$. Finally, $V_2(\pc(\sstate'')) = \ktrue$ by assumption, therefore $\simstate{V_2}{\sstate''}{\heap}{\perms}{\env}$ by induction.

    \case\label{case:produce-soundness-ifa} \refrule{SProduceIfA} -- $\sproduce{\sstate}{\sif{e}{\phi_1}{\phi_2}}{\sstate''}$:

      By \refrule{SProduceIfA} $\seval{\sstate}{e}{t}{\sstate'}{\_}$ and $\sproduce{\sstate'[\pc = \pc(\sstate') \kand t]}{\phi_1}{\sstate''}$. Let $V_1$ and $V'$ be the respective corresponding valuations extending $V$ and $V_1$, respectively. Then $V'$ is the corresponding valuation in this case.

      By lemma \ref{lem:produce-subpath}, $\pc(\sstate'') \implies \pc(\sstate') \kand t$. By assumptions $V'(\pc(\sstate'')) = \ktrue$, thus $V'(\pc(\sstate') \kand t) = V_1(\pc(\sstate') \kand t) = \ktrue$. In addition, $V_1(\pc(\sigma')) = \ktrue$ by lemma \ref{lem:eval-progress}. Therefore, $V_1(t) = \ktrue$.

      Since $\assertion{\heap}{\perms}{\env}{\sif{e}{\phi_1}{\phi_2}}$, $\eval{\heap}{\env}{e}{v}$ for some $v$ by \refrule{AssertIfA} or \refrule{AssertIfB}. As $\simheap{V}{\sheap(\sstate)}{\heap}{\perms \setminus \efoot{\heap}{\perms}{\sif{e}{\phi_1}{\phi_2}}}$ and $\simstate{V}{\sstate}{\heap}{\perms}{\env}$, then $\efoot{\heap}{\perms}{\sif{e}{\phi_1}{\phi_2}} \subseteq \perms$. Thus, $\efoot{\heap}{\perms}{e} \subseteq \perms$ by definition \ref{def:efoot}. Then by lemma \ref{lem:eval-correspondence}, $\simstate{V_1}{\sstate'}{\heap}{\perms}{\env}$ and $\eval{\heap}{\env}{e}{V_1(t)}$. Therefore $v = V_1(t) = \ktrue$.  

      Then, by lemma \ref{lem:eval-unchanged} $\sheap(\sigma') = \sheap(\sigma)$, so 
      \[
      \simheap{V_1}{\sheap(\sstate')}{\heap}{\perms \setminus \efoot{\heap}{\perms}{\sif{e}{\phi_1}{\phi_2}}}.
      \]
      Then $\efoot{\heap}{\env}{\phi_1} \subseteq \efoot{\heap}{\env}{\sif{e}{\phi_1}{\phi_2}}$ by definition \ref{def:efoot}. Therefore, since $V_1(\pc(\sstate') \kand t) = \ktrue$,  
      \[
      \simheap{V_1}{\sheap(\sstate'[\pc = \pc(\sstate') \kand t])}{\heap}{\perms \setminus \efoot{\heap}{\env}{\phi_1}}.
      \] 
      Also, $\assertion{\heap}{\perms}{\env}{\phi_1}$ by \refrule{AssertIfA} since $\eval{\heap}{\env}{e}{\ktrue}$ and $\assertion{\heap}{\perms}{\env}{\sif{e}{\phi_1}{\phi_2}}$. Finally, $V'(\pc(\sstate'')) = \ktrue$ by assumption.
      
      Thus $\simstate{V'}{\sstate''}{\heap}{\perms}{\env}$ by induction.

    \case \refrule{SProduceIfB} -- $\sproduce{\sstate}{\sif{e}{\phi_1}{\phi_2}}{\sstate''}$: Similar to case \ref{case:produce-soundness-ifa}.
  \end{enumcases}
\end{proof}

\begin{lemma}[Progress] \label{lem:produce-progress}
  If $V(\pc(\sstate)) = \ktrue$, $\assertion{\heap}{\perms}{\env}{\gform}$, and $\simstate{V}{\sstate}{\heap}{\perms}{\env}$, then $\sproduce{\sstate}{\gform}{\sstate'}$ for some $\sstate'$ where $V'(\pc(\sstate')) = \ktrue$ and $V' = V[\sproduce{\sstate}{\gform}{\sstate'} \mid \heap, \env]$.
\end{lemma}
\begin{proof}
  Let $\perms' = \perms \cup \efoot{\heap}{\env}{\gform}$. Then by lemma \ref{lem:simstate-monotonicity} $\simstate{V}{\sstate}{\heap}{\perms'}{\env}$ and by lemma \ref{lem:assert-efoot-subset} $\assertion{\heap}{\perms'}{\env}{\gform}$.
  Then suppose $V(\pc(\sstate)) = \ktrue$ and complete the proof by induction on the syntax forms of $\gform$:

  \begin{enumcases}
    \case $\simprecise{\phi}$ -- $\phi \in \Formula$:

      Let $\hat{\sstate} = \sstate[\imp = \top]$. Then $V(\pc(\hat{\sstate})) = V(\pc(\sstate)) = \ktrue$. Also, since $\assertion{\heap}{\perms}{\env}{\simprecise{\phi}}$, by \refrule{AssertImprecise} $\assertion{\heap}{\perms}{\env}{\phi}$. Thus by induction $\sproduce{\hat{\sstate}}{\phi}{\sstate'}$ for some $\sstate'$ where $V'(\pc(\sstate')) = \ktrue$. Then $\sproduce{\sstate}{\simprecise{\phi}}{\sstate'}$ by \refrule{SProduceImprecise}, and $V'(\pc(\sstate')) = \ktrue$.

    \case $\phi_1 * \phi_2$ -- $\phi_1, \phi_2 \in \Formula$:

      Since $\assertion{\heap}{\perms}{\env}{\phi_1 * \phi_2}$, $\assertion{\heap}{\perms}{\env}{\phi_1}$ and $\assertion{\heap}{\perms}{\env}{\phi_2}$ by \refrule{AssertConjunction} and lemma \ref{lem:assert-monotonicity}.

      By induction $\sproduce{\sstate}{\phi_1}{\sstate'}$, with corresponding valuation $V'$ where $V'(\pc(\sstate')) = \ktrue$. Then by induction $\sproduce{\sstate'}{\phi_2}{\sstate''}$, with corresponding valuation $V''$, for some $\sstate''$ where $V''(\pc(\sstate'')) = \ktrue$. By \refrule{SProduceConjunction}, $\sproduce{\sstate}{\phi_1 * \phi_2}{\sstate''}$, and $V''(\pc(\sstate'')) = \ktrue$.

    \case $p(\multiple{e})$ -- $p \in \Predicate, \multiple{e \in \Expr}$:

      By lemma \ref{lem:eval-progress}, $\multiple{\seval{\sstate}{e}{t}{\sstate'}{\_}}$ for some $t \in \SExpr$ and $\sstate'$ where $V'(\pc(\sstate')) = \ktrue$, and $V'$ is the corresponding valuation extending $V$. Let $\sstate'' = \sstate'[\sheap = \sheap(\sstate'); \pair{p}{\multiple{t}}]$, then by \refrule{SProducePredicate}, $\sproduce{\sstate}{p(\multiple{e})}{\sstate''}$. Since $\sstate''$ and $\sstate'$ 
      differ only in their $\sheap$ components, $V'(\pc(\sstate'')) = V'(\pc(\sstate')) = \ktrue$.

    \case $e \in \Expr$:

      By lemma \ref{lem:eval-progress}, $\seval{\sstate}{e}{t}{\sstate'}{\_}$ for some $t \in \SExpr$. Let $V'$ be the corresponding valuation and $\sstate'' = \sstate'[\pc = \pc(\sstate') \kand t]$. 
      
      Since $\assertion{\heap}{\perms}{\env}{e}$, $\eval{\heap}{\env}{e}{\ktrue}$ by \refrule{AssertValue}. Then, since $\simstate{V}{\sstate}{\heap}{\perms'}{\env}$, and $\efoot{\heap}{\env}{e} \subseteq \perms'$, by lemma \ref{lem:eval-correspondence} $\eval{\heap}{\env}{e}{V'(t)}$. Then $V'(t) = \ktrue$ and $V'(\pc(\sstate')) = \ktrue$. Finally, $V'(\sstate'') = V'(\sstate') \wedge V'(t) = V(\sstate') \wedge \ktrue = \ktrue$.

    \case $\sif{e}{\phi_1}{\phi_2}$ -- $e \in \Expr, \phi_1, \phi_2 \in \Formula$:

      By lemma \ref{lem:eval-progress}, $\seval{\sstate}{e}{t}{\sstate'}{\_}$ for some $t \in \SExpr$ and $\sstate'$. Let $V_1$ be the valuation corresponding to this derivation. In addition, $V_1(\pc(\sstate')) = \ktrue$.

      Then one of the following rules must apply to derive $\assertion{\heap}{\perms}{\env}{\sif{e}{\phi_1}{\phi_2}}$:

      \subcase \refrule{AssertIfA}:

        Then $\eval{\heap}{\env}{e}{\ktrue}$. Since $\simstate{V}{\sstate}{\heap}{\perms'}{\env}$ and $\efoot{\heap}{\env}{e} \subseteq \perms'$, by lemma \ref{lem:eval-correspondence} $\simstate{V_1}{\sstate'}{\heap}{\perms'}{\env}$ and $\eval{\heap}{\env}{e}{V_1(t)}$. Therefore, $V_1(t) = \ktrue$ and $V_1(\pc(\sstate')) = \ktrue$.
        
        Therefore $V_1(\pc(\sstate') \kand t) = \ktrue$. Also, $\assertion{\heap}{\perms}{\env}{\phi_1}$ by \refrule{AssertIfA}, and so by lemma \ref{lem:assert-monotonicity} $\assertion{\heap}{\perms'}{\env}{\phi_1}$. 
        
        Then, since $\simstate{V_1}{\sstate'}{\heap}{\perms'}{\env}$, by induction $\sproduce{\sstate'[\pc = \pc(\sstate') \kand t]}{\phi_1}{\sstate''}$ for some $\sstate''$ with corresponding valuation $V'$ extending $V_1$, where $V'(\pc(\sstate'')) = \ktrue$.

        Now by \refrule{SProduceIfA}, $\sproduce{\sstate}{\sif{e}{\phi_1}{\phi_2}}{\sstate''}$. By definition $V'$ is the corresponding valuation extending $V$, and as shown before, $V'(\pc(\sstate'')) = \ktrue$.

      \subcase \refrule{AssertIfB}:

        Then $\eval{\heap}{\env}{e}{\kfalse}$. Since $\simstate{V}{\sstate}{\heap}{\perms'}{\env}$ and $\efoot{\heap}{\env}{e} \subseteq \perms'$, by lemma \ref{lem:eval-correspondence} $\simstate{V_1}{\sstate'}{\heap}{\perms'}{\env}$ and $\eval{\heap}{\env}{e}{V_1(t)}$. Therefore $V_1(t) = \kfalse$ and $V_1(\pc(\sstate')) = \ktrue$.
        
        Therefore $V_1(\pc(\sstate) \kand \kneg t) = \ktrue$. Also, $\assertion{\heap}{\perms}{\env}{\phi_2}$ by \refrule{AssertIfB}, and so by lemma \ref{lem:assert-monotonicity} $\assertion{\heap}{\perms'}{\env}{\phi_1}$.

        Then, since $\simstate{V_1}{\sstate'}{\heap}{\perms'}{\env}$, by induction $\sproduce{\sstate'[\pc = \pc(\sstate') \kand \kneg t]}{\phi_1}{\sstate''}$ for some $\sstate''$ with corresponding valuation $V'$ extending $V_1$, where $V'(\pc(\sstate'')) = \ktrue$.

        Now by \refrule{SProduceIfB}, $\sproduce{\sstate}{\sif{e}{\phi_1}{\phi_2}}{\sstate'}$. By definition $V'$ is the corresponding valuation extending $V$, and as shown before, $V'(\pc(\sstate')) = \ktrue$.

    \case $\kacc(e.f)$ where $e \in \Expr, f \in \Field$

      By lemma \ref{lem:eval-progress}, $\seval{\sstate}{e}{t}{\sstate'}{\_}$ for some $t$ and $V'(\pc(\sstate')) = \ktrue$, where $V_1$ is the corresponding valuation extending $V$. 
     
      Let $\sstate'' = \sstate'[\sheap = \sheap(\sstate'); \triple{f}{t}{\ffresh}]$.  Then, by \refrule{SProduceField}, $\sproduce{\sstate}{\kacc(e.f)}{\sstate''}$ with corresponding valuation $V'$ extending $V_1$. Since $\sstate''$ and $\sstate'$ differ only in their $\sheap$ component, $V'(\pc(\sstate'')) = V'(\pc(\sstate')) = \ktrue$.
  \end{enumcases}
\end{proof}

\subsection{Consume}\label{sec:soundness-consume}
\begin{definition}\label{def:consume-valuation}
  For a judgement $\sconsume{\sstate}{\sstate_E}{\gform}{\sstate'}{\scheck}{\sperms}$, given an initial valuation $V$, heap $\heap$, and environment $\env$, the \textbf{corresponding valuation} is denoted
  $$V[\sconsume{\sstate}{\sstate_E}{\gform}{\sstate'}{\scheck}{\sperms} \mid \heap, \env].$$
  This valuation is defined as follows, depending on the rule that proves the derivation. Values are referenced using the respective name from the rule definition.

  Note that the corresponding valuation always extends the initial valuation and is defined for all $\ffresh$ symbolic values in the judgement.
  \begin{itemize}
    \item \refrule{SConsumeImprecision}:
      \begin{align*}
        &V[\sconsume{\sstate}{\sstate_E}{\simprecise{\phi}}{\quintuple{\top}{\pc(\sstate')}{\senv(\sstate')}{\emptyset}{\emptyset}}{\scheck}{\sperms} \mid \heap, \env] := \\
        &\quad V[\sconsume{\sstate}{\sstate_E}{\phi}{\sstate'}{\scheck}{\sperms} \mid \heap, \env]
      \end{align*}
    \item \refrule{SConsumeValue}:
      $$V[\sconsume{\sstate}{\sstate_E}{e}{\sstate}{\scheck}{\emptyset} \mid \heap, \env] := V[\seval{\sstate_E}{e}{t}{\_}{\scheck} \mid \heap, \env]$$
    \item \refrule{SConsumeValueImprecise}:
      $$V[\sconsume{\sstate}{\sstate_E}{e}{\sstate[\pc = \pc(\sstate) \kand t]}{\scheck; t}{\emptyset} \mid \heap, \env] := V[\seval{\sstate_E}{e}{t}{\_}{\scheck} \mid \heap, \env]$$
    \item \refrule{SConsumeValueFailure}:
      $$V[\sconsume{\sstate}{\sstate_E}{e}{\sstate}{\set{\bot}}{\emptyset} \mid \heap, \env] := V[\seval{\sstate_E}{e}{t}{\_}{\scheck} \mid \heap, \env]$$
    \item \refrule{SConsumePredicate}:
      $$V[\sconsume{\sstate}{\sstate_E}{p(\multiple{e})}{\sstate[\sheap = \sheap', \oheap = \emptyset]}{\_}{\_} \mid \heap, \env] := V[\multiple{\seval{\sstate_E}{e}{t}{\_}{\scheck}} \mid \heap, \env]$$
    \item \refrule{SConsumePredicateImprecise}:
      $$V[\sconsume{\sstate}{\sstate_E}{p(\multiple{e})}{\sstate[\heap = \emptyset, \oheap = \emptyset]}{\_}{\_} \mid \heap, \env] := V[\multiple{\seval{\sstate_E}{e}{t}{\_}{\scheck}} \mid \heap, \env]$$
    \item \refrule{SConsumePredicateFailure}:
      $$V[\sconsume{\sstate}{\sstate_E}{p(\multiple{e})}{\sstate}{\set{\bot}}{\_} \mid \heap, \env] := V[\multiple{\seval{\sstate_E}{e}{t}{\_}{\scheck}} \mid \heap, \env]$$
    \item \refrule{SConsumeAcc}, \refrule{SConsumeAccOptimistic}, \refrule{SConsumeAccImprecise}, \\ \refrule{SConsumeAccFailure}:
      \begin{align*}
        &V[\sconsume{\sstate}{\sstate_E}{\kacc(e.f)}{\sstate[\sheap = \sheap', \oheap = \oheap']}{\_}{\_} \mid \heap, \env] := \\
        &\quad V[\seval{\sstate_E}{e}{t_e}{\_}{\scheck} \mid \heap, \env]
      \end{align*}
    \item \refrule{SConsumeConjunction}:
      \begin{align*}
        &V[\sconsume{\sstate}{\sstate_E}{\phi_1 * \phi_2}{\sstate''}{\_}{\_} \mid \heap, \env] := \\
        &\quad V[\sconsume{\sstate}{\sstate_E}{\phi_1}{\sstate'}{\scheck_1}{\sperms_1} \mid \heap, \env] \\
        &\hspace{1.85em} [\sconsume{\sstate'}{\sstate_E[\pc = \pc(\sstate')]}{\phi_2}{\sstate''}{\scheck_2}{\sperms_2} \mid \heap, \env]
      \end{align*}
    \item \refrule{SConsumeConditionalA}:
      \begin{align*}
        &V[\sconsume{\sstate}{\sstate_E}{\sif{e}{\phi_1}{\phi_2}}{\sstate'}{\scheck \cup \scheck'}{\_} \mid \heap, \env] := \\
        &\quad V[\sconsume{\sstate[\pc = \pc']}{\sstate_E[\pc = \pc']}{\phi_1}{\sstate'}{\scheck'}{\sperms} \mid \heap, \env]
      \end{align*}
    \item \refrule{SConsumeConditionalB}:
      \begin{align*}
        &V[\sconsume{\sstate}{\sstate_E}{\sif{e}{\phi_1}{\phi_2}}{\sstate'}{\scheck \cup \scheck'}{\_} \mid \heap, \env] := \\
        &\quad V[\sconsume{\sstate[\pc = \pc']}{\sstate_E[\pc = \pc']}{\phi_2}{\sstate'}{\scheck'}{\sperms} \mid \heap, \env]
      \end{align*}
  \end{itemize}
\end{definition}

\begin{definition}\label{def:cons-valuation}
  For a judgement $\scons{\sstate}{\gform}{\sstate'}{\scheck}$, given an initial valuation $V$, heap $\heap$, and environment $\env$, the \textbf{corresponding valuation} is denoted
  $$V[\scons{\sstate}{\gform}{\sstate'}{\scheck} \mid \heap, \env].$$

  The is defined by
  $$V[\scons{\sstate}{\gform}{\sstate'}{\scheck} \mid \heap, \env] := V[\sconsume{\sstate}{\sstate}{\gform}{\sstate'}{\scheck}{\_} \mid \heap, \env]$$
  where $\sconsume{\sstate}{\sstate}{\gform}{\sstate'}{\scheck}{\_}$ is the judgement used when applying \textsc{SConsume} to derive $\scons{\sstate}{\gform}{\sstate'}{\scheck}$.
\end{definition}

\begin{lemma}[Consume results in more specific path condition (long form)]\label{lem:consume-subpath}
  If $\sconsume{\sstate}{\sstate_E}{\gform}{\sstate'}{\_}{\_}$, then $\pc(\sstate') \implies \pc(\sstate)$.
\end{lemma}

\begin{lemma}[Consume results in more specific path condition (short form)]\label{lem:cons-subpath}
  If $\scons{\sstate}{\gform}{\sstate'}{\_}$, then $\pc(\sstate') \implies \pc(\sstate)$.
\end{lemma}

\begin{lemma}[Soundness of $\fremf$ for precise heaps and imprecise states]\label{lem:sheap-remf-imp}
  If $\simstate{V}{\sstate}{\heap}{\perms}{\env}$ and $\imp(\sstate)$ then $\simheap{V}{\fremf(\sheap(\sstate), \sstate, t, f)}{\heap}{\perms \setminus \set{(V(t), f)}}$ for any $t$, $f$.
\end{lemma}

\begin{lemma}\label{lem:consume-unchanged}
  If $\sconsume{\sstate}{\sstate_E}{\gform}{\sstate'}{\scheck}{\sperms}$ then $\senv(\sstate') = \senv(\sstate)$.
\end{lemma}

\begin{lemma}\label{lem:cons-unchanged}
  If $\scons{\sstate}{\gform}{\sstate'}{\scheck}$ then $\senv(\sstate') = \senv(\sstate)$.
\end{lemma}

\begin{lemma}[Soundness of $\fremfp$ for precise heaps]\label{lem:sheap-remfp-prec}
  If $\simstate{V}{\sstate}{\heap}{\perms}{\env}$ and $\triple{f}{t}{\_} \in \sheap(\sstate)$, then $\simheap{V}{\fremfp(\sheap(\sstate), \sstate, t, f)}{\heap}{\perms \setminus \set{\pair{V(t)}{f}}}$ for any $t$, $f$.
\end{lemma}

\begin{lemma}[Soundness of $\fremf$ for optimistic heaps]\label{lem:oheap-remf}
  If $\sstate$ is well-formed and $\simstate{V}{\sstate}{\heap}{\perms}{\env}$ then $\simheap{V}{\fremf(\oheap(\sstate), \sstate, t, f)}{\heap}{\perms \setminus \set{\pair{V(t)}{f}}}$ for any $t$, $f$.
\end{lemma}

\begin{proof}
    If $\lnot \imp(\sstate)$, then $\oheap(\sstate) = \emptyset$ because $\sstate$ is well-formed, so trivially $\simheap{V}{\heap}{\emptyset}{\oheap(\sstate)}$. Then, by definition, $\fremf(\oheap(\sstate), \sstate, t, f) \subseteq \oheap(\sstate) = \emptyset$, so $\simheap{V}{\heap}{\emptyset}{\fremf(\oheap(\sstate), \sstate, t, f)}$. Then, since $\emptyset \subseteq \perms \setminus \{ \pair{V(t)}{f} \}$, $\simheap{V}{\heap}{\perms \setminus \{ \pair{V(t)}{f} \}}{\fremf(\oheap(\sstate), \sstate, t, f)}$ by Lemma~\ref{lem:sim-oheap-monotonicity}.

    If $\imp(\sstate)$ then $\oheap(\sstate)$ may be non-empty. Let $\oheap' = \fremf(\oheap(\sstate), \sstate, t, f)$ and let $\perms' = \perms \setminus \{ \pair{V(t)}{f} \}$. Then, by definition of $\fremf$, $\oheap' \subseteq \oheap(\sstate)$, so $\simheap{V}{\heap}{\perms}{\oheap'}$. Then, by the definition of optimistic heap correspondence, $\heap(V(t), f) = V(t')$ for all $\triple{f}{t}{t'} \in \oheap'$. Using this and the fact that $\oheap'$ contains no predicate chunks to the definition of $\fremf$, it suffices to show that $\pair{V(s)}{g} \in \perms'$ for $\triple{g}{s}{s'} \in \oheap'$. By definition of $\fremf$,

    \begin{align*}
        \lnot \falias(\sstate, t, f, s, g) &\implies (f \neq g) \lor \lnot \fsat(\pc(\sstate) \kand [t == s]) \\
                                           &\implies (f \neq g) \lor (V(\pc(\sstate) \kand t == t') = \kfalse) \\
                                           &\implies (f \neq g) \lor (V(\pc(\sstate)) = \kfalse) \lor (V(t) \neq V(s)) \\
                                           &\implies (f \neq g) \lor (V(t) \neq V(s)) \\
                                           &\implies \pair{V(t)}{f} \neq \pair{V(s)}{g}.
    \end{align*}

    Furthermore, $\pair{V(s)}{g} \in \perms$ because $\simheap{V}{\heap}{\perms}{\oheap'}$, so $\pair{V(s)}{g} \in \perms'$ and $\simheap{V}{\heap}{\perms'}{\oheap'}$ as desired.
\end{proof}

\begin{lemma}[Soundness of $\sperms$ calculation]\label{lem:consume-assert-vfoot}
  If $\phi$ is a precise formula, $\simstate{V}{\sstate_E}{\heap}{\perms_E}{\env}$, $\sconsume{\sstate}{\sstate_E}{\phi}{\sstate'}{\scheck}{\sperms}$ with corresponding valuation $V'$, $V'(\pc(\sstate')) = \ktrue$, and $\assertion{\heap}{\perms'}{\env}{\phi}$ where $\perms' \subseteq \perms_E$, then $\assertion{\heap}{\vfoot{V'}{\heap}{\sperms}}{\env}{\phi}$ and $\vfoot{V}{\heap}{\sperms} \subseteq \perms'$.
\end{lemma}
\begin{proof}
  By induction on $\sconsume{\sstate}{\sstate_E}{\phi}{\sstate'}{\scheck}{\sperms}$:

  \begin{enumcases}
    \case \refrule{SConsumeImprecision} -- $\sconsume{\sstate}{\sstate_E}{\simprecise{\phi}}{\quintuple{\top}{\pc(\sstate')}{\senv(\sstate')}{\emptyset}{\emptyset}}{\scheck}{\sperms}$:
      $\simprecise{\phi}$ is not precise, therefore this rule cannot apply.

    \case \refrule{SConsumeValue}, \refrule{SConsumeValueImprecise}, \refrule{SConsumeValueFailure} -- $\sconsume{\sstate}{\sstate_E}{e}{\sstate}{\_}{\emptyset}$:

      Since $\assertion{\heap}{\perms'}{\env}{e}$, $\eval{\heap}{\env}{e}{\ktrue}$ by \refrule{AssertValue}. Therefore $\assertion{\heap}{\vfoot{V}{\heap}{\emptyset}}{\env}{e}$ by \refrule{AssertValue}.

      Also, $\vfoot{V}{\heap}{\emptyset} = \emptyset \subseteq \perms'$.

    \case \refrule{SConsumePredicate}, \refrule{SConsumePredicateImprecise}, \refrule{SConsumePredicateFailure} -- $\sconsume{\sstate}{\sstate_E}{p(\multiple{e})}{\_}{\_}{\set{\pair{p}{\multiple{t}}}}$:

      By the respective rule, $\multiple{\seval{\sstate_E}{e}{t}{\_}{\_}}$ for some $\multiple{t}$. The corresponding valuation for this case extends the corresponding valuation for all of these derivation.

      Since $\assertion{\heap}{\perms}{\env}{p(\multiple{e})}$, $\multiple{\eval{\heap}{\env}{e}{v}}$ for some $\multiple{v}$ by \refrule{AssertPredicate}. By assumptions \\
      $\simstate{V}{\sstate_E}{\heap}{\perms_E}{\env}$. Let $\hat{\perms_E} = \perms_E \cup \efoot{\heap}{\env}{p(\multiple{e})}$. Then by lemma \ref{lem:simstate-monotonicity} $\simstate{V}{\sstate_E}{\heap}{\hat{\perms_E}}{\env}$. Thus by lemma \ref{lem:eval-correspondence} $\multiple{v = V'(t)}$, i.e., $\multiple{\eval{\heap}{\env}{e}{V'(t)}}$.

      Let $\multiple{x} = \fpredparams(p)$. By \refrule{AssertPredicate} $\assertion{\heap}{\perms'}{[\multiple{x \mapsto V'(t)}]}{\fpred(p)}$. Also, $\vfoot{V}{\heap}{\pair{p}{\multiple{t}}} = \efoot{\heap}{[\multiple{x \mapsto V'(t)}]}{\fpred(p)}$. Therefore $\assertion{\heap}{\vfoot{V}{\heap}{\pair{p}{\multiple{t}} \cap \perms'}}{[\multiple{x \mapsto V'(t)}]}{\fpred(p)}$ by lemma \ref{lem:efoot-assert}.

      But $\fpred(p)$ must be a specification, thus $\vfoot{V}{\heap}{\pair{p}{\multiple{t}}} = \efoot{\heap}{[\multiple{x \mapsto V'(t)}]}{\fpred(p)} \subseteq \vfoot{V}{\heap}{\pair{p}{\multiple{t}}} \cap \perms'$ by lemma \ref{lem:efoot-subset-spec}. Therefore $\vfoot{V}{\heap}{\set{\pair{p}{\multiple{t}}}} \subseteq \perms'$.

      Then $\assertion{\heap}{\vfoot{V}{\heap}{\set{\pair{p}{\multiple{t}}}}}{[\multiple{x \mapsto V'(t)}]}{\fpred(p)}$, and thus $\assertion{\heap}{\vfoot{V}{\heap}{\set{\pair{p}{\multiple{t}}}}}{\env}{p(\multiple{e})}$ by \refrule{AssertPredicate}.

    \case \refrule{SConsumeAcc}, \refrule{SConsumeAccOptimistic}, \refrule{SConsumeAccImprecise}, \refrule{SConsumeAccFailure} -- $\sconsume{\sstate}{\sstate_E}{\kacc(e.f)}{\sstate'}{\_}{\set{\pair{t_e}{f}}}$:

      By the respective rule, $\seval{\sstate_E}{e}{t_e}{\_}{\_}$ for some $t_e$. The corresponding valuation for this case extends the corresponding valuation for this derivation.

      Since $\assertion{\heap}{\perms}{\env}{\kacc(e.f)}$, $\eval{\heap}{\env}{e}{v}$ for some $v$ by \refrule{AssertAcc}.  Let \\ $\hat{\perms_E} = \perms_E \cup \efoot{\heap}{\env}{\kacc(e.f)}$. Then by lemma \ref{lem:simstate-monotonicity} $\simstate{V}{\sstate_E}{\heap}{\hat{\perms_E}}{\env}$. Thus by lemma \ref{lem:eval-correspondence} $v = V'(t_e)$, i.e. $\eval{\heap}{\env}{e}{V'(t_e)}$.

      By \refrule{AssertAcc} $\pair{V'(t_e)}{f} \in \perms'$. Also, $\set{\pair{V'(t_e)}{f}}  = \vfoot{V'}{\heap}{\set{\pair{t_e}{f}}}$, therefore $\vfoot{V'}{\heap}{\set{\pair{t_e}{f}}} \subseteq \perms'$ and also $\assertion{\heap}{\vfoot{V'}{\heap}{\set{\pair{t_e}{f}}}}{\env}{\kacc(e.f)}$ by \refrule{AssertAcc}.

    \case \refrule{SConsumeConjunction}, \refrule{SConsumeConjunctionImprecise} -- $\sconsume{\sstate}{\sstate_E}{\phi_1 * \phi_2}{\sstate''}{\_}{\sperms_1 \cup \sperms_2}$:

      By the respective rule, $\sconsume{\sstate}{\sstate_E}{\phi_1}{\sstate'}{\_}{\sperms_1}$ and $\sconsume{\sstate'}{\sstate_E[\pc = \pc(\sstate')]}{\phi_2}{\sstate''}{\_}{\sperms_2}$. Let $V'$ be the corresponding valuation for this case, thus $V'$ extends the corresponding valuations for these judgements.

      By assumptions, $V'(\pc(\sstate'')) = \ktrue$ and $\pc(\sstate'') \implies \pc(\sstate')$ by \ref{lem:consume-subpath}. Therefore $\simstate{V'}{\sstate_E[\pc = \pc(\sstate')]}{\heap}{\perms_E}{\env}$.

      Then $\assertion{\heap}{\perms_1}{\env}{\phi_1}$ and $\assertion{\heap}{\perms_2}{\env}{\phi_2}$ where $\perms_1 \cup \perms_2 \subseteq \perms'$ and $\perms_1 \cap \perms_2 = \emptyset$ by \refrule{AssertConjunction}, since $\assertion{\heap}{\perms'}{\env}{\phi_1 * \phi_2}$. Then by induction $\assertion{\heap}{\vfoot{V'}{\heap}{\sperms_1}}{\env}{\phi_1}$, $\assertion{\heap}{\vfoot{V'}{\heap}{\sperms_2}}{\env}{\phi_2}$, $\vfoot{V'}{\heap}{\sperms_1} \subseteq \perms_1$, and $\vfoot{V'}{\heap}{\sperms_2} \subseteq \perms_1$.

      Now $\vfoot{V'}{\heap}{\sperms_1} \cap \vfoot{V'}{\heap}{\sperms_2} = \emptyset$ and $\vfoot{V'}{\heap}{\sperms_1} \cup \vfoot{V'}{\heap}{\sperms_2} = \vfoot{V'}{\heap}{\sperms_1 \cup \sperms_2} \subseteq \perms'$. Therefore $\assertion{\heap}{\vfoot{V'}{\heap}{\sperms_1 \cup \sperms_2}}{\env}{\phi_1 * \phi_2}$ by \refrule{AssertConjunction}.
  
    \case\label{case:consume-assert-vfoot-ifa} \refrule{SConsumeConditionalA} -- $\sconsume{\sstate}{\sstate_E}{\sif{e}{\phi_1}{\phi_2}}{\sstate'}{\scheck \cup \scheck'}{\sperms}$:

      By \refrule{SConsumeConditionalA} $\seval{\sstate_E}{e}{t}{\_}{\_}$ for some $t$. Let $V_1$ be the corresponding valuation. 
      
      Also, $\sconsume{\sstate[\pc = \pc']}{\sstate_E[\pc = \pc']}{\phi_1}{\sstate'}{\_}{\sperms}$ where $\pc' = \pc(\sstate) \kand t$. Let $V'$ be the corresponding valuation for this case, which extends the corresponding valuation for this judgement and $V_1$.

      Since $\assertion{\heap}{\perms'}{\env}{\sif{e}{\phi_1}{\phi_2}}$, by \refrule{AssertIfA} $\eval{\heap}{\env}{e}{v}$ for some $v$. Let $\hat{\perms_E} = \perms_E \cup \efoot{\heap}{\env}{e}$. Then by lemma \ref{lem:simstate-monotonicity} $\simstate{V}{\sstate_E}{\heap}{\hat{\perms_E}}{\env}$.
      Thus, by lemma \ref{lem:eval-correspondence} $v = V_1(t)$. Also, since $V'(\pc(\sstate')) = \ktrue$ and $\pc(\sstate') \implies \pc' = \pc(\sstate) \kand t$ by lemma \ref{lem:consume-subpath}, $V_1(t) = \ktrue$. Therefore $\eval{\heap}{\env}{e}{\ktrue}$. 

      In addition, $\sstate_E[\pc = \pc']$ and $\sstate_E$ differ only in their $\pc$ component, so $\simstate{V_1}{\sstate_E[\pc = \pc']}{\heap}{\perms_E}{\env}$. Also by \refrule{AssertIfA} $\assertion{\heap}{\perms'}{\env}{\phi_1}$. Therefore by induction $\assertion{\heap}{\vfoot{V}{\heap}{\sperms}}{\env}{\phi_1}$ and $\vfoot{V}{\heap}{\sperms} \subseteq \perms'$.

      Finally, $\assertion{\heap}{\vfoot{V}{\heap}{\sperms}}{\env}{\sif{e}{\phi_1}{\phi_2}}$ by \refrule{AssertIfA}.

    \case \refrule{SConsumeConditionalB} -- $\sconsume{\sstate}{\sstate_E}{\sif{e}{\phi_1}{\phi_2}}{\sstate'}{\scheck \cup \scheck'}{\sperms}$:
      Similar to case \ref{case:consume-assert-vfoot-ifa}.
  \end{enumcases}
\end{proof}

\begin{lemma}[Soundness of consume for precise formulas]\label{lem:consume-soundness-precise}
  Let $\phi$ be some precise formula, $V$ be some initial valuation, $\heap$ be some heap, $\env$ be some environment, $\perms_E$ and $\perms$ be sets of permissions such that $\perms \subseteq \perms_E$, and $\sstate$ and $\sstate_E$ be well-formed symbolic states such that $\simstate{V}{\sstate}{\heap}{\perms}{\env}$ and $\simstate{V}{\sstate_E}{\heap}{\perms_E}{\env}$.

  Then, if $\sconsume{\sstate}{\sstate_E}{\phi}{\sstate'}{\scheck}{\sperms}$ with corresponding valuation $V'$, $\rtassert{V'}{\heap}{\perms_E}{\scheck}$, and $V'(\pc(\sstate')) = \ktrue$, then
  \begin{equation}\label{eq:consume-prec-assert-e}
    \assertion{\heap}{\perms_E}{\env}{\phi}, \quad
    \simstate{V'}{\sstate'}{\heap}{\perms \setminus \vfoot{V}{\heap}{\sperms}}{\env}, \quad\text{and}~
    \ifrm{\heap}{\perms_E}{\env}{\phi}.
  \end{equation}

  Furthermore, if the above conditions hold and $\scheck \cap \SPerm = \emptyset$, then
  \begin{equation}\label{eq:consume-prec-assert}
    \assertion{\heap}{\perms}{\env}{\phi}.
  \end{equation}
\end{lemma}
\begin{proof}
  Suppose that $\sconsume{\sstate}{\sstate_E}{\phi}{\sstate'}{\scheck}{\sperms}$ with corresponding valuation $V'$, $\rtassert{V'}{\heap}{\perms_E}{\scheck}$, and $V'(\pc(\sstate')) = \ktrue$. Complete the proof by induction on $\sconsume{\sstate}{\sstate_E}{\phi}{\scheck}{\sperms}$:

  \begin{enumcases}
    \case \refrule{SConsumeImprecision} -- $\sconsume{\sstate}{\sstate_E}{\simprecise{\phi}}{\quintuple{\top}{\pc(\sstate')}{\senv(\sstate')}{\emptyset}{\emptyset}}{\scheck}{\sperms}$:
      Since $\simprecise{\phi}$ is imprecise, this rule cannot apply.

    \case \refrule{SConsumeValue} -- $\sconsume{\sstate}{\sstate_E}{e}{\sstate}{\scheck}{\emptyset}$:

      By \refrule{SConsumeValue} $\seval{\sstate_E}{e}{t}{\sstate_E'}{\scheck}$ for some $\sstate_E'$. Let $V'$ be the corresponding valuation, with initial valuation $V$. Then $V'$ is the corresponding valuation for this case. 

      By lemma \ref{lem:eval-progress}, $V(\pc(\sstate_E')) = \ktrue$. Then, since $\rtassert{V'}{\heap}{\perms_E}{\scheck}$, by lemma \ref{lem:seval-soundness} $\eval{\heap}{\env}{e}{V'(t)}$. Also by \refrule{SConsumeValue} $\pc(\sstate) \implies t$. Therefore $V'(t) = \ktrue$, and therefore $\eval{\heap}{\env}{e}{\ktrue}$. Thus $\assertion{\heap}{\perms_E}{\env}{e}$ by \refrule{AssertValue}.

      $\vfoot{V}{\heap}{\emptyset} = \emptyset$, therefore $\simstate{V'}{\sstate}{\heap}{\perms \setminus \vfoot{V}{\heap}{\emptyset}}{\env}$ since $V'$ extends $V$.

      Finally, $\frm{\heap}{\perms_E}{\env}{e}$ by lemma \ref{lem:seval-soundness}. Therefore $\ifrm{\heap}{\perms_E}{\env}{e}$ by \refrule{IFrameExpression}, which completes the proof of \eqref{eq:consume-prec-assert-e}.

      As shown before, $\eval{\heap}{\env}{e}{\ktrue}$, thus $\assertion{\heap}{\perms}{\env}{e}$ by \refrule{AssertValue}, which proves \eqref{eq:consume-prec-assert}.

    \case \refrule{SConsumeValueImprecise} -- $\sconsume{\sstate}{\sstate_E}{e}{\sstate[\pc = \pc(\sstate) \kand t]}{\scheck; t}{\emptyset}$:

      By \refrule{SConsumeValueImprecise} $\seval{\sstate_E}{e}{\sstate_E'}{t}{\scheck}$ for some $\sstate_E'$. Let $V'$ be the valuation corresponding to this derivation, with initial valuation $V$. Then $V'$ is the valuation corresponding to this case.

      Let $\sstate' = \sstate[\pc = \pc(\sstate) \kand t]$.

      By assumptions, $\rtassert{V'}{\heap}{\perms_E}{\set{t}}$ by lemma \ref{lem:scheck-monotonicity}. Thus $V'(t) = \ktrue$ by \refrule{CheckValue}. Also $V'(\pc(\sstate_E')) = \ktrue$ by lemma \ref{lem:eval-progress}. Then, $\eval{\heap}{\env}{e}{V'(t)}$ by lemma \ref{lem:seval-soundness}; therefore $\eval{\heap}{\env}{e}{\ktrue}$. Thus, by \refrule{AssertValue}, $\assertion{\heap}{\perms_E}{\env}{e}$.

      $\vfoot{V}{\heap}{\emptyset} = \emptyset$, therefore $\simstate{V'}{\sstate}{\heap}{\perms \setminus \vfoot{V}{\heap}{\emptyset}}{\env}$ since $V'$ extends $V$.

      Furthermore, since $V'(t) = \ktrue$, $V'(\pc(\sstate')) = V(\pc(\sstate)) \kand V'(t) = \ktrue$. Since $\sstate'$ and $\sstate$ differ only in their $\pc$ components, $\simstate{V'}{\sstate'}{\heap}{\perms \setminus \vfoot{V}{\heap}{\emptyset}}{\env}$.

      Finally, $\frm{\heap}{\perms_E}{\env}{e}$ by lemma \ref{lem:seval-soundness}. Therefore $\ifrm{\heap}{\perms_E}{\env}{e}$ by \refrule{IFrameExpression}.

    \case \refrule{SConsumeValueFailure} -- $\sconsume{\sstate}{\sstate_E}{e}{\sstate}{\set{\bot}}{\emptyset}$:

      $\rtassert{V'}{\heap}{\perms_E}{\scheck \cup \set{\bot}}$ is a contradiction, thus the lemma vacuously holds.

    \case\label{case:consume-prec-pred} \refrule{SConsumePredicate} -- $\sconsume{\sstate}{\sstate_E}{p(\multiple{e})}{\sstate'}{\bigcup \multiple{\scheck}}{\set{\pair{p}{\multiple{t}}}}$:

      By \refrule{SConsumePredicate}, for each $e$, $\seval{\sstate_E}{e}{t}{\sstate_E'}{\scheck}$ for some $\sstate_E'$, $t$, and $\scheck$. Let $V'$ be corresponding valuation for this case, thus $V'$ extends the respective individual corresponding valuations and for each $\scheck$, $\rtassert{V'}{\heap}{\perms}{\scheck}$ by lemma \ref{lem:scheck-monotonicity}. In addition, by lemma \ref{lem:eval-progress} $V'(\pc(\sstate_E')) = \ktrue$.

      Therefore, for each $e$ and corresponding $t$, $\eval{\heap}{\env}{e}{V'(t)}$ by lemma \ref{lem:seval-soundness}. By \refrule{SConsumePredicate}, for each $t$, $\pc(\sstate) \implies t \keq t'$ for some $t'$, thus $V'(t) = V'(t')$ since $V'(\pc(\sstate)) = \ktrue$. Therefore $\eval{\heap}{\env}{e_i}{V'(t_i')}$.

      By \refrule{SConsumePredicate} $\pair{p}{\multiple{t'}} \in \sheap(\sstate)$. Since $\simheap{V}{\sstate}{\heap}{\perms}$, $\assertion{\heap}{\perms}{[\multiple{x \mapsto V'(t')}]}{\fpred(p)}$ where $\multiple{x} = \fpredparams(p)$. Thus by \refrule{AssertPredicate} $\assertion{\heap}{\perms}{\env}{p(\multiple{e})}$, which proves \eqref{eq:consume-prec-assert}. Also, $\assertion{\heap}{\perms_E}{\env}{p(\multiple{e})}$ by lemma \ref{lem:assert-monotonicity}.

      Let $\perms' = \vfoot{V}{\heap}{\pair{p}{\multiple{t}}}$, therefore $\perms' = \efoot{\heap}{[\multiple{x \mapsto V'(t)}]}{\fpred(p)} = \efoot{\heap}{[\multiple{x \mapsto V'(t')}]}{\fpred(p)}$.

      By \refrule{SConsumePredicate} $\sstate' = \sstate[\sheap = \sheap', \oheap = \emptyset]$ where $\sheap = \sheap'; \pair{p}{\multiple{t}}$. Thus $\simheap{V'}{\sheap(\sstate')}{\heap}{\perms}$ since $\sheap(\sstate') \subset \sheap(\sstate)$.

      Then $\vfoot{V'}{\heap}{h} \cap \perms' = \emptyset$ for all $h \in \sheap(\sstate')$ since $\simheap{V'}{\sheap(\sstate)}{\heap}{\perms}$, $\pair{p}{\multiple{t}} \in \sheap(\sstate)$, and $\pair{p}{\multiple{t}} \notin \sheap(\sstate')$. Therefore, by lemma \ref{lem:disjoint-sim-heap-subset}, $\simheap{V'}{\sheap(\sstate')}{\heap}{\perms \setminus \perms'}$.

      Also, since $\oheap(\sstate') = \emptyset$, $\simheap{V'}{\oheap(\sstate')}{\heap}{\perms \setminus \perms'}$.

      Therefore $\simstate{V'}{\sstate'}{\heap}{\perms \setminus \perms'}{\env}$ since $\sstate'$ and $\sstate$ differ only in their $\sheap$ and $\oheap$ components.

      By lemmas \ref{lem:scheck-monotonicity} $\rtassert{V'}{\heap}{\perms_E}{\scheck}$ for each $\scheck$, thus $\frm{\heap}{\perms_E}{\env}{e}$ for each $e$ by lemma \ref{lem:seval-soundness}, thus $\ifrm{\heap}{\perms_E}{\env}{p(\multiple{e})}$ by \refrule{IFramePredicate}. Thus \eqref{eq:consume-prec-assert-e} holds.

    \case \refrule{SConsumePredicateImprecise} -- $\sconsume{\sstate}{\sstate_E}{p(\multiple{e})}{\sstate'}{\bigcup \multiple{\scheck}; \pair{p}{\multiple{t}}}{\set{\pair{p}{\multiple{t}}}}$:

      By \refrule{SConsumePredicateImprecise}, for each $e$, $\seval{\sstate_E}{e}{t}{\sstate_E'}{\scheck}$ for some $\sstate_E'$, $t$, and $\scheck$. Let $V'$ be corresponding valuation for this case, thus $V'$ extends the respective individual corresponding valuations and for each $\scheck$, $\rtassert{V'}{\heap}{\perms}{\scheck}$ by lemma \ref{lem:scheck-monotonicity}. In addition, $V'(\pc(\sstate_E')) = \ktrue$ by lemma \ref{lem:eval-progress}.

      Therefore, for each $e$ and corresponding $t$, $\eval{\heap}{\env}{e}{V'(t)}$ by lemma \ref{lem:seval-soundness}. By \\
      \refrule{SConsumePredicateImprecise}, for each $t$, $\pc(\sstate) \implies t \keq t'$ for some $t'$, thus $V'(t) = V'(t')$ since $V'(\pc(\sstate)) = \ktrue$. Therefore $\eval{\heap}{\env}{e_i}{V'(t_i')}$.

      Also, $\rtassert{V'}{\heap}{\perms_E}{\set{\pair{p}{\multiple{t}}}}$ by assumptions and lemma \ref{lem:scheck-monotonicity}. Thus $\assertion{\heap}{\perms_E}{[\multiple{x \mapsto V'(t)}]}{\fpred(p)}$ by \refrule{CheckPred}. Therefore $\assertion{\heap}{\perms_E}{\env}{p(\multiple{e})}$ by \refrule{AssertPredicate}.

      By \refrule{SConsumePredicateImprecise} $\sstate' = \sstate[\sheap = \emptyset, \oheap = \emptyset]$, thus $\simheap{V'}{\sheap(\sstate')}{\heap}{\perms \setminus \vfoot{V'}{\heap}{\set{\pair{p}{\multiple{t}}}}}$ and $\simheap{V'}{\oheap(\sstate')}{\heap}{\perms \setminus \vfoot{V'}{\heap}{\set{\pair{p}{\multiple{t}}}}}$. Therefore $\simstate{V'}{\sstate'}{\heap}{\perms \setminus \vfoot{V}{\heap}{\set{\pair{p}{\multiple{t}}}}}{\env}$ since $\sstate'$ and $\sstate$ differ only in their $\sheap$ and $\oheap$ components.

      For each $e$, $\frm{\heap}{\perms_E}{\env}{e}$ by lemma \ref{lem:seval-soundness}. Therefore $\ifrm{\heap}{\perms_E}{\env}{p(\multiple{e})}$ by \refrule{IFramePredicate}, which completes the proof of \eqref{eq:consume-prec-assert-e}.

      Also, $\pair{p}{\multiple{t}} \in \SPerm$ is in the resulting set of checks, which contradicts the premises of \eqref{eq:consume-prec-assert}. Therefore it is vacuously true.

    \case \refrule{SConsumePredicateFailure} -- $\sconsume{\sstate}{\sstate_E}{p(\multiple{e})}{\sstate}{\set{\bot}}{\set{\pair{p}{\multiple{t}}}}$:

      $\rtassert{V'}{\heap}{\perms_E}{\set{\bot}}$ is a contradiction, thus the lemma vacuously holds.

    \case\label{case:consume-prec-acc} \refrule{SConsumeAcc} -- $\sconsume{\sstate}{\sstate_E}{\kacc(e.f)}{\sstate[\sheap = \sheap', \oheap = \oheap']}{\scheck}{\set{\pair{t_e}{f}}}$:

      By \refrule{SConsumeAcc} $\seval{\sstate_E}{e}{t_e}{\sstate_E'}{\scheck}$. Let $V'$ be the corresponding valuation, thus $V'$ is the corresponding valuation for this case.
      
      By lemma \ref{lem:eval-progress}, $V'(\pc(\sstate_E')) = \ktrue$. Thus $\eval{\heap}{\env}{e}{V'(t_e)}$ by lemma \ref{lem:seval-soundness}. Also $\pc(\sstate) \implies t_e' \keq t_e$ by \refrule{SConsumeAcc}, thus $V'(t_e') = V'(t_e)$ since $V'(\pc(\sstate)) = \ktrue$. Therefore, $\eval{\heap}{\env}{e}{V'(t_e')}$.

      Since $\triple{f}{t_e'}{t} \in \sheap(\sstate)$ by \refrule{SConsumeAcc} and $\simheap{V}{\sheap(\sstate)}{\heap}{\perms}$, $\pair{V'(t_e')}{f} \in \perms$. Thus $\assertion{\heap}{\perms}{\env}{\kacc(e.f)}$ by \refrule{AssertAcc}, which proves \eqref{eq:consume-prec-assert}. Therefore $\assertion{\heap}{\perms_E}{\env}{\kacc(e.f)}$ by lemma \ref{lem:assert-monotonicity} since $\perms \subseteq \perms_E$.

      Let $\sheap' = \fremfp(\sheap(\sstate), \sstate, t_e, f)$. Therefore $\simheap{V'}{\sheap'}{\heap}{\perms \setminus \set{\pair{V'(t_e)}{f}}}$ by lemma \ref{lem:sheap-remfp-prec}. Also, $\set{\pair{V'(t_e')}{f}} = \set{\pair{V'(t_e)}{f}} = \vfoot{V'}{\heap}{\set{\pair{t_e}{f}}}$. Thus $\simheap{V'}{\sheap'}{\heap}{\perms \setminus \vfoot{V'}{\heap}{\set{\pair{t_e}{f}}}}$.

      Likewise, let $\oheap' = \fremf(\oheap(\sstate), \sstate, t_e, f)$. Then similarly $\simheap{V'}{\oheap'}{\heap}{\perms \setminus \vfoot{V}{\heap}{\set{\pair{t_e}{f}}}}$ by lemma \ref{lem:oheap-remf}.

      By \refrule{SConsumeAcc}, $\sstate' = \sstate[\sheap = \sheap', \oheap = \oheap']$. Now, since $\sstate$ and $\sstate'$ differ only in their $\sheap$ and $\oheap$ components, using the properties of $\sheap'$ and $\oheap'$ shown above, $\simstate{V'}{\sstate'}{\heap}{\perms \setminus \vfoot{V}{\heap}{\set{\pair{t_e}{f}}}}{\env}$.

      Finally, $\frm{\heap}{\perms_E}{\env}{e}$ by lemma \ref{lem:seval-soundness}. Therefore $\ifrm{\heap}{\perms_E}{\env}{\kacc(e.f)}$ by \refrule{IFrameAcc}, which completes the proof of \eqref{eq:consume-prec-assert-e}.

    \case\label{case:consume-prec-acc-optimistic} \refrule{SConsumeAccOptimistic} --  $\sconsume{\sstate}{\sstate_E}{\kacc(e.f)}{\sstate'}{\scheck}{\set{\pair{t_e}{f}}}$:
      Similar to case \ref{case:consume-prec-acc}, except to show that $\simheap{V}{\sheap'}{\heap}{\perms \setminus \set{\pair{V'(t_e')}{f}}}$.

      Since $\triple{f}{t_e'}{t} \in \oheap(\sstate)$, $\oheap(\sstate) \ne \emptyset$. Therefore, since $\sstate$ is well-formed, $\imp(\sstate) = \top$. Let $\sheap' = \fremf(\sheap(\sstate), \sstate, t_e', f)$. Therefore $\simheap{V}{\sheap'}{\heap}{\perms \setminus \set{\pair{V'(t_e')}{f}}}$ by lemma \ref{lem:sheap-remf-imp}.

      Continue as in case \ref{case:consume-prec-acc}.

    \case \refrule{SConsumeAccImprecise} -- $\sconsume{\sstate}{\sstate_E}{\kacc(e.f)}{\sstate'}{\scheck; \pair{t_e}{f}}{\set{\pair{t_e}{f}}}$:

      By \refrule{SConsumeAcc} $\seval{\sstate_E}{e}{t_e}{\sstate_E'}{\scheck}$. Let $V'$ be the corresponding valuation, thus $V'$ is the corresponding valuation for this case.

      Then $\rtassert{V'}{\heap}{\perms_E}{\scheck}$ by assumptions and lemma \ref{lem:scheck-monotonicity}. In addition, $V'(\pc(\sstate_E')) = \ktrue$ by lemma \ref{lem:eval-progress}. Thus $\eval{\heap}{\env}{e}{V'(t_e)}$ by lemma \ref{lem:seval-soundness}.

      Also, $\rtassert{V'}{\heap}{\perms_E}{\pair{t_e}{f}}$ by assumptions and lemma \ref{lem:scheck-monotonicity}. Then $\pair{V'(t_e)}{f} \in \perms_E$ by \refrule{CheckAcc}. Therefore $\assertion{\heap}{\perms_E}{\env}{\kacc(e.f)}$ by \refrule{AssertAcc}.

      Let $\sheap' = \fremf(\sheap(\sstate), \sstate, t_e, f)$. By \refrule{SConsumeAccImprecise} $\imp(\sstate)$. Thus $\simheap{V'}{\sheap'}{\heap}{\perms \setminus \set{\pair{V'(t_e)}{f}}}$ by lemma \ref{lem:sheap-remf-imp}. Also, $\set{\pair{V'(t_e)}{f}} = \vfoot{V'}{\heap}{\set{\pair{t_e}{f}}}$. Thus $\simheap{V'}{\sheap'}{\heap}{\perms \setminus \vfoot{V'}{\heap}{\set{\pair{t_e}{f}}}}$.

      Likewise, let $\oheap' = \fremf(\oheap(\sstate), \sstate, t_e, f)$. Then $\simheap{V'}{\oheap'}{\heap}{\perms \setminus \vfoot{V}{\heap}{\pair{t_e}{f}}}$ by lemma \ref{lem:oheap-remf}.

      By \refrule{SConsumeAcc}, $\sstate' = \sstate[\sheap = \sheap', \oheap = \oheap']$. Now, since $\sstate$ and $\sstate'$ differ only in their $\sheap$ and $\oheap$ components, using the properties of $\sheap'$ and $\oheap'$ shown above, $\simstate{V'}{\sstate'}{\heap}{\perms \setminus \vfoot{V}{\heap}{\set{\pair{t_e}{f}}}}{\env}$.

      By lemma \ref{lem:seval-soundness} $\frm{\heap}{\perms_E}{\env}{e}$, therefore $\ifrm{\heap}{\perms_E}{\kacc(e.f)}$ by \refrule{IFrameAcc}. Thus \eqref{eq:consume-prec-assert-e} holds.

      Also, $\pair{t_e}{f} \in \SPerm$ is in the resulting set of checks, which contradicts the premises of \eqref{eq:consume-prec-assert}, therefore it vacuously holds.

    \case \refrule{SConsumeAccFailure} -- $\sconsume{\sstate}{\sstate_E}{\kacc(e.f)}{\sstate}{\set{\bot}}{\set{\pair{t_e}{f}}}$:

    $\rtassert{V'}{\heap}{\perms_E}{\set{\bot}}$ is a contradiction, thus the lemma vacuously holds.

    \case\label{case:consume-prec-conj} \refrule{SConsumeConjunction} -- $\sconsume{\sstate}{\sstate_E}{\phi_1 * \phi_2}{\sstate''}{\scheck_1 \cup \scheck_2}{\sperms_1 \cup \sperms_2}$:

      By \refrule{SConsumeConjunction} $\sconsume{\sstate}{\sstate_E}{\phi_1}{\sstate'}{\scheck_1}{\sperms_1}$ and $\sconsume{\sstate'}{\sstate_E[\pc = \pc(\sstate')]}{\phi_2}{\sstate''}{\scheck_2}{\sperms_2}$. Let $V_1$ and $V'$ be the respective corresponding valuations, with initial valuations $V$ and $V_1$, respectively. Then $V'$ is the corresponding valuation for this case.

      By lemma \ref{lem:consume-subpath}, $\pc(\sstate'') \implies \pc(\sstate')$. Thus $V'(\pc(\sstate')) = V_1(\pc(\sstate')) = \ktrue$. Also, $\rtassert{V_1}{\heap}{\perms_E}{\scheck_1}$ by lemma \ref{lem:scheck-monotonicity}, since $V'$ extends $V_1$.

      By \refrule{SConsumeConjunction} $(\scheck_1 \cup \scheck_2) \cap \SPerm = \emptyset$, thus $\scheck_1 \cap \SPerm = \emptyset$.

      Let $\perms_1 = \vfoot{V}{\heap}{\sperms_1}$. By induction, using \eqref{eq:consume-prec-assert}, $\assertion{\heap}{\perms}{\env}{\phi_1}$. Thus $\assertion{\heap}{\perms_1}{\env}{\phi_1}$ and $\perms_1 \subseteq \perms$ by lemma \ref{lem:consume-assert-vfoot}.
      
      Also $\simstate{V_1}{\sstate'}{\heap}{\perms \setminus \perms_1}{\env}$ by induction, $\simstate{V_1}{\sstate_E[\pc = \pc(\sstate')]}{\heap}{\perms_E}{\env}$ since $V_1$ extends $V$ and $V_1(\pc(\sstate')) = \ktrue$, and $\rtassert{V}{\heap}{\perms_E}{\scheck_2}$ by lemma \ref{lem:scheck-monotonicity}. Finally, by assumptions $V'(\pc(\sstate'')) = \ktrue$, and $(\perms \setminus \perms_1) \subseteq \perms \subseteq \perms_E$. Thus by induction $\simstate{V'}{\sstate''}{\heap}{(\perms \setminus \perms_1) \setminus \vfoot{V}{\heap}{\sperms_2}}{\env}$.

      Also by induction, using \eqref{eq:consume-prec-assert}, $\assertion{\heap}{\perms \setminus \perms_1}{\env}{\phi_2}$.

      Now $(\perms \setminus \perms_1) \subseteq \perms$, $\perms_1 \subseteq \perms$, and $(\perms \setminus \perms_1) \cap \perms_1 = \emptyset$. Therefore $\assertion{\heap}{\perms}{\env}{\phi_1 * \phi_2}$ by \refrule{AssertConjunction}, which proves \eqref{eq:consume-prec-assert}. Then by lemma \ref{lem:assert-monotonicity} $\assertion{\heap}{\perms_E}{\env}{\phi_1 * \phi_2}$.

      As shown before, $\simstate{V'}{\sstate''}{\heap}{(\perms \setminus \perms_1) \setminus \vfoot{V}{\heap}{\sperms_2}}{\env}$, and $(\perms \setminus \perms_1) \setminus \vfoot{V}{\heap}{\sperms_2} = \perms \setminus (\vfoot{V}{\heap}{\sperms_1} \cup \vfoot{V}{\heap}{\sperms_2}) = \perms \setminus \vfoot{V}{\heap}{\sperms_1 \cup \sperms_2}$, therefore $\simstate{V'}{\sstate''}{\heap}{\perms \setminus \vfoot{V}{\heap}{\sperms_1 \cup \sperms_2}}{\env}$.

      By induction $\ifrm{\heap}{\perms_E}{\env}{\phi_1}$ and $\ifrm{\heap}{\perms_E}{\env}{\phi_2}$. Therefore $\ifrm{\heap}{\perms_E}{\env}{\phi_1 * \phi_2}$ by \refrule{IFrameConjunction}, which completes the proof of \eqref{eq:consume-prec-assert-e}.

    \case \refrule{SConsumeConjunctionImprecise} -- $\sconsume{\sstate}{\sstate_E}{\phi_1 * \phi_2}{\sstate''}{\scheck_1 \cup \scheck_2; \fsep(\sperms_1, \sperms_2)}{\sperms_1 \cup \sperms_2}$:
    
      Similar to case \ref{case:consume-prec-conj}, except when showing that $\assertion{\heap}{\perms_E}{\env}{\phi_1 * \phi_2}$ and when proving \eqref{eq:consume-prec-assert}:

      By induction $\assertion{\heap}{\perms_E}{\env}{\phi_1}$ and $\assertion{\heap}{\perms_E}{\env}{\phi_2}$. Thus by lemma \ref{lem:consume-assert-vfoot} $\assertion{\heap}{\vfoot{V'}{\heap}{\sperms_1}}{\env}{\phi_1}$, $\assertion{\heap}{\vfoot{V'}{\heap}{\sperms_2}}{\env}{\phi_2}$, and $\vfoot{V'}{\heap}{\sperms_1} \cup \vfoot{V'}{\heap}{\sperms_2} \subseteq \perms_E$.
      
      By assumptions $\rtassert{V'}{\heap}{\perms_E}{\fsep(\sperms_1, \sperms_2)}$. Then by \refrule{CheckSep} $\vfoot{V'}{\heap}{\sperms_1} \cap \vfoot{V'}{\heap}{\sperms_2} = \emptyset$. Therefore $\assertion{\heap}{\perms_E}{\env}{\phi_1 * \phi_2}$ by \refrule{AssertConjunction}.

      By \refrule{SConsumeConjunctionImprecise} $(\scheck_1 \cup \scheck_2) \cap \SPerm \ne \emptyset$. Therefore the premises of \eqref{eq:consume-prec-assert} do not hold, therefore it is vacuously true.

    \case\label{case:consume-prec-cond-a} \refrule{SConsumeConditionalA} -- $\sconsume{\sstate}{\sstate_E}{\sif{e}{\phi_1}{\phi_2}}{\sstate'}{\scheck \cup \scheck'}{\sperms}$:

      By \refrule{SConsumeConditionalA}, $\seval{\sstate_E}{e}{t}{\sstate_E'}{\scheck}$ and $\sconsume{\sstate[\pc = \pc']}{\sstate_E[\pc = \pc']}{\phi_1}{\sstate'}{\scheck'}{\sperms}$ where $\pc' = \pc(\sstate) \kand t$. Let $V_1$ and $V'$ be the respective corresponding valuations, with initial valuations $V$ and $V_1$, respectively. Then $V'$ is the corresponding valuation for this case.

      By lemma \ref{lem:consume-subpath} $\pc(\sstate') \implies \pc' = \pc(\sstate) \kand t$, thus $V'(\pc(\sstate) \kand t) = V_1(\pc(\sstate) \kand t) = \ktrue$. Therefore $\simstate{V_1}{\sstate[\pc = \pc(\sstate) \kand t]}{\heap}{\perms}{\env}$ and $\simstate{V_1}{\sstate_E[\pc = \pc(\sstate) \kand t]}{\heap}{\perms}{\env}$.

      By assumptions and lemma \ref{lem:scheck-monotonicity} $\rtassert{V_1}{\heap}{\perms_E}{\scheck}$. In addition, by lemma \ref{lem:eval-progress} $V_1(\pc(\sstate_E')) = \ktrue$. Thus $\eval{\heap}{\env}{e}{V_1(t)}$ by lemma \ref{lem:seval-soundness}. Furthermore, $V_1(t) = \ktrue$ since $\pc(\sstate) \kand t \implies t$, thus $\eval{\heap}{\env}{e}{\ktrue}$. Finally, $\assertion{\heap}{\perms_E}{\env}{\phi_1}$ by induction. Therefore $\assertion{\heap}{\perms_E}{\env}{\sif{e}{\phi_1}{\phi_2}}$ by \refrule{AssertIfA}.

      Also by induction $\simstate{V'}{\sstate'}{\heap}{\perms \setminus \vfoot{V'}{\heap}{\sperms}}{\env}$.

      Finally, $\frm{\heap}{\perms_E}{\env}{\phi_1}$ by induction, and $\frm{\heap}{\perms_E}{\env}{e}$ by lemma \ref{lem:seval-soundness}. As shown before, $\eval{\heap}{\env}{e}{\ktrue}$. Therefore $\frm{\heap}{\perms_E}{\env}{\sif{e}{\phi_1}{\phi_2}}$ by \refrule{FrameIfA}, which completes the proof of \eqref{eq:consume-prec-assert-e}.

      Now suppose that $(\scheck \cup \scheck') \cap \SPerm = \emptyset$, thus $\scheck' \cap \SPerm = \emptyset$. Then by induction $\assertion{\heap}{\perms}{\env}{\phi}$, and as before $\eval{\heap}{\env}{e}{\ktrue}$, therefore $\assertion{\heap}{\perms}{\env}{\sif{e}{\phi_1}{\phi_2}}$ by \refrule{AssertIfA}, which completes the proof of \eqref{eq:consume-prec-assert}.

    \case \refrule{SConsumeConditionalB} -- $\sconsume{\sstate}{\sstate_E}{\sif{e}{\phi_1}{\phi_2}}{\sstate'}{\scheck \cup \scheck'}{\sperms}$: Similar to case \ref{case:consume-prec-cond-a}.

  \end{enumcases}
\end{proof}

\begin{lemma}[Soundness of consume (long form)]\label{lem:consume-soundness}
  Let $\gform$ be some specification, $\sstate$ and $\sstate_E$ some well-formed symbolic states such that $\pc(\sstate) \implies \pc(\sstate_E)$, and $\triple{\heap}{\perms}{\env}$ some evaluation state such that $\simstate{V}{\sstate}{\heap}{\perms}{\env}$ and $\simstate{V}{\sstate_E}{\heap}{\perms}{\env}$.
  
  If $\sconsume{\sstate}{\sstate_E}{\gform}{\sstate'}{\scheck}{\sperms}$ with corresponding valuation $V'$, $\rtassert{V'}{\heap}{\perms_E}{\scheck}$, and $V'(\pc(\sstate')) = \ktrue$, then
  $$\assertion{\heap}{\perms_E}{\env}{\gform} \quad\text{and}\quad
    \simstate{V'}{\sstate'}{\heap}{\perms \setminus \efoot{\heap}{\env}{\gform}}{\env}.
  $$
\end{lemma}

\begin{lemma}[Soundness of consume (short form)]\label{lem:cons-soundness}
  Let $\gform$ be some specification, $\sstate$ be some well-formed symbolic state, $\triple{\heap}{\perms}{\env}$ some evaluation state, and $V$ be some valuation such that \\
  $\simstate{V}{\sstate}{\heap}{\perms}{\env}$.

  If $\scons{\sstate}{\gform}{\sstate'}{\scheck}$ with corresponding valuation $V'$, $\rtassert{V'}{\heap}{\perms}{\scheck}$, and $V'(\pc(\sstate')) = \ktrue$ then
  $$
    \assertion{\heap}{\perms}{\env}{\gform} \quad\text{and}\quad
    \simstate{V'}{\sstate'}{\heap}{\perms \setminus \efoot{\heap}{\env}{\gform}}{\env}.
  $$
\end{lemma}

\begin{lemma}[Progress of consume (long form)]\label{lem:consume-progress}
  For any heap $\heap$, $\sstate$, $\sstate_E$, $\gform$, and valuation $V$, if $V(\pc(\sstate_E)) = \ktrue$ then $\sconsume{\sstate}{\sstate_E}{\gform}{\sstate'}{\_}{\_}$ for some $\sstate'$ such that $V'(\pc(\sstate')) = \ktrue$ where $V'$ is the corresponding valuation.
\end{lemma}
\begin{proof}
  By induction on the syntax forms of $\phi$:

  \begin{enumcases}
    \case $e \in \Expr$:
      By lemma \ref{lem:eval-progress}, $\seval{\sstate_E}{e}{t}{\_}{\_}$ for some $t$. Then one of the following cases applies to yield $\sconsume{\sstate}{\sstate_E}{e}{\sstate}{\_}{\_}$:

      \subcase $\pc(\sstate) \implies t$: Then \refrule{SConsumeValue} applies.

      \subcase $\imp(\sstate)$ and $\pc(\sstate) \notimplies t$: Then \refrule{SConsumeValueImprecise} applies.

      \subcase $\neg \imp(\sstate)$ and $\pc(\sstate) \notimplies t$: Then \refrule{SConsumeValueFailure} applies.

    \case $p(\multiple{e})$ -- $p \in \Predicate, \multiple{e \in \Expr}$:

      By lemma \ref{lem:eval-progress}, for each $e$, $\seval{\sstate_E}{e}{t}{\_}{\_}$ for some $t$. Then one of the following cases applies to yield $\sconsume{\sstate}{\sstate_E}{p(\multiple{e})}{\sstate'}{\_}{\_}$ where $\pc(\sstate') = \pc(\sstate)$:

      \subcase $p(\multiple{t'}) \in \sheap(\sstate)$ and $\pc(\sstate) \implies \multiple{t \keq t'}$ for some $\multiple{t'}$: Then \refrule{SConsumePredicate} applies.

      \subcase $\imp(\sstate)$ and $\nexistential{\pair{p}{\multiple{t'}} \in \sheap(\sstate)}{\bigwedge \multiple{\pc(\sstate) \implies t \keq t'}}$: Then \refrule{SConsumePredicateImprecise} applies.

      \subcase $\neg\imp(\sstate)$ and $\nexistential{\pair{p}{\multiple{t'}} \in \sheap(\sstate)}{\bigwedge \multiple{\pc(\sstate) \implies t \keq t'}}$: Then \refrule{SConsumePredicateFailure} applies.

    \case $\kacc(e.f)$ -- $e \in \Expr, f \in \Field$:

      By lemma \ref{lem:eval-progress}, $\seval{\sstate_E}{e}{t_e}{\_}{\_}$ for some $t_e$. Note that $\fremf$ and $\fremfp$ are defined for all inputs. Then one of the following cases applies to yield $\sconsume{\sstate}{\sstate_E}{\kacc(e.f)}{\sstate'}{\_}{\_}$ where $\pc(\sstate') = \pc(\sstate')$:

      \subcase $\triple{f}{t_e'}{t} \in \sheap(\sstate)$ and $\pc(\sstate) \implies t_e' \keq t_e$ for some $t_e'$ and $t$: Then \refrule{SConsumeAcc} applies.

      \subcase $\nexistential{t_e', t}{\triple{f}{t_e}{t} \in \sheap(\sstate) \wedge (\pc(\sstate) \implies t_e' \keq t_e)}$ and $\triple{f}{t_e'}{t} \in \sheap(\sstate)$ for some $t_e'$ and $t$ where $\pc(\sstate) \implies t_e' \keq t_e$: Then \\
      \refrule{SConsumeAccOptimistic} applies.

      \subcase $\nexistential{t_e', t}{\triple{f}{t_e}{t} \in \sheap(\sstate) \cup \oheap(\sstate) \wedge (\pc(\sstate) \implies t_e' \keq t_e)}$ and $\imp(\sstate)$: Then \\
      \refrule{SConsumeAccImprecise} applies.

      \subcase $\nexistential{t_e', t}{\triple{f}{t_e}{t} \in \sheap(\sstate) \cup \oheap(\sstate) \wedge (\pc(\sstate) \implies t_e' \keq t_e)}$ and $\neg \imp(\sstate)$: Then \\
      \refrule{SConsumeAccFailure} applies.

    \case $\phi_1 * \phi_2$ -- $\phi_1, \phi_2 \in \Formula$

      By induction, $\sconsume{\sstate}{\sstate_E}{\phi_1}{\sstate'}{\_}{\_}$ for some $\sstate'$ such that $V'(\pc(\sstate')) = \ktrue$ where $V'$ is the corresponding valuation.

      Then also by induction, $\sconsume{\sstate'}{\sstate_E[\pc = \pc(\sstate')]}{\phi_2}{\sstate''}{\_}{\_}$ for some $\sstate''$ such that $V''(\pc(\sstate'')) = \ktrue$ where $V''$ is the corresponding valuation, with initial valuation $V'$. Then one of the following cases applies to yield $\sconsume{\sstate}{\sstate_E}{\phi_1 * \phi_2}{\sstate''}{\_}{\_}$:

      \subcase $(\scheck_1 \cup \scheck_2) \cap \SPerm \ne \emptyset$: Then \refrule{SConsumeConjunctionImprecise} applies.

      \subcase $(\scheck_1 \cup \scheck_2) \cap \SPerm = \emptyset$: Then \refrule{SConsumeConjunction} applies.

    \case $\sif{e}{\phi_1}{\phi_2}$ -- $e \in \Expr, \phi_1, \phi_2 \in \Formula$:

      By lemma \ref{lem:eval-progress}, $\seval{\sstate_E}{e}{t}{\_}{\_}$ for some $t$. Let $V'$ be the valuation corresponding to this derivation. Then since this is a well-typed program, one of the following cases must apply:

      \subcase $V'(t) = \ktrue$: Let $\pc' = \pc(\sstate) \kand t$. Then $V(\pc') = \ktrue$ and by induction, $\sconsume{\sstate[\pc = \pc']}{\sstate_E[\pc = \pc']}{\phi_1}{\sstate'}{\_}{\_}$ for some $\sstate'$ where $V''(\pc(\sstate')) = \ktrue$ for the corresponding derivation $V''$ with initial valuation $V'$. Then by \refrule{SConsumeConditionalA}, $\sconsume{\sstate}{\sstate_E}{\sif{e}{\phi_1}{\phi_2}}{\sstate'}{\_}{\_}$, and $V''$ is the corresponding valuation for this derivation with initial valuation $V$.

      \subcase $V'(t) = \kfalse$: Let $\pc' = \pc(\sstate) \kand \kneg t$. Then $V(\pc') = \ktrue$ and by induction, $\sconsume{\sstate[\pc = \pc']}{\sstate_E[\pc = \pc']}{\phi_2}{\sstate'}{\_}{\_}$ for some $\sstate'$ where $V''(\pc(\sstate')) = \ktrue$ for the corresponding derivation $V''$ with initial valuation $V'$. Then by \refrule{SConsumeConditionalB}, $\sconsume{\sstate}{\sstate_E}{\sif{e}{\phi_1}{\phi_2}}{\sstate'}{\_}{\_}$, and $V''$ is the corresponding valuation for this derivation with initial valuation $V$.

  \end{enumcases}
\end{proof}

\begin{lemma}[Progress of consume (short form)]\label{lem:cons-progress}
  For any heap $\heap$, $\sstate$, $\gform$, and valuation $V$, if $V(\pc(\sstate)) = \ktrue$ then $\scons{\sstate}{\gform}{\sstate'}{\_}$ for some $\sstate'$ such that $V'(\pc(\sstate')) = \ktrue$ where $V'$ is the corresponding valuation.
\end{lemma}

\subsection{Progress}\label{sec:soundness-progress}

\begin{definition}\label{def:sguard-valuation}
  For a derivations $\sguard{\vstate}{\sstate}{\scheck}{\sperms}$, given an initial valuation $V$, heap $\heap$, and environment $\env$ the \textbf{corresponding valuation} is denoted as
    $$V[\sguard{\vstate}{\sstate}{\scheck}{\sperms} \mid \heap, \env].$$
  This function is defined as follows, depending on the rule that proves the derivation. Values are referenced using the respective name from the rule definition.
  \begin{itemize}
    \item \refrule{SGuardInit}:
      $$V[\sguard{\initsym}{\quintuple{\bot}{\emptyset}{\emptyset}{\emptyset}{\ktrue}}{\emptyset}{\emptyset} \mid \heap, \env] := V$$
    \item \refrule{SGuardSeq}:
      $$V[\sguard{\triple{\sstate}{\sseq{\kskip}{s}}{\gform}}{\sstate}{\emptyset}{\emptyset} \mid \heap, \env] := V$$
    \item \refrule{SGuardAssign}:
      $$V[\sguard{\triple{\sstate}{\sseq{x = e}{s}}{\gform}}{\sstate}{\scheck}{\emptyset} \mid \heap, \env] := V[\seval{\sstate}{e}{\_}{\sstate'}{\scheck} \mid \heap, \env]$$
    \item \refrule{SGuardAssignField}:
      \begin{align*}
        &V[\sguard{\triple{\sstate}{\sseq{x.f = e}{s}}{\gform}}{\sstate''}{\scheck' \cup \scheck''}{\emptyset} \mid \heap, \env] := \\
          &\quad V[\seval{\sstate}{e}{\_}{\sstate'}{\scheck'} \mid \heap, \env]
          [\scons{\sstate'}{\kacc(x.f)}{\sstate''}{\scheck''} \mid \heap, \env]
      \end{align*}
    \item \refrule{SGuardAlloc}:
      $$V[\sguard{\triple{\sstate}{\sseq{x = \salloc{S}}{s}}{\gform}}{\sstate}{\emptyset}{\emptyset} \mid \heap, \env] := V$$
    \item \refrule{SGuardCall}:
      \begin{align*}
        &V[\sguard{\triple{\sstate}{\sseq{y \kassign m(\multiple{e})}{s}}{\gform}}{\sstate''[\senv = \senv(\sstate)]}{\multiple{\scheck} \cup \scheck'}{\frem(\sstate'', \fpre(m))} \mid \heap, \env] := \\
          &\quad V\multiple{[\seval{\sstate}{e}{t}{\sstate'}{\scheck} \mid \heap, \env]}[\scons{\sstate'[\senv = [\multiple{x \mapsto t}]]}{\fpre(m)}{\sstate''}{\scheck'} \mid \heap, \env]
      \end{align*}
    \item \refrule{SGuardAssert}:
      $$V[\sguard{\triple{\sstate}{\sseq{\sassert{\phi}}{s}}{\gform}}{\sstate'}{\scheck}{\emptyset} \mid \heap] := V[\scons{\sstate}{\simprecise{\phi}}{\sstate'}{\scheck} \mid \heap, \env]$$
    \item \refrule{SGuardFold}:
      \begin{align*}
        &V[\sguard{\triple{\sstate}{\sseq{\sfold{p(\multiple{e})}}{s}}{\gform}}{\sstate''[\senv = \senv(\sstate)]}{\_}{\emptyset} \mid \heap, \env] := \\
        &\quad V\multiple{[\seval{\sstate}{e}{t}{\sstate'}{\scheck} \mid \heap, \env]}[\scons{\sstate'[\senv = [\multiple{x \mapsto t}]]}{\fpred(p)}{\sstate''}{\scheck'} \mid \heap, \env]
      \end{align*}
    \item \refrule{SGuardUnfold}:
      \begin{align*}
        &V[\sguard{\triple{\sstate}{\sseq{\sunfold{p(\multiple{e})}}{s}}{\gform}}{\sstate''}{\scheck' \cup \bigcup \multiple{\scheck}}{\emptyset} \mid \heap, \env] := \\
        &\quad V\multiple{[\seval{\sstate}{e}{t}{\sstate'}{\scheck} \mid \heap, \env]}[\scons{\sstate'}{p(\multiple{e})}{\sstate''}{\scheck'} \mid \heap, \env]
      \end{align*}
    \item \refrule{SGuardIf}:
      \begin{align*}
        &V[\sguard{\triple{\sstate}{\sseq{\sif{e}{s_1}{s_2}}{s}}{\gform}}{\sstate'}{\scheck}{\emptyset} \mid \heap, \env] := \\
        &\quad V[\seval{\sstate}{e}{\_}{\sstate'}{\scheck} \mid \heap, \env]
      \end{align*}
    \item \refrule{SGuardWhile}:
      \begin{align*}
        &V[\sguard{\triple{\sstate}{\sseq{\swhile{e}{\gform}{s}}{s'}}{\gform'} }{\sstate'[\pc = \pc(\sstate'')]}{\_}{\_} \mid \heap, \env] := \\
        &\quad V_0[\multiple{t \mapsto V_0(\senv(\sstate')(x))}][\sproduce{\sstate'[\senv = \senv(\sstate')[\multiple{x \mapsto t}]]}{\gform}{\sstate''} \mid \heap, \env]
      \end{align*}
      where $V_0 = V[\scons{\sstate}{\gform}{\sstate'}{\scheck'} \mid \heap, \env]$.
    \item \refrule{SGuardFinish}:
      $$V[\sguard{\triple{\sstate}{\kskip}{\gform}}{\sstate'}{\scheck}{\emptyset} \mid \heap, \env] := V[\scons{\sstate}{\gform}{\sstate'}{\scheck} \mid \heap, \env]$$
  \end{itemize}
\end{definition}

\begin{lemma}\label{lem:assert-after-rem}
  If $\assertion{\heap}{\perms}{\env}{\gform}$ and $\simstate{V}{\sstate}{\heap}{\perms \setminus \efoot{\heap}{\perms}{\gform}}{\env}$, then $\assertion{\heap}{\perms \setminus \vfoot{V'}{\heap}{\frem(\sstate', \gform')}}{\env}{\gform}$.
\end{lemma}

\begin{theorem}[Progress, part 1]\label{thm:dtrans-progress}
  Let $\Gamma$ be some dynamic state validated by $\vstate$ and valuation $V$. If $\sguard{\vstate}{\sstate'}{\scheck}{\sperms}$ with corresponding valuation $V'$ extending $V$, $V'(\pc(\sstate')) = \ktrue$, and $\pair{\heap}{\perms(\Gamma)} \vdash_{V'} \scheck$ then
  $$\dtrans{\prog}{\vfoot{V'}{\heap(\Gamma)}{\sperms}}{\Gamma}{\Gamma'}$$
  for some $\Gamma'$.

  In other words, if the dynamic state satisfies the matching symbolic checks, then dynamic execution can proceed.
\end{theorem}

\begin{proof}
  We proceed by cases on $\sguard{\vstate}{\sstate'}{\scheck}{\sperms}$.
  \begin{enumcases}
    \case \refrule{SGuardInit}: Result is trival by \refrule{ExecInit}.

    \case \refrule{SGuardSeq}: Then $\Gamma = \pair{\heap}{\triple{\perms}{\env}{\sseq{\kskip}{s}} \cdot \stack'}$ for some $\heap, \perms, \env, s, \stack'$, thus \\
    $\dexec{\heap}{\triple{\perms}{\env}{\sseq{\kskip}{s}} \cdot \stack'}{\vfoot{V'}{\heap}{\emptyset}}{\heap}{\triple{\perms}{\env}{s} \cdot \stack'}$ by \refrule{ExecSeq}, and the result is immediate from \refrule{ExecStep}.

    \case \refrule{SGuardAssign}: Then $\Gamma = \pair{\heap}{\triple{\perms}{\env}{\sseq{x = e}{s}} \cdot \stack'}$ for some $\heap, \perms, \env, \stack'$.

    By \refrule{SGuardAssign} $\seval{\sstate(\vstate)}{e}{t}{\sstate'}{\scheck}$ for some $t, \sstate', \scheck$. Also, $V' = V[\seval{\sstate}{e}{t}{\sstate'}{\scheck} \mid \heap]$. Since $\Gamma$ corresponds to $\vstate$, by definition \ref{def:vstate-corresponds} $\simstate{V}{\sstate(\vstate)}{\heap}{\perms}{\env}$. By assumptions $\rtassert{V'}{\heap}{\perms}{\scheck}$, and $V'(\pc(\sstate')) = \ktrue$. Then $\eval{\heap}{\env}{e}{V'(t)}$ and $\frm{\heap}{\perms}{\env}{e}$ by lemma \ref{lem:seval-soundness}.

    Therefore $\dexec{\heap}{\triple{\perms}{\env}{\sseq{x = e}{s}} \cdot \stack'}{\vfoot{V'}{\heap}{\emptyset}}{\heap}{\triple{\perms}{\env[x \mapsto V'(t)]}{s} \cdot \stack'}$ by \refrule{ExecAssign}, and the result is immediate from \refrule{ExecStep}.

    \case \refrule{SGuardAssignField}: Then $\Gamma = \pair{\heap}{\triple{\perms}{\env}{\sseq{x.f = e}{s}} \cdot \stack'}$ for some $\heap, \perms, \env, x, f, e, s, \stack'$.

    By \refrule{SGuardAssignField} $\seval{\sstate(\vstate)}{e}{t}{\sstate'}{\scheck_1}$ for some $t, \sstate', \scheck_1$, and $\scons{\sstate'}{\kacc(x.f)}{\sstate''}{\scheck_2}$ for some $\sstate'', \scheck_2$. Also, $V' = V[\seval{\sstate(\vstate)}{e}{t}{\sstate'}{\scheck_1} \mid \heap][\scons{\sstate'}{\kacc(x.f)}{\sstate''}{\scheck_2} \mid \heap]$.

    Since $\Gamma$ corresponds to $\vstate$, by definition \ref{def:vstate-corresponds} $\simstate{V}{\sstate(\vstate)}{\heap}{\perms}{\env}$. By assumptions $\rtassert{V'}{\heap}{\perms}{\scheck_1 \cup \scheck_2}$, thus $\rtassert{V'}{\heap}{\perms}{\scheck_1}$ and $\rtassert{V'}{\heap}{\perms}{\scheck_2}$ by lemma \ref{lem:scheck-monotonicity}, and $V'(\pc(\sstate'')) = \ktrue$, thus $V'(\pc(\sstate')) = \ktrue$ by lemma \ref{lem:cons-subpath}.
    
    Now $\eval{\heap}{\env}{e}{V'(t)}$ and $\frm{\heap}{\perms}{\env}{e}$ by lemma \ref{lem:seval-soundness}. Let $\ell = \env(x)$, then $\eval{\heap}{\env}{x}{\env(x)}$ by \refrule{EvalVar}. Finally, $\assertion{\heap}{\perms}{\env}{\kacc(x.f)}$ by lemma \ref{lem:cons-soundness}.

    Let $\heap' = \heap[\pair{\ell}{f} \mapsto V'(t)]$. Then $\dexec{\heap}{\triple{\perms}{\env}{\sseq{x.f = e}{s}} \cdot \stack'}{\vfoot{V'}{\heap}{\emptyset}}{\heap'}{\triple{\perms}{\env}{s} \cdot \stack'}$ by \refrule{ExecAssignField}, and the result is immediate from \refrule{ExecStep}.

    \case \refrule{SGuardAlloc}: Then $\Gamma = \pair{\heap}{\triple{\perms}{\env}{\sseq{x = \kalloc(S)}{s}} \cdot \stack'}$ for some $\heap, \perms, \env, x, S, s, \stack'$.

    Let $\ell = \ffresh$, $\multiple{T~f} = \fstruct(S)$, and $\heap' = \heap[\multiple{(\ell, f) \mapsto \fdefault(T)}]$. Then \\
    $\dexec{\heap}{\triple{\perms}{\env}{\sseq{x = \kalloc(S)}{s}} \cdot \stack'}{\vfoot{V'}{\heap}{\emptyset}}{\heap'}{\triple{\perms}{\env}{s} \cdot \stack'}$ by \refrule{ExecAlloc}, and the result is immediate from \refrule{ExecStep}.

    \case \refrule{SGuardCall}:
    Then $\Gamma = \pair{\heap}{\triple{\perms}{\env}{\sseq{y \kassign m(\multiple{e})}{s}}} \cdot \stack$ for some $\heap, \perms, \env, y, m, \multiple{e}, s, \stack$. Let $\multiple{x} = \fparams(m)$.

    By \refrule{SGuardCall} $\multiple{\seval{\sstate(\vstate)}{e}{t}{\sstate'}{\scheck}}$ for some $t, \sstate', \scheck$, and $\scons{\sstate'[\senv = [\multiple{x \mapsto t}]]}{\fpre(m)}{\sstate''}{\scheck'}$ for some $\sstate'', \scheck'$.

    Also, by definition $V' = V\multiple{[\seval{\sstate(\vstate)}{e}{t}{\sstate'}{\scheck} \mid \heap]}[\scons{\sstate'[\senv = \multiple{x \mapsto t}]}{\fpre(m)}{\sstate''}{\scheck'} \mid \heap]$.

    Since $\Gamma$ corresponds to $\vstate$, by definition \ref{def:vstate-corresponds} $\simstate{V}{\sstate(\vstate)}{\heap}{\perms}{\env}$. By assumptions, $V'(\pc(\sstate'')) = \ktrue$ thus $V'(\pc(\sstate')) = \ktrue$ by lemma \ref{lem:cons-subpath}, and $\rtassert{V'}{\heap}{\perms}{\scheck \cup \scheck'}$, thus $\rtassert{V'}{\heap}{\perms}{\scheck}$ and $\rtassert{V'}{\heap}{\perms}{\scheck'}$ by lemma \ref{lem:scheck-monotonicity}.

    Then $\multiple{\eval{\heap}{\env}{e}{V'(t)}}$, $\multiple{\frm{\heap}{\perms}{\env}{e}}$, and $\simstate{V'}{\sstate'}{\heap}{\perms}{\env}$ by lemma \ref{lem:seval-soundness}.

    Let $\senv' = [\multiple{x \mapsto t}]$ and $\env' = [\multiple{x \mapsto V'(t)}]$. Since $\simstate{V'}{\sstate'}{\heap}{\perms}{\env}$ and $\simenv{V'}{\senv'}{\env'}$ by construction, $\simstate{V'}{\sstate'[\senv = \senv']}{\heap}{\perms}{\env'}$.

    Let $\perms_V = \vfoot{V'}{\heap}{\frem(\sstate', \fpre(m))}$ and $\perms_E = \efoot{\heap}{\env}{\fpre(m)}$.

    As noted before, $\scons{\sstate'[\senv = \senv']}{\fpre(m)}{\sstate''}{\scheck'}$ by \refrule{SExecCall}. Also, $V'(\pc(\sstate'')) = \ktrue$ and $\rtassert{V'}{\heap}{\perms}{\scheck'}$. Therefore $\simstate{V'}{\sstate'}{\heap}{\perms \setminus \perms_E}{\env'}$ and $\assertion{\heap}{\perms}{\env}{\fpre(m)}$ by lemma \ref{lem:cons-soundness}.

    Then $\assertion{\heap}{\perms \setminus \perms_V}{\env}{\fpre(m)}$ by lemma \ref{lem:assert-after-rem}.

    Let $\perms' = \foot{\heap}{\perms \setminus \perms_V}{\env'}{\fpre(m)}$. Then $\dexec{\heap}{\triple{\perms}{\env}{\sseq{y \kassign m(\multiple{e})}{s}} \cdot \stack}{\perms_V}{\heap}{\triple{\perms'}{\env'}{\sseq{\fbody(m)}{\kskip}} \cdot \triple{\perms \setminus \perms'}{\env}{\sseq{y \kassign m(\multiple{e})}{s}} \cdot \stack}$ by \refrule{ExecCallEnter}, and the result is immediate from \refrule{ExecStep}.

    \case \refrule{SGuardAssert}:
    Then $\Gamma = \pair{\heap}{\triple{\perms}{\env}{\sseq{\sassert{\phi}}{s}} \cdot \stack}$ for some $\heap, \perms, \env, \phi, s, \stack$.

    By \refrule{SGuardAssert} $\scons{\sstate(\vstate)}{\simprecise{\phi}}{\sstate'}{\scheck}$, also $V' = V[\scons{\sstate(\vstate)}{\simprecise{\phi}}{\sstate'}{\scheck} \mid \heap]$.

    Since $\Gamma$ corresponds to $\vstate$, by definition \ref{def:vstate-corresponds} $\simstate{V}{\sstate(\vstate)}{\heap}{\perms}{\env}$. Also by assumptions, $\rtassert{V'}{\heap}{\perms}{\scheck}$, and $V'(\pc(\sstate')) = \ktrue$. Thus $\assertion{\heap}{\perms}{\env}{\phi}$ by lemma \ref{lem:cons-soundness} since $\simprecise{\phi}$ is a specification.
    
    Therefore $\dexec{\heap}{\triple{\perms}{\env}{\sseq{\sassert{\phi}}{s}} \cdot \stack}{\vfoot{V'}{\heap}{\emptyset}}{\heap}{\triple{\perms}{\env}{s} \cdot \stack}$ by \refrule{ExecAssert}, and the result is immediate from \refrule{ExecStep}.

    \case \refrule{SGuardFold}:
    Then $s(\vstate) = s(\Gamma) = \sfold{p(\multiple{e})}$ for some $p, \multiple{e}$. Thus \refrule{ExecFold} trivially applies, and the result is immediate from \refrule{ExecStep}.

    \case \refrule{SGuardUnfold}:
    Then $s(\vstate) = s(\Gamma) = \sunfold{p(\multiple{e})}$ for some $p, \multiple{e}$. Thus \refrule{ExecUnfold} trivially applies, and the result is immediate from \refrule{ExecStep}.

    \case \refrule{SGuardIf}:
    Then $s(\vstate) = s(\Gamma) = \sseq{\sif{e}{s_1}{s_2}}{s}$ for some $e, s_1, s_2$, and thus $\Gamma = \pair{\heap}{\triple{\perms}{\env}{\sseq{\sif{e}{s_1}{s_2}}{s}} \cdot \stack}$ for some $\heap, \perms, \env, \stack$.

    By \refrule{SGuardIf} $\seval{\sstate(\vstate)}{e}{t}{\sstate'}{\scheck}$ for some $t, \sstate', \scheck$, and also $V' = V[\seval{\sstate(\vstate)}{e}{t}{\sstate'}{\scheck} \mid \heap]$.

    Now by assumptions $V'(\pc(\sstate')) = \ktrue$ and $\rtassert{V'}{\heap}{\perms}{\scheck}$. Then $\eval{\heap}{\env}{e}{V'(t)}$ and $\frm{\heap}{\perms}{\env}{e}$ by lemma \ref{lem:seval-soundness}.

    Now, since we assume a well-typed program, $V'(t) = \ktrue$ or $\kfalse$. Then either \refrule{ExecIfA} or \refrule{ExecIfB} applies, and the result is immediate from \refrule{ExecStep}.

    \case \refrule{SGuardWhile}:
      Then $s(\vstate) = s(\Gamma) = \sseq{\swhile{e}{\gform}{s}}{s'}$ for some $e$, $\gform$, $s$, $s'$, and thus $\Gamma = \pair{\heap}{\triple{\perms}{\env}{\sseq{\swhile{e}{\gform}{s}}{s'}} \cdot \stack}$ for some $\heap$, $\perms$, $\env$, $\stack$.

      Let $\multiple{x} = \fmodified(s)$. By \refrule{SGuardWhile} $\scons{\sstate}{\gform}{\sstate'}{\scheck'}$, $\sproduce{\sstate'[\senv = \senv(\sstate')[\multiple{x \mapsto \ffresh}]]}{\gform}{\sstate''}$, and $\seval{\sstate''}{e}{t}{\_}{\scheck''}$. Then by definition \ref{def:sguard-valuation} $V'$ extends the corresponding valuation for these judgements.

      By assumptions $V'(\pc(\sstate'')) = \ktrue$, thus $V'(\pc(\sstate')) = \ktrue$ by lemma \ref{lem:produce-subpath}.

      Also by assumptions $\rtassert{V'}{\heap}{\perms}{\scheck' \cup \scheck''}$, thus $\rtassert{V'}{\heap}{\perms}{\scheck'}$ and $\rtassert{V'}{\heap}{\perms}{\scheck''}$ by lemma \ref{lem:scheck-monotonicity}.

      Therefore $\simstate{V'}{\sstate'}{\heap}{\perms \setminus \efoot{\heap}{\perms}{\gform}}{\env}$ and $\assertion{\heap}{\perms}{\env}{\gform}$ by lemma \ref{lem:cons-soundness}.   

      Let $\xperms = \vfoot{V'}{\heap}{\frem(\sstate', \gform)}$. Then by lemma \ref{lem:assert-after-rem} $\assertion{\heap}{\perms \setminus \xperms}{\env}{\gform}$.

      Let $\multiple{t}$ be the list of fresh values used in $\sproduce{\sstate'[\senv = \senv(\sstate')[\multiple{x \mapsto \ffresh}]]}{\gform}{\sstate''}$ when applying \refrule{SGuardWhile}. Let $\senv' = \senv(\sstate')[\multiple{x \mapsto t}]$. Then by definition \ref{def:sguard-valuation}, and since $\simenv{V'}{\senv(\sstate')}{\env}$, for each pair of $x$ and $t$, $V'(\senv'(x)) = V'(t) = V'(\senv(\sstate')(x)) = \env(x)$.
      
      Therefore $\simenv{V'}{\senv'}{\env}$, and thus $\simstate{V'}{\sstate'[\senv = \senv']}{\heap}{\perms \setminus \efoot{\heap}{\perms}{\gform}}{\env}$.

      Now by lemma \ref{lem:produce-soundness} $\simstate{V'}{\sstate''}{\heap}{\perms}{\gform}$.

      Also, as noted before, $\rtassert{V'}{\heap}{\perms}{\scheck''}$. In addition, since $V'(\pc(\sstate'')) = \ktrue$, by lemma \ref{lem:eval-progress} $\seval{\sstate''}{e}{t}{\sstate'''}{\scheck''}$ where $V'(\pc(\sstate''')) = \ktrue$. Therefore $\eval{\heap}{\env}{e}{V'(t)}$ by lemma \ref{lem:seval-soundness}.

      Let $\perms' = \efoot{\heap}{\perms \setminus \xperms}{\env}$.

      Now, since we assume the program to be properly typed, one of the following subcases apply:

      \subcase $V'(t) = \ktrue$: Then by \refrule{ExecWhileSkip} $\dexec{\heap}{\triple{\perms}{\env}{\sseq{\swhile{e}{\gform}{s}}{s'}} \cdot \stack}{\xperms}{\heap}{\triple{\perms'}{\env}{\sseq{s}{\kskip}} \cdot \triple{\perms \setminus \perms'}{\env}{\sseq{\swhile{e}{\gform}{s}}{s'}} \cdot \stack}$ and the result is immediate from \refrule{ExecStep}.

      \subcase $V'(t) = \kfalse$: Then by \refrule{ExecWhileSkip} $\dexec{\prog}{\heap}{\triple{\perms}{\env}{\sseq{\swhile{e}{\gform}{s}}{s'}} \cdot \stack}{\xperms}{\heap}{\triple{\perms}{\env}{s} \cdot \stack}$ and the result is immediate from \refrule{ExecStep}.

    \case \refrule{SGuardFinish}:
    Then $s(\vstate) = s(\Gamma) = \kskip$ and thus $\Gamma = \pair{\heap}{\triple{\perms}{\env}{\kskip} \cdot \stack}$ for some $\heap, \perms, \env, \stack$.

    By \refrule{SGuardFinish} $\scons{\sstate(\vstate)}{\gform(\vstate)}{\sstate'}{\scheck}$, and $V' = V[\scons{\sstate(\vstate)}{\gform(\vstate)}{\sstate'}{\scheck} \mid \heap]$. By assumptions, $V'(\pc(\sstate')) = \ktrue$ and $\rtassert{V'}{\heap}{\perms}{\scheck}$. Therefore $\assertion{\heap}{\perms}{\env}{\gform(\vstate)}$.

    Since $\Gamma$ is a valid state, the partial state $\pair{\heap}{\stack}$ must be validated by $\vstate$ and $V$, thus one of the following subcases must apply:

    \subcase $\stack = \nilsym$ -- then $\Gamma = \pair{\heap}{\triple{\perms}{\env}{\kskip} \cdot \nilsym}$. Then \refrule{ExecFinal} trivially applies to yield the result.

    \subcase $\stack = \triple{\perms_0}{\env_0}{\sseq{m(\multiple{e})}{s_0}} \cdot \stack'$ for some $\perms_0, \env_0, m, \multiple{e}, s_0, \stack'$ and $\gform(\vstate) = \fpost(m)$.

    Since $\gform(\vstate) = \fpost(m)$, $\assertion{\heap}{\perms}{\env}{\fpost(m)}$. Then \refrule{ExecCallExit} applies, and the result is immediate from \refrule{ExecStep}.

    \subcase $\stack = \triple{\perms_0}{\env_0}{\sseq{\swhile{e}{\gform}{s_0}}{s_0'}} \cdot \stack'$ for some $\perms_0, \env_0, e, \gform, s_0, s_0', \stack'$ and $\gform(\vstate) = \gform$.

    Since $\gform(\vstate) = \fpost(m)$, $\assertion{\heap}{\perms}{\env}{\gform}$. Then \refrule{ExecWhileFinish} applies, and the result is immediate from \refrule{ExecStep}.

  \end{enumcases}
\end{proof}

\begin{theorem}[Progress, part 2]\label{thm:guard-progress}
  Let $\Gamma$ be some well-formed dynamic state validated by $\vstate$ and valuation $V$. Then if $\Gamma \ne \finalsym$ and $\rtassert{V'}{\heap(\Gamma)}{\perms}{\scheck}$,
  $$\vstate \rightharpoonup \sstate', \scheck, \sperms$$
  for some $\sstate'$, $\scheck$, $\sperms$ such that $V'(\pc(\sstate')) = \ktrue$ where $V'$ is the corresponding valuation extending $V'$.

  In other words, there is always some matching guard that computes the necessary checks.
\end{theorem}
\begin{proof}
  First, if $\Gamma = \initsym$, then \refrule{SGuardInit} applies to yield the desired result.

  Otherwise, $\Gamma = \pair{\heap}{\stack}$ for some non-empty stack. Therefore $\vstate = \triple{\sstate}{s}{\gform}$ for some $\sstate$, $s$, and $\gform$ such that $\simstate{V}{\sstate}{\heap}{\perms(\Gamma)}{\env(\Gamma)}$ and $s = s(\Gamma)$.
  
  Then one of the following cases apply since $\stack$ is a well-formed stack:

  \begin{enumcases}
    \case $s = \kskip$: By lemma \ref{lem:cons-progress} $\scons{\sstate}{\gform}{\sstate'}{\scheck}$ for some $\sstate'$, $\scheck$ such that $V'(\pc(\sstate')) = \ktrue$ where $V'$ is the corresponding valuation. Then by \refrule{SGuardFinish} $\sguard{\vstate}{\sstate'}{\scheck}{\emptyset}$ and $V'$ is the corresponding valuation for this judgement.

    \case $s = \sseq{s'}{s''}$ for some $s'$, $s''$: We complete the proof by proving the following statement by induction on the syntax form of $s'$:

    If $s = \sseq{s'}{s''}$ then $\vstate \rightharpoonup \sstate', \scheck, \sperms$ for some $\sstate'$, $\scheck$, $\sperms$ such that $V'(\pc(\sstate')) = \ktrue$ where $V'$ is the corresponding valuation extending $V'$.

    \subcase $s' = \sseq{s_1}{s_2}$:
      By lemma \ref{lem:stmt-rearrangement}, $s' = \sseq{s_1'}{s_2'}$ where $s_1'$ is not a sequence statement. Then $\sseq{s}{s'} = \sseq{s_1'}{\sseq{s_2'}{s'}}$.

      Then the inductive hypothesis applies, which completes the proof.

    \subcase $s' = \kskip$:
      Then $\sguard{\vstate}{\sstate}{\emptyset}{\emptyset}$ by \textsc{SGuardSeq}.

    \subcase $s' = x \kassign e$:
      By lemma \ref{lem:eval-progress}, $\seval{\sstate}{e}{\_}{\sstate'}{\scheck}$ for some $\sstate'$ and $\scheck$ such that $V_1(\pc(\sstate')) = \ktrue$ for the corresponding valuation $V_1$.

      Then $\sguard{\vstate}{\sstate'}{\scheck}{\emptyset}$ by \refrule{SGuardAssign}. By definition the corresponding valuation extends $V_1$, thus $V'(\pc(\sstate')) = \ktrue$.

    \subcase $x.f \kassign e$:

      By lemma \ref{lem:eval-progress} $\seval{\sstate}{e}{\_}{\sstate'}{\scheck'}$ for some $\sstate'$ and $\scheck'$ such that $V_1(\pc(\sstate')) = \ktrue$ where $V_1$ is the corresponding valuation extending $V$.

      By lemma \ref{lem:cons-progress} $\scons{\sstate'}{\kacc(x.f)}{\sstate''}{\scheck''}$ for some $\sstate''$ and $\scheck''$ such that $V_2(\pc(\sstate'')) = \ktrue$ for corresponding valuation $V_2$, with initial valuation $V_1$.

      Then $\sguard{\vstate}{\sstate''}{\scheck' \cup \scheck''}{\emptyset}$ by \refrule{SGuardAssignField}. By definition the corresponding valuation $V'$ extends $V_2$, therefore $V'(\pc(\sstate'')) = \ktrue$.

    \subcase $x = \salloc{S}$:

      Then $\sguard{\vstate}{\sstate}{\emptyset}{\emptyset}$ by \refrule{SGuardAlloc}.

    \subcase $y \kassign m(e_1, \cdots, e_n)$:

      Let $\sstate_0 = \sstate$ and $V_0 = V$, then for each $e_i$, $\seval{\sstate_{i-1}}{e}{t_i}{\sstate_i}{\_}$ for some $\sstate_i$ such that $V_i(\pc(\sstate_i)) = \ktrue$ for corresponding valuation $V_i$, with initial valuation $V_{i-1}$, by lemma \ref{lem:eval-progress}.

      Let $x_1, \cdots, x_n = \fparams(m)$. By lemma \ref{lem:cons-progress}, $\scons{\sstate_n[\senv = [\multiple{x_i \mapsto t_i}]]}{\fpre(m)}{\sstate'}{\_}$ for some $\sstate'$ such that $V'(\pc(\sstate')) = \ktrue$ for corresponding valuation $V'$, with initial valuation $V_n$.

      Then $\sguard{\vstate}{\sstate'[\senv = \senv(\sstate)]}{\scheck_1 \cup \cdots \cup \scheck_n \cup \scheck'}{\frem(\sstate'', \fpre(m))}$ by \refrule{SGuardCall} and $V'$ is the corresponding valuation

    \subcase $\sassert{\gform}$:

      By lemma \ref{lem:cons-progress} $\scons{\sstate}{\gform}{\sstate'}{\scheck}$ for some $\sstate'$ and $\scheck$ where $V'(\sstate') = \ktrue$ for the corresponding valuation $V'$.

      Then by \refrule{SGuardAssert} $\sguard{\vstate}{\sstate'}{\scheck}{\emptyset}$ and $V'$ is the corresponding valuation.

    \subcase $\sif{e}{s_1}{s_2}$:

      By lemma \ref{lem:eval-progress} $\seval{\sstate}{e}{t}{\sstate'}{\scheck}$ for some $t$, $\sstate'$ such that $V'(\pc(\sstate')) = \ktrue$ where $V'$ is the corresponding valuation.

      Then $\sguard{\vstate}{\sstate'}{\scheck}{\emptyset}$ by \refrule{SGuardIf} and $V'$ is the corresponding valuation.

    \subcase $\swhile{e}{\gform}{s}$ for some $e$, $\gform$, $s$:

      By lemma \ref{lem:cons-progress} $\scons{\sstate}{\gform}{\sstate'}{\scheck'}$ for some $\sstate'$ and $\scheck'$ such that $V_1(\pc(\sstate')) = \ktrue$ where $V'$ is the corresponding valuation.

      Let $\multiple{x} = \fmodified(s)$ and $\sstate'' = \sstate'[\senv = \senv(\sstate')[\multiple{x \mapsto \ffresh}]]$.
      
      Let $V_2 = V_1[\multiple{t \mapsto V_1(\senv(\sstate')(x))}]$. Then $V_2(\pc(\sstate'')) = V_1(\pc(\sstate')) = \ktrue$.
      
      Then by lemma \ref{lem:produce-progress} $\sproduce{\sstate''}{\gform}{\sstate'''}$ for some $\sstate'''$ such that $V_3(\pc(\sstate''')) = \ktrue$ where $V_3$ is the corresponding valuation extending $V_2$.
      
      By lemma \ref{lem:eval-progress} $\seval{\sstate'''}{e}{t}{\_}{\_}$ for some $t$. Let $V'$ be the corresponding valuation extending $V_3$, thus $V'(\pc(\sstate''')) = V_3(\pc(\sstate''')) = \ktrue$.

      Then $\sguard{\vstate}{\sstate'[\pc = \pc(\sstate''')]}{\scheck' \cup \scheck''}{\frem(\sstate', \gform)}$ by \textsc{SGuardWhile} and $V'$ is the corresponding valuation.

    \subcase $\sfold{p(e_1, \cdots, e_n)}$:

      Let $\sstate_0 = \sstate$ and $V_0 = V$. For each $e_i$, $\seval{\sstate_{i-1}}{e}{t_i}{\sstate_i}{\scheck_i}$ by lemma \ref{lem:eval-progress} for some $\sstate_i$ and $\scheck_i$ such that $V_i(\pc(\sstate_i)) = \ktrue$ where $V_i$ is the corresponding valuation.

      Let $x_1, \cdots, x_n = \fpredparams(p)$. By lemma \ref{lem:cons-progress} $\scons{\sstate_n[\senv = [\multiple{x_i \mapsto t_n}]]}{\fpred(p)}{\sstate'}{\scheck'}$ for some $\sstate'$ and $\scheck'$ such that $V'(\pc(\sstate')) = \ktrue$ where $V'$ is the corresponding valuation extending $V_n$.

      Then $\sguard{\vstate}{\sstate'[\senv = \senv(\sstate)]}{\scheck_1 \cup \cdots \cup \scheck_n \cup \scheck'}{\emptyset}$ by \refrule{SGuardFold} and $V'$ is the corresponding valuation.

    \subcase $\sunfold{p(e_1, \cdots, e_n)}$:

      Let $\sstate_0 = \sstate$ and $V_0 = V$. For each $e_i$, $\seval{\sstate_{i-1}}{e}{t_i}{\sstate_i}{\scheck_i}$ by lemma \ref{lem:eval-progress} for some $\sstate_i$ and $\scheck_i$ such that $V_i(\pc(\sstate_i)) = \ktrue$ where $V_i$ is the corresponding valuation.

      By lemma \ref{lem:cons-progress} $\scons{\sstate_n}{p(e_1, \cdots, e_n)}{\sstate'}{\scheck'}$ for some $\sstate'$ and $\scheck'$ such that $V'(\pc(\sstate')) = \ktrue$ where $V'$ is the corresponding valuation extending $V_n$.

      Then $\sguard{\vstate}{\sstate'}{\scheck_1 \cup \cdots \cup \scheck_n \cup \scheck'}{\emptyset}$ by \refrule{SGuardUnfold} and $V'$ is the corresponding valuation.

  \end{enumcases}

\end{proof}

\subsection{Preservation}\label{sec:soundness-preservation}

\begin{lemma}\label{lem:preservation-heap-env-unchanged}
  Suppose $\Gamma = \pair{\heap}{\triple{\perms}{\env}{s} \cdot \stack}$ and $\Gamma' = \pair{\heap}{\triple{\perms'}{\env}{s'} \cdot \stack}$.

  If $\Gamma$ is validated by $\vstate$ and $V$, $\dtrans{\prog}{\Gamma}{\_}{\Gamma'}$ with valuation $V'$, and $\strans{\prog}{\vstate}{\vstate'}$ for some $\vstate'$ such that $\vstate'$ corresponds to $\Gamma'$, $\gform(\vstate') = \gform(\vstate)$, and $\dom(\senv(\vstate')) \supseteq \dom(\senv(\vstate))$, then $\Gamma'$ is a valid state.
\end{lemma}

\begin{lemma}\label{lem:pres-add-heap}
  If $\simstate{V}{\sstate}{\heap}{\perms \setminus \vfoot{V}{\heap}{\triple{t}{f}{t'}}}{\env}$, $\pair{V(t)}{f} \in \perms$, and $\heap(V(t), f) = V(t')$, then $\simstate{V}{\sstate[\sheap = \sheap(\sstate); \triple{t}{f}{t'}]}{\heap}{\perms}{\env}$.
\end{lemma}

\begin{lemma}\label{lem:eval-heap-efoot-unchanged}
  If $\universal{\pair{\ell}{f} \in \efoot{\heap}{\perms}{e}}{\heap'(\ell, f) = \heap(\ell, f)}$ and $\eval{\heap}{\env}{e}{v}$ then $\eval{\heap'}{\env}{e}{v}$.
\end{lemma}

\begin{proof}
    By induction on the derivation of $\eval{\heap'}{\env}{e}{v}$.

    \begin{enumcases}
        \case \refrule{EvalLiteral} -- $\eval{\heap}{\env}{l}{l}$: $\eval{\heap'}{\env}{l}{l}$ by \refrule{EvalLiteral}.
    
        \case \refrule{EvalVar} -- $\eval{\heap}{\env}{x}{\env(x)}$: $\eval{\heap'}{\env}{x}{\env(x)}$ by \refrule{EvalVar}.
    
        \case\label{case:eval-heap-efoot-unchanged-anda} \refrule{EvalAndA} -- $\eval{\heap}{\env}{e_1 \kand e_2}{\kfalse}$:
    
          By inversion of \refrule{EvalAndA} $\eval{\heap}{\env}{e_1}{\kfalse}$. Then $\efoot{\heap}{\env}{e_1 \kand e_2} = \efoot{\heap}{\env}{e_1}$ by definition. Thus $\universal{\pair{\ell}{f} \in \efoot{\heap}{\perms}{e_1}}{\heap'(\ell, f) = \heap(\ell, f)}$, and therefore $\eval{\heap'}{\env}{e_1}{\kfalse}$ by induction.
    
        \case\label{case:eval-heap-efoot-unchanged-andb} \refrule{EvalAndB} -- $\eval{\heap}{\env}{e_1 \kand e_2}{v_2}$:
    
          By inversion of \refrule{EvalAndB} $\eval{\heap}{\env}{e_1}{\ktrue}$ and $\eval{\heap}{\env}{e_2}{v_2}$. Now $\efoot{\heap}{\env}{e_1 \kand e_2} = \efoot{\heap}{\env}{e_1} \cup \efoot{\heap}{\env}{e_2}$.
    
          Thus $\universal{\pair{\ell}{f} \in \efoot{\heap}{\perms}{e_1}}{\heap'(\ell, f) = \heap(\ell, f)}$ since $\efoot{\heap}{\env}{e_1} \subseteq \efoot{\heap}{\env}{e_1 \kand e_2}$. Therefore $\eval{\heap'}{\env}{e_1}{\ktrue}$ by induction. Similarly, $\eval{\heap'}{\env}{e_2}{v_2}$.
    
          Therefore $\eval{\heap'}{\env}{e_1 \kand e_2}{v_2}$ by \refrule{EvalAndB}.
    
        \case \refrule{EvalOrA} -- $\eval{\heap}{\env}{e_1 \kor e_2}{\ktrue}$: Similar to case \ref{case:eval-heap-efoot-unchanged-anda}.
    
        \case \refrule{EvalOrB} -- $\eval{\heap}{\env}{e_1 \kor e_2}{v_2}$: Similar to case \ref{case:eval-heap-efoot-unchanged-andb}.
    
        \case \refrule{EvalOp} -- $\eval{\heap}{\env}{e_1 \oplus e_2}{v_1 \oplus v_2}$:
    
          By inversion of \refrule{EvalOp} $\eval{\heap}{\env}{e_1}{v_1}$ and $\eval{\heap}{\env}{e_2}{v_2}$. Now $\efoot{\heap}{\env}{e_1 \oplus e_2} = \efoot{\heap}{\env}{e_1} \cup \efoot{\heap}{\env}{e_2}$.
    
          Thus $\universal{\pair{\ell}{f} \in \efoot{\heap}{\perms}{e_1}}{\heap'(\ell, f) = \heap(\ell, f)}$ since $\efoot{\heap}{\env}{e_1} \subseteq \efoot{\heap}{\env}{e_1 \kand e_2}$. Therefore $\eval{\heap'}{\env}{e_1}{v_1}$ by induction. Similarly, $\eval{\heap'}{\env}{e_2}{v_2}$.
    
          Therefore $\eval{\heap'}{\env}{e_1 \oplus e_2}{e_1 \oplus e_2}$ by \refrule{EvalOp}.
    
        \case \refrule{EvalNeg} -- $\eval{\heap}{\env}{\neg v}$:
    
          By inversion of \refrule{EvalNeg} $\eval{\heap}{\env}{e}{v}$. Also, $\efoot{\heap}{\env}{\kneg e} = \efoot{\heap}{\env}{e}$ by definition, thus $\universal{\pair{\ell}{f} \in \efoot{\heap}{\perms}{e}}{\heap'(\ell, f) = \heap(\ell, f)}$. Therefore $\eval{\heap'}{\env}{e}{v}$ by induction.
    
          Therefore $\eval{\heap'}{\env}{\kneg e}{\neg v}$ by \refrule{EvalNeg}.
    
        \case \refrule{EvalField} -- $\eval{\heap}{\env}{e.f}{\heap(\ell, f)}$:
    
          By inversion of \refrule{EvalField} $\eval{\heap}{\env}{e}{\ell}$. Then $\efoot{\heap}{\env}{e.f} = \efoot{\heap}{\env}{e}; \pair{\ell}{f}$.
    
          Thus $\universal{\pair{\ell}{f} \in \efoot{\heap}{\perms}{e}}{\heap'(\ell, f) = \heap(\ell, f)}$ since $\efoot{\heap}{\env}{e} \subseteq \efoot{\heap}{\env}{e.f}$. Therefore $\eval{\heap'}{\env}{e}{\ell}$ by induction.
    
          Now $\eval{\heap'}{\env}{e.f}{\heap'(\ell, f)}$. Also, since $\pair{\ell}{f} \in \efoot{\heap}{\env}{e.f}$, $\heap'(\ell, f) = \heap(\ell, f)$. Therefore $\eval{\heap'}{\env}{e.f}{\heap(\ell, f)}$.

        \definecolor{shadecolor}{gray}{0.9}
        \begin{snugshade} \case \refrule{EvalUnfolding} -- By the induction hypothesis, $\eval{\heap'}{\env}{e_0}{v}$, so by \refrule{EvalUnfolding}, $\eval{\heap'}{\env}{\sunfolding{\pe}{e_0}}{v}$ as desired.

        \case \refrule{EvalFunction} -- By the induction hypothesis, $\multiple{\eval{\heap'}{\env}{e}{v}}$ and $\eval{\heap'}{\env'}{\ffuncbody(f)}{v'}$, so by definition, $\eval{\heap'}{\env}{\fe}{v'}$.
        \end{snugshade}
    \end{enumcases}
\end{proof}

\begin{lemma}\label{lem:frm-heap-efoot-unchanged}
  If $\universal{\pair{\ell}{f} \in \efoot{\heap}{\env}{e}}{\heap'(\ell, f) = \heap(\ell, f)}$ and $\frm{\heap}{\perms}{\env}{e}$ then $\frm{\heap'}{\perms}{\env}{e}$.
\end{lemma}

\begin{proof}
    By induction on the derivation of $\frm{\heap'}{\perms}{\env}{e}$.

    \begin{enumcases}
        \case \refrule{EFrameExpression} -- $\efrm{\heap}{\perms}{\env}{e}$:
    
          By inversion of \refrule{EFrameExpression} $\frm{\heap}{\perms}{\env}{e}$ and then $\frm{\heap'}{\perms}{\env}{e}$ by assumptions and lemma \ref{lem:frm-heap-efoot-unchanged}. Therefore $\efrm{\heap}{\perms}{\env}{e}$ by \refrule{EFrameExpression}.
    
        \case \refrule{EFrameConjunction} -- $\efrm{\heap}{\perms}{\env}{\phi_1 * \phi_2}$:
    
          By inversion of \refrule{EFrameConjunction} $\efrm{\heap}{\perms}{\env}{\phi_1}$ and $\efrm{\heap}{\perms}{\env}{\phi_2}$. Also, $\efoot{\heap}{\env}{\phi_1 * \phi_2} = \efoot{\heap}{\env}{\phi_1} \cup \efoot{\heap}{\env}{\phi_2}$ by definition.
    
          Now $\universal{\pair{\ell}{f} \in \efoot{\heap}{\env}{\phi_1}}{\heap'(\ell, f) = \heap(\ell, f)}$ since $\efoot{\heap}{\perms}{\phi_1} \subseteq \efoot{\heap}{\perms}{\gform}$. Therefore \\
          $\efrm{\heap'}{\perms}{\env}{\phi_1}$ by induction. Similarly, $\efrm{\heap'}{\perms}{\env}{\phi_2}$.
    
          Therefore $\efrm{\heap'}{\perms}{\env}{\phi_1 * \phi_2}$ by \refrule{EFrameConjunction}.
    
        \case \refrule{EFramePredicate} -- $\efrm{\heap}{\perms}{\env}{p(\multiple{e})}$:
    
          By inversion of \refrule{EFramePredicate} $\multiple{\frm{\heap}{\perms}{\env}{e}}$, $\multiple{\eval{\heap}{\env}{e}{v}}$, and $\efrm{\heap}{\perms}{[\multiple{x \mapsto v}]}{\fpred(p)}$, where $\multiple{x} = \fpredparams(p)$.
    
          Now $\efoot{\heap}{\env}{p(\multiple{e})} = \efoot{\heap}{[\multiple{x \mapsto v}]}{\fpred(p)} \cup \bigcup \multiple{\efoot{\heap}{\env}{e}}$.
    
          Then, for each $e$ and corresponding $x$, $\universal{\pair{\ell}{f} \in \efoot{\heap}{\env}{e}}{\heap'(\ell, f) = \heap(\ell, f)}$ since $\efoot{\heap}{\env}{e} \subseteq \efoot{\heap}{\env}{p(\multiple{e})}$. Therefore $\eval{\heap'}{\env}{e}{v}$ by lemma \ref{lem:eval-heap-efoot-unchanged} and $\frm{\heap'}{\perms}{\env}{e}$ by \ref{lem:frm-heap-efoot-unchanged}.
    
          Also, $\universal{\pair{\ell}{f} \in \efoot{\heap}{[\multiple{x \mapsto v}]}{\fpred(p)}}{\heap'(\ell, f) = \heap(\ell, f)}$, thus by induction \\
          $\efrm{\heap'}{\perms}{[\multiple{x \mapsto v}]}{\fpred(p)}$.
    
          Therefore $\efrm{\heap'}{\perms}{\env}{p(\multiple{e})}$ by \refrule{EFramePredicate}.
    
        \case\label{case:efrm-heap-efoot-unchanged-ifa} \refrule{EFrameConditionalA} -- $\efrm{\heap}{\perms}{\env}{\sif{e}{\phi_1}{\phi_2}}$:
    
          By inversion of \refrule{EFrameConditionalA} $\eval{\heap}{\env}{e}{\ktrue}$, $\frm{\heap}{\perms}{\env}{e}$, and $\efrm{\heap}{\perms}{\env}{\phi_1}$.
    
          Now $\efoot{\heap}{\env}{\sif{e}{\phi_1}{\phi_2}} = \efoot{\heap}{\env}{e} \cup \efoot{\heap}{\env}{\phi_1}$.
          
          Then $\universal{\pair{\ell}{f} \in \efoot{\heap}{\env}{e}}{\heap'(\ell, f) = \heap(\ell, f)}$ since $\efoot{\heap}{\env}{e} \subseteq \efoot{\heap}{\env}{\sif{e}{\phi_1}{\phi_2}}$. Therefore $\eval{\heap'}{\env}{e}{\ktrue}$ by lemma \ref{lem:eval-heap-efoot-unchanged} and $\frm{\heap'}{\perms}{\env}{e}$ by lemma \ref{lem:frm-heap-efoot-unchanged}.
    
          Also, Then $\universal{\pair{\ell}{f} \in \efoot{\heap}{\env}{\phi_1}}{\heap'(\ell, f) = \heap(\ell, f)}$ since \\
          $\efoot{\heap}{\env}{\phi_1} \subseteq \efoot{\heap}{\env}{\sif{e}{\phi_1}{\phi_2}}$. Therefore $\efrm{\heap'}{\perms}{\env}{\phi_1}$ by induction.
    
          Therefore $\efrm{\heap'}{\perms}{\env}{\sif{e}{\phi_1}{\phi_2}}$ by \refrule{EFrameConditionalA}.
    
        \case \refrule{EFrameConditionalB} -- $\efrm{\heap}{\perms}{\env}{\sif{e}{\phi_1}{\phi_2}}$: Similar to case \ref{case:efrm-heap-efoot-unchanged-ifa}.
    
        \case \refrule{EFrameAcc} -- $\efrm{\heap}{\perms}{\env}{\kacc(e.f)}$:
    
          By inversion of \refrule{EFrameAcc} $\frm{\heap}{\perms}{\env}{e}$.
    
          Also, $\efoot{\heap}{\env}{e} \subseteq \efoot{\heap}{\env}{\kacc(e.f)}$ by definition, thus $\universal{\pair{\ell}{f} \in \efoot{\heap}{\env}{e}}{\heap'(\ell, f) = \heap(\ell, f)}$. Therefore $\frm{\heap'}{\perms}{\env}{e}$ by lemma \ref{lem:frm-heap-efoot-unchanged}.
    
          Thus $\efrm{\heap'}{\perms}{\env}{\kacc(e.f)}$ by \refrule{EFrameAcc}.

        \definecolor{shadecolor}{gray}{0.9}
        \begin{snugshade} \case \refrule{FrameUnfolding} -- By definition, $\ifrm{\heap}{\perms}{\env}{\pe}$, $\assertion{\heap}{\perms}{\env}{\pe}$, and $\frm{\heap}{\perms}{\env}{e_0}$. By the induction hypothesis, $\frm{\heap'}{\perms}{\env}{e_0}$. By \refrule{IFramePredicate}, $\multiple{\frm{\heap}{\perms}{\env}{e}}$, so by the induction hypothesis, $\multiple{\frm{\heap'}{\perms}{\env}{e}}$ and hence $\ifrm{\heap'}{\perms}{\env}{\pe}$. By \refrule{AssertPredicate},
        $\multiple{\eval{\heap}{\env}{e}{v}}$ and $\assertion{\heap}{\perms}{[\multiple{x \mapsto v}]}{\fpred(p)}$ where $\multiple{x} = \fpredparams(p)$. By Lemma~\ref{lem:frm-heap-efoot-unchanged}, $\multiple{\eval{\heap'}{\env}{e}{v}}$. Then, since $p$ is a predicate, its body must be a specification so one of the following cases applies.

            \subcase $\fpred(p)$ is precise and self-framing -- by Lemma~\ref{lem:assert-heap-efoot-unchanged}, $\assertion{\heap'}{\perms}{\env}{\fpred(p)}$.
            
            \subcase $\fpred(p) = \simprecise{\phi}$ -- By \refrule{AssertImprecise}, $\assertion{\heap}{\perms}{\env}{\phi}$ and $\efrm{\heap}{\perms}{[\multiple{x \mapsto v}]}{\phi}$. By Lemma~\ref{lem:assert-heap-efoot-unchanged}, $\assertion{\heap'}{\perms}{\env}{\phi}$, and by Lemma~\ref{lem:efrm-heap-efoot-unchanged}, $\efrm{\heap'}{\perms}{[\multiple{x \mapsto v}]}{\phi}$. By \refrule{AssertImprecise}, $\assertion{\heap'}{\perms}{\env}{\fpred(p)}$.

        In either case, $\assertion{\heap'}{\perms}{\env}{\pe}$ by \refrule{AssertPredicate}, so $\frm{\heap'}{\perms}{\env}{\sunfolding{\pe}{e_0}}$ as desired.

        \case \refrule{FrameFunction} -- By definition, $\ifrm{\heap}{\perms}{\env}{\ffuncpre(f)}$, $\assertion{\heap}{\perms}{\env}{\ffuncpre(f)}$, and $\multiple{\frm{\heap}{\perms}{\env}{e}}$. By the induction hypothesis, $\multiple{\frm{\heap'}{\perms}{\env}{e}}$ and $\ifrm{\heap'}{\perms}{\env}{\ffuncpre(f)}$. By Lemma~\ref{lem:assert-heap-efoot-unchanged}, $\assertion{\heap'}{\perms}{\env}{\ffuncpre(f)}$, so $\frm{\heap'}{\perms}{\env}{\fe}$.
        \end{snugshade}
    \end{enumcases}
\end{proof}

\begin{lemma}\label{lem:efrm-heap-efoot-unchanged}
  If $\universal{\pair{\ell}{f} \in \efoot{\heap}{\perms}{\gform}}{\heap'(\ell, f) = \heap(\ell, f)}$ and $\efrm{\heap}{\perms}{\env}{\gform}$ then $\efrm{\heap'}{\perms}{\env}{\gform}$.
\end{lemma}

\begin{lemma}\label{lem:assert-heap-efoot-unchanged}
  If $\assertion{\heap}{\perms}{\env}{\gform}$ and $\universal{\pair{\ell}{f} \in \efoot{\heap}{\env}{\gform}}{\heap'(\ell, f) = \heap(\ell, f)}$ then $\assertion{\heap'}{\perms}{\env}{\gform}$.
\end{lemma}

\begin{lemma}\label{lem:assert-heap-perms-unchanged}
  If $\assertion{\heap}{\perms}{\env}{\gform}$ for some specification $\gform$ and $\universal{\pair{\ell}{f} \in \perms}{\heap'(\ell, f) = \heap(\ell, f)}$ then $\assertion{\heap'}{\perms}{\env}{\gform}$.
\end{lemma}

\begin{lemma}\label{lem:pres-modify-heap}
  If $\simstate{V}{\sstate}{\heap}{\perms}{\env}$ and $\universal{\pair{\ell}{f} \in \perms}{\heap'(\ell, f) = \heap(\ell, f)}$ then $\simstate{V}{\sstate}{\heap'}{\perms}{\env}$.
\end{lemma}

\begin{proof}
    From the assumptions, clearly $\simenv{V}{\senv(\sstate)}{\env}$. It suffices to show that $\simheap{V}{\sheap(\sstate)}{\heap'}{\perms}$ and $\simheap{V}{\oheap(\sstate)}{\heap'}{\perms}$.

  Let $\sheap = \sheap(\sstate)$. From assumptions, clearly $\simheap{V}{\sheap}{\heap}{\perms}$. Therefore
  \begin{gather}
    \universal{\triple{f}{t}{t'} \in \sheap}{\heap(V(t), f) = V(t')} \label{eq:pres-modify-heap-1}\\
    \universal{\triple{f}{t}{t'} \in \sheap}{\pair{V(t)}{f} \in \perms} \label{eq:pres-modify-heap-2}\\
    \universal{\pair{p}{\multiple{t}} \in \sheap}{\assertion{\heap}{\perms}{[\multiple{x \mapsto V(t)}]}{\fpred(p)}} \label{eq:pres-modify-heap-3} \\
    \universal{h_1, h_2 \in \sheap^2}{h_1 \ne h_2 \implies \vfoot{V}{\heap}{h_1} \cap \vfoot{V}{\heap}{h_2} = \emptyset} \label{eq:pres-modify-heap-4}
  \end{gather}

  Let $\triple{f}{t}{t'}$ be an arbitrary field value in $\sheap$. Then by \eqref{eq:pres-modify-heap-1} $\heap(V(t), f) = V(t')$. Also, by \eqref{eq:pres-modify-heap-2} $\pair{V(t)}{f} \in \perms$, thus by our initial assumptions $\heap'(V(t), f) = \heap(V(t), f) = V(t')$. Therefore
  \begin{equation}\label{eq:pres-modify-heap-5}
    \universal{\triple{f}{t}{t'} \in \sheap}{\heap'(V(t), f) = V(t')}.
  \end{equation}

  let $\pair{p}{\multiple{t}}$ be an arbitrary predicate instance in $\sheap$. Then by \eqref{eq:pres-modify-heap-3} $\assertion{\heap}{\perms}{[\multiple{x \mapsto V(t)}]}{\fpred(p)}$. Since $\fpred(p)$ is a specification, and $\universal{\pair{\ell}{f} \in \perms}{\heap'(\ell, f) = \heap(\ell, f)}$, then \\
  $\assertion{\heap'}{\perms}{[\multiple{x \mapsto V'(t)}]}{\fpred(p)}$ by lemma \ref{lem:assert-heap-perms-unchanged}. Therefore
  \begin{equation}\label{eq:pres-modify-heap-6}
    \universal{\pair{p}{\multiple{t}} \in \sheap}{\assertion{\heap'}{\perms}{[\multiple{x \mapsto V'(t)}]}{\fpred(p)}}.
  \end{equation}

  Therefore by \eqref{eq:pres-modify-heap-5}, \eqref{eq:pres-modify-heap-2}, \eqref{eq:pres-modify-heap-6}, and \eqref{eq:pres-modify-heap-4}, $\simheap{V}{\sheap}{\heap'}{\perms}$.

\definecolor{shadecolor}{gray}{0.9}
\begin{snugshade}
    Let $\oheap = \oheap(\sstate)$ and let $\triple{f}{t}{t'} \in \oheap$ be an arbitrary field chunk. Then, $\heap(V(t), f) = V(t')$ and $\pair{V(t)}{f} \in \perms$ by definition of heap correspondence. By the assumptions, $\heap'(V(t), f) = V(t)$, so

    $$
        \forall \triple{f}{t}{t'} \in \oheap : \heap'(V(t), f) = V(t'),
    $$
    $$
        \forall \triple{f}{t}{t'} \in \oheap : \pair{V(t)}{f} \in \perms.
    $$

    Then, let $\pair{p}{\multiple{t}}$ be an arbitrary predicate chunk in $\oheap$. Then, by definition of heap correspondence, $\assertion{\heap}{\perms}{[\multiple{x \mapsto V'(t)}]}{\fpred(p)}$. Then, since $\fpred(p)$ is a specification and $H'$ coincides with $H$, $\assertion{\heap'}{\perms}{[\multiple{x \mapsto V'(t)}]}{\fpred(p)}$ by Lemma~\ref{lem:pres-heap-change-partial}.
\end{snugshade}
\end{proof}

\begin{lemma}\label{lem:pres-heap-change-partial}
  If $\pair{\heap}{\triple{\perms}{\env}{s} \cdot \stack}$ is a well-formed state, the partial state $\pair{\heap}{\stack}$ is validated by $\vstate$ and $V$, and $\heap' = \heap[\pair{\ell}{f} \mapsto v]$ for some $\ell$, $f$, and $v$ such that $\pair{\ell}{f} \in \perms$ or $\ell$ is a fresh value unused in $\stack$, then the partial state $\pair{\heap'}{\stack}$ is validated by $\vstate$ and $V$.
\end{lemma}

\begin{lemma}\label{lem:rem-subset}
  If $\simstate{V}{\sstate}{\heap}{\perms}{\env}$ then $\vfoot{V}{\heap}{\frem(\sstate, \gform)} \subseteq \perms$.
\end{lemma}

\begin{lemma}\label{lem:rem-simstate}
  Let $\gform$ be some specification, $\xperms = \vfoot{V}{\heap}{\frem(\sstate, \gform)}$ and $\perms' = \foot{\heap}{\perms \setminus \xperms}{\env}{\gform}$.

  If $\simstate{V}{\sstate}{\heap}{\perms \setminus \efoot{\heap}{\env}{\gform}}{\env}$ and $\assertion{\heap}{\perms \setminus \xperms}{\env}{\gform}$, then $\simstate{V}{\sstate}{\heap}{\perms \setminus \perms'}{\env}$.
\end{lemma}

\begin{lemma}\label{lem:dexec-preservation}
  Let $\pair{\heap}{\stack}$ be some dynamic state validated by $\vstate$ and valuation $V$ for some program $\prog$. If $\sguard{\vstate}{\sstate'}{\scheck}{\sperms}$ with $V' = V[\sguard{\vstate}{\sstate'}{\scheck}{\sperms} \mid \heap]$, $V'(\pc(\sstate')) = \ktrue$, $\rtassert{V'}{\heap}{\perms(\stack)}{\scheck}$, and $\dexec{\heap}{\stack}{\vfoot{V'}{\heap}{\sperms}}{\heap'}{\stack'}$
  then $\Gamma' = \pair{\heap'}{\stack'}$ is a valid state.

  Note: This is a simplification of theorem \ref{thm:dtrans-preservation}, restricted to the particular case of a normal program step (in contrast to a step from $\initsym$ or to $\finalsym$).
\end{lemma}

\begin{proof}
  We proceed by cases on $\dexec{\heap}{\stack}{\vfoot{V'}{\heap}{\sperms}}{\heap'}{\stack'}$.

  \begin{enumcases}
    \case \refrule{ExecSeq}:
      We have
      $$\dexec{\heap}{\triple{\perms}{\env}{\sseq{\kskip}{s}} \cdot \stack^*}{\vfoot{V'}{\heap}{\sperms}}{\heap}{\triple{\perms}{\env}{s} \cdot \stack^*}.$$

      Since the initial state is validated by $\vstate$ and $V$, $\vstate = \triple{\sstate}{\sseq{\kskip}{s}}{\gform}$ for some $\sstate, \gform$, where $\simstate{V}{\sstate}{\heap}{\perms}{\env}$.
      
      Now $\sexec{\sstate(\vstate)}{\sseq{\kskip}{s}}{s}{\sstate}$ by \refrule{SExecSeq}. Therefore $\strans{\prog}{\vstate}{\vstate'}$ by \refrule{SVerifyStep} where $\vstate' = \triple{\sstate}{s}{\gform}$.

      Now, by lemma \ref{lem:preservation-heap-env-unchanged}, it suffices to show that $\vstate'$ corresponds to $\pair{\heap}{\triple{\perms}{\env}{s} \cdot \stack^*}$ (with valuation $V$).

      Since $\sstate(\vstate') = \sstate$, we have $\simstate{V}{\sstate(\vstate')}{\heap}{\perms}{\env}$, and $s(\vstate') = s$ by definition. Therefore $\vstate'$ corresponds to $\pair{\heap}{\triple{\perms}{\env}{s} \cdot \stack^*}$, which completes the proof.

    \case \refrule{ExecAssign}:
      We have
      \begin{gather*}
        \dexec{\heap}{\triple{\perms}{\env}{\sseq{x = e}{s}} \cdot \stack^*}{\vfoot{V'}{\heap}{\sperms}}{\heap}{\triple{\perms}{\env[x \mapsto v]}{s} \cdot \stack^*} \\
        \text{where}\quad \eval{\heap}{\env}{e}{v}
      \end{gather*}

      Since the initial state is validated by $\vstate$, $\vstate = \triple{\sstate}{\sseq{x = e}{s}}{\gform}$ for some $\sstate, \gform$ where $\simstate{V}{\sstate}{\heap}{\perms}{\env}$.

      The only guard that applies is \refrule{SGuardAssign}, so we have, for some $\sstate', t$:
      \begin{gather*}
        \sguard{\triple{\sstate}{\sseq{x = e}{s}}{\gform}}{\sstate}{\scheck}{\sperms} \\
        \text{where }
        \seval{\sstate}{e}{t}{\sstate'}{\scheck} \quad\text{and (by assumptions) }\rtassert{V'}{\heap}{\env}{\scheck}.
      \end{gather*}
      where $V'$ is the corresponding valuation, extending $V$.

      Let $\sstate'' = \sstate[\senv = \senv(\sstate)[x \mapsto t]]$, then $\sexec{\sstate}{\sseq{x = e}{s}}{s}{\sstate''}$ by \refrule{SExecAssign}.
      
      Let $\vstate' = \triple{\sstate''}{s}{\gform}$. Then $\strans{\prog}{\vstate}{\vstate'}$ by \refrule{SVerifyStep}, thus $\vstate'$ is reachable from $\prog$.

      We want to show that $\pair{\heap}{\triple{\perms}{\env[x \mapsto v]}{s} \cdot \stack^*}$ is validated by $\vstate'$ with valuation $V'$.

      \textit{Part \ref{def:state-valid-reachable}:} As shown above, $\vstate'$ is reachable from $\prog$ with valuation $V'$.

      \textit{Part \ref{def:state-valid-correspond}:}
      By lemma \ref{lem:seval-soundness} $\eval{\heap}{\env}{e}{V'(t)}$, therefore $V'(t) = v$. Also, $\simenv{V}{\senv(\sstate)}{\env}$, thus $\simenv{V'}{\senv(\sstate)[x \mapsto V'(t)]}{\env[x \mapsto v]}$. Rewriting using definitions, we get $\simenv{V'}{\senv(\sstate'')}{\env[x \mapsto v]}$.

      Also by lemma \ref{lem:seval-soundness}, $\simstate{V'}{\sstate'}{\heap}{\perms}{\env}$. Thus, since $\senv$ is the only component changed from $\sstate'$ to $\sstate''$ and $\simenv{V'}{\senv(\sstate'')}{\env[x \mapsto v]}$, $\simstate{V'}{\sstate''}{\heap}{\perms}{\env[x \mapsto v]}$.

      Also, by definition $\vstate' = \triple{\sstate''}{s}{\gform}$.

      Therefore $\vstate'$ corresponds to $\pair{\heap}{\triple{\perms}{\env[x \mapsto v]}{s} \cdot \stack^*}$.

      \textit{Part \ref{def:state-valid-partial}:}
      Since the initial state is validated by $\vstate$ and $V$, the partial state $\pair{\heap}{\stack^*}$ is validated by $\vstate$ and valuation $V$. Therefore one of \ref{def:partial-valid-nil}, \ref{def:partial-valid-call}, \ref{def:partial-valid-while} applies. We want to show that the partial state $\pair{\heap}{\stack^*}$ is validated by $\vstate'$ and valuation $V'$.

      \begin{itemize}
        \item \textit{Case \ref{def:partial-valid-nil}:} Then $\stack^* = \nilsym$ and trivially $\pair{\heap}{\nilsym}$ is validated by $\vstate'$ and valuation $V'$.
        
        \item \textit{Case \ref{def:partial-valid-call}:} Then $\stack^* = \triple{\env_0}{\perms_0}{\sseq{y \kassign m(e_1, \cdots, e_k)}{s_0}} \cdot \stack_0$ for some $\env_0$, $\perms_0$, $y$, $m$, $k$, $e_1, \cdots, e_k$, $s_0$, $\stack_0$. Also, there is some $\vstate_0$, $V_0$, $x_1, \cdots, x_k$, $t_1, \cdots, t_k$ and $\sstate_0, \cdots, \sstate_k, \sstate'$ such that
        \begin{gather*}
          \text{The partial state $\pair{\heap}{\stack_0}$ is validated by $\vstate_0$ and $V_0$}, \\
          \vstate_0 \text{ is reachable from $\prog$ with valuation $V_0$}, \quad s(\vstate_0) = s(\stack^*) \\
          x_1, \cdots, x_k = \fparams(m), \\
          \sstate_0 = \sstate(\vstate_0), \quad \seval{\sstate_0}{e_1}{t_1}{\sstate_1}{\_}, \quad\cdots,\quad \seval{\sstate_{k-1}}{e_k}{t_k}{\sstate_k}{\_}, \\
          \scons{\sstate_k}{\fpre(m)}{\sstate'}{\_}, \quad \simstate{V_0}{\sstate'[\senv = \senv(\sstate_0)]}{\heap}{\perms_0}{\env_0}, \quad\text{and} \\
          \universal{1 \le i \le k}{V(\senv(\vstate)(x_i)) = V_0(t_i)}, \\
          \gform(\vstate) = \fpost(m).
        \end{gather*}
  
        We want to show that the partial state $\pair{\heap}{\triple{\env_0}{\perms_0}{\sseq{y \kassign m(e_1, \cdots, e_k)}{s}} \cdot \stack_0}$ is validated by $\vstate'$ and valuation $V'$. Immediately from above we can conclude that
        \begin{gather*}
          \text{The partial state $\pair{\heap}{\stack_0}$ is validated by $\vstate_0$ and $V_0$}, \\
          \vstate_0 \text{ is reachable from $\prog$ with valuation $V_0$}, \quad s(\vstate_0) = s(\stack^*) \\
          x_1, \cdots, x_k = \fparams(m), \\
          \sstate_0 = \sstate(\vstate_0), \quad \seval{\sstate_0}{e_1}{t_1}{\sstate_1}{\_}, \quad\cdots,\quad \seval{\sstate_{k-1}}{e_k}{t_k}{\sstate_k}{\_}, \quad\text{and} \\
          \scons{\sstate_k}{\fpre(m)}{\sstate'}{\_}, \quad \simstate{V_0}{\sstate'[\senv = \senv(\sstate_0)]}{\heap}{\perms_0}{\env_0}.
        \end{gather*}
        Also, the frame $\triple{\perms}{\env}{\sseq{x = e}{s}}$ must be executing the body of $m$, since it is in the stack immediately above the frame that contains $y \kassign m(e_1, \cdots, e_k)$. Therefore, since $x_1, \cdots, x_k$ are all parameters of $m$, $y$ must be distinct from all of $x_1, \cdots, x_k$, since we do not allow assignment to parameters in a well-formed program. Thus
        \begin{align*}
          \forall 1 \le i \le k : V'(\senv(\vstate')(x_i))
            &= V'((\senv(\sstate)[x \mapsto t])(x_i)) = V'(\senv(\sstate)(x_i)) \\
            &= V'(\senv(\vstate)(x_i)) = V(\senv(\vstate)(x_i)) \\
            &= V_0(t_i).
        \end{align*}
  
        Finally, $\gform(\vstate') = \gform(\vstate)$ by definition, thus
        $$\gform(\vstate') = \gform(\vstate) = \fpost(m).$$
  
        Therefore the partial state $\pair{\heap}{\stack^*}$ is validated by $\vstate'$ and $V'$ in this case.
        
        \item \textit{Case \ref{def:partial-valid-while}:}
        Then $\stack^* = \triple{\env_0}{\perms_0}{\sseq{\swhile{e_0}{\gform_0}{s_0}}{s_0'}} \cdot \stack_0$ for some $\env_0$, $\perms_0$, $e$, $\gform_0$, $s_0$, $s_0'$, $\stack_0$, and there exists some $\vstate_0$, $V_0$, and $\sstate_0'$ such that:
        \begin{gather*}
          \text{The partial state $\pair{\heap}{\stack_0}$ is validated by $\vstate_0$ and $V_0$} \\
          \vstate_0 \text{ is reachable from $\prog$ with valuation $V_0$} \quad s(\vstate_0) = s(\stack^*) \\
          \scons{\sstate_0}{\gform_0}{\sstate_0'}{\_}, \quad
          \simstate{V_0}{\sstate_0'}{\heap}{\perms_0}{\env_0} \quad\text{and}\\
          \gform(\vstate) = \gform_0.
        \end{gather*}
  
        Now, by definition of $\vstate''$, $\gform(\vstate') = \gform(\vstate) = \gform_0$. Therefore, using the other assumptions given above, the partial state $\pair{\heap}{\stack^*}$ is validated by $\vstate'$ and $V'$ in this case.
      \end{itemize}

      Therefore definition part \ref{def:state-valid-partial} is satisfied.

      Therefore all parts of definition \ref{def:state-valid} are satisfied. Thus $\pair{\heap}{\triple{\perms}{\env[x \mapsto v]}{s} \cdot \stack^*}$ is validated by $\vstate'$ with valuation $V'$.

    \case \refrule{ExecAssignField}:

      We have
      \begin{align}
        &\dexec{\heap}{\triple{\perms}{\env}{\sseq{x.f = e}{s}} \cdot \stack^*}{\vfoot{V'}{\heap}{\sperms}}{\heap'}{\triple{\perms}{\env}{s} \cdot \stack^*} \\
        &\text{where}\quad \eval{\heap}{\env}{x}{\ell}, \quad \eval{\heap}{\env}{e}{v}, \quad \assertion{\heap}{\perms}{\env}{\kacc(x.f)}, \label{eq:dexec-pres-assign-field-eval-assert} \\
        &\hspace{3.5em} \frm{\heap}{\perms}{\env}{e}, \quad\text{and}\quad \heap' = \heap[\pair{\ell}{f} \mapsto v].
      \end{align}
      Since the initial state is validated by $\vstate$, $\vstate = \triple{\sstate}{\sseq{x.f = e}{s}}{\gform}$ for some $\sstate, \gform$ where $\simstate{V}{\sstate}{\heap}{\perms}{\env}$.

      The only guard rule that applies is \refrule{SGuardAssign}, so we have
      \begin{align}
        &\sguard{\triple{\sstate}{\sseq{x.f = e}{s}}{\gform}}{\sstate}{\scheck' \cup \scheck''}{\emptyset} \\
        &\text{where}\quad \seval{\sstate}{e}{t}{\sstate'}{\scheck'}, \quad \scons{\sstate'}{\kacc(x.f)}{\sstate''}{\scheck''} \\
        &\text{and (by assumptions)}\quad \rtassert{V'}{\heap}{\perms}{\scheck' \cup \scheck''}, \quad V'(\pc(\sstate'')) = \ktrue
      \end{align}

      Furthermore, since $\pc(\sstate'') \implies \sstate'$ by lemma \ref{lem:cons-subpath}, and by lemma \ref{lem:scheck-monotonicity},
      $$V'(\pc(\sstate')) = \ktrue, \quad \rtassert{V'}{\heap}{\perms}{\scheck'}, \quad\text{and}\quad \rtassert{V'}{\heap}{\perms}{\scheck''}.$$

      Now, by \refrule{SExecAssignField},
      \begin{align*}
        \sexec{\sstate}{\sseq{x.f = e}{s}}{s}{\sstate'''}
        \quad\text{where}\quad
        &\sstate''' = \sstate''[\sheap = \sheap'], \quad\text{and}\\
        &\sheap' = \sheap(\sstate''); \triple{\senv(\sstate'')(x)}{f}{t}.
      \end{align*}

      Let $\vstate' = \triple{\sstate'''}{s}{\gform}$.  We want to show that $\pair{\heap'}{\triple{\perms}{\env}{s} \cdot \stack^*}$ is validated by $\vstate'$ and $V'$.

      \textit{Part \ref{def:state-valid-reachable}:}
      By \refrule{SVerifyStep}, $\strans{\prog}{\vstate}{\vstate'}$. Therefore $\vstate'$ is reachable from program $\prog$ with valuation $V'$.

      \textit{Part \ref{def:state-valid-correspond}:}
      We want to show that $\vstate'$ corresponds to $\pair{\heap'}{\triple{\perms}{\env}{s} \cdot \stack^*}$. Since $\vstate' = \triple{\sstate'''}{s}{\_}$ by construction, it suffices to show that $\simstate{V'}{\sstate'''}{\heap'}{\perms}{\env}$.

      By lemma \ref{lem:seval-soundness}, $\simstate{V'}{\sstate'}{\heap}{\perms}{\env}$. By lemma \ref{lem:cons-soundness}, $\simstate{V'}{\sstate''}{\heap}{\perms \setminus \efoot{\heap}{\env}{\kacc(x.f)}}{\env}$.

      Since $\simenv{V}{\senv(\sstate)}{\env}$, $\env(x) = V(\senv(\sstate)(x)) = V'(\senv(\sstate)(x))$. Also, $\eval{\heap}{\env}{x}{\env(x)}$ by \refrule{EvalVar} and $\eval{\heap}{\env}{x}{\ell}$ by \eqref{eq:dexec-pres-assign-field-eval-assert}, thus
      $$V'(\senv(\sstate)(x)) = \env(x) = \ell.$$
      Also, $\senv(\sstate'') = \senv(\sstate)$ by lemmas \ref{lem:eval-unchanged} and \ref{lem:cons-unchanged}. Thus $\eval{\heap}{\env}{x}{V'(\senv(\sstate'')(x))}$. Therefore \\
      $\efoot{\heap}{\env}{\kacc(x.f)} = \set{\pair{\ell}{f}} = \set{ \pair{V'(\senv(\sstate'')(x))}{f} } = \vfoot{V'}{\heap}{\triple{\senv(\sstate'')(x)}{f}{t'}}$. Now,
      $$\simstate{V'}{\sstate''}{\heap}{\perms \setminus \vfoot{V'}{\heap}{\triple{\senv(\sstate'')(x)}{f}{t'}}}{\env}.$$

      Also, $v = V'(t)$ by lemma \ref{lem:seval-soundness}, thus $\heap'(V'(\senv(\sstate'')(x)), f) = \heap'(\ell, f) = v = V'(t)$. Now by lemma \ref{lem:pres-add-heap},
      $$\simstate{V'}{\sstate'''}{\heap'}{\perms}{\env}$$
      which is sufficient to show that $\vstate'$ corresponds to $\pair{\heap'}{\triple{\perms}{\env}{s} \cdot \stack^*}$.

      \textit{Part \ref{def:state-valid-partial}:}

      By \eqref{eq:dexec-pres-assign-field-eval-assert} $\assertion{\heap}{\perms}{\env}{\kacc(x.f)}$. This assertion must be given by \refrule{AssertAcc}, therefore $\pair{\ell}{f} \in \perms$.

      Now we show by induction that all partial states are validated, in other words we want to show that the partial state $\pair{\heap'}{\stack^*}$ is validated by $\vstate'$. By assumptions, $\pair{\heap}{\stack^*}$ is validated by $\vstate$ with valuation $V$. Also, by lemma \ref{lem:pres-heap-change-partial}, $\pair{\heap'}{\stack^*}$ is validated by $\vstate$ with $V$.
      
      Thus, one of the following cases apply:

      \begin{itemize}
        \item \textit{Case \ref{def:partial-valid-nil}:} Then $\stack^* = \nilsym$, and trivially $\pair{\heap'}{\nilsym}$ is validated by $\vstate'$ and $V$.
        
        \item \textit{Case \ref{def:partial-valid-call}:}
        Then $\stack^* = \triple{\env}{\perms}{\sseq{y \kassign m(e_1, \cdots, e_k)}{s}} \cdot \stack_0$ for some $\env$, $\perms$, $y$, $m$, $e_1, \cdots, e_k$, and there exists some $\vstate_0$, $V_0$, $x_1, \cdots, x_k$, $t_1, \cdots, t_k$, $\sstate_0, \cdots, \sstate_k$, and $\sstate_0'$ such that:
        \begin{gather*}
          \text{The partial state $\pair{\heap'}{\stack_0}$ is validated by $\vstate_0$ and $V_0$},\\
          \vstate_0 \text{ is reachable from $\prog$ with valuation $V_0$}, \quad s(\vstate_0) = s(\stack^*)\\
          x_1, \cdots, x_k = \fparams(m), \\
          \sstate_0 = \sstate(\vstate_0), \quad \seval{\sstate_0}{e_1}{t_1}{\sstate_1}{\_}, \quad\cdots,\quad \seval{\sstate_{k-1}}{e_k}{t_k}{\sstate_k}{\_}, \\
          \universal{1 \le i \le k}{V(\senv(\vstate)(x_i)) = V_0(t_i)}, \\
          \scons{\sstate_k[\senv = [x_1 \mapsto t_1, \cdots, x_k \mapsto t_k]]}{\fpre(m)}{\sstate_0'}{\_}, \\
          \simstate{V_0}{\sstate_0'[\senv = \senv(\sstate_0)]}{\heap'}{\perms}{\env}, \quad\text{and} \\
          \gform(\vstate) = \fpost(m).
        \end{gather*}

        We want to show that the partial state $\pair{\heap'}{\triple{\env}{\perms}{\sseq{y \kassign m(e_1, \cdots, e_k)}{s}} \cdot \stack_0}$ is validated by $\vstate'$. From above,
        \begin{gather*}
          \text{The partial state $\pair{\heap'}{\stack_0}$ is validated by $\vstate_0$ and $V_0$}, \\
          \vstate_0 \text{ is reachable from $\prog$ with valuation $V_0$}, \quad s(\vstate_0) = s(\stack^*)\\
          x_1, \cdots, x_k = \fparams(m), \\
          \sstate_0 = \sstate(\vstate_0), \quad \seval{\sstate_0}{e_1}{t_1}{\sstate_1}{\_}, \quad\cdots,\quad \seval{\sstate_{k-1}}{e_k}{t_k}{\sstate_k}{\_}, \\
          \scons{\sstate_k[\senv = [x_1 \mapsto t_1, \cdots, x_k \mapsto t_k]]}{\fpre(m)}{\sstate_0'}{\_}, \quad\text{and}\\
          \simstate{V_0}{\sstate_0'[\senv = \senv(\sstate_0)]}{\heap'}{\perms}{\env}.
        \end{gather*}

        $\senv(\vstate') = \senv(\sstate''') = \senv(\sstate'')$ by definition. Also, $\senv(\sstate'') = \senv(\sstate') = \senv(\sstate')$ by lemmas \ref{lem:eval-unchanged} and \ref{lem:cons-unchanged}. Also, $V'$ extends $V$. Thus
        $$\universal{1 \le i \le k}{V'(\senv(\vstate')(x_i)) = V(\senv(\vstate)(x_i)) = V_0(t_i)}.$$
        Also, by definition
        $$\gform(\vstate') = \gform(\vstate) = \fpost(m).$$
        Therefore the partial state $\pair{\heap'}{\stack^*}$ is validated by $\vstate'$.

        \item \textit{Case \ref{def:partial-valid-while}:}
        Then $\stack^* = \triple{\env}{\perms}{\sseq{\swhile{e}{\gform}{s}}{s'}} \cdot \stack_0$ for some $\env$, $\perms$, $e$, $\gform$, $s$, $s'$, and $\stack_0$, and there exists some $\vstate_0$, $V_0$, and $\sstate_0'$ such that:
        \begin{gather*}
          \text{The partial state $\pair{\heap'}{\stack_0}$ is validated by $\vstate$ and $V$} \\
          \vstate_0 \text{ is reachable from $\prog$ with valuation $V_0$} \quad s(\vstate_0) = s(\stack^*) \\
          \scons{\sstate(\vstate_0)}{\gform}{\sstate_0'}{\_}, \quad\text{and}\quad
          \simstate{V_0}{\sstate_0'}{\heap'}{\perms}{\env} \\
          \gform(\vstate) = \gform
        \end{gather*}

        Now, $\gform(\vstate') = \gform(\vstate)$ by definition, thus $\gform(\vstate') = \gform$. Therefore the partial state $\pair{\heap'}{\stack^*}$ is validated by $\vstate'$.
      \end{itemize}
      Now we have shown that $\vstate'$ corresponds to the resulting state, and therefore $\gamma'$ is validated by $\vstate'$.

    \case \refrule{ExecAlloc}:
      We have
      \begin{align*}
        &\dexec{\heap}{\triple{\perms}{\env}{\sseq{x = \kalloc(S)}{s}} \cdot \stack^*}{\vfoot{V'}{\heap}{\sperms}}{\heap'}{\triple{\perms'}{\env[x \mapsto \ell]}{s} \cdot \stack^*} \\
        &\text{where} \quad \multiple{T~f} = \fstruct(S), \quad \ell = \ffresh, \quad \heap' = \heap[\multiple{\pair{\ell}{f} \mapsto \fdefault(T)}], \\
        &\hspace{3.5em} \perms' = \perms \cup \set{\multiple{\pair{\ell}{f}}}
      \end{align*}

      Since the initial state is validated by $\vstate$, $\vstate = \triple{\sstate}{\sseq{x = \kalloc}{s}}{\gform}$ for some $\sstate, \gform$ where $\simstate{V}{\sstate}{\heap}{\perms}{\env}$.

      By \refrule{SExecAlloc} \begin{align*}
        \sexec{\sstate}{\sseq{x = \kalloc(S)}{s}}{s}{\sstate'}, ~\text{where}~
        &\sheap(\sstate') = \sheap(\sstate); \multiple{\triple{t}{f}{\fdefault(T)}}, \\
        &\senv(\sstate') = \senv(\sstate)[x \mapsto t], \quad\text{and} \\
        &t = \ffresh.
      \end{align*}
      Let $\vstate' = \triple{\sstate'}{s}{\gform}$, and $V' = V[t \mapsto \ell]$. We want to show that $\pair{\heap'}{\triple{\perms'}{\env[x \mapsto \ell]}{s} \cdot \stack^*}$ is validated by $\vstate'$ with $V$.

      \textit{Part \ref{def:state-valid-reachable}:}
      By \refrule{SVerifyStep}, $\prog \vdash \vstate \to \vstate'$, and all fresh values added to $\vstate'$ are defined in $V'$. Therefore, $\vstate'$ is reachable from $\prog$ with valuation $V'$.

      \textit{Part \ref{def:state-valid-correspond}:}
      We want to show that $\vstate'$ corresponds to ${\heap'}{\triple{\perms'}{\env[x \mapsto \ell]}{s} \cdot \stack^*}$. By definition $s(\vstate') = s$ and $\sstate(\vstate') = \sstate'$, thus it suffices to show $\simstate{V'}{\sstate'}{\heap'}{\perms'}{\env[x \mapsto \ell]}$.

      By assumptions, $\simenv{V}{\senv(\sstate)}{\env}$. Also,
      $$V'(\senv(\sstate')(x)) = V'((\senv(\sstate)[x \mapsto t])(x)) = V'(t) = \ell = (\env[x \mapsto \ell])(x).$$
      Therefore
      $$\simenv{V'}{\senv(\sstate)[x \mapsto t]}{\env[x \mapsto \ell]}.$$

      By assumptions, $\simstate{V}{\sstate}{\heap}{\perms}{\env}$. Since $V \subseteq V'$, $\simstate{V'}{\sstate}{\heap}{\perms}{\env}$.
      Since $\ell$ is a fresh value, $\ell \notin \perms$. Thus
      $$\universal{\pair{\ell}{f} \in \perms}{\heap'(\ell, f) = \heap(\ell, f)}.$$
      Thus by lemma \ref{lem:pres-add-heap}, $\simstate{V'}{\sstate}{\heap'}{\perms}{\env}$. Also, since $\perms \subseteq \perms'$, by lemma \ref{lem:simstate-monotonicity} $\simstate{V'}{\sstate}{\heap'}{\perms'}{\env}$.

      Let $\triple{f'}{t'}{t''} \in \sheap(\sstate')$. If $\triple{f'}{t'}{t''} \in \sheap(\sstate)$, then since $\simstate{V'}{\sstate}{\heap'}{\perms'}{\env}$, $\heap'(V'(t'), f') = V'(t'')$ and $\pair{V'(t')}{f'} \in \perms'$.

      Otherwise, $t' = t$ and $f' = f$ for some $T~f \in \fstruct(S)$, and thus $t'' = \fdefault(T)$. Now $\heap(V'(t'), f') = \heap'(V'(t), f) = \heap'(\ell, f) = \fdefault(T) = t''.$
      
      Therefore $$\universal{\triple{f'}{t'}{t''} \in \sheap(\sstate')}{\heap'(V'(t'), f') = V'(t'') ~\text{and}~ \pair{V'(t')}{f'} \in \perms'}.$$

      Also, since $\simstate{V'}{\sstate}{\heap'}{\perms'}{\env}$ and $\universal{\pair{p}{\multiple{t}} \in \sheap(\sstate')}{\pair{p}{\multiple{t}} \in \sheap(\sstate)}$,
      $$\universal{\pair{p}{\multiple{t}} \in \sheap(\sstate')}{\assertion{\heap'}{\perms'}{[\multiple{x \mapsto V(t)}]}{\fpred(p)}}$$
      where $\multiple{x} = \fpredparams(p)$.

      Let $h_1, h_2 \in \sheap(\sstate')$ and suppose $h_1 \ne h_2$. If $h_1 \in \sheap(\sstate)$ and $h_2 \in \sheap(\sstate)$, then in this case $\vfoot{V'}{\heap'}{h_1} \cap \vfoot{V'}{\heap'}{h_2} = \emptyset$ since $\simstate{V'}{\sstate}{\heap'}{\perms'}{\env}$.

      Otherwise, WLOG $h_1 = \triple{t}{f_1}{\fdefault(T_1)}$ for some $T_1~f_1 \in \fstruct(S)$. Thus 
      \begin{align*}
        \vfoot{V'}{\heap'}{h_1} &= \vfoot{V'}{\heap'}{\triple{t}{f_1}{\fdefault(T_1)}} &h_1 = \triple{t}{f_1}{\fdefault(T_1)} \\
          &= \set{ \pair{V'(t)}{f_1} } &\text{defn.} \\
          &= \set{ \pair{\ell}{f_1} } &\text{defn. $V'$}
      \end{align*}
      If $h_2 \in \sheap(\sstate)$, then
      \begin{align*}
        \vfoot{V'}{\heap'}{h_2} &= \vfoot{V}{\heap}{h_2} & V \subseteq V', \heap \subseteq \heap' \\
          &\subseteq \perms &\text{Lemma \ref{lem:sim-heap-contains}}
      \end{align*}
      Now $\pair{\ell}{f_1} \notin \perms$ since $\ell$ is a fresh value. Therefore $\vfoot{V'}{\heap'}{h_1} \cap \vfoot{V'}{\heap'}{h_2} = \emptyset$ in this case.
      
      Otherwise, $h_2 = \triple{t}{f_2}{\fdefault(T_2)}$ for some $T_2~f_2 \in \fstruct(S)$. Then $f_1 \ne f_2$ since $h_1 \ne h_2$. Therefore
      \begin{align*}
        \vfoot{V'}{\heap'}{h_1} \cap \vfoot{V'}{\heap'}{h_2} &= \vfoot{V'}{\heap'}{\triple{t}{f_1}{\fdefault(T_1)}} \cap \vfoot{V'}{\heap'}{\triple{t}{f_2}{\fdefault(T_2)}} \\
          &= \set{\pair{V'(t)}{f_1}} \cap \set{\pair{V'(t)}{f_2}} \\
          &= \emptyset
      \end{align*}

      Therefore,
      $$\universal{h_1, h_2 \in \sheap(\sstate')}{\vfoot{V'}{\heap'}{h_1} \cap \vfoot{V'}{\heap'}{h_2} = \emptyset}.$$

      Now, since we have shown all requirements, we can conclude that $$\simheap{V'}{\sheap(\sstate')}{\heap'}{\perms'}.$$

      Finally, since $\sheap$ is the only component that differs between $\sstate'$ and $\sstate$, $\simstate{V'}{\sstate}{\heap'}{\perms'}{\env}$, $\simenv{V'}{\senv(\sstate)[x \mapsto t]}{\env[x \mapsto \ell]}$, and $\simheap{V'}{\sheap(\sstate')}{\heap'}{\perms'}$,
      $$\simstate{V'}{\sstate'[\senv = \senv(\sstate)[x \mapsto t]]}{\heap'}{\perms'}{\env[x \mapsto \ell]}$$
      which is what we wanted to show.

      \textit{Part \ref{def:state-valid-partial}:}
      First, we show that the partial state $\pair{\heap}{\stack^*}$ is validated by $\vstate'$ and $V'$.

      Since the initial state is validated by $\vstate$ and $V$, the partial state $\pair{\heap}{\stack^*}$ is validated by $\vstate$ and valuation $V$. Therefore one of \ref{def:partial-valid-nil}, \ref{def:partial-valid-call}, \ref{def:partial-valid-while} applies. We want to show that the partial state $\pair{\heap}{\stack^*}$ is validated by $\vstate'$ and valuation $V'$.

      \begin{itemize}
        \item \textit{Case \ref{def:partial-valid-nil}:} Then $\stack^* = \nilsym$ and trivially $\pair{\heap}{\nilsym}$ is validated by $\vstate'$ and valuation $V'$.
        
        \item \textit{Case \ref{def:partial-valid-call}:} Then $\stack^* = \triple{\env_0}{\perms_0}{\sseq{y \kassign m(e_1, \cdots, e_k)}{s_0}} \cdot \stack_0$ for some $\env_0$, $\perms_0$, $y$, $m$, $k$, $e_1, \cdots, e_k$, $s_0$, $\stack_0$. Also, there is some $\vstate_0$, $V_0$, $x_1, \cdots, x_k$, $t_1, \cdots, t_k$ and $\sstate_0, \cdots, \sstate_k, \sstate'$ such that
        \begin{gather*}
          \text{The partial state $\pair{\heap}{\stack_0}$ is validated by $\vstate_0$ and $V_0$}, \\
          \vstate_0 \text{ is reachable from $\prog$ with valuation $V_0$}, \quad s(\vstate_0) = s(\stack^*) \\
          x_1, \cdots, x_k = \fparams(m), \\
          \sstate_0 = \sstate(\vstate_0), \quad \seval{\sstate_0}{e_1}{t_1}{\sstate_1}{\_}, \quad\cdots,\quad \seval{\sstate_{k-1}}{e_k}{t_k}{\sstate_k}{\_}, \\
          \scons{\sstate_k}{\fpre(m)}{\sstate'}{\_}, \quad \simstate{V_0}{\sstate'[\senv = \senv(\sstate_0)]}{\heap}{\perms_0}{\env_0}, \quad\text{and} \\
          \universal{1 \le i \le k}{V(\senv(\vstate)(x_i)) = V_0(t_i)}, \\
          \gform(\vstate) = \fpost(m).
        \end{gather*}
  
        We want to show that the partial state $\pair{\heap}{\triple{\env_0}{\perms_0}{\sseq{y \kassign m(e_1, \cdots, e_k)}{s}} \cdot \stack_0}$ is validated by $\vstate'$ and valuation $V'$. Immediately from above we can conclude that
        \begin{gather*}
          \text{The partial state $\pair{\heap}{\stack_0}$ is validated by $\vstate_0$ and $V_0$}, \\
          \vstate_0 \text{ is reachable from $\prog$ with valuation $V_0$}, \quad s(\vstate_0) = s(\stack^*) \\
          x_1, \cdots, x_k = \fparams(m), \\
          \sstate_0 = \sstate(\vstate_0), \quad \seval{\sstate_0}{e_1}{t_1}{\sstate_1}{\_}, \quad\cdots,\quad \seval{\sstate_{k-1}}{e_k}{t_k}{\sstate_k}{\_}, \quad\text{and} \\
          \scons{\sstate_k}{\fpre(m)}{\sstate'}{\_}, \quad \simstate{V_0}{\sstate'[\senv = \senv(\sstate_0)]}{\heap}{\perms_0}{\env_0}.
        \end{gather*}
        Also, the frame $\triple{\perms}{\env}{\sseq{x = \kalloc(S)}{s}}$ must be executing the body of $m$, since it is in the stack immediately above the frame that contains $y \kassign m(e_1, \cdots, e_k)$. Therefore, since $x_1, \cdots, x_k$ are all parameters of $m$, $y$ must be distinct from all of $x_1, \cdots, x_k$, since we do not allow assignment to parameters in a well-formed program. Thus
        \begin{align*}
          \forall 1 \le i \le k : V'(\senv(\vstate')(x_i))
            &= V'((\senv(\sstate)[x \mapsto t])(x_i)) = V'(\senv(\sstate)(x_i)) \\
            &= V'(\senv(\vstate)(x_i)) = V(\senv(\vstate)(x_i)) \\
            &= V_0(t_i).
        \end{align*}
  
        Finally, $\gform(\vstate') = \gform(\vstate)$ by definition, thus
        $$\gform(\vstate') = \gform(\vstate) = \fpost(m).$$
  
        Therefore the partial state $\pair{\heap}{\stack^*}$ is validated by $\vstate'$ and $V'$ in this case.
        
        \item \textit{Case \ref{def:partial-valid-while}:}
        Then $\stack^* = \triple{\env_0}{\perms_0}{\sseq{\swhile{e_0}{\gform_0}{s_0}}{s_0'}} \cdot \stack_0$ for some $\env_0$, $\perms_0$, $e$, $\gform_0$, $s_0$, $s_0'$, $\stack_0$, and there exists some $\vstate_0$, $V_0$, and $\sstate_0'$ such that:
        \begin{gather*}
          \text{The partial state $\pair{\heap}{\stack_0}$ is validated by $\vstate_0$ and $V_0$} \\
          \vstate_0 \text{ is reachable from $\prog$ with valuation $V_0$}, \quad s(\vstate_0) = s(\stack^*) \\
          \scons{\sstate_0}{\gform_0}{\sstate_0'}{\_}, \quad
          \simstate{V_0}{\sstate_0'}{\heap}{\perms_0}{\env_0} \quad\text{and}\\
          \gform(\vstate) = \gform_0.
        \end{gather*}
  
        Now, by definition of $\vstate''$, $\gform(\vstate') = \gform(\vstate) = \gform_0$. Therefore, using the other assumptions given above, the partial state $\pair{\heap}{\stack^*}$ is validated by $\vstate'$ and $V'$ in this case.

      \end{itemize}

      Therefore the partial state $\pair{\heap}{\stack^*}$ is validated by $\vstate'$ and $V'$.

      Now, by lemma \ref{lem:pres-heap-change-partial}, the partial state $\pair{\heap'}{\stack^*}$ is validated by $\vstate'$ and $V'$, which is what we need to show for this part.

      Therefore $\pair{\heap'}{\triple{\perms'}{\env[x \mapsto \ell]}{s} \cdot \stack^*}$ is validated by $\vstate'$ with $V$, which completes the proof.

    \case \refrule{ExecCallEnter}: We have
      \begin{align*}
        &\pair{\heap}{\triple{\perms}{\env}{\sseq{y \kassign m(\multiple{e})}{s}} \cdot \stack^*}, \vfoot{V'}{\heap}{\sperms} \to \\
        &\quad \pair{\heap}{\triple{\perms'}{\env'}{\sseq{\fbody(m)}{\kskip}} \cdot \triple{\perms \setminus \perms'}{\env}{\sseq{y \kassign m(\multiple{e})}{s}} \cdot \stack^*} \\
        &\text{where}~ \multiple{x} = \fparams(m), \quad \multiple{\eval{\heap}{\env}{e}{v}}, \quad \multiple{\frm{\heap}{\perms}{\env}{e}}, \\
        &\hspace{2em} \env' = [\multiple{x \mapsto v}], \quad
        \assertion{\heap}{\perms \setminus \xperms}{\env'}{\fpre(m)}, \\
        &\hspace{2em} \xperms = \vfoot{V'}{\heap}{\sperms}, \quad\text{and}\quad
        \perms' = \foot{\heap}{\perms \setminus \xperms}{\env'}{\fpre(m)}.
      \end{align*}

      By assumptions, the initial state is validated by some $\vstate$ and valuation $V$, thus $\vstate = \triple{\sstate}{\sseq{y \kassign m(\multiple{e})}{s}}{\gform}$ for some $\sstate$, $\gform$ where $\simstate{V'}{\sstate}{\heap}{\perms}{\env}$.

      The only guard rule that applies is \refrule{SGuardCall}, so we have
      \begin{gather*}
        \multiple{\seval{\sstate}{e}{t}{\sstate'}{\scheck}}, \quad \scons{\sstate'[\senv = [\multiple{x \mapsto t}]]}{\fpre(m)}{\sstate''}{\scheck'}, \\
        \sperms = \frem(\sstate'', \fpre(m)), \quad
        \text{and by assumptions,}\quad \rtassert{V}{\heap}{\perms}{\multiple{\scheck} \cup \scheck'}
      \end{gather*}

      For some $k$, let $x_1, \cdots, x_k = \multiple{x}$, $e_1, \cdots, e_k = \multiple{e}$, $v_1, \cdots, v_k = \multiple{v}$, and $t_1 = \ffresh, \cdots, t_k = \ffresh$.

      Also, let $\sstate_0 = \quintuple{\bot}{\emptyset}{\emptyset}{[x_1 \mapsto t_1, \cdots, x_k \mapsto t_k]}{\ktrue}$, and let $V_0 = [t_1 \mapsto v_1, \cdots, t_k \mapsto v_k]$. Then $\simstate{V_0}{\sstate_0}{\heap}{\perms'}{\env'}$.

      $\efoot{\heap}{\env'}{\fpre(m)} \subseteq \perms'$: If $\fpre(m)$ is completely precise, then $\perms' = \foot{\heap}{\perms \setminus \xperms}{\env'}{\fpre(m)} = \\ \efoot{\heap}{\env'}{\fpre(m)}$. Otherwise, $\perms' = \foot{\heap}{\perms \setminus \xperms}{\env'}{\fpre(m)} = \perms \setminus \xperms$, but also $\efoot{\heap}{\env'}{\fpre(m)} \subseteq \perms' = \perms \setminus \xperms$ by lemma \ref{lem:efoot-subset-spec} since $\assertion{\heap}{\perms \setminus \xperms}{\env'}{\fpre(m)}$.

      Now $\assertion{\heap}{\perms'}{\env'}{\fpre(m)}$ by lemma \ref{lem:assert-efoot-subset}, since $\assertion{\heap}{\perms \setminus \xperms}{\env'}{\fpre(m)}$.

      Now $\assertion{\heap}{\perms'}{\env'}{\fpre(m)}$, $\simstate{V_0}{\sstate_0}{\heap}{\perms'}{\env'}$. Thus by lemma \ref{lem:produce-progress}, for some $\sstate_0'$,
      $$\sproduce{\sstate_0}{\fpre(m)}{\sstate_0'} ~\text{and}~V_0'(\pc(\sstate_0)) \quad\text{where}\quad V_0' = V_0[\sproduce{\sstate_0}{\fpre(m)}{\sstate_0'} \mid \heap].$$

      Let $\vstate_0' = \triple{\sstate_0'}{\sseq{\fbody(m)}{\kskip}}{\fpost(m)}$. We want to show that
      $$\Gamma' = \pair{\heap}{\triple{\perms'}{\env'}{\sseq{\fbody(m)}{\kskip}} \cdot \triple{\perms \setminus \perms'}{\env}{\sseq{y \kassign m(\multiple{e})}{s}} \cdot \stack^*}$$ is validated by $\vstate_0'$ with $V_0'$.

      \textit{Part \ref{def:state-valid-reachable}:}
      By \refrule{SVerifyMethod}, $\strans{\prog}{\initsym}{\vstate_0'}$. Therefore $\vstate_0'$ is reachable from $\prog$ with valuation $V_0'$.

      \textit{Part \ref{def:state-valid-correspond}:}
      We want to show that $\Gamma'$ corresponds to $\vstate_0'$.

      By definition $s(\vstate_0') = \sseq{\fbody(m)}{\kskip} = s(\Gamma')$. Therefore, since $\sstate(\vstate_0') = \sstate_0'$, it suffices to show $\simstate{V_0'}{\sstate_0'}{\heap}{\perms'}{\env'}$.

      Since $\sheap(\sstate_0) = \oheap(\sstate_0) = \emptyset$, $\simheap{V_0}{\sheap(\sstate_0)}{\heap}{\perms' \setminus \efoot{\heap}{\env'}{\fpre(m)}}$ and $\simheap{V_0}{\oheap(\sstate_0)}{\heap}{\perms' \setminus \efoot{\heap}{\env'}{\fpre(m)}}$. Also, for each $1 \le i \le k$, $V_0(\senv(\sstate_0)(x_i)) = v_i = \env'(x_i)$, thus $\simenv{V_0}{\senv(\sstate_0)}{\env'}$. Finally, $V_0(\pc(\sstate_0)) = V_0(\ktrue) = \ktrue$. Therefore $\simstate{V_0}{\sstate_0}{\heap}{\perms' \setminus \efoot{\heap}{\env'}{\fpre(m)}}{\env'}$.

      Also, as shown before, $\assertion{\heap}{\perms'}{\env'}{\fpre(m)}$. Therefore, by lemma \ref{lem:produce-soundness},
      $$\simstate{V_0'}{\sstate_0'}{\heap}{\perms'}{\env'}.$$

      \textit{Part \ref{def:state-valid-partial}:}
      We want to show that the partial state $\pair{\heap}{\triple{\perms \setminus \perms'}{\env}{\sseq{y \kassign m(\multiple{e})}{s}} \cdot \stack^*}$ is validated by $\vstate_0'$ with $V_0'$, thus it suffices to show that case \ref{def:partial-valid-call} is satisfied.

      Since $\triple{\perms}{\env}{\sseq{y \kassign m(\multiple{e})}{s}} \cdot \stack^*$ was validated by $\vstate$ and $V$, the partial state $\pair{\heap}{\stack^*}$ is validated by $\vstate$ and $V$.

      Also, by assumptions, $\vstate$ is reachable from $\prog$ with valuation $V$ and $s(\vstate) = \sseq{y \kassign m(\multiple{e})}{s}$ as shown before.

      Furthermore, by assumptions, $x_1, \cdots, x_k = \multiple{x} = \fparams(m)$

      Now let $\sstate_0 = \sstate = \sstate(\vstate)$, then $\multiple{\seval{\sstate}{e}{t}{\sstate'}{\scheck}}$, which was shown before, represents the series of judgements
      $$\seval{\sstate_0}{e_1}{t_1}{\sstate_1}{\scheck_1}, \quad\cdots\quad \seval{\sstate_{k-1}}{e_k}{t_k}{\sstate_k}{\stack_k}$$
      where $\sstate_k = \sstate'$. Also, as shown before,
      $$\scons{\sstate_k[\senv = [x_1 \mapsto t_1, \cdots, x_k \mapsto t_k]]}{\fpre(m)}{\sstate''}{\scheck'}.$$
      Note that by definition \ref{def:sguard-valuation} $V' = V[\sguard{\vstate}{\sstate'}{\scheck}{\sperms} \mid \heap]$ is  the valuation corresponding to the series of judgements above, extending $V$.

      By lemmas \ref{lem:eval-subpath} and \ref{lem:cons-subpath}, $\pc(\sstate'') \implies \pc(\sstate_k) \implies \cdots \implies \pc(\sstate_1)$. Thus, since $V'(\pc(\sstate'')) = \ktrue$ by assumption,
      $$V'(\pc(\sstate'')) = V'(\pc(\sstate_k)) = \cdots = V'(\pc(\sstate_1)) = \ktrue.$$
      Therefore, by lemmas \ref{lem:seval-soundness} and \ref{lem:cons-soundness},
      $$\simstate{V'}{\sstate_1}{\heap}{\perms}{\env}, \quad\cdots,\quad \simstate{V'}{\sstate_k}{\heap}{\perms}{\env}, \quad \simstate{V'}{\sstate''}{\heap}{\perms \setminus \efoot{\heap}{\env'}{\fpre(m)}}{\env'},$$
      Furthermore, since $\xperms = \vfoot{V'}{\heap}{\sperms} = \vfoot{V'}{\heap}{\frem(\sstate'', \fpre(m))}$ and $\assertion{\heap}{\perms \setminus \xperms}{\env'}{\fpre(m)}$, we can apply lemma \ref{lem:rem-simstate} to get
      $$\simstate{V'}{\sstate''[\senv = \senv(\sstate_0)]}{\heap}{\perms \setminus \perms'}{\env'}.$$
      Since $\simenv{V}{\senv(\sstate_0)}{\env}$ (by assumptions and since $\sstate_0 = \sstate$),
      $$\simstate{V'}{\sstate''[\senv = \senv(\sstate_0)]}{\heap}{\perms \setminus \perms'}{\env}.$$

      For each $1 \le i \le k$, $\eval{\heap}{\env}{e_i}{V'(t_i)}$ by lemma \ref{lem:seval-soundness}, and $\eval{\heap}{\env}{e_i}{v_i}$ as shown before, thus $V'(t_i) = v_i$. Thus
      \begin{align*}
        \universal{1 \le i \le k}{V_0'(\senv(\vstate_0')(x_i)) &= V_0'(\senv(\sstate_0')(x_i))} &\text{by defn.} \\
          &= V_0'(\senv(\sstate_0)(x_i)) &\text{Lemma \ref{lem:cons-unchanged}} \\
          &= V_0(\senv(\sstate_0)(x_i)) &V \subseteq V' \\
          &= v_i &\text{by def.}\\
          &= V'(t_i). &\text{shown above}
      \end{align*}

      Finally, by definition $\gform(\vstate) = \fpost(m)$.

      Therefore the partial state $\pair{\heap}{\triple{\perms'}{\env'}{\sseq{\fbody(m)}{\kskip}} \cdot \triple{\perms \setminus \perms'}{\env}{\sseq{y \kassign m(\multiple{e})}{s}} \cdot \stack^*}$ is validated by $\vstate_0'$ with $V_0'$.

      Therefore $\Gamma'$ is validated by $\vstate_0'$ with $V_0'$.

    \case \refrule{ExecCallExit}: We have
      \begin{gather}
        \dexec{\heap}{\triple{\perms}{\env}{\kskip} \cdot \triple{\perms'}{\env'}{\sseq{y \kassign m(\multiple{e})}{s}} \cdot \stack}{\vfoot{V'}{\heap}{\sperms}}{\heap}{\triple{\perms''}{\env''}{s} \cdot \stack^*} \label{eq:dexec-pres-return-exec}\\
        \text{where } \assertion{\heap}{\perms}{\env}{\fpost(m)}, \quad
        \env'' = \env'[y \mapsto \env(\kresult)], \label{eq:dexec-pres-return-1} \\
        \text{and }\perms'' = \perms' \cup \foot{\heap}{\perms}{\env}{\fpost(m)}.
      \end{gather}

      By assumptions, the initial state is validated by some $\vstate$ and valuation $V$, thus $\vstate = \triple{\sstate}{\kskip}{\gform}$ for some $\sstate, \gform$ where $\simstate{V'}{\sstate}{\heap}{\perms}{\env}$.

      Also, by \ref{def:state-valid-partial}, the partial state $\pair{\heap}{\triple{\perms'}{\env'}{\sseq{y \kassign m(\multiple{e})}{s}} \cdot \stack}$ is validated by $\vstate$ and $V$. Thus \ref{def:partial-valid-call} must apply, and thus there is some $\vstate'$ reachable from $\prog$ and valuation $V$ such that $\vstate' = \triple{\sstate_0}{\sseq{y \kassign m(\multiple{e})}{s}}{\gform'}$ for some $\sstate_0, \gform'$. Also, we can let $e_1, \cdots, e_k = \multiple{e}$ and then there are sequences $\sstate_1, \cdots, \sstate_k$, $x_1, \cdots, x_k$, and $t_1, \cdots, t_k$ where
      \begin{align}
        &\seval{\sstate_0}{e_1}{t_1}{\sstate_1}{\_}, \quad\cdots,\quad
        \seval{\sstate_{k-1}}{e_k}{t_k}{\sstate_k}{\_}, &\text{by \eqref{eq:partial-valid-call-eval}}\label{eq:dexec-pres-return-seval} \\
        &\scons{\sstate_k[\senv = [\multiple{x_i \mapsto t_i}]]}{\fpre(m)}{\sstate'}{\_}, \label{eq:dexec-pres-return-cons} \\
        &\text{and}\quad
        \simstate{V'}{\sstate'[\senv = \senv(\sstate_0)]}{\heap}{\perms'}{\env'} &\text{by \eqref{eq:partial-valid-call-sim}} \label{eq:dexec-pres-return-2}
      \end{align}
      where $V'$ is a valuation corresponding to this series of judgements.

      Let
      \begin{gather*}
        t = \ffresh, \quad
        \hat{V}' = V'[t \mapsto \env(\kresult)], \quad
        \hat{\env} = [x_1 \mapsto \env(x_1), \cdots, x_k \mapsto \env(x_k)], \\
        \text{and}\quad \hat{\senv} = [x_1 \mapsto t_1, \cdots, x_k \mapsto t_k, \kresult \mapsto t]
      \end{gather*}

      We have $\assertion{\heap}{\perms}{\env}{\fpost(m)}$ by \eqref{eq:dexec-pres-return-1}. Since $\hat{\env}$ is simply the restriction of $\env$ to $\fparams(m)$ and $\kresult$, and $\fpost(m)$ may only reference variables in $\fparams(m)$ as well as $\kresult$, $\assertion{\heap}{\perms}{\hat{\env}}{\fpost(m)}$ and $\efoot{\heap}{\hat{\env}}{\fpost(m)} = \efoot{\heap}{\env}{\fpost(m)}$.

      By lemma \ref{lem:efoot-subset-spec} $\efoot{\heap}{\env}{\fpost(m)} \subseteq \perms$. Recall that $\perms'' = \perms \cup \foot{\heap}{\perms}{\env}{\fpost(m)}$. If $\fpost(m)$ is completely precise, then $\foot{\heap}{\perms}{\env}{\fpost(m)} = \efoot{\heap}{\env}{\fpost(m)}$. Otherwise, $\foot{\heap}{\perms}{\env}{\fpost(m)} = \perms$, but $\efoot{\heap}{\env}{\fpost(m)} \subseteq \perms$ as shown before. In both cases, $\efoot{\heap}{\env}{\fpost(m)} = \efoot{\heap}{\hat{\env}}{\fpost(m)} \subseteq \perms''$.

      Therefore by lemma \ref{lem:assert-efoot-subset}
      \begin{equation}\label{eq:dexec-pres-return-3}
        \assertion{\heap}{\perms''}{\hat{\env}}{\fpost(m)}.
      \end{equation}

      Note that, for all $1 \le i \le k$,
      \begin{align*}
        \hat{V}'(t_i) &= V'(t_i) &\text{by definition} \\
          &= V(\senv(\vstate)(x_i)) &\text{by \eqref{eq:partial-valid-call-params}} \\
          &= \env(x_i) &\text{since $\simenv{V}{\senv(\vstate)}{\env}$, since initial valid by $\vstate$ and $V$} \\
          &= \hat{\env}(x_i) &\text{by definition}
      \end{align*}
      Thus $\simenv{\hat{V}'}{\senv'}{\hat{\env}}$. Also, $\simstate{V'}{\sstate'[\senv = \senv(\sstate_0)]}{\heap}{\perms'}{\env'}$ by \eqref{eq:dexec-pres-return-2}. Therefore $\simstate{\hat{V}'}{\sstate'[\senv = \hat{\senv}]}{\heap}{\perms'}{\hat{\env}}$. Finally, since $\perms'' \subseteq \perms'$, by lemma \ref{lem:simstate-monotonicity},
      \begin{equation}\label{eq:dexec-pres-return-4}
        \simstate{\hat{V}'}{\sstate'[\senv = \hat{\senv}]}{\heap}{\perms''}{\hat{\env}}.
      \end{equation}

      Now, by lemma \ref{lem:produce-progress}, \eqref{eq:dexec-pres-return-3}, and \eqref{eq:dexec-pres-return-4},
      \begin{equation}\label{eq:dexec-pres-return-prod}
        \sproduce{\sstate'[\senv = \hat{\senv}]}{\fpost(m)}{\sstate''}, \quad\text{and}\quad V''(\pc(\sstate''')) = \ktrue
      \end{equation}
      where $V''$ is the corresponding valuation extending $\hat{V}'$.

      Let $\sstate''' = \sstate''[\senv = \senv(\sstate_0)[y \mapsto t]]$.

      Now \eqref{eq:dexec-pres-return-seval}, \eqref{eq:dexec-pres-return-cons}, and \eqref{eq:dexec-pres-return-prod}, and the definition of $\sstate'''$ satisfy the antecedent for \refrule{SExecCall}, therefore
      $$\sexec{\sstate_0}{\sseq{y \kassign m(e_1, \cdots, e_k)}{s}}{s}{\sstate'''}.$$
      Let $\vstate'' = \triple{\sstate'''}{s}{\gform'}$, now by \refrule{SVerifyStep}
      $$\strans{\prog}{\vstate'}{\vstate''}.$$

      We want to show that $\pair{\heap}{\triple{\perms''}{\env''}{s} \cdot \stack^*}$ is validated by $\vstate''$.
      
      Part \ref{def:state-valid-reachable}: Since $\strans{\prog}{\vstate'}{\vstate''}$, $\vstate''$ is reachable from $\prog$. Let its corresponding valuation be $V'''$.

      Part \ref{def:state-valid-correspond}: By definition, $s(\vstate'') = s$.

      By \eqref{eq:dexec-pres-return-2} $\simstate{V'}{\sstate'[\senv = \senv(\sstate_0)]}{\heap}{\perms'}{\env'}$. Since the initial state must be well-formed, $\perms$ and $\perms'$ are disjoint, and as shown before, $\efoot{\heap}{\env}{\fpost(m)} = \efoot{\heap}{\hat{\env}}{\fpost(m)} \subseteq \perms$, therefore $\perms' \setminus \efoot{\heap}{\hat{\env}}{\fpost(m)} = \perms'$. Also, $\hat{V}' \subseteq V'$. Thus
      $$\simstate{\hat{V}'}{\sstate'[\senv = \senv(\sstate_0)]}{\heap}{\perms' \setminus \efoot{\heap}{\hat{\env}}{\fpost(m)}}{\env'}.$$

      Also, as shown before, $\simenv{\hat{V}'}{\hat{\senv}}{\hat{\env}}$, therefore
      $$\simstate{\hat{V}'}{\sstate'[\senv = \hat{\senv}]}{\heap}{\perms' \setminus \efoot{\heap}{\hat{\env}}{\fpost(m)}}{\hat{\env}}.$$

      Then, since $\perms' \subseteq \perms''$, $\perms' \setminus \efoot{\heap}{\hat{\env}}{\fpost(m)} \subseteq \perms'' \setminus \efoot{\heap}{\hat{\env}}{\fpost(m)}$, and thus by lemma \ref{lem:simstate-monotonicity}
      $$\simstate{\hat{V}'}{\sstate'[\senv = \hat{\senv}]}{\heap}{\perms'' \setminus \efoot{\heap}{\hat{\env}}{\fpost(m)}}{\hat{\env}}.$$

      Now, since it was shown in \eqref{eq:dexec-pres-return-3} that $\assertion{\heap}{\perms''}{\hat{\env}}{\fpost(m)}$ and in \eqref{eq:dexec-pres-return-prod} that $V''(\pc(\sstate'')) = \ktrue$, by lemma \ref{lem:produce-soundness}
      $$\simstate{V''}{\sstate''}{\heap}{\perms''}{\hat{\env}}.$$

      Now by \eqref{eq:dexec-pres-return-2} $\simenv{V'}{\senv(\sstate_0)}{\env'}$, then since $V' \subseteq V''$, $\simenv{V''}{\senv(\sstate_0)}{\env'}$. Now, since $\senv(\sstate''') = \senv(\sstate_0)[y \mapsto t]$, to show $\simenv{V''}{\senv(\sstate''')}{\env''}$ it suffices to show that $V''(\senv(\sstate''')(y)) = \env''(y)$.

      But now $V''(\senv(\sstate''')(y)) = V''(t) = \hat{V}'(t) = \env(\kresult) = \env''(y)$, which is what we needed to show. Therefore, since $\senv$ is the only component changed between $\sstate''$ and $\sstate'''$,
      $$\simstate{V''}{\sstate'''}{\heap}{\perms''}{\env''}.$$

      Therefore, since $\sstate(\vstate'') = \sstate'''$, we have shown that $\vstate''$ corresponds to $\pair{\heap}{\triple{\perms''}{\env''}{s} \cdot \stack^*}$ with valuation $V''$.

      Part \ref{def:state-valid-partial}: We need to show that the partial state $\pair{\heap}{\stack^*}$ is validated by $\vstate''$ and $V''$. We already have that $\pair{\heap}{\stack^*}$ is validated by $\vstate'$ and $V'$. Thus one of \ref{def:partial-valid-nil}, \ref{def:partial-valid-call}, or \ref{def:partial-valid-while} must apply.

      \textit{If \ref{def:partial-valid-nil} applies:} Then $\stack^* = \nilsym$, thus trivially the partial state $\pair{\heap}{\stack^*}$ is validated by $\vstate''$ and $V''$.

      \textit{If \ref{def:partial-valid-call} applies:}
      Then $\stack^* = \triple{\perms_0}{\env_0}{\sseq{y' \kassign m'(e_1', \cdots, e_{k'}')}{s'}} \cdot \stack^*_0$ for some $k', y', m', e_1', \cdots, e_{k'}', s', \stack^*_0$. Also, there exists some $\vstate_0', V_0', x_1', \cdots, x_{k'}', t_1', \cdots, t_{k'}', \sstate_0, \cdots, \sstate_{k'}, \sstate'$ such that
      \begin{gather*}
        \text{The partial state $\pair{\heap}{\stack_0^*}$ is validated by $\vstate_0'$ and $V_0'$},\\
        \vstate_0' \text{ is reachable from $\prog$ with valuation $V_0'$}, \quad s(\vstate_0') = s(\stack^*), \\
        x_1', \cdots, x_{k'}' = \fparams(m),\\
        \sstate_0 = \sstate(\vstate_0'), \quad \seval{\sstate_0}{e_1'}{t_1'}{\sstate_1}{\_}, \quad\cdots,\quad \seval{\sstate_{k'-1}}{e_{k'}'}{t_{k'}'}{\sstate_{k'}}{\_},\\
        \scons{\sstate_{k'}}{\fpre(m')}{\sstate'}{\_}, \quad \simstate{V_0'}{\sstate'[\senv = \senv(\sstate_0)]}{\heap}{\perms_0}{\env_0}, \\
        \universal{1 \le i \le k'}{V'(\senv(\vstate')(x_i)) = V_0'(t_i')}, \quad\text{and} \\
        \gform(\vstate') = \fpost(m').
      \end{gather*}

      We want to show that the partial state $\pair{\heap}{\stack^*}$ is validated by $\vstate''$ and $V''$. Immediately from above,
      \begin{gather*}
        \text{The partial state $\pair{\heap}{\stack_0^*}$ is validated by $\vstate_0'$ and $V_0'$},\\
        \vstate_0' \text{ is reachable from $\prog$ with valuation $V_0'$}, \quad s(\vstate_0') = s(\stack^*), \\
        x_1', \cdots, x_{k'}' = \fparams(m),\\
        \sstate_0 = \sstate(\vstate_0'), \quad \seval{\sstate_0}{e_1'}{t_1'}{\sstate_1}{\_}, \quad\cdots,\quad \seval{\sstate_{k'-1}}{e_{k'}'}{t_{k'}'}{\sstate_{k'}}{\_}, \\
        \scons{\sstate_{k'}}{\fpre(m')}{\sstate'}{\_}, \quad \simstate{V_0'}{\sstate'[\senv = \senv(\sstate_0)]}{\heap}{\perms_0}{\env_0}.
      \end{gather*}
      Also, the frame $\triple{\perms'}{\env'}{\sseq{y \kassign m(e_1, \cdots, e_k)}{s}}$ must be executing the body of $m'$, since it is in the stack immediately above the frame that contains $y' \kassign m'(e_1', \cdots, e_{k'}')$. Therefore, since $x_1', \cdots, x_{k'}$ are all parameters of $m'$, $y$ must be distinct from all of $x_1', \cdots, x_{k'}'$, since we do not allow assignment to parameters in a well-formed program. Thus
      $$\forall 1 \le i \le k' : V''(\senv(\vstate'')(x_i)) = V''(\senv(\vstate')(x_i)) = V'(\senv(\vstate')(x_i)) = V_0'(t_i).$$

      Finally, $\gform(\vstate'') = \gform(\vstate')$ by definition, thus
      $$\gform(\vstate'') = \gform(\vstate') = \fpost(m').$$

      Therefore the partial state $\pair{\heap}{\stack^*}$ is validated by $\vstate''$ and $V''$ in this case.

      \textit{If \ref{def:partial-valid-while} applies:}
      Then $\stack^* = \triple{\env_0}{\perms_0}{\sseq{\swhile{e}{\gform_0}{s_0}}{s_0'}} \cdot \stack_0^*$ for some $\env_0$, $\perms_0$, $e$, $\gform_0$, $s_0$, $s_0'$, $\stack_0^*$, and there exists some $\vstate_0'$, $V_0'$, and $\sstate_0'$ such that:
      \begin{gather*}
        \text{The partial state $\pair{\heap}{\stack_0^*}$ is validated by $\vstate_0'$ and $V_0'$} \\
        \vstate_0' \text{ is reachable from $\prog$ with valuation $V_0'$}, \quad s(\vstate_0') = s(\stack^*), \\
        \scons{\sstate_0}{\gform_0}{\sstate_0'}{\_}, \quad
        \simstate{V_0'}{\sstate_0'}{\heap}{\perms_0}{\env_0} \quad\text{and}\\
        \gform(\vstate') = \gform_0.
      \end{gather*}

      Now, by definition of $\vstate''$, $\gform(\vstate'') = \gform(\vstate') = \gform_0$. Therefore, using the other assumptions given above, the partial state $\pair{\heap}{\stack^*}$ is validated by $\vstate''$ and $V''$ in this case.

      Therefore definition part \ref{def:state-valid-partial} is satisfied.

      Therefore all parts of definition \ref{def:state-valid} are satisfied. Thus $\pair{\heap}{\triple{\perms''}{\env''}{s} \cdot \stack^*}$ is validated by $\vstate''$, as we wanted to show.

    \case \refrule{ExecAssert}: We have
      \begin{gather*}
        \dexec{\heap}{\triple{\perms}{\env}{\sseq{\sassert{\phi}}{s}} \cdot \stack^*}{\vfoot{V'}{\heap}{\sperms}}{\heap}{\triple{\perms}{\env}{s} \cdot \stack^*} \\
        \text{where}\quad \assertion{\heap}{\perms}{\env}{\simprecise{\phi}}.
      \end{gather*}

      Since the initial state is validated by $\vstate$, $\vstate = \triple{\sstate}{\sseq{\sassert{\phi}}{s}}{\gform}$ for some $\sstate$, $\gform$ where $\simstate{V}{\sstate}{\heap}{\perms}{\env}$.

      The only guard rule that applies is \refrule{SGuardAssert}, so we have, for some $\sstate'$,
      \begin{gather*}
        \scons{\sstate}{\simprecise{\phi}}{\sstate'}{\scheck} \\
        \text{and by assumptions } V'(\pc(\sstate')) = \ktrue \quad\text{and}\quad \rtassert{V'}{\heap}{\perms}{\scheck}.
      \end{gather*}
      Also, by definition $V' = V[\scons{\sstate}{\simprecise{\phi}}{\sstate'}{\scheck} \mid \heap]$.

      Thus by lemma \ref{lem:cons-soundness} $\simstate{V'}{\sstate'}{\heap}{\perms \setminus \efoot{\heap}{\env}{\simprecise{\phi}}}{\env}$.

      Thus by lemma \ref{lem:simstate-monotonicity}, $\simstate{V'}{\sstate'}{\heap}{\perms}{\env}$. Also, as noted before, $\assertion{\heap}{\perms}{\env}{\simprecise{\phi}}$. Therefore, by lemma \ref{lem:produce-soundness}, for some $\sstate''$,
      $$\sproduce{\sstate'}{\simprecise{\phi}}{\sstate''} \quad\text{and}\quad V''(\pc(\sstate'')) = \ktrue.$$

      Now, by \refrule{SExecAssert}, $\sexec{\sstate}{\sseq{\sassert{\phi}}{s}}{s}{\sstate[\pc = \pc(\sstate'')]}$.

      Let $\vstate' = \triple{\sstate[\pc = \pc(\sstate'')]}{s}{\gform}$. We want to show that $\pair{\heap}{\triple{\perms}{\env}{s} \cdot \stack^*}$ is validated by $\vstate'$ and $V''$.

      By \refrule{SVerifyStep}, $\vstate'$ is reachable from $\prog$ with valuation $V''$.

      By assumptions, $\simstate{V}{\sstate}{\heap}{\perms}{\env}$, and $V''(\pc(\sstate'')) = \ktrue$, thus $\simstate{V''}{\sstate[\pc = \pc(\sstate'')]}{\heap}{\perms}{\env}$. Also, by definition, $\vstate' = \triple{\sstate[\pc = \pc(\sstate'')]}{s}{\gform}$. Therefore $\vstate'$ corresponds to $\pair{\heap}{\triple{\perms}{\env}{s} \cdot \stack^*}$ with valuation $V''$.

      Finally, by definition $\senv(\vstate') = \senv(\vstate)$ and $\gform(\vstate') = \gform = \gform(\vstate)$.

      Therefore $\pair{\heap}{\triple{\perms}{\env}{s} \cdot \stack^*}$ is a valid state by lemma \ref{lem:preservation-heap-env-unchanged}.

    \case\label{case:pres-dexec-ifa} \refrule{ExecIfA}: We have
      \begin{gather*}
        \dexec{\heap}{\triple{\perms}{\env}{\sseq{\sif{e}{s_1}{s_2}}{s}} \cdot \stack^*}{\vfoot{V'}{\heap}{\sperms}}{\heap}{\triple{\perms}{\env}{\sseq{s_1}{s}} \cdot \stack^*} \\
        \text{where}\quad \eval{\heap}{\env}{e}{\ktrue} \quad\text{and}\quad \frm{\heap}{\perms}{\env}{e}
      \end{gather*}

      Since the initial state is validated by $\vstate$, $\vstate = \triple{\sstate}{\sseq{\sif{e}{s_1}{s_2}}{s}}{\gform}$ for some $\sstate$, $\gform$ where $\simstate{V}{\sstate}{\heap}{\perms}{\env}$.

      The only guard rule that applies is \refrule{SGuardIf}, so we have, for some $\sstate'$,
      \begin{gather*}
        \seval{\sstate}{e}{t}{\sstate'}{\scheck} \\
        \text{and by assumptions}\quad V'(\pc(\sstate')) = \ktrue \quad\text{and}\quad \rtassert{V'}{\heap}{\perms}{\scheck}
      \end{gather*}
      where $V' = V[\seval{\sstate}{e}{t}{\sstate'}{\scheck} \mid \heap]$.

      Now by \refrule{SExecIfA},
      $$\sexec{\sstate}{\sseq{\sif{e}{s_1}{s_2}}{s}}{\sseq{s_1}{s}}{\sstate'[\pc = \pc(\sstate') \kand t]}.$$
      
      Let $\vstate' = \triple{\sstate'[\pc = \pc(\sstate') \kand t]}{\sseq{s_1}{s}}{\gform}$. Then by \refrule{SVerifyStep}, $\vstate'$ is reachable from $\prog$ with valuation $V'$.

      Now $\eval{\heap}{\env}{e}{V'(t)}$ and $\eval{\heap}{\env}{e}{\ktrue}$, thus $V'(t) = \ktrue$. Also, $\simstate{V'}{\sstate'}{\heap}{\perms}{\env}$, thus $V'(\pc(\sstate')) = \ktrue$, and then $V'(\pc(\sstate') \kand t) = V'(\pc(\sstate')) \wedge V'(t) = \ktrue$. Therefore $\simstate{V'}{\sstate'[\pc = \pc(\sstate') \kand t]}{\heap}{\perms}{\env}$.

      Also, by definition, $\vstate' = \triple{\sstate[\pc = \pc(\sstate') \kand t]}{\sseq{s_1}{s}}{\gform}$. Therefore $\vstate'$ corresponds to \\
      $\pair{\heap}{\triple{\perms}{\env}{\sseq{s_1}{s}} \cdot \stack^*}$ with valuation $V''$.

      Finally, $\senv(\vstate') = \senv(\sstate') = \senv(\sstate) = \senv(\vstate)$ by lemma \ref{lem:eval-unchanged} and $\gform(\vstate') = \gform = \gform(\vstate)$ by definition.

      Therefore ${\heap}{\triple{\perms}{\env}{\sseq{s_1}{s}} \cdot \stack^*}$ is a valid state by lemma \ref{lem:preservation-heap-env-unchanged}.

    \case \refrule{ExecIfB}: Similar to case \ref{case:pres-dexec-ifa}, but using \refrule{SExecIfB}.

    \case \refrule{ExecWhileEnter}: We have
      \begin{align*}
        &\pair{\heap}{\triple{\perms}{\env}{\sseq{\swhile{e}{\gform}{s}}{s'}} \cdot \stack^*} \to \\
        &\quad \pair{\heap}{\triple{\perms'}{\env}{\sseq{s}{\kskip}} \cdot \triple{\perms \setminus \perms'}{\env}{\sseq{\swhile{e}{\gform}{s}}{s'}} \cdot \stack^*} \\
        &\text{where}~ \eval{\heap}{\env}{e}{\ktrue}, \quad \assertion{\heap}{\perms \setminus \xperms}{\env}{\gform}, \\
        &\hspace{3em} \xperms = \vfoot{V}{\heap}{\sperms}, \quad\text{and}\quad \perms' = \foot{\heap}{\perms \setminus \xperms}{\env}{\gform}
      \end{align*}

      Since the initial state is validated by $\vstate$, $\vstate = \triple{\sstate}{\sseq{\swhile{e}{\gform}{s}}{s'}}{\gform_0}$ for some $\sstate$, $\gform$ where $\simstate{V}{\sstate}{\heap}{\perms}{\env}$.

      The only guard rule that applies is \refrule{SGuardWhile}, so we have, for some $\sstate'$, $\sstate''$, $k$, $x_1, \cdots, x_k$, $t_1, \cdots, t_k$, and $t$,
      \begin{gather}
        \scons{\sstate}{\gform}{\sstate'}{\scheck'}, \quad x_1, \cdots, x_k = \fmodified(s), \quad t_1 = \ffresh, \cdots, t_k = \ffresh, \label{eq:pres-dexec-while-enter-cons} \\
        \sproduce{\sstate'[\senv = \senv(\sstate')[x_1 \mapsto t_1, \cdots, x_k \mapsto t_k]]}{\gform}{\sstate''}, \\
        \seval{\sstate''}{e}{t}{\sstate'''}{\scheck''}, \quad \sperms = \frem(\sstate', \gform), \\
        \text{and by assumptions}\quad V'(\pc(\sstate'')) = \ktrue \quad\text{and}\quad \rtassert{V'}{\heap}{\perms}{\scheck' \cup \scheck''}
      \end{gather}
      where $V'$ is the corresponding valuation for these judgements (see definition \ref{def:sguard-valuation}).

      By lemma \ref{lem:scheck-monotonicity}
      \begin{equation} \label{eq:pres-dexec-while-enter-rt}
        \rtassert{V'}{\heap}{\perms}{\scheck'} \quad\text{and}\quad \rtassert{V'}{\heap}{\perms}{\scheck''}.
      \end{equation}

      Let $t_1' = \ffresh, \cdots, t_k' = \ffresh$ and $\sstate_0 = \quintuple{\bot}{\emptyset}{\emptyset}{\senv(\sstate)[x_1 \mapsto t_1', \cdots, x_k \mapsto t_k']}{\pc(\sstate)}$.

      Let $V_0 = V'[t_1' \mapsto \env(x_1), \cdots, t_k' \mapsto \env(x_k)]$. Now, for any $x \in \dom(\senv(\sstate_0))$, if $x = x_i$ for some $i$, then $V_0(\senv(x)) = V_0(\senv(\sstate_0)(x_i)) = V_0(t_i) = \env(x_i) = \env(x)$. Otherwise, $x \in \dom(\senv(\sstate))$ and thus $V_0(\senv(\sstate)(x)) = V(\senv(\sstate)(x)) = \env(x)$ since $\simenv{V}{\senv(\sstate)}{\env}$. Therefore $\simenv{V_0}{\senv(\sstate_0)}{\env}$.

      Also, since $\sheap(\sstate_0) = \oheap(\sstate_0) = \emptyset$, $\simheap{V_0}{\sheap(\sstate_0)}{\heap}{\perms' \setminus \efoot{\heap}{\env}{\gform}}$ and $\simheap{V_0}{\oheap(\sstate_0)}{\heap}{\perms' \setminus \efoot{\heap}{\env}{\gform}}$. Finally, $V_0(\pc(\sstate_0)) = V(\pc(\sstate)) = \ktrue$ since $\simstate{V}{\sstate}{\heap}{\perms}{\env}$.

      Therefore
      $$\simstate{V_0}{\sstate_0}{\heap}{\perms' \setminus \efoot{\heap}{\env}{\gform}}{\env}$$
      and then also $\simstate{V_0}{\sstate_0}{\heap}{\perms'}{\env}$ by lemma \ref{lem:simstate-monotonicity}.

      Furthermore, by assumptions, $\assertion{\heap}{\perms \setminus \xperms}{\env}{\gform}$, thus $\assertion{\heap}{\perms'}{\env}{\gform}$ by lemma \ref{lem:foot-assert}, since $\perms' = \foot{\heap}{\perms \setminus \xperms}{\env}{\gform}$.

      Therefore, by lemma \ref{lem:produce-progress} $$\sproduce{\sstate_0}{\gform}{\sstate_0'} \quad\text{and}\quad V_0'(\pc(\sstate_0')) = \ktrue$$
      where $V_0' = V_0[\sproduce{\sstate_0}{\gform}{\sstate_0'} \mid \heap, \env]$. Also, by lemma \ref{lem:produce-soundness},
      $$\simstate{V_0'}{\sstate_0'}{\heap}{\perms'}{\env}.$$

      Now by lemma \ref{lem:eval-progress},
      \[\seval{\sstate_0'}{e}{t_0}{\_}{\_}\]
      and let $V_0'' = V_0'[\seval{\sstate_0'}{e}{t_0}{\_}{\_} \mid \heap, \env]$.

      Let $\vstate_0 = \triple{\sstate_0'[\pc = \pc(\sstate_0') \kand t_0]}{\sseq{s'}{\kskip}}{\gform}$.
      
      We want to show that
      $$\Gamma' = \pair{\heap}{\triple{\perms'}{\env}{\sseq{s}{\kskip}} \cdot \triple{\perms \setminus \perms'}{\env}{\sseq{\swhile{e}{\gform}{s}}{s'}} \cdot \stack^*}$$
      is validated by $\vstate_0$ and $V_0''$.

      \textit{Part \ref{def:state-valid-reachable}:}
      By \refrule{SVerifyLoopBody} $\strans{\prog}{\vstate}{\vstate_0}$. Therefore $\vstate_0$ is reachable from $\prog$ with valuation $V_0''$.

      \textit{Part \ref{def:state-valid-correspond}:}
      By lemma \ref{lem:simstate-monotonicity} $\simstate{V_0'}{\sstate_0'}{\heap}{\perms' \cup \efoot{\heap}{\env}{e}}{\env}$. Then since $\eval{\heap}{\env}{e}{\ktrue}$, by lemma \ref{lem:eval-correspondence} $V_0''(t_0) = \ktrue$. Therefore $V_0''(\pc(\sstate_0') \kand t_0) = V'(\pc(\sstate_0')) \wedge V'(t_0') = \ktrue$. Thus, since we have already shown $\simstate{V_0'}{\sstate_0'}{\heap}{\perms'}{\env}$,
      $$\simstate{V_0''}{\sstate_0'[\pc = \pc(\sstate_0') \kand t_0]}{\heap}{\perms'}{\env}.$$
      
      Also, $s(\vstate_0) = \sseq{s}{\kskip}$ by definition. Therefore $\Gamma'$ corresponds to $\vstate_0$ with $V_0''$.

      \textit{Part \ref{def:state-valid-partial}:}
      We want to show that the partial state
      $$\Gamma^* = \pair{\heap}{\triple{\perms \setminus \perms'}{\env}{\sseq{\swhile{e}{\gform}{s}}{s'}} \cdot \stack^*}$$
      is validated by $\vstate_0$ and $V_0''$.

      By assumptions, $\pair{\heap}{\triple{\perms}{\env}{\sseq{\swhile{e}{\gform}{s}}{s'}} \cdot \stack^*}$ was validated by $\vstate$ and $V$. Therefore the partial state $\pair{\heap}{\stack^*}$ was validated by $\vstate$ and $V$, $\vstate$ is reachable from $\prog$ with $V$, $s(\vstate) = \sseq{\swhile{e}{\gform}{s}}{s'}$, and $\simstate{V}{\sstate}{\heap}{\perms}{\env}$.
      
      By \eqref{eq:pres-dexec-while-enter-cons} $\scons{\sstate}{\gform}{\sstate'}{\scheck}$ and by \eqref{eq:pres-dexec-while-enter-rt} $\rtassert{V'}{\heap}{\perms}{\scheck}$. Therefore $\simstate{V'}{\sstate'}{\heap}{\perms \setminus \efoot{\heap}{\env}{\gform}}{\env}$.

      Furthermore, since $\xperms = \vfoot{V'}{\heap}{\sperms} = \vfoot{V'}{\heap}{\frem(\sstate', \gform)}$, $\perms' = \foot{\heap}{\perms \setminus \xperms}{\env}{\gform}$, and $\assertion{\heap}{\perms \setminus \xperms}{\env}{\gform}$, we can apply lemma \ref{lem:rem-simstate} to get
      $$\simstate{V'}{\sstate'}{\heap}{\perms \setminus \perms'}{\env}.$$

      Finally, $\gform(\vstate_0) = \gform$ by definition.

      Therefore $\Gamma^*$ is validated by $\vstate_0$ and $V_0''$ since we have satisfied all requirements of case \ref{def:partial-valid-while}.

      Therefore $\Gamma'$ is validated by $\vstate_0$ and $V_0''$, which completes the proof.

    \case \refrule{ExecWhileSkip}: We have
      \begin{gather*}
        \dexec{\heap}{\triple{\perms}{\env}{\sseq{\swhile{e}{\gform}{s}}{s'}} \cdot \stack^*}{\xperms}{\heap}{\triple{\perms}{\env}{s} \cdot \stack^*} \\
        \text{where}\quad \xperms = \vfoot{V'}{\heap}{\sperms}, \quad \eval{\heap}{\env}{e}{\kfalse}, \quad
        \frm{\heap}{\perms}{\env}{e}, \quad\text{and} \\
        \assertion{\heap}{\perms \setminus \xperms}{\env}{\gform}
      \end{gather*}

      Since the initial state is validated by $\vstate$, $\vstate = \triple{\sstate}{\sseq{\swhile{e}{\gform}{s}}{s'}}{\gform'}$ for some $\sstate$, $\gform'$ where $\simstate{V}{\sstate}{\heap}{\perms}{\env}$.

      The only guard rule that applies is \refrule{SGuardWhile}, so we have, for some $\sstate'$, $\sstate''$, $k$, $x_1, \cdots, x_k$, $t_1, \cdots, t_k$, and $t$,
      \begin{gather}
        \scons{\sstate}{\gform}{\sstate'}{\scheck'}, \quad x_1, \cdots, x_k = \fmodified(s), \quad t_1 = \ffresh, \cdots, t_k = \ffresh, \label{eq:pres-dexec-while-skip-cons} \\
        \sproduce{\sstate'[\senv = \senv(\sstate')[x_1 \mapsto t_1, \cdots, x_k \mapsto t_k]]}{\gform}{\sstate''},\label{eq:pres-dexec-while-skip-prod} \\
        \seval{\sstate''}{e}{t}{\sstate'''}{\scheck''}, \quad \sperms = \frem(\sstate', \gform), \\
        \text{and by assumptions}\quad V'(\pc(\sstate'')) = \ktrue \quad\text{and}\quad \rtassert{V'}{\heap}{\perms}{\scheck' \cup \scheck''} \label{eq:pres-dexec-while-skip-path}
      \end{gather}
      where $V'$ is the corresponding valuation for these judgements (see definition \ref{def:sguard-valuation}).

      Let $\vstate' = \triple{\sstate''[\pc = \pc(\sstate'') \kand \kneg t]}{s'}{\gform'}$.

      By \refrule{SExecWhileSkip}
      $$\sexec{\sstate}{\sseq{\swhile{e}{\gform}{s}}{s'}}{s'}{\sstate''[\pc = \pc(\sstate'') \kand \kneg t]}.$$
      Therefore by \refrule{SVerifyStep} $\strans{\prog}{\vstate}{\vstate'}$. Thus $\vstate'$ is reachable from $\prog$ with valuation $V'$.

      By lemma \ref{lem:scheck-monotonicity}
      \begin{equation} \label{eq:pres-dexec-while-skip-rt}
        \rtassert{V'}{\heap}{\perms}{\scheck'} \quad\text{and}\quad \rtassert{V'}{\heap}{\perms}{\scheck''}.
      \end{equation}

      By lemmas \ref{lem:cons-subpath} and \ref{lem:produce-subpath}, $\pc(\sstate'') \implies \pc(\sstate')$. Therefore $V'(\pc(\sstate')) = \ktrue$. Now, by lemma \ref{lem:cons-soundness},
      $$\simstate{V'}{\sstate'}{\heap}{\perms \setminus \efoot{\heap}{\perms}{\gform}}{\env}.$$

      By definition \ref{def:sguard-valuation}, for all $1 \le i \le k$, $V'(t_i) = V(\senv(\sstate)(x_i))$. By assumptions, $\simenv{V}{\senv(\sstate)}{\env}$, thus $V(\senv(\sstate)(x_i)) = \env(x_i)$. Also, as shown above, $\simenv{V'}{\senv(\sstate')}{\env}$.

      Let $\senv' = \senv(\sstate')[x_1 \mapsto t_1, \cdots, x_k \mapsto t_k]$. Now, for any $x \in \dom(\senv')$, if $x = x_i$ for some $i$ then $V'(\senv'(x)) = V'(t_i) = \env(x_i)$. Otherwise, $x \in \dom(\senv(\sstate'))$ and thus $V'(\senv'(x)) = V'(\senv(\sstate')(x)) = \env(x)$. Therefore $\simenv{V'}{\senv'}{\env}$.

      Therefore $\simstate{V'}{\sstate'[\senv = \senv']}{\heap}{\perms \setminus \efoot{\heap}{\perms}{\gform}}{\env}$. Using the definition of $\senv'$ and \eqref{eq:pres-dexec-while-skip-prod}, $\sproduce{\sstate'[\senv = \senv']}{\gform}{\sstate''}$, and by \eqref{eq:pres-dexec-while-skip-path}, $V'(\pc(\sstate'')) = \ktrue$.

      In addition, as shown before, $\assertion{\heap}{\perms \setminus \xperms}{\env}{\gform}$, thus by lemma \ref{lem:assert-monotonicity} $\assertion{\heap}{\perms}{\env}{\gform}$.

      Now by lemma \ref{lem:produce-soundness}, $\simstate{V'}{\sstate''}{\heap}{\perms}{\env}$.

      In addition, by lemma \ref{lem:simstate-monotonicity}, $\simstate{V'}{\sstate''}{\heap}{\perms \cup \efoot{\heap}{\env}{e}}{\env}$. Then, since $\seval{\sstate''}{e}{t}{\_}{\_}$ and $\eval{\heap}{\env}{e}{\kfalse}$, by lemma \ref{lem:eval-correspondence} $V'(t) = \kfalse$. Therefore $V'(\pc(\sstate'') \kand \kneg t) = V'(\pc(\sstate'')) \wedge \neg V'(t) = \ktrue$. Therefore
      $$\simstate{V'}{\sstate''[\pc = \pc(\sstate'') \kand \kneg t]}{\heap}{\perms}{\env}.$$

      Now, by definition, $s(\vstate') = s'$. Therefore $\pair{\heap}{\triple{\perms}{\env}{s} \cdot \stack^*}$ corresponds to $\vstate'$ with valuation $V'$.

      By definition $\gform(\vstate') = \gform' = \gform(\vstate)$. Also, $\senv(\vstate') = \senv(\sstate'') = \senv(\sstate')[x_1 \mapsto t_1, \cdots, x_k \mapsto t_k]$ by lemma \ref{lem:produce-unchanged} and $\senv(\sstate') = \senv(\sstate) = \senv(\vstate)$ by lemma \ref{lem:cons-unchanged}. Therefore $\dom(\senv(\vstate')) \supseteq \dom(\senv(\vstate))$.

      Thus by lemma \ref{lem:preservation-heap-env-unchanged} $\pair{\heap}{\triple{\perms}{\env}{s} \cdot \stack^*}$ is a valid state.

    \case \refrule{ExecWhileFinish}: We have
    \begin{align*}
      &\pair{\heap}{\triple{\perms'}{\env'}{\kskip} \cdot \triple{\perms}{\env}{\sseq{\swhile{e}{\gform}{s}}{s'}} \cdot \stack^*} \to \\
      &\quad \pair{\heap}{\triple{\perms''}{\env'}{\sseq{\swhile{e}{\gform}{s}}{s'}} \cdot \stack^*} \\
      &\text{where}\quad \assertion{\heap}{\perms'}{\env'}{\gform} \quad\text{and}\quad \perms'' = \perms \cup \foot{\heap}{\perms'}{\env'}{\gform}.
    \end{align*}

    By assumptions, the initial state is validated by some $\vstate$ and valuation $V'$, thus $\vstate' = \triple{\sstate'}{\kskip}{\gform'}$ for some $\sstate'$, $\gform$ where $\simstate{V}{\sstate'}{\heap}{\perms'}{\env'}$.

    Also, by \ref{def:state-valid-partial}, the partial state $\pair{\heap}{\triple{\perms}{\env}{\sseq{\swhile{e}{\gform}{s}}{s'}} \cdot \stack^*}$ is validated by $\vstate'$ and $V'$. Thus \ref{def:partial-valid-while} must apply, and thus there is some $\vstate_0$, $V_0$, $\sstate_0$, $\sstate_0'$, $\gform_0$ such that
    \begin{gather*}
      \vstate_0 = \triple{\sstate_0}{\sseq{\swhile{e}{\gform}{s}}{s'}}{\gform_0} \\
      \text{$\vstate_0$ is reachable from $\prog$ with valuation $V_0$} \\
      \scons{\sstate_0}{\gform}{\sstate_0'}{\_}, \quad \simstate{V_0'}{\sstate_0'}{\heap}{\perms}{\env}, \quad \gform' = \gform, \\
      \text{and}\quad V_0' = V_0[\scons{\sstate}{\gform}{\sstate'}{\_} \mid \heap].
    \end{gather*}

    We have $\assertion{\heap}{\perms'}{\env'}{\gform}$, thus by lemma \ref{lem:foot-assert}$\assertion{\heap}{\foot{\heap}{\perms'}{\env'}{\gform}}{\env'}{\gform}$, and then by lemma \ref{lem:assert-monotonicity} $\assertion{\heap}{\perms''}{\env'}{\gform}$.

    Also by lemma \ref{lem:simstate-monotonicity} $\simstate{V_0'}{\sstate_0'}{\heap}{\perms''}{\env}$.

    For some $k$, let $x_1, \cdots, x_k$ be the list of variables in $\fmodified(s)$. Let $t_1 = \ffresh, \cdots, t_k = \ffresh$, let $\hat{\senv} = \senv(\sstate_0)[x_1 \mapsto t_1, \cdots, x_k \mapsto t_k]$, and let $\hat{V}_0' = V_0'[t_1 \mapsto \env'(x_1), \cdots, t_k \mapsto \env'(x_k)]$.

    $\env'$ is contained in the stack frame executing the loop body, which is $s$, thus for all $x \in \dom(\env')$, either $\env(x) = \env'(x)$ or $x \in \fmodified(m)$.

    Also, since $\simenv{V_0'}{\senv(\sstate_0')}{\env}$ and $\senv(\sstate_0') = \senv(\sstate_0)$ by lemma \ref{lem:cons-unchanged}, $\simenv{V_0}{\senv(\sstate_0)}{\env}$.

    Now, for any $x \in \dom(\hat{\senv})$, if $x = x_i$ for some $i$, then $\hat{V}_0'(\hat{\senv}(x)) = \hat{V}_0'(\hat{\senv}(x_i)) = \env'(x_i) = \env'(x)$. Otherwise, $x \notin \fmodified(s)$, thus $\env'(x) = \env(x)$, and $x \in \dom(\senv(\sstate_0))$. Thus $\hat{V}_0'(\hat{\senv}(x)) = V_0(\senv(\sstate_0)(x)) = \env(x) = \env'(x_i)$ since $\simenv{V_0}{\senv(\sstate_0)}{\env}$. Therefore $\simenv{\hat{V}_0'}{\hat{\senv}}{\env'}$. Thus $\simstate{\hat{V}_0'}{\sstate_0'[\senv = \senv']}{\heap}{\perms''}{\env'}$.
    
    Therefore by lemma \ref{lem:produce-progress} $\sproduce{\sstate_0'[\senv = \senv']}{\gform}{\sstate_0''}$ for some $\sstate_0''$ such that $V_0''(\pc(\sstate_0'')) = \ktrue$ where $V_0'' = \hat{V}_0'[\sproduce{\sstate_0'}{\gform}{\sstate_0''} \mid \heap]$.

    Let $\vstate_0' = \triple{\sstate_0''}{\sseq{\swhile{e}{\gform}{s}}{s'}}{\gform_0}$. We want to show that
    $$\Gamma' = \pair{\heap}{\triple{\perms''}{\env'}{\sseq{\swhile{e}{\gform}{s}}{s'}} \cdot \stack^*}$$
    is validated by $\vstate_0'$ and $V_0''$.

    \textit{Part \ref{def:state-valid-reachable}:}
    By \refrule{SVerifyLoop}, and since $\vstate_0$ is reachable from $\prog$, $\prog \vdash \vstate_0 \to \vstate_0'$. Therefore $\vstate_0'$ is reachable from $\prog$ with valuation $V_0''$.

    \textit{Part \ref{def:state-valid-correspond}:}
    Since $\assertion{\heap}{\perms'}{\env'}{\gform}$, by lemma \ref{lem:efoot-subset-spec} $\efoot{\heap}{\env'}{\gform} \subseteq \perms'$. Also, since the stack is well-formed, $\perms'$ and $\perms$ are disjoint, thus $\perms \setminus \efoot{\heap}{\env'}{\gform} = \perms$. Therefore $\simstate{V_0'}{\sstate_0'}{\heap}{\perms \setminus \efoot{\heap}{\env'}{\gform}}{\env}$, and now since $\perms \subseteq \perms''$, by lemma \ref{lem:simstate-monotonicity}, and since $\simenv{\hat{V}_0'}{\senv'}{\env'}$,
    $$\simstate{\hat{V}_0'}{\sstate_0'[\senv = \senv']}{\heap}{\perms'' \setminus \efoot{\heap}{\env'}{\gform}}{\env'}.$$

    Now, by lemma \ref{lem:produce-soundness},
    $$\simstate{V_0''}{\sstate_0''}{\heap}{\perms''}{\env'}.$$

    Also, $s(\vstate_0') = s(\Gamma')$ by construction. Therefore $\Gamma'$ corresponds to $\vstate_0'$.

    \textit{Part \ref{def:state-valid-partial}:} We need to show that the partial state $\pair{\heap}{\stack^*}$ is validated by $\vstate_0'$ and $V_0''$. We already have that $\pair{\heap}{\stack^*}$ is validated by $\vstate_0$ and $V_0$. Thus one of \ref{def:partial-valid-nil}, \ref{def:partial-valid-call}, or \ref{def:partial-valid-while} must apply.

    \textit{If \ref{def:partial-valid-nil} applies:} Then $\stack^* = \nilsym$, thus trivially the partial state $\pair{\heap}{\stack^*}$ is validated by $\vstate_0'$ and $V_0''$.

    \textit{If \ref{def:partial-valid-call} applies:}
    Then $\stack^* = \triple{\perms^*}{\env^*}{\sseq{y \kassign m(e_1, \cdots, e_k)}{s^*}} \cdot \stack^*_1$ for some $\perms^*$, $\env^*$, $k$, $y$, $m$, $e_1, \cdots, e_k$, $s^*$, $\stack^*_1$. Also, there exists some $\vstate^*$, $V^*$, $x_1, \cdots, x_k$, $t_1, \cdots, t_k$, $\sstate_0, \cdots, \sstate_k$, $\sstate'$ such that
    \begin{gather*}
      \text{The partial state $\pair{\heap}{\stack^*_1}$ is validated by $\vstate^*$ and $V^*$},\\
      \vstate^* \text{ is reachable from $\prog$ with valuation $V^*$}, \quad s(\vstate^*) = s(\stack^*), \\
      x_1, \cdots, x_k = \fparams(m),\\
      \sstate_0 = \sstate(\vstate^*), \quad \seval{\sstate_0}{e_1}{t_1}{\sstate_1}{\_}, \quad\cdots,\quad \seval{\sstate_{k-1}}{e_k}{t_k}{\sstate_k}{\_},\\
      \scons{\sstate_k}{\fpre(m)}{\sstate'}{\_}, \quad \simstate{V^*}{\sstate'[\senv = \senv(\sstate_0)]}{\heap}{\perms^*}{\env^*}, \\
      \universal{1 \le i \le k}{V_0(\senv(\vstate_0)(x_i)) = V^*(t_i)}, \quad\text{and} \\
      \gform(\vstate_0) = \fpost(m).
    \end{gather*}

    We want to show that the partial state $\pair{\heap}{\stack^*}$ is validated by $\vstate''$ and $V''$. Immediately from above,
    \begin{gather*}
      \text{The partial state $\pair{\heap}{\stack^*_1}$ is validated by $\vstate^*$ and $V^*$},\\
      \vstate^* \text{ is reachable from $\prog$ with valuation $V^*$}, \quad s(\vstate^*) = s(\stack^*), \\
      x_1, \cdots, x_k = \fparams(m),\\
      \sstate_0 = \sstate(\vstate^*), \quad \seval{\sstate_0}{e_1}{t_1}{\sstate_1}{\_}, \quad\cdots,\quad \seval{\sstate_{k-1}}{e_k}{t_k}{\sstate_k}{\_},\\
      \scons{\sstate_k}{\fpre(m)}{\sstate'}{\_}, \quad \simstate{V^*}{\sstate'[\senv = \senv(\sstate_0)]}{\heap}{\perms^*}{\env^*}.
    \end{gather*}
    Also, the frame $\triple{\perms}{\env}{\sseq{\swhile{e}{\gform}{s}}{s'}}$ must be executing the body of $m$, since it is in the stack immediately above the frame that contains $y \kassign m(e_1, \cdots, e_k)$. Therefore, since $x_1, \cdots, x_k$ are all parameters of $m$, $\fmodified(s)$ cannot contain any of $x_1, \cdots, x_k$, since we do not allow assignment to parameters in a well-formed program. Thus
    \begin{align*}
      \forall 1 \le i \le k : V_0''(\senv(\vstate_0')(x_i))
        &= V_0''(\senv(\sstate_0'')(x_i)) &\text{defn. $\vstate_0'$} \\
        &= V_0''(\hat{\senv}(x_i)) &\text{Lemma \ref{lem:produce-unchanged}} \\
        &= V_0''(\senv(\sstate_0')(x_i)) &x_i \notin \fmodified(s) \\
        &= V_0''(\senv(\sstate_0)(x_i)) &\text{Lemma \ref{lem:cons-unchanged}} \\
        &= V_0(\senv(\sstate_0)(x_i)) &V_0 \subseteq V_0'' \\
        &= V^*(t_i) &\text{prev. assump.}
    \end{align*}

    Finally, $\gform(\vstate_0') = \gform(\vstate_0)$ by definition, thus
    $$\gform(\vstate_0') = \gform(\vstate_0) = \fpost(m).$$

    Therefore the partial state $\pair{\heap}{\stack^*}$ is validated by $\vstate_0'$ and $V_0''$ in this case.

    \textit{If \ref{def:partial-valid-while} applies:}
    Then $\stack^* = \triple{\env^*}{\perms^*}{\sseq{\swhile{e^*}{\gform^*}{s^*}}{s^{*\prime}}} \cdot \stack_1^*$ for some $\env^*$, $\perms^*$, $e^*$, $\gform^*$, $s^*$, $s^{*\prime}$, $\stack_1^*$, and there exists some $\vstate^*$, $V^*$, and $\sstate'$ such that:
    \begin{gather*}
      \text{The partial state $\pair{\heap}{\stack_1^*}$ is validated by $\vstate^*$ and $V^*$} \\
      \vstate^* \text{ is reachable from $\prog$ with valuation $V^*$} \quad s(\vstate^*) = s(\stack^*), \\
      \scons{\sstate(\vstate^*)}{\gform^*}{\sstate'}{\_}, \quad
      \simstate{V^*}{\sstate'}{\heap}{\perms^*}{\env^*} \quad\text{and}\\
      \gform(\vstate^*) = \gform^*.
    \end{gather*}

    Now, by definition of $\vstate_0'$, $\gform(\vstate_0') = \gform(\vstate_0) = \gform^*$. Therefore, using the other assumptions given above, the partial state $\pair{\heap}{\stack^*}$ is validated by $\vstate_0'$ and $V_0''$ in this case.

    Therefore definition part \ref{def:state-valid-partial} is satisfied.

    Therefore all parts of definition \ref{def:state-valid} are satisfied. Thus $\Gamma'$ is validated by $\vstate_0'$ and $V_0''$, as we wanted to show.

    \case \refrule{ExecFold}: We have
      $$\dexec{\heap}{\triple{\perms}{\env}{\sseq{\sfold{p(\multiple{e})}}{s}} \cdot \stack^*}{\vfoot{V'}{\heap}{\sperms}}{\heap}{\triple{\perms}{\env}{s} \cdot \stack^*}$$

      By assumptions, the initial state is validated by some $\vstate$ and valuation $V'$, thus \\
      $\vstate = \triple{\sstate}{\sseq{\sfold{p(\multiple{e})}}{s}}{\gform}$ for some $\sstate$, $\gform$ where $\simstate{V}{\sstate}{\heap}{\perms}{\env}$.

      The only guard that applies is \refrule{SGuardFold}, thus we have
      \begin{gather*}
        \multiple{\seval{\sstate}{e}{t}{\sstate'}{\scheck}}, \quad \multiple{x} = \fpredparams(p), \\
        \scons{\sstate'[\senv = [x_i \mapsto t_i]]}{\fpred(p)}{\sstate''}{\scheck'}, \\
        \text{and by assumptions}\quad \rtassert{V'}{\heap}{\perms}{\multiple{\scheck} \cup \scheck'} \quad\text{and}\quad V'(\pc(\sstate'')) = \ktrue
      \end{gather*}
      where $V'$ is the valuation corresponding to this series of judgements, extending $V$ (see definition \ref{def:sguard-valuation}).

      Let $e_1, \cdots, e_n = \multiple{e}$, $t_1, \cdots, t_n = \multiple{t}$, and $\scheck_1, \cdots, \scheck_n = \multiple{\scheck}$. Let $\sstate_0$, then for some $\sstate_1, \cdots, \sstate_n$ we have
      $$\seval{\sstate_0}{e_1}{t_1}{\sstate_1}{\scheck_1}, \cdots, \seval{\sstate_{n-1}}{e_n}{t_n}{\sstate_n}{\scheck_n}$$
      where $\sstate_n = \sstate'$. By lemmas \ref{lem:eval-subpath} and \ref{lem:cons-subpath} $\pc(\sstate'') \implies \pc(\sstate_n) \implies \cdots \implies \pc(\sstate_1)$. Therefore $V'(\pc(\sstate'')) = V'(\pc(\sstate_n)) = \cdots = V'(\pc(\sstate_1)) = \ktrue$. Also, by lemma \ref{lem:scheck-monotonicity} we have $\rtassert{V'}{\heap}{\perms}{\scheck_i}$ for all $1 \le i \le n$. Thus, by lemma \ref{lem:seval-soundness}
      \begin{gather*}
        \eval{\heap}{\env}{e_1}{V'(t_1)}, \quad\cdots,\quad \eval{\heap}{\env}{e_n}{V'(t_n)} \\
        \text{and}\quad \simstate{V'}{\sstate_1}{\heap}{\perms}{\env}, \quad\cdots,\quad \simstate{V'}{\sstate_n}{\heap}{\perms}{\env}.
      \end{gather*}
      Therefore $\simstate{V'}{\sstate'}{\heap}{\perms}{\env}$.

      Let $\senv' = [\multiple{x \mapsto t}]$ and $\env' = [\multiple{x \mapsto V'(t)}]$. Then, by construction, $\simenv{V'}{\senv'}{\env'}$. Therefore $\simstate{V'}{\sstate'[\senv = \senv']}{\heap}{\perms}{\env}$.

      From above we have $\scons{\sstate'[\senv = \senv']}{\fpred(p)}{\sstate''}{\scheck'}$ and $\rtassert{\heap}{\perms}{\env}{\scheck''}$ by lemma \ref{lem:scheck-monotonicity}. Thus, by lemma \ref{lem:cons-soundness}
      $\simstate{V'}{\sstate''}{\heap}{\perms \setminus \efoot{\heap}{\env'}{\fpred(p)}}{\env'}$ and thus
      $$\simstate{V'}{\sstate''[\senv = \senv(\sstate)]}{\heap}{\perms \setminus \efoot{\heap}{\env'}{\fpred(p)}}{\env}.$$

      Let $\sheap' = \sheap(\sstate''); \pair{p}{\multiple{t}}$. Expanding definitions,
      $$\vfoot{V'}{\heap}{\pair{p}{\multiple{t}}} = \efoot{\heap}{\env'}{\fpred(p)}.$$
      Now $\simstate{V'}{\sstate''[\senv = \senv(\sstate)]}{\heap}{\perms \setminus \vfoot{V'}{\heap}{\pair{p}{\multiple{t}}}}{\env}$, thus
      $$\universal{h_1, h_2 \in \sheap(\sstate'')}{h_1 \ne h_2 \implies \vfoot{V'}{\heap}{h_1} \cap \vfoot{V'}{\heap}{h_2} = \emptyset}$$
      and by lemma \ref{lem:disjoint-sim-heap-subset},
      $$\universal{h \in \sheap(\sstate'')}{\vfoot{V'}{\heap}{h} \cap \vfoot{V'}{\heap}{\pair{p}{\multiple{t}}} = \emptyset}.$$
      From these we can deduce that
      $$\universal{h_1, h_2 \in \sheap'}{h_1 \ne h_2 \implies \vfoot{V'}{\heap}{h_1} \cap \vfoot{V'}{\heap}{h_2} = \emptyset}.$$

      Also, from lemma \ref{lem:cons-soundness}, $\assertion{\heap}{\perms}{[\multiple{x \mapsto V'(t)}]}{\fpred(p)}$ since $\env' = [\multiple{x \mapsto V'(t)}]$.

      Since $\simstate{V'}{\sstate''[\senv = \senv(\sstate)]}{\heap}{\perms \setminus \vfoot{V'}{\heap}{\pair{p}{\multiple{t}}}}{\env}$,
      $$\universal{\pair{p}{\multiple{t}} \in \sheap(\sstate'')}{\assertion{\heap}{\perms}{[\multiple{x \mapsto V'(t)}]}{\fpred(p)}}.$$
      From these we can deduce that
      $$\universal{\pair{p}{\multiple{t}} \in \sheap'}{\assertion{\heap}{\perms}{[\multiple{x \mapsto V'(t)}]}{\fpred(p)}}.$$

      Since field values are unchanged between $\sheap(\sstate'')$ and $\sheap'$,
      \begin{gather*}
        \universal{\triple{f}{t}{t'} \in \sheap'}{\pair{V(t)}{f} \in \perms} \quad\text{and} \\
        \universal{\triple{f}{t}{t'} \in \sheap'}{\heap(V(t), f) = V(t')}.
      \end{gather*}

      Therefore $\simheap{V'}{\sheap'}{\heap}{\perms}$, and thus
      $$\simstate{V'}{\sstate''[\senv = \senv(\sstate), \sheap = \sheap']}{\heap}{\perms}{\env}.$$

      Let $\vstate' = \triple{\sstate''[\senv = \senv(\sstate), \sheap = \sheap']}{s}{\gform}$. By \refrule{SExecFold} $\sexec{\sstate}{\sseq{\sfold{p(\multiple{e})}}{s}}{s}{\sstate(\vstate')}$ (after expanding definitions). Therefore $\strans{\prog}{\vstate}{\vstate'}$ by \refrule{SVerifyStep}. Therefore $\vstate'$ is reachable from $\prog$ with valuation $V'$.

      Also, as shown before, $\simstate{V'}{\sstate(\vstate')}{\heap}{\perms}{\env}$, and by definition $s(\vstate') = s$. Therefore  $\vstate'$ corresponds to $\pair{\heap}{\triple{\perms}{\env}{s} \cdot \stack^*}$.

      By definition $\senv(\vstate') = \senv(\sstate) = \senv(\vstate)$ and $\gform(\vstate') = \gform = \gform(\vstate)$. Therefore by lemma \ref{lem:preservation-heap-env-unchanged} $\pair{\heap}{\triple{\perms}{\env}{s} \cdot \stack^*}$ is a valid state.

    \case \refrule{ExecUnfold}: We have
      $$\dexec{\heap}{\triple{\perms}{\env}{\sseq{\sunfold{p(\multiple{e})}}{s}} \cdot \stack^*}{\vfoot{V'}{\heap}{\sperms}}{\heap}{\triple{\perms}{\env}{s} \cdot \stack^*}$$

      By assumptions, the initial state is validated by some $\vstate$ and valuation $V'$, thus \\
      $\vstate = \triple{\sstate}{\sseq{\sunfold{p(\multiple{e})}}{s}}{\gform}$ for some $\sstate$, $\gform$ where $\simstate{V}{\sstate}{\heap}{\perms}{\env}$.

      The only guard that applies is \refrule{SGuardUnfold}, thus we have
      The only guard that applies is \refrule{SGuardFold}, thus we have
      \begin{gather*}
        \multiple{\seval{\sstate}{e}{t}{\sstate'}{\scheck}}, \quad
        \scons{\sstate'}{p(\multiple{e})}{\sstate''}{\scheck'}, \\
        \text{and by assumptions}\quad \rtassert{V'}{\heap}{\perms}{\scheck' \cup \bigcup \multiple{\scheck}} \quad\text{and}\quad V'(\pc(\sstate'')) = \ktrue
      \end{gather*}

      Let $e_1, \cdots, e_n = \multiple{e}$, $t_1, \cdots, t_n = \multiple{t}$, and $\scheck_1, \cdots, \scheck_n = \multiple{\scheck}$. Let $\sstate_0$, then for some $\sstate_1, \cdots, \sstate_n$ we have
      $$\seval{\sstate_0}{e_1}{t_1}{\sstate_1}{\scheck_1}, \cdots, \seval{\sstate_{n-1}}{e_n}{t_n}{\sstate_n}{\scheck_n}$$
      where $\sstate_n = \sstate'$. By lemmas \ref{lem:eval-subpath} and \ref{lem:cons-subpath} $\pc(\sstate'') \implies \pc(\sstate_n) \implies \cdots \implies \pc(\sstate_1)$. Therefore $V'(\pc(\sstate'')) = V'(\pc(\sstate_n)) = \cdots = V'(\pc(\sstate_1)) = \ktrue$. Also, by lemma \ref{lem:scheck-monotonicity} we have $\rtassert{V'}{\heap}{\perms}{\scheck_i}$ for all $1 \le i \le n$. Thus, by lemma \ref{lem:seval-soundness}
      \begin{gather*}
        \eval{\heap}{\env}{e_1}{V'(t_1)}, \quad\cdots,\quad \eval{\heap}{\env}{e_n}{V'(t_n)} \\
        \text{and}\quad \simstate{V'}{\sstate_1}{\heap}{\perms}{\env}, \quad\cdots,\quad \simstate{V'}{\sstate_n}{\heap}{\perms}{\env}.
      \end{gather*}
      Therefore $\simstate{V'}{\sstate'}{\heap}{\perms}{\env}$.

      Thus by lemma \ref{lem:cons-soundness}
      $$\assertion{\heap}{\env}{\perms}{p(\multiple{e})} \quad\text{and}\quad \simstate{V'}{\sstate'}{\heap}{\perms \setminus \efoot{\heap}{\env}{p(\multiple{e})}}.$$

      Let $\multiple{x} = \fpredparams(p)$, $\senv' = [\multiple{x \mapsto t}]$, and $\env' = [\multiple{x \mapsto V'(t)}]$. Then, by construction, $\simenv{V'}{\senv'}{\env'}$. Therefore $\simstate{V'}{\sstate'[\senv = \senv']}{\heap}{\perms \setminus \efoot{\heap}{\env}{p(\multiple{e})}}{\env'}$.

      Now, by definition, $\efoot{\heap}{\env}{p(\multiple{e})} = \efoot{\heap}{\env'}{\fpred(p)} \cup \bigcup \multiple{\efoot{\heap}{\env}{e}}$.

      Therefore $\efoot{\heap}{\env'}{\fpred(p)} \subseteq \perms \setminus \efoot{\heap}{\env}{p(\multiple{e})}$, thus by lemma \ref{lem:simstate-monotonicity},
      $$\simstate{V'}{\sstate'[\senv = \senv']}{\heap}{\perms \setminus \efoot{\heap}{\env'}{\fpred(p)}}{\env'}$$
      and $\simstate{V'}{\sstate}{\sstate'[\senv = \senv']}{\heap}{\perms}{\env'}$.

      Since $\assertion{\heap}{\perms}{\env}{p(\multiple{e})}$, by \refrule{AssertPredicate} $\assertion{\heap}{\perms}{\env'}{\fpred(p)}$.

      Therefore by lemma \ref{lem:produce-progress}
      $$\sproduce{\sstate'[\senv = \senv']}{\fpred(p)}{\sstate''} \quad\text{and}\quad V''(\pc(\sstate'')) = \ktrue$$
      where $V'' = V'[\sproduce{\sstate'[\senv = \senv']}{\fpred(p)}{\sstate''} \mid \heap]$. Also, by lemma \ref{lem:produce-soundness} $\simstate{V''}{\sstate''}{\heap}{\perms}{\env'}$, and thus
      $$\simstate{V''}{\sstate''[\senv = \senv(\sstate)]}{\heap}{\perms}{\env}.$$

      Now, by \refrule{SExecUnfold}, $\sexec{\sstate}{\sseq{\sunfold{p(\multiple{e})}}{s}}{s}{\sstate''[\senv = \senv(\sstate)]}$.

      Let $\vstate' = \triple{\sstate''[\senv = \senv(\sstate)]}{s}{\gform}$. By \refrule{SVerifyStep} $\strans{\prog}{\vstate}{\vstate'}$. Therefore $\vstate'$ is reachable from $\prog$ with valuation $V''$.

      Also, as shown before, $\simstate{V''}{\sstate(\vstate')}{\heap}{\perms}{\env}$. Furthermore, $s(\vstate') = s$ by definition. Thus $\pair{\heap}{\triple{\perms}{\env}{s} \cdot \stack^*}$ corresponds to $\vstate'$ with valuation $V''$.

      By definition $\senv(\vstate') = \senv(\sstate) = \senv(\vstate)$ and $\gform(\vstate') = \gform = \gform(\vstate)$. Therefore, by lemma \ref{lem:preservation-heap-env-unchanged} $\pair{\heap}{\triple{\perms}{\env}{s} \cdot \stack^*}$ is a valid state.
  \end{enumcases}
\end{proof}

\begin{theorem}[Preservation]\label{thm:dtrans-preservation}
  Let $\Gamma$ be some dynamic state validated by the $\vstate$ and valuation $V$ for some program $\prog$. If $\sguard{\vstate}{\sstate'}{\scheck}{\sperms}$ with corresponding valuation $V$ extending $V'$, $V'(\pc(\sstate')) = \ktrue$, $\rtassert{V'}{\heap}{\perms(\Gamma)} {\scheck}$, and $\dtrans{\prog}{\Gamma}{\vfoot{V'}{\heap(\Gamma)}{\sperms}}{\Gamma'}$
  then $\Gamma'$ is a valid state.

  In other words, if the dynamic state satisfies the matching symbolic checks, and dynamic execution procedes, then the resulting state is valid.
\end{theorem}

\begin{proof}
  We proceed by cases on the judgement $\dtrans{\prog}{\Gamma}{\vfoot{V'}{\heap(\Gamma)}{\sperms}}{\Gamma'}$.

  \begin{enumcases}
    \case \refrule{ExecInit}: Then $\Gamma = \initsym$ and $\Gamma' = \pair{\emptyset}{\triple{\emptyset}{\emptyset}{s(\prog)} \cdot \nilsym}$. Since $\Gamma$ is validated by $\vstate$, then $\vstate = \initsym$.

    Let $\vstate' = \triple{\quintuple{\bot}{\emptyset}{\emptyset}{\emptyset}{\ktrue}}{s(\prog)}{\ktrue}$. Then by \refrule{SVerifyInit} $\strans{\prog}{\initsym}{\vstate'}$.

    Since $\Gamma'$ has a stack with a sole stack frame, and clearly $\vstate'$ is reachable from $\prog$, in order to show that $\Gamma'$ is validated by $\vstate'$ it suffices to show that $\Gamma'$ corresponds to $\vstate'$. In turn, since clearly $s(\Gamma') = s(\vstate')$, it suffices to show that $\simstate{V}{\sstate(\vstate')}{\heap(\Gamma')}{\perms(\Gamma')}{\env(\Gamma')}$. Since all the requisite sets or values are trivial, this is immediate from the definition.

    \case \refrule{ExecFinal}: Then $\Gamma' = \finalsym$ and $\Gamma = \pair{\heap}{\triple{\perms}{\env}{\kskip} \cdot \nilsym}$. Since $\Gamma$ is validated by $\vstate$, $s(\Gamma) = s(\vstate) = \kskip$. By lemma \ref{lem:cons-progress}, $\scons{\sstate(\vstate)}{\gform(\vstate)}{\sstate'}{\_}$ for some $\sstate'$. Thus $\strans{\prog}{\vstate}{\finalsym}$ by \refrule{SVerifyFinal}.

    \case \refrule{ExecStep}: Then $\Gamma = \pair{\heap}{\stack}$ and $\Gamma' = \pair{\heap'}{\stack'}$ for some $\heap, \stack, \heap', \stack'$ where \\
    $\dexec{\heap}{\stack}{\vfoot{V'}{\heap}{\sperms}}{\heap'}{\stack'}$. Therefore $\Gamma'$ is a valid state by lemma \ref{lem:dexec-preservation}.
  \end{enumcases}
\end{proof}

\end{document}